\documentclass[letterpaper,12pt,titlepage,oneside,final]{book}
 
\newcommand{\href}[1]{#1} % does nothing, but defines the command so the
\usepackage{ifthen}
\newboolean{PrintVersion}
\setboolean{PrintVersion}{false} 
\usepackage{amsmath,amssymb,amstext} % Lots of math symbols and environments
\usepackage[pdftex]{graphicx} % For including graphics N.B. pdftex graphics driver 

\usepackage{microtype}

\usepackage{color}
\usepackage[usenames,dvipsnames]{xcolor}
 
\newcommand{\abs}[1]{\left|#1\right|}
\newcommand{\norm}[1]{\left\|#1\right\|}
\def\Tr{\text{Tr}}

\newcommand{\bra}[1]{\left\langle{#1}\right\vert}
\newcommand{\ket}[1]{\left\vert{#1}\right\rangle}

\usepackage{amsthm}
\newtheorem{theorem}{Theorem}[section]
\newtheorem{lemma}[theorem]{Lemma}

\newtheorem{claim}[theorem]{Claim}
\newtheorem*{reclaim}{Claim}
\newtheorem{proposition}[theorem]{Proposition}

\newtheorem{definition}[theorem]{Definition}

\newcommand{\eqnref}[1]{\hyperref[#1]{{(\ref*{#1})}}}
\newcommand{\thmref}[1]{\hyperref[#1]{{Theorem~\ref*{#1}}}}
\newcommand{\lemref}[1]{\hyperref[#1]{{Lemma~\ref*{#1}}}}
\newcommand{\corref}[1]{\hyperref[#1]{{Corollary~\ref*{#1}}}}
\newcommand{\defref}[1]{\hyperref[#1]{{Definition~\ref*{#1}}}}
\newcommand{\secref}[1]{\hyperref[#1]{{Section~\ref*{#1}}}}
\newcommand{\chapref}[1]{\hyperref[#1]{{Chapter~\ref*{#1}}}}
\newcommand{\figref}[1]{\hyperref[#1]{{Figure~\ref*{#1}}}}
\newcommand{\tabref}[1]{\hyperref[#1]{{Table~\ref*{#1}}}}
\newcommand{\remref}[1]{\hyperref[#1]{{Remark~\ref*{#1}}}}
\newcommand{\appref}[1]{\hyperref[#1]{{Appendix~\ref*{#1}}}}
\newcommand{\claimref}[1]{\hyperref[#1]{{Claim~\ref*{#1}}}}
\newcommand{\propref}[1]{\hyperref[#1]{{Proposition~\ref*{#1}}}}
\newcommand{\exampleref}[1]{\hyperref[#1]{{Example~\ref*{#1}}}}
\newcommand{\conjref}[1]{\hyperref[#1]{{Conjecture~\ref*{#1}}}}

\usepackage{graphicx}
\usepackage{caption}
\usepackage{subcaption}
\usepackage{multirow}
\usepackage{tabularx}

\newcommand{\transpose}[1]{#1^{\intercal}}

\newcommand{\lket}[1]{\ket{\overline{#1}}}
\newcommand{\indicator}{\mathcal{I}}
\def\N {{\bf N}}

\def\malig{\text{mal}}

\def\chiIn{\chi_{\text{in}}}
\def\chiOut{\chi_{\text{out}}}
\def\zetaIn{\zeta_{\text{in}}}
\def\zetaOut{\zeta_{\text{out}}}

\newcommand{\eventFont}[1]{{\text{\sf{#1}}}}
\def\out{{\eventFont{out}}}
\def\incorrect{{\eventFont{incorrect}}}
\def\good{{\eventFont{good}}}

\def\bad{{\eventFont{bad}}}

\def\acceptX{\eventFont{accept}}

\def\goodZ{\good_Z}

\def\badZ{\bad_Z}
\def\goodX{\good_X}
\def\badX{\bad_X}
\def\accept{\eventFont{accept}}

\def\kGood{k_{\text{good}}}

\def\gammaTH{\gamma_{\text{th}}}
\def\Pevent{\mathcal{P}}

\def\kMax{k_\text{max}}

\def\threshOverlap{1.32 \times 10^{-3}}
\def\threshSteane{1.24 \times 10^{-3}}

\newcommand{\logical}[1]{{\overline{#1}}}
\newcommand{\CCZ}{\text{CCZ}}
\newcommand{\CNOT}{\text{CNOT}}
\newcommand{\Clifford}{\text{Clifford}}
\newcommand{\Id}{I}
\newcommand{\KMMring}{\mathbb{Z}[i,\frac{1}{\sqrt 2}]}
\newcommand{\Aring}{\mathbb{Z}[i,\sqrt 2]}

\definecolor{BlueGreen}{RGB}{0,75,75}
\usepackage[pdftex,letterpaper=true,pagebackref=false]{hyperref} % with basic options
\hypersetup{
    plainpages=false,       % needed if Roman numbers in frontpages
    unicode=false,          % non-Latin characters in Acrobat’s bookmarks
    pdftoolbar=true,        % show Acrobat’s toolbar?
    pdfmenubar=true,        % show Acrobat’s menu?
    pdffitwindow=false,     % window fit to page when opened
    pdfstartview={FitH},    % fits the width of the page to the window
    pdftitle={},    % title: CHANGE THIS TEXT!
    pdfauthor={Adam Paetznick},    % author: CHANGE THIS TEXT! and uncomment this line
    pdfsubject={quantum computing},  % subject: CHANGE THIS TEXT! and uncomment this line
    pdfnewwindow=true,      % links in new window
    colorlinks=true,        % false: boxed links; true: colored links
    linkcolor=Black,         % color of internal links
    citecolor=BlueGreen,        % color of links to bibliography
    filecolor=BlueGreen,      % color of file links
    urlcolor=BlueGreen           % color of external links
}
\ifthenelse{\boolean{PrintVersion}}{   % for improved print quality, change some hyperref options
\hypersetup{	% override some previously defined hyperref options
    citecolor=black,%
    filecolor=black,%
    linkcolor=black,%
    urlcolor=black}
}{} % end of ifthenelse (no else)

\let\origdoublepage\cleardoublepage
\newcommand{\clearemptydoublepage}{%
  \clearpage{\pagestyle{empty}\origdoublepage}}
\let\cleardoublepage\clearemptydoublepage

\begin{document}

% For a large document, it is a good idea to divide your thesis
% into several files, each one containing one chapter.
% To illustrate this idea, the "front pages" (i.e., title page,
% declaration, borrowers' page, abstract, acknowledgements,
% dedication, table of contents, list of tables, list of figures,
% nomenclature) are contained within the file "uw-ethesis-frontpgs.tex" which is
% included into the document by the following statement.
%----------------------------------------------------------------------
% FRONT MATERIAL
%----------------------------------------------------------------------
% T I T L E   P A G E
% -------------------
% Last updated May 24, 2011, by Stephen Carr, IST-Client Services
% The title page is counted as page `i' but we need to suppress the
% page number.  We also don't want any headers or footers.
\pagestyle{empty}
\pagenumbering{roman}

% The contents of the title page are specified in the "titlepage"
% environment.
\begin{titlepage}
        \begin{center}
        \vspace*{1.0cm}

        \Huge
        {\bf Resource optimization for fault-tolerant quantum computing }

        \vspace*{1.0cm}

        \normalsize
        by \\

        \vspace*{1.0cm}

        \Large
        Adam Paetznick \\

        \vspace*{3.0cm}

        \normalsize
        A thesis \\
        presented to the University of Waterloo \\ 
        in fulfillment of the \\
        thesis requirement for the degree of \\
        Doctor of Philosophy \\
        in \\
        Computer Science \\

        \vspace*{2.0cm}

        Waterloo, Ontario, Canada, 2013 \\

        \vspace*{1.0cm}

        \end{center}
\end{titlepage}

\noindent\textbf{Copyright notice.}
\chapref{chap:transversal} contains material from~\cite{Paetznick2013a}, which is copyrighted by the American Physical Society.
Chapters~\ref{chap:ancilla} and~\ref{chap:threshold} contain material from~\cite{Paetznick2011} which is copyrighted by Rinton Press.
\\ \\
Remaining material is:
\copyright\ Adam Paetznick 2013 \\

% The rest of the front pages should contain no headers and be numbered using Roman numerals starting with `ii'
\pagestyle{plain}
\setcounter{page}{2}

\cleardoublepage % Ends the current page and causes all figures and tables that have so far appeared in the input to be printed.
% In a two-sided printing style, it also makes the next page a right-hand (odd-numbered) page, producing a blank page if necessary.

% D E C L A R A T I O N   P A G E
% -------------------------------
  % The following is the sample Delaration Page as provided by the GSO
  % December 13th, 2006.  It is designed for an electronic thesis.
  \noindent
I hereby declare that I am the sole author of this thesis. 
This is a true copy of the thesis, including any required final revisions, as accepted by my examiners.

  \bigskip
  
  \noindent
I understand that my thesis may be made electronically available to the public.

\cleardoublepage
%\newpage

% A B S T R A C T
% ---------------

\begin{center}\textbf{Abstract}\end{center}

Quantum	computing offers the potential for efficiently solving otherwise classically difficult problems, with applications in material and drug design, cryptography, theoretical physics, number theory and more.  However, quantum systems are notoriously fragile; interaction with the surrounding environment and lack of precise control constitute noise,
which makes construction of a reliable quantum computer extremely challenging.
Threshold theorems show that by adding enough redundancy, reliable and arbitrarily long quantum computation is possible so long as the amount of noise is relatively low---below a ``threshold'' value.
The amount of redundancy required is reasonable in the asymptotic sense, but in absolute terms the resource overhead of existing protocols is enormous when compared to current experimental capabilities. 
%Overhead for fault-tolerant quantum computation comes from a variety of sources including encoded gates, error correction, and physical limitations on connectivity and the allowed set of quantum operations.

In this thesis we examine a variety of techniques for reducing the resources required for fault-tolerant quantum computation.  First, we show how to simplify universal encoded computation by using only transversal gates and standard error correction procedures, circumventing existing no-go theorems.  The cost of certain error correction procedures is dominated by preparation of special ancillary states. We show how to simplify ancilla preparation, reducing the cost of error correction by more than a factor of four.  Using this optimized ancilla preparation, we then develop improved techniques for proving rigorous lower bounds on the noise threshold.  The techniques are specifically intended for analysis of relatively large codes such as the $23$-qubit Golay code, for which we compute a lower bound on the threshold error rate of $0.132$ percent per gate for depolarizing noise. This bound is the best known for any scheme.

Additional overhead can be incurred because quantum algorithms must be translated into sequences of gates that are actually available in the quantum computer.  In particular, arbitrary single-qubit rotations must be decomposed into a discrete set of fault-tolerant gates.  We find that by using a special class of non-deterministic circuits, the cost of decomposition can be reduced by as much as a factor of four over state-of-the-art techniques, which typically use deterministic circuits.

Finally, we examine global optimization of fault-tolerant quantum circuits.   Physical connectivity constraints require that qubits are moved close together before they can interact, but such movement can cause data to lay idle, wasting time and space.  We adapt techniques from VLSI in order to minimize time and space usage for computations in the surface code, and we develop a software prototype to demonstrate the potential savings.    

\cleardoublepage
%\newpage

% A C K N O W L E D G E M E N T S
% -------------------------------

\begin{center}\textbf{Acknowledgements}\end{center}
I must begin by thanking my supervisor, Ben Reichardt, for his support over the past four years.  Ben is responsible for teaching me much of what I know about fault-tolerant quantum computation.  In addition he has served as a tremendous guide in terms of academic writing and speaking, and navigation of the academic world in general.  Much of my writing and speaking style is due to Ben's advice.

I would also like to thank Richard Cleve for supporting me throughout, but especially for support in the past two years during which Ben has been at USC. Special thanks also to Michele Mosca for bringing me into the quantum circuits and Torque group.  I am grateful to other members of the Torque team for their enthusiasm and support including especially Martin Roetteler and Rich Lazarus.  The surface code was largely a mystery until it was marvelously explained to me by Austin Fowler.  His tenacity for finding practical solutions to important problems has inspired me to try to do the same.

Some of my most valuable discussions and collaborations occurred during internships away from Waterloo.  I would like to thank all of those in the quantum computing group at HRL, and Jim Harrington and Bryan Fong in particular, for their hospitality and support.  Thanks also to the QuArC group at Microsoft Research including: Krysta Svore, Alex Bocharov, Dave Wecker and Nathan Wiebe.

Much of my financial support has come from the Mike and Ophelia Lazaridis fellowship, for which I am very grateful.

Of course, I must also acknowledge the support of my peers, at Waterloo and elsewhere, for their friendship and for helpful suggestions and conversations.  This includes: Vadym Kliuchnikov, Cody Jones, Peter Brooks, Robin Kothari, Alessandro Cosentino, Matt Amy, Vinayak Pathak, Lucy Zhang, David Gosset, Rajat Mittal, Ansis Rosmanis, Stacy Jeffery, Moritz Ernst, Tomas Jochym-O'Connor, Jaimie Sikora, Sevag Gharibian, Sarvagya Upadhyay, Laura Mancinska, Abel Molina and Shelby Kimmel.  To my many other friends including Troy Borneman, Chad Daley, Mike Wesolowski, Kurt Schreiter, Mike Zhang, Daniel Park, Chris Wood, Holger Haas, Shane Farnsworth, Andrew Achkar and Halle Revell, thank you for making the experience in Waterloo an enjoyable one for me and my wife Marion.

Finally, my personal and academic successes are due largely to the influence, love and support of my parents, Duane and Phyllis, and my brother Brandon and his wife Heather.  Thank you for your unwavering encouragement, especially during the first few years in Waterloo, which were difficult for both Marion and myself.  I am similarly grateful to Marion's parents, Bill and Sue, and my sister-in-law Gwen and her husband Rich.

\cleardoublepage
%\newpage

% D E D I C A T I O N
% -------------------

\begin{center}\textbf{Dedication}\end{center}

\noindent
To my loving wife Marion,
\\*\\*
I would not have even considered this pursuit had it not been for your enthusiastic encouragement and support.  Conferences and internships have kept us apart for long stretches, and while I have been traveling all over the world, you have been working the extra jobs to keep us afloat.
When we are together, you fill me with life and laughter.
This thesis is as much a product of your time, effort and love as it is of mine.
I love you.

\cleardoublepage
%\newpage

% T A B L E   O F   C O N T E N T S
% ---------------------------------
\renewcommand\contentsname{Table of Contents}
\tableofcontents
\cleardoublepage
\phantomsection
%\newpage

% L I S T   O F   T A B L E S
% ---------------------------
\addcontentsline{toc}{chapter}{List of Tables}
\listoftables
\cleardoublepage
\phantomsection		% allows hyperref to link to the correct page
%\newpage

% L I S T   O F   F I G U R E S
% -----------------------------
\addcontentsline{toc}{chapter}{List of Figures}
\listoffigures
\cleardoublepage
\phantomsection		% allows hyperref to link to the correct page
%\newpage

% L I S T   O F   S Y M B O L S
% -----------------------------
% To include a Nomenclature section
% \addcontentsline{toc}{chapter}{\textbf{Nomenclature}}
% \renewcommand{\nomname}{Nomenclature}
% \printglossary
% \cleardoublepage
% \phantomsection % allows hyperref to link to the correct page
% \newpage

% Change page numbering back to Arabic numerals
\pagenumbering{arabic}

%----------------------------------------------------------------------
% MAIN BODY
%----------------------------------------------------------------------
\chapter{Motivation and results}
\label{chap:motivation}

The discovery of quantum mechanics in the early $1900$s represented a fundamental departure from previous understanding of the natural world.
In a similar way, quantum computers, conceived by Feynman in $1982$, represent a fundamental shift from the traditional way of solving computational problems~\cite{Feyn82}.  Feynman observed that simulation of quantum mechanics, though an apparently difficult task for (classical) computers, is accomplished tautologically by natural physical systems.  Consequently, a computing device operating according to the laws of quantum mechanics could have a distinct advantage over its classical counterparts.

Indeed, simulation of quantum mechanical systems is of enormous practical importance, with potential applications in drug design, materials science, protein folding and more~(see, e.g., \cite{Kassal2010}). Feynman's original ideas have since been refined and show that exponential speedups for simulation of quantum mechanical systems are indeed possible, in theory~\cite{Abrams1997,Boghosian1998,Zalk98b}.

Exponential improvements are not limited to simulation, though. In $1994$, Shor developed a polynomial-time algorithm for factoring large numbers, a problem which is widely believed to be intractable for classical computers~\cite{Shor1994}.
Other exponential speedups exist including algorithms for solving linear systems of equations~\cite{Harr09}, and other mathematical problems~\cite{Kedlaya2004,Janzing2006,Hallgren2007,Alagic2010}. Finding new algorithms is a subject of active research~\cite{Mosca2008,Childs2010}.

To date, however, quantum computers capable of outperforming classical devices do not exist.  The limited number of experimental efforts that have been attempted, while encouraging, fall well short of the scale necessary for real-world applications~\cite{Ladd2010}.  Some modern technologies such as transistors and optical drives do exploit aspects of quantum mechanics; recently, quantum mechanics been used to develop highly secure communication devices~\cite{Stucki2011}.  But none of these devices are sophisticated enough to execute quantum algorithms.

Executing large-scale algorithms on a quantum computer is a daunting task.  Quantum algorithms rely on the ability to create and maintain highly entangled quantum states.  Interaction with the environment quickly causes \emph{decoherence}, which destroys entanglement.  Decoherence can be delayed by carefully isolating the quantum information from its environment.  However, too much isolation also prevents (wanted) access to the quantum system, making control and readout difficult.
At the same time, coherently controlling a large quantum mechanical system for the duration of an algorithm requires extreme accuracy. Such stringent control requirements, combined with the inherent fragility of quantum information, raise concerns about the feasibility of constructing a quantum computer.

Is accurate large-scale quantum computation possible?
It turns out that, by incorporating enough redundancy, quantum computation with arbitrary accuracy is possible, at least in principle~\cite{Aharonov1996a}.  In practice, the engineering challenges are significant and the necessary amount of redundancy can be overwhelmingly large.  In this thesis, we will discuss the challenges and propose a variety of methods for reducing resource requirements.
 
\section{The role of noise in a quantum computer}
\label{sec:motivation.noise}
Errors in a quantum computer originate from two sources.  First, control of the quantum system may be imperfect. For example, operations in a quantum computer can be described by rotations about a set of fixed axes.  Over time, small over- or under-rotations can accumulate, resulting in data corruption.  Second, the surrounding environment may interact undesirably with the system.  For example, data stored in an electron can be altered by interaction with surrounding magnetic fields.
Collectively, imperfect control and environmental interactions represent \emph{noise} in a quantum computer.

Noise is not exclusive to quantum systems.
Classical devices can also suffer from errors due to imperfections, or external physical phenomena.  However, most electronics can be manufactured so that errors are vanishingly rare.  When this is not possible, errors can be suppressed by adding redundancy.  Error-correcting codes use a large number of physical bits in order to represent some smaller number of ``logical'' bits~\cite{MacWilliams1988}.  As long as the number of physical bit errors is small enough, the information inside of the code can be retrieved accurately.

Indeed, a very simple kind of error protection is used in dynamic random-access-memory (DRAM), which is ubiquitous in modern electronics.  Each bit in DRAM is stored in a small capacitor as an electric charge, which may leak away over time.  To avoid data loss, each charge is periodically ``refreshed'' by reading it and then rewriting it.
Unfortunately, directly refreshing quantum bits is not possible.  Merely reading a quantum bit, or \emph{qubit}, has the effect of changing its state.  %In fact, \emph{any} disturbance can dramatically affect the state of a qubit.  Thus, the quantum system must be extremely well isolated.  At the same time, the system must be readily accessible to the computer to allow for highly accurate manipulation.  It is these two competing requirements that make noise such a difficult problem.

One might hope that quantum hardware could be manufactured to reduce noise to acceptable levels.  However, most quantum algorithms will require billions of operations and many hundreds or thousands of qubits.  Controlling such a large number of qubits, each with an error rate below one part in a billion is far beyond the capability of current technology, and is likely to remain so for the foreseeable future.

The inability to refresh is due, in part, to the fact that quantum information cannot be cloned~\cite{Woot82}.
One might expect that the use of error-correcting codes for quantum information is therefore also prohibited.  Nevertheless, quantum information can be protected by combining classical error-correcting codes in a novel way~\cite{Shor1997}.
Indeed, so long as the probability of an error is below a constant \emph{threshold} value, it is possible to use error-correcting codes to protect quantum information during arbitrarily long computations~\cite{Aharonov1996a}. 

Error correction is not the only technique available for protecting quantum information. Decoherence-free subspaces and dynamical decoupling are capable of improving the fidelity of quantum operations~\cite{Palma1996,Duan1997,Viola1999,Ban1998}.  However, these methods have limitations and are generally regarded as complementary to active error correction, which is where we will focus our attention.

\section{Requirements imposed by error-correcting codes
\label{sec:motivation.overhead}
}
Quantum error-correcting codes permit high-quality protection of quantum information from noise, but is it enough?  At a minimum, the amount of noise that can be tolerated by error correction must meet or exceed the amount of noise in the physical system.  Threshold theorems tell us that arbitrary accuracy \emph{is} possible even if error rates are constant, but small enough~\cite{Aharonov1996a,Kitaev1997a,Knill1996a,Reichardt2005,Terhal2005,Aliferis2005,Aharonov2006,Ng2009,Preskill2012a}.  What is the noise threshold for quantum computing, and can it be physically achieved?

Initial estimates of the threshold error rate were around $0.01$ percent per gate~\cite{Zalk96}, but have been subsequently improved to as high as one- to three-percent per gate~\cite{Knill2004,Raussendorf2007,Wang2011}.   This range of error rates meets or approaches gate fidelities reported by a variety of experimental efforts for small-scale systems~\cite{Ladd2010,Monz11,Chow2012,Ghosh2013}.
It seems, therefore, that quantum error-correcting codes have the capacity to protect quantum information in realistic conditions.  

But there is a second, potentially more alarming concern. In principle, quantum computation with error correction is efficient. If the size of the ideal circuit is $n$ then the corresponding fault-tolerant circuit need only be a factor of poly(log $n$) larger. However the constants involved can be quite large, and numerical studies have shown that the resource requirements can be astoundingly large in absolute terms.
A single encoded quantum gate can require millions or billions of physical gates~\cite{Knill2004,Raussendorf2007a,Paetznick2011,Jones2012b}.
In addition, many proposed quantum computing architectures impose limitations on the placement of and interactions between qubits. Imposing geometric constraints of this kind only increases the overhead costs. 

The necessary resources depend on the algorithm, desired level of accuracy, clock speed, noise properties of the hardware and so on.  Regardless of other factors, though, estimates are often dominated by the resources required for error correction. For example, under an error-correction scheme proposed by Knill, a quantum algorithm consisting of ten billion operations would require a resource overhead factor of about one million when the error rate per (physical) gate is $10^{-3}$~\cite{Knill2004}.  That is, if the size of the original algorithm is $n$, then the size of the quantum computer would need to be roughly $10^6 n$.
For other size and error parameters, the gate and qubit overhead can range from one-thousand to one-billion fold, or more.

These kinds of resource requirements place a huge burden on the construction of a quantum computer.  Even if billions of qubits can be coherently controlled, such large overhead is clearly undesirable.
The fear is that the overhead required to protect quantum information is so large as to make quantum computers wholly impractical, or to effectively negate any algorithmic speedups over classical computers.
The most important goal of the quantum circuit designer, therefore, is to reduce resource requirements to manageable levels.

\section{Summary of new results
\label{sec:motivation.results}
}

% This proposal considers the following question:
% \begin{center}
% \emph{How can we minimize the resources required to accurately implement a quantum algorithm?}
% \end{center}

Resource overhead in fault-tolerant circuits is incurred in a variety of ways, including large error-correction circuits, large gate costs, low encoding rates and more.  Thus one should consider a variety of optimization strategies in order to address each problem.
Accordingly, this thesis proposes a number of new techniques for reducing the resources required to accurately implement quantum algorithms, subject to realistic constraints imposed by quantum computing hardware.

\subsection*{Universality with transversal gates}
Fault-tolerant computation involves performing operations on the data while it is encoded.  For most quantum error-correcting codes, there is a small set of operations that can be performed easily, and another set of operations that are much more difficult but are required in order to implement quantum algorithms.
The Toffoli gate, for example, is used heavily in classical subroutines but usually involves costly decomposition into a sequence of other gates.
In~\chapref{chap:transversal} we show that a particular family of quantum codes admits a simple ``transversal'' implementation of the controlled-controlled-$Z$ gate.  A relatively cheap implementation of Toffoli can then be obtained with the help of encoded Hadamards, which we show can also be implemented transversally.  Toffoli and Hadamard are universal for quantum computation~\cite{Shi02}, and so only these simple transversal gates are necessary.

\subsection*{Smaller error correction circuits}
Error correction dominates the resource costs of many fault-tolerance schemes. Reducing the cost of error-correction therefore reduces the total cost by nearly the same amount.  \chapref{chap:ancilla} examines methods for efficiently preparing so-called ``stabilizer states'', which comprise the bulk of the cost for several types of error correction.  These methods can be applied to a large class of quantum error-correcting codes, and are particularly effective for codes of medium to large size. For example, the cost of error-correction for the $23$-qubit Golay code can be reduced by more than a factor of four when compared to previous methods.

\subsection*{Improved noise thresholds}
Computational accuracy increases rapidly as the physical noise rate drops below the threshold.  Thus, an effective way to reduce resource requirements is to \emph{increase} the noise threshold by improving lower bounds.  \chapref{chap:threshold} describes a technique for more accurately calculating lower bounds on the noise threshold when noise is modeled as a Pauli channel.  We calculate a threshold error rate of $0.132$ percent per gate for depolarizing noise, the best lower bound currently known.
Our proof uses malignant set counting~\cite{Aliferis2005}, extensively tailored for our optimized error-correction circuits and for Pauli channel noise.  Instead of assuming adversarial (i.e., worst-case) noise at higher levels of code concatenation, the counting procedure keeps track of multiple types of malignant events to create a transformed independent noise model for each level, allowing for a more accurate analysis.

\subsection*{Low-cost approximations of single-qubit unitaries}
Fault-tolerance schemes offer a universal but finite set of gates from which to implement quantum algorithms.  An arbitrary unitary requested by an algorithm must be approximated by decomposition into a sequence of fault-tolerant gates. 
Traditional approximation methods output a deterministic sequence of gates~\cite{Dawson2005,Fowl04c,Selinger2012a,Kliuchnikov2012b}. In~\chapref{chap:repeat} we explore the use of non-deterministic but repeatable quantum circuits.  By optimized direct computer search, we find a large number of such circuits and show how to use them to reduce the cost of approximating a single-qubit unitary by about a factor of three.

\subsection*{Circuit optimization subject to geometric constraints}
Resource calculations often ignore geometric connectivity constraints imposed by a quantum computer.  Fault-tolerant quantum circuits encoded in the \emph{surface code} automatically respect two-dimensional nearest-neighbor constraints but do not consider global dimensions of the computer, wasting both space and time.  To solve this problem, \chapref{chap:braidpack} proposes two algorithms for placing fault-tolerant quantum circuits onto a two-dimensional qubit lattice of fixed, but arbitrary size.  The algorithms exploit topological properties of the surface code in order to transform the initial circuit into one that fits compactly into the lattice geometry.
\chapter{The mechanics of a quantum computer
\label{chap:mechanics}
}

Classical computers operate based on the laws of electricity and magnetism.  However, the physical details are usually abstracted and, instead, operations are described in terms of bits and logic gates.  Similarly, though quantum computers operate based on the laws of quantum mechanics, we will use abstractions such as qubits and quantum gates.  In this chapter, we summarize the mathematics of quantum computation.  This summary introduces only the concepts that are necessary for quantum error correction and fault tolerance. For a more complete treatment, the reader is referred to any of several textbooks~\cite{Nielsen2000,Kita02,Kaye2007}.

\section{States
\label{mechanics.states}
}
The content, or \emph{state}, of a classical computer is described by bits.  A bit is value either zero or one, or alternatively, a bit is a vector
\begin{equation}
\vec{v} = a\vec{0} + b\vec{1}
\enspace,
\end{equation}
where $a,b\in \{0,1\}$ and such that $a + b = 1$.  A string, or \emph{register}, of $n$ bits is then a length $n$ vector over the field $\mathbb{Z}_2 = \{0,1\}$, i.e., an ordered collection of bits.

The state of a quantum computer is described by \emph{qubits}.  Like a bit, a qubit is a vector
\begin{equation}
\ket\psi = a\ket{0} + b\ket{1}
\enspace,
\end{equation}
except that the ``amplitudes'' $a, b \in \mathbb{C}$ are now free to take complex values and must satisfy the normalization condition $\abs{a}^2 + \abs{b}^2 = 1$. The notation $\ket{\cdot}$, is called a ``ket'' and is conventional for quantum states. Measurement of a qubit yields a bit, the value of which is determined by a probability distribution defined by $a$ and $b$.  The normalization condition ensures that the total probability is equal to one.  See~\secref{sec:mechanics.measurement}.  

A register of qubits is a unit vector in a $2^n$-dimensional vector space over the complex field $\mathbb{C}$.  However, unlike a classical register, a register of $n$ qubits has length $2^n$, one entry for each of the possible bit strings of length $n$.  This is akin to a probabilistic classical register which may take one of $2^n$ possible values according to a probability distribution.  In this way, a qubit register is a generalization of a probabilistic register in which the coefficients are complex and could be negative, for example.  The normalization condition for a register $\sum_i a_i \ket{x_i}$ is $\sum_i \abs{a_i}^2 = 1$.

Any two $n$-qubit registers $\ket{\psi} = \sum_i a_i\ket{i}$ and $\ket{\phi} = \sum_i b_i\ket{i}$ obey the inner product
\begin{equation}
\langle \ket{\psi}, \ket{\phi} \rangle = \langle\phi|\psi\rangle = \sum_i a_ib^*_i
\enspace.
\end{equation}
The normalization condition enforces that a quantum register has inner product one with itself, i.e., $\langle{\psi}|{\psi}\rangle = 1$.

Registers of qubits can be joined together by tensor product.  For example, the tensor product of the state $\ket\psi$ and $\ket\phi$ defined above is given by
\begin{equation}
\ket\psi \otimes \ket\phi = \sum_{i,j} a_i b_j \ket{i}\otimes\ket{j}
\enspace .
\end{equation} 
Often the $\otimes$ notation is dropped, instead using the shorthand $\ket\psi\ket\phi$, or sometimes $\ket{\psi, \phi}$.  The tensor product of $k$ identical registers $\ket\psi$ is denoted by $\ket{\psi}^{\otimes k}$, or sometimes $\ket{\psi^k}$.

\section{Operations}

Computers map input states to output states through a series of operations called \emph{gates}.  A classical gate takes some number of bit registers as input, and outputs one or more bit registers as output.  A quantum gate is similar, but manipulates registers of qubits.

A quantum gate operating on $n$ qubits can be described by a $2^n \times 2^n$ \emph{unitary} matrix.  A matrix $U$ is unitary if and only if $UU^\dagger = I$, where $U^\dagger$ is the matrix obtained by transposing $U$ and then taking the entry-wise complex conjugate, and $I$ is the identity matrix of appropriate dimension.  Unitary operations are reversible.  That is, the inputs of a quantum gate $U$ can be obtained from the outputs by performing the gate $U^\dagger$. 

Like registers, quantum gates can be joined by tensor product.  Again, the $\otimes$ notation is sometimes dropped for visual clarity.  This can create an ambiguity between matrix multiplication $UV$ and the tensor product $U\otimes V$.  When the intended product cannot be inferred from the context we will use $\otimes$ explicitly.

\subsection{Pauli operators
\label{sec:mechanics.operations.paulis}
}
One particularly important class of unitary gates is the single-qubit Pauli operators.  There are four such Pauli operators:
\begin{subequations}
\label{eq:single-qubit-paulis}
\begin{align}
I &= \begin{pmatrix} 1&0 \\ 0&1 \end{pmatrix} \enspace ,\\
X &= \begin{pmatrix} 0&1 \\ 1&0 \end{pmatrix} \enspace ,\\
Y &= \begin{pmatrix} 0&-i \\ i&0 \end{pmatrix} \enspace ,\\
Z &= \begin{pmatrix} 1&0 \\ 0&-1 \end{pmatrix} \enspace .
\end{align}
\end{subequations}
The square of any Pauli is equal to the identity $I$, and except for $I$, the Paulis pairwise anticommute.  That is, $PQ = -QP$ for $P,Q \in \{X,Y,Z\}$ and $P\neq Q$.

The Paulis are orthogonal under the Hilbert-Schmidt matrix inner product
\begin{equation}
\langle U, V \rangle := \Tr(U^\dagger V)
\enspace .
\end{equation}
Accordingly, they form an orthogonal basis for the set of $2\times 2$ complex matrices.  Any $2\times 2$ unitary $U$ can be written as a linear combination
\begin{equation}
U = \cos(\theta)I - i\sin(\theta)(aX + bY + cZ)
\enspace,
\end{equation}
for $\theta\in [0,\pi]$ and nonnegative real values $a,b,c$ such that $\sqrt{a^2 + b^2 + c^2} = 1$

The set of tensor products of Pauli operators forms a group under multiplication.  The product of any two Pauli operators is a Pauli operator, up to a possible unit phase $\{\pm 1,\pm i\}$.  The extra phase can usually be ignored, and the corresponding group is called the \emph{Pauli group}.

% \subsection{Schmidt decomposition}
% A $2^n\times 2^n$ unitary acts on all $n$ input qubits.  However it is sometimes useful to rewrite a unitary in terms of unitaries on smaller registers.
% \begin{theorem}[Schmidt decomposition]
% \label{thm:schmidt-decomposition}
% Let $U$ be a unitary on $n$ qubits, and let $m_1 + m_2 = n$. Then there exist orthogonal unitaries $\{V_i\}$ on $m_1$ qubits, orthogonal unitaries $\{W_i\}$ on $m_2$ qubits, and nonnegative real coefficients $c_i$ such that
% \begin{equation}
% U = \sum_i c_i V_i \otimes W_i
% \enspace .
% \end{equation}
% \end{theorem}
% Schmidt decomposition is normally phrased in terms of quantum registers, not unitary operations.  This version of the theorem is more useful for our purposes, though. \thmref{thm:schmidt-decomposition} will used in~\chapref{chap:qec} for characterizing errors that occur on quantum states.

\section{Measurement
\label{sec:mechanics.measurement}
}
Results of a quantum operation or quantum algorithm are obtained by \emph{measuring} quantum registers.  Let $\{\ket{\phi_i}\}$ be an orthonormal basis for a quantum register $\ket\psi$ such that $\ket\psi = \sum_i a_i \ket{\phi_i}$.  The measurement of $\ket\psi$ with respect to this basis yields outcome $i$ with probability $\abs{a_i}^2$.  For example, measurement of the single-qubit state $a\ket{0} + b\ket{1}$ yields outcome zero with probability $\abs{a}^2$ and outcome one with probability $\abs{b}^2$.  Since $\ket{0}$ and $\ket{1}$ are eigenstates of $Z$, this is called a $Z$-basis measurement.

We may alternatively measure in the $X$ eigenbasis $\{\ket{+} = \frac{1}{\sqrt{2}}(\ket 0 + \ket 1), \ket{-} = \frac{1}{\sqrt{2}}(\ket 0 + \ket 1)\}$.  Measurement in the $X$ basis is equivalent to first performing the Hadamard gate
\begin{equation}
H = \frac{1}{\sqrt 2} \begin{pmatrix} 1&1 \\ 1&-1 \end{pmatrix}
\end{equation}
and then measuring in the $Z$ basis, since $H\ket{+} = \ket 0$ and $H\ket{-} = \ket 1$.

Measurement in other bases, and measurement of multi-qubit registers is physically possible in principle.  However, we will use only single-qubit $Z$-basis and $X$-basis measurement in this thesis.

\section{Entanglement
\label{sec:mechanics.entanglement}
}
Unlike bits of a classical register, qubits in a quantum register need not be independent of each other.  Consider the so-called ``Bell-state'' on two-qubits
\begin{equation}
\label{eq:bell-state}
\ket{\psi} = \frac{1}{\sqrt{2}}(\ket{00} + \ket{11})
\enspace .
\end{equation}
This state is a ``superposition'' of two cases, one in which both qubits have value zero, and one in which both qubits have value one.  

If we measure the first qubit of $\ket\psi$, then we get a classical bit, either zero or one.  But in this case, we know that the value of the second qubit must be equal to the value of the first qubit.  That is, if we measure zero on the first qubit, then the value of the second qubit must also be zero.  Similarly, if we measure a one on the first qubit, then the second qubit must also have value one.

A state such as~\eqnref{eq:bell-state} in which qubit values are not independent is said to be \emph{entangled}.  Entangled states are an important part of many quantum algorithms and are used heavily in quantum error-correcting codes.

\section{Universality
\label{sec:mechanics.universality}
}
Any quantum algorithm can be expressed as a sequence of unitary operations and single-qubit measurements.  However, rather than construct a quantum computer capable of executing an infinite number of possible unitary operations, it is more practical to decompose quantum algorithms into a finite, but universal set of gates.

\begin{definition}[Universality]
\label{def:universality}
A set of quantum gates $G$ is \emph{universal} if for any unitary $U$ and $\epsilon > 0$, there exists some $k$ and $V = G_1 G_2\ldots G_k$ such that $G_1, G_2\ldots G_k \in G$ and $\norm{V-U} \leq \epsilon$.
\end{definition}

Informally~\defref{def:universality} says that a universal gate set is one from which any unitary $U$ can be approximated to any desired error tolerance $\epsilon$.  The choice of norm $\norm{V-U}$ is largely arbitrary; when necessary, the choice of norm will be stated explicitly.  It can be shown that the set of arbitrary single-qubit gates with the addition of any non-trivial multi-qubit gate---i.e., one that cannot be expressed as the product of single-qubit gates---is universal~\cite{DiVi95}.  Thus the problem of universality can be reduced to just the single-qubit case.

\subsection{The Clifford group
\label{sec:mechanics.universality.clifford}
}
One special class of quantum gates is the \emph{Clifford} gates.
A gate $G$ on $n$ qubits is Clifford if and only if $e^{i\theta} G^\dagger P G \in \mathcal{P}^{\otimes n}$ for all $P \in \mathcal{P}^{\otimes n}$ and some unit phase $e^{i\theta}$, where $\mathcal{P} = \{I,X,Y,Z\}$.  That is, the Clifford gates are those that map Pauli operators to Pauli operators under conjugation.  The Clifford operators form a group.  The single-qubit Clifford group has size $24$ and can be generated by $\{H, S= \left( \begin{smallmatrix}1&0 \\ 0&i \end{smallmatrix} \right) \}$.  The entire Clifford group can be generated by adding a single two-qubit gate, usually
\begin{equation}
\text{CNOT} = \begin{pmatrix} 1&0&0&0 \\ 0&1&0&0 \\ 0&0&0&1 \\ 0&0&1&0 \end{pmatrix}
\enspace .
\end{equation}
The first input of the CNOT is called the \emph{control} and the second input of the CNOT is called the \emph{target}. The CNOT gate flips the value of the target qubit only if the state of the control qubit is $\ket 1$.

The Clifford group is important in the study of fault-tolerant quantum computing for two reasons.  First, many quantum error-correcting codes permit very simple and robust encoded versions of Clifford gates. Second, and more importantly, it is particularly easy to calculate the effect Pauli errors as they propagate through sequences of Clifford gates.  Indeed, the Clifford group contains several important quantum gates including $H$ and CNOT, but quantum computations that contain only Cliffords can be efficiently simulated by a classical computer, a result known as the Gottesman-Knill theorem.  In fact, the Clifford group is strictly \emph{less} powerful than (universal) classical computation~\cite{Aaronson2004a}.

Propagation of Pauli errors through Clifford gates is used heavily throughout this thesis.  For convenience, we give the relevant equations explicitly for $X$ and $Z$.  Propagation for $Y$ follows from $Y = iXZ$.
\begin{subequations}\begin{align}
HX &= ZH,\\ HZ &= XH,\\
SX &= YS,\\ SZ &= ZS,\\
\CNOT(I\otimes X) &= (I\otimes X)\CNOT,\\
\CNOT(X\otimes I) &= (X\otimes X)\CNOT,\\
\CNOT(I\otimes Z) &= (Z\otimes Z)\CNOT,\\
\CNOT(Z\otimes I) &= (Z\otimes I)\CNOT
\enspace .
\end{align}
\end{subequations}

\subsection{Non-Clifford gates
\label{sec:mechanics.universality.non-clifford}
}
The relatively meager computational power of the Clifford group implies that Clifford gates alone cannot be universal for quantum computation.  It turns out, however, that the addition of any non-Clifford gate is sufficient for universality (see, e.g., \cite{Campbell2012} Appendix D).  The most common choice is the single-qubit gate
\begin{equation}
T = \begin{pmatrix} 1&0 \\ 0&e^{i\pi/4} \end{pmatrix}
\enspace .
\end{equation}
Note that $T^2 = S$.  There are other sensible choices, however.  For example the three-qubit Toffoli gate, defined by $\ket{a,b,c} \mapsto \ket{a,b,c\oplus(a\cdot b)}$, is universal for classical computation and is therefore also useful in constructing classical reversible subroutines such as addition.  Some other alternatives are discussed in~\chapref{chap:repeat}. 

\section{Circuits
\label{sec:mechanics.circuits}
}
It is often convenient and helpful to describe sequences of quantum gates visually, as circuits.  Technically, a quantum circuit is a directed acyclic graph in which the vertices represent quantum gates, and the edges represent qubits. \figref{fig:teleportation} shows an example of a circuit composed of gates from $\{\CNOT,H,X,Z\}$. 

A circuit can be partitioned into time-steps in which each qubit is involved in at most one gate. By convention, time goes from left to right.  Note that this is the opposite of the convention for matrix multiplication, in which gates are applied on the state $\ket\psi$ from right to left.  In~\figref{fig:teleportation}, the Hadamard gate is applied first, followed by a CNOT on qubits two and three and then a CNOT on qubits one and two.

Measurements output classical bits, indicated by the double lines. Quantum gates can be conditionally applied based on classical measurement values.  In this example, the $X$ gate is applied only if the $Z$-basis measurement on the second qubit is one, and the $Z$ gate is applied only if the $X$-basis measurement on the first qubit is one.

\begin{figure}
\centering
\includegraphics{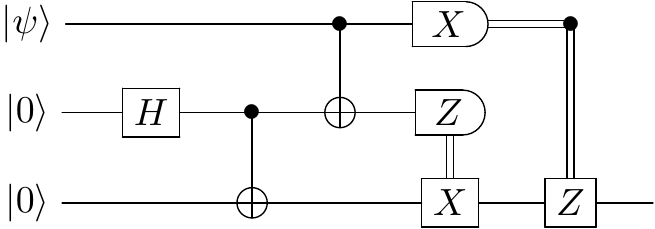}
\caption[Teleportation (two-qubit).]{\label{fig:teleportation}
An example of a quantum circuit.  The circuit takes three qubits as input, and outputs a single qubit.  CNOT gates are indicated by vertical lines between qubits; the black dot indicates the control, and the $\oplus$ indicates the target.  Measurements are represented by ``D'' shapes, and the basis ($X$ or $Z$) is indicated.  Classically-controlled gates are denoted by double lines.  This particular circuit performs ``teleportation'', transferring $\ket\psi$ from the first qubit to the third qubit. 
}
\end{figure}

\section{Teleportation
\label{sec:mechanics.teleportation}
}
The circuit shown in~\figref{fig:teleportation} demonstrates a uniquely quantum concept called \emph{teleportation}~\cite{Bennett1993}.  Teleportation can be useful for transporting quantum information quickly over large distances.
The effect of this circuit is to transfer the input state $\ket\psi$ of the first qubit on to the third qubit. Initially the second and third qubits must be located close together in order to execute the first CNOT gate.  The third qubit can then be transported to any desired location.  Upon executing the remainder of the circuit, the state of the first qubit is instantly transported to the location of the third qubit, up to Pauli corrections based on the measurement outcomes.

Teleportation is used frequently in fault-tolerant circuits, but for a different reason.  Consider the circuit shown in~\figref{fig:one-qubit-teleportation}.  This circuit also teleports the state $\ket\psi$, but requires only one additional qubit~\cite{Zhou2000}.  After teleportation, a $Z$-axis rotation
\begin{equation}
R_Z(\theta) = \cos(\theta/2)I - i\sin(\theta/2)Z
\end{equation}
is applied to the output.  Next, observe that 
\begin{equation}
R_Z(\theta) X = R_Z(\theta) X R_Z(-\theta) R_Z(\theta) = R_Z(2\theta)X R_Z(\theta)
\enspace .
\end{equation}  
Therefore, the $Z$-axis rotation may be shifted to the left of the conditional $X$ correction, and to the left of the CNOT gate (since $Z$ has no effect on the control of a CNOT).  The $R_Z(\theta)$ gate can now be performed ``offline'' on the ancillary qubit, before interaction with the state $\ket\psi$.  The technique of preparing a gate offline by commuting through the teleportation circuit is called \emph{gate teleportation}~\cite{Gottesman1999a}.

Of course, the conditional correction $R_Z(2\theta) X$ in the gate teleportation circuit is now more complicated than it was before.  However, there are certain cases in which fault-tolerantly executing $R_Z(2\theta)$ is far easier than executing $R_Z(\theta)$.  Offline preparation of the more difficult $R_Z(\theta)$ allows for more efficient error suppression, as we will see in~\chapref{chap:fault}.

\begin{figure}
\centering
\begin{subfigure}[b]{.42\textwidth}
\includegraphics[width=\textwidth]{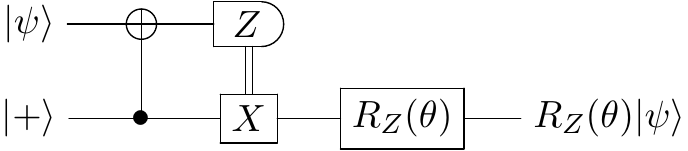}
\caption{\label{fig:one-qubit-teleportation}}
\end{subfigure}
\hfill
\begin{subfigure}[b]{.5\textwidth}
\includegraphics[width=\textwidth]{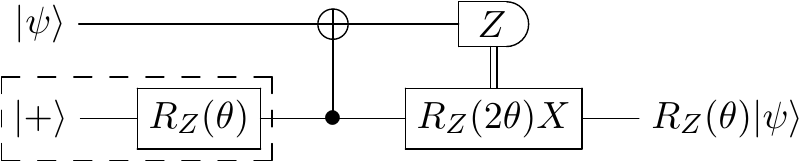}
\caption{\label{fig:gate-teleportation}}
\end{subfigure}
\caption[One-qubit teleportation.]{Two modifications of the teleportation circuit shown in~\figref{fig:teleportation}. (a) One-qubit teleportation. The input $\ket\psi$ is teleported using just one ancilla qubit, prepared as $\ket+$.  After teleportation, a $Z$-axis rotation is applied to the output. (b) Gate teleportation. Using the relation $R_Z(\theta)X = R_Z(2\theta)X R_Z(\theta)$, the $Z$-axis rotation can be shifted to the left and a new conditional correction is required.}
\end{figure}

\chapter{Protecting quantum information
\label{chap:qec}
}

In~\chapref{chap:motivation} we discussed the fragility of the information stored in quantum bits.  An unprotected quantum system interacts freely with its environment, causing the information that it contains to be corrupted or lost.  Before a qubit can be used for computation, it must be protected against noise.

In this chapter, we detail a major tool for protecting quantum information, quantum error-correcting codes.  Quantum codes use many physical qubits to represent one logical qubit, thereby reducing the impact of an error on any one of the physical qubits.  Quantum information is more complicated than classical information, and likewise quantum errors are more complicated than classical errors.  Nonetheless, it is still possible to use the wealth of classical coding theory to develop quantum codes.

\section{First things first: classical error correction
\label{sec:qec.classical}
}
Classical codes operate by adding redundancy.  For example, the simplest classical code is the two-bit repetition code in which a single \emph{logical} bit is encoded using two noisy bits. The logical value zero is encoded as $00$ and the logical one is encoded as $11$.  An error on either one of the two noisy bits will result in a value of $01$ or $10$.  This single error can be detected by taking the parity of the two bits (i.e., the sum of the bits modulo two); in this case an odd parity indicates an error.  By adding third bit of repetition, single bit-flips can be \emph{corrected}.  For example, the value $010$ can be restored by flipping the second bit back to zero.  The errant bit can be identified by taking the parity of each pair of bits.  An odd parity for the first two and the last two bits indicates an error on the middle bit.

A common simplifying assumption is that errors occur identically and independently on each bit.  If the probability of an error on a single bit is $p$, then the probability of a simultaneous error on two bits is $p^2$.  Since the three-bit repetition code can correct any single-bit error, an uncorrectable error occurs only when there are simultaneous errors on two or more of the bits.  The probability $p_L$ of this uncorrectable, or ``logical'' error is given by
\begin{equation}
p_L = 3p^2(1-p) + p^3
\enspace,
\end{equation}
where there are $\binom{3}{2}=3$ ways for two errors to occur.  So long as 
\begin{equation}
p < 3p^2(1-p) + p^3,
\end{equation}
which is true for $p < 0.5$, then the encoding yields a net improvement over just a single bit.\footnote{In this case, net improvement can also be obtained for $p> 0.5$ by inverting the correction procedure.}

The repetition code can be extended to correct larger numbers of errors by simply adding more bits.  The number of simultaneously correctable errors is given by $\lfloor (n-1)/2 \rfloor$ where $n$ is the number bits in the code.  In the limit of large $n$, each additional bit increases the number of correctable errors by one-half.

\subsection{Linear codes
\label{sec:qec.classical.linear}
}
Improved efficiency can be obtained by encoding more than one logical bit at a time.  Linear codes are defined by a $k\times n$ binary matrix $G$, where $n$ is the number of bits of the code and $k$ is the number of encoded logical bits.  The logical value $x$ is encoded into a \emph{codeword} $c$ by binary (i.e., sum modulo two) matrix-vector multiplication
\begin{equation}
c = \transpose{G} x
\enspace,
\end{equation}
where $x$ and $c$ are treated as column vectors.

All codewords satisfy a set of linear constraints called parity checks, defined by a $(n-k)\times n$ binary matrix $H$ such that
\begin{equation}
H\transpose{G} = 0
\enspace,
\end{equation}
which implies that $Hc = 0$ for all codewords $c$.

The parity check matrix $H$ is useful in identifying errors since for any codeword $c$ and any $n$-bit vector $e$,
\begin{equation}
H(c+e) = Hc + He = He
\enspace .
\end{equation}
The $(n-k)$-bit vector $He$ identifies the parity checks violated by the error $e$ and is called the error \emph{syndrome}.  Each syndrome can be associated with a recovery operation $e'$ that returns the vector ($c+e$) to a codeword, i.e., $H(c+e+e') = 0$.

The \emph{distance} of a linear code is defined as the minimum Hamming weight of any nonzero codeword.  The distance corresponds to the minimum number of bits that must be flipped to transform one codeword into another---i.e., the minimum Hamming distance between codewords.  The all zero vector is always a codeword of any linear code, and so the minimum Hamming distance cannot be larger than the minimum weight (nonzero) codeword.  Conversely, for any two codewords $c_1, c_2$, the linear combination $c = c_1 + c_2$ is also a codeword and the Hamming weight of $c$ is equal to the Hamming distance of $c_1$ and $c_2$.  Thus, the Hamming distance between $c_1$ and $c_2$ is at least the code distance.  The three-bit repetition code, for example, has distance three since $111$ has weight three.

A code with distance $d$ can detect up to $d-1$ bit errors.  This fact follows from the definition of minimum distance. Any vector $c+e$ that is not a codeword yields a nonzero syndrome, and so applying an error $e$ to a codeword $c$ results in a syndrome of zero only if $e$ has Hamming weight at least $d$.  A linear code can correct up to $t = \lfloor (d-1)/2 \rfloor$ bit errors.  The correction procedure takes a vector $c+e$ and replaces it with the closest (in Hamming distance) codeword $c'$. Informally, an error of weight $k$ moves the data $k$ steps away from the codeword.  So long as $k$ is less than halfway to any other codeword the correction procedure will succeed.  Again, the three-bit repetition code can detect errors up to weight two, but can only correct errors of weight one.

A linear code using $n$ noisy bits to encode $k$ logical bits to a distance of $d$ is denoted by $[n,k,d]$.
Perhaps the most well known class of linear codes is the family of $[2^r-1,2^r-r-1,3]$ Hamming codes, for $r\geq 2$~\cite{MacWilliams1988}.  The three-bit ($r=2$) Hamming code corresponds to the three-bit repetition code discussed above.  Parity check matrices for the seven-bit and $15$-bit Hamming codes are shown in~\tabref{tab:hamming-codes}.

\begin{table}
\centering
\begin{subtable}[b]{.3\textwidth}
\centering
$\begin{pmatrix} 
 0&0&0&1&1&1&1\\
 0&1&1&0&0&1&1\\
 1&0&1&0&1&0&1
\end{pmatrix}$
\caption{\label{tab:7-bit-Hamming-parity-checks}}
\end{subtable}
\hfill
\begin{subtable}[b]{.69\textwidth}
\centering
$\begin{pmatrix} 
 0&0&0&0&0&0&0&1&1&1&1&1&1&1&1\\
 0&0&0&1&1&1&1&0&0&0&0&1&1&1&1\\
 0&1&1&0&0&1&1&0&0&1&1&0&0&1&1\\
 1&0&1&0&1&0&1&0&1&0&1&0&1&0&1
\end{pmatrix}$
\caption{\label{tab:15-bit-Hamming-parity-checks}}
\end{subtable}
\caption[Hamming code parity checks]{\label{tab:hamming-codes}
Parity check matrices for the (a) $[7,4,3]$ and (b) $[15,11,3]$ Hamming codes.
}
\end{table}

\subsection{Dual codes
\label{sec:qec.classical.dual}
}

The generator matrix $G$ and the parity check matrix $H$ are interchangeable.  Just as $G$ defines the codewords of a linear code, $H$ defines the codewords of a different code, called the \emph{dual}.  The parity checks of the dual code are then given by $G$.  Alternatively, given a linear code $C$, the codewords of the dual code are given by the orthogonal complement of $C$ defined by the set $C^\perp = \{g: |g\cdot c| = 0\mod 2, \forall c\in C\}$.

\section{Quantum errors
\label{sec:qec.errors}
}

Unfortunately, classical codes cannot be used directly to protect quantum information, primarily because in addition to bit flips, qubits can suffer from more exotic kinds of errors. For example, consider the state $\ket + = \frac{1}{\sqrt{2}}(\ket{0} + \ket{1})$. If the Pauli operator $Z$ is accidentally applied to this state then it becomes $\ket{-} = \frac{1}{\sqrt{2}}(\ket{0} - \ket{1})$.  This kind of error is called a \emph{phase-flip}, since the relative phase between $\ket{0}$ and $\ket{1}$ has been swapped from $+1$ to $-1$.

\subsection{Discretization}
On the surface the problem appears to be even worse than just dealing with bit-flip and phase-flip errors. Consider the operator
\begin{equation}
\label{eq:diagonal-qubit-error}
E_\theta = \begin{pmatrix} 1&0 \\ 0&e^{-i2\theta} \end{pmatrix}
\enspace ,
\end{equation}
where where $\theta \in [0,\pi)$.
Accidental application of $E_\theta$ introduces one of an infinite number of continuous phase errors $e^{-i2\theta}$.  Bit-flip errors may be similarly continuous.

However, we may rewrite~\eqnref{eq:diagonal-qubit-error} as 
\begin{equation}
\label{eq:discrete-qubit-error}
E_\theta = e^{-i\theta}(\cos(\theta)I + i\sin(\theta)Z)
\enspace .
\end{equation}  
When written in this way, what was a continuous phase error now appears as a \emph{discrete} $Z$ error, but with a continuous amplitude.  Up to a global phase, the state is either left unchanged with amplitude $\cos(\theta)$ or incurs a phase-flip with amplitude $i\sin(\theta)$.  The global phase $e^{-i\theta}$ is generally unimportant, since it has no effect on measurement outcomes.

More generally, an error can be modeled as a unitary transformation $U_E$ on the joint state of the quantum computer $\ket\psi$ and its surrounding environment $\ket{E}$.  Using the fact that the Pauli operators form a basis for single-qubit operators, $U_E$ can be decomposed as
\begin{equation}
\label{eq:general-error}
U_E = \sum_{i,j} e_{ij} P_i \otimes E_j
\enspace ,
\end{equation}
where each $P_i$ is a tensor product of Pauli operators and $E_j$ acts only on the environment. 
The result of an error $U_E$ on the joint state is then given by
\begin{equation}
\label{eq:discrete-general-error}
U_E \ket\psi \ket E = \left( \sum_{i} P_i\ket\psi \right) \sum_j e_{ij} E_j\ket{E} 
\enspace .
\end{equation}
Again, as in~\eqnref{eq:discrete-qubit-error}, the error is written as a discrete sum over Pauli operators.

Equation~\eqnref{eq:discrete-general-error} implies that task of protecting quantum information can be reduced to the task of guarding against products of Pauli errors.  Additionally, since $Y = iXZ$, each tensor of Paulis can be expressed using only $X$ and $Z$, up to an unimportant global phase.  In other words, quantum errors can be expressed solely in terms bit-flips and phase-flips on individual qubits.

\subsection{Leakage and loss}
The error expressed in~\eqnref{eq:general-error} is not entirely general in that it does not directly account for leakage and loss errors.  Leakage occurs when the state $\ket\psi$ goes outside of the expected $2^n$ dimensional state space.  For example, a qubit may be represented physically by the first two energy levels of an ion. Thermal excitations could cause the ion to jump to a higher energy level, in which case the state would have to be represented by a \emph{qutrit}
\begin{equation}
\ket\psi = a_0\ket{0} + a_1\ket{1} + a_2\ket{2}
\enspace ,
\end{equation}
where the state $\ket{2}$ represents leakage outside of the qubit space.  Similarly, loss occurs when a qubit is removed or otherwise disappears from the computer.  This could happen if an ion is spontaneously ejected from a trap.

Left unchecked, leakage and loss errors can have serious consequences for protection of quantum information~\cite{Ghosh2013a}. However, they can usually be controlled with a small amount of effort~\cite{Pres98b,Fowler2013c}.  We will not consider leakage and loss errors in this thesis.

\section{Quantum error-correcting codes
\label{sec:qec.qecc}
}

Equation~\eqnref{eq:discrete-general-error} shows us that the state of a quantum register after being subjected to noise can be expressed as a superposition of the original state over a discrete set of bit-flip and phase-flip errors.  Informally then, the goal of a quantum error-correcting code is to project the register onto one of those superposition states, identify the error and reverse it.

More formally, let $\{\psi_i\}$ be an orthonormal basis for the codewords of a quantum error-correcting code $C$, and let $\{E_i\}$ be a set of errors against which we would like to protect.  The conditions under which the code $C$ can correct errors $\{E_i\}$ are given by the following theorem~\cite{Bennett1996,Knill2000}.
\begin{theorem}[Quantum error correction condition]
\label{thm:qecc-condition}
A code $C$ with codewords $\{\ket{\psi_i\}}$ can correct the set of errors $\{E_a\}$ if and only if
\begin{equation}
\label{eq:qecc-condition}
\bra{\psi_i}E_a^\dagger E_b\ket{\psi_j} = C_{ab}\delta_{ij}
\enspace ,
\end{equation}
where $\delta_{ij}$ equals one if $i=j$ and equals zero otherwise, and $C_{ab} \in \mathbb{C}$ is independent of $i$ and $j$.  
\end{theorem}

\thmref{thm:qecc-condition} can be understood by considering a code with just two codewords $\{\lket{0}, \lket{1}\}$, where the notation $\lket a$ indicates the encoded logical state $\ket a$.  Then~\eqnref{eq:qecc-condition} requires that $E_a\lket 0$ and $E_b\lket 1$ are orthogonal.  If this were not the case, then an error $E_a$ on $\lket 0$ and $E_b$ on $\lket 1$ would yield overlapping states, and measurement of the error could confuse the two cases.  In particular, if $E_a^\dagger E_b$ is a logical operator (say $\overline{X}$), then~\eqnref{eq:qecc-condition} is certainly violated.  Likewise, consider an error $E$ with the property that
$\bra{\overline{0}}E^\dagger E\lket{0} \neq \bra{\overline{1}}E^\dagger E\lket{1}$, again violating~\eqnref{eq:qecc-condition}.  Then $E$ changes the relative amplitudes of $\lket 0$ and $\lket 1$ so that $E(\ket{0} + \ket{1}) \propto \lket{0} + \delta\lket 1$ for some $\delta$.  But $\lket{0} + \delta\lket 1$ is itself a codeword, so the error $E$ cannot be distinguished from a valid logical operation. 

\subsection{Stabilizer codes
\label{sec:qec.qecc.stabilizer}
}
The most widely studied class of quantum error-correcting codes is \emph{stabilizer codes}, the quantum analog of classical linear codes~\cite{Gott96,Cald97}.  A stabilizer code is defined by a stabilizer group $\mathcal M$ for which each element is a tensor product of Pauli operators.  The set of codewords is given by $\{\ket \psi : M\ket\psi = \ket\psi, M \in \mathcal M\}$; each codeword is a $+1$-eigenvector of all of the elements in the stabilizer group. Since $\mathcal M$ is a group, the stabilizers can be specified by a set of generating elements called \emph{stabilizer generators}.  The stabilizer generators are directly analogous to the parity checks of a classical linear code.

Error correction can be performed by measuring each of the stabilizer generators in order to determine the error syndrome.  The number of simultaneous single-qubit errors that the code can correct is given by $\lfloor (d-1)/2 \rfloor$, where $d$ is the code distance.  Define the \emph{normalizer} of $\mathcal{M}$, $\mathcal{N} := \{P\in \mathcal{P}^n : PS=SP, S\in \mathcal{M}\}$, as the set of $n$-qubit Pauli group elements that commute with all of the stabilizers. The distance of the code is then equivalent to the minimum weight non-identity element of $\mathcal{N} \setminus \mathcal{M}$. Here, the weight of an operator is defined as the number of $X$ , $Y$ and $Z$ operators in its tensor product decomposition.

A stabilizer group on $n$ physical qubits with $m$ generators encodes $n-m$ logical qubits.  The $2^m$ syndromes partition the $2^n$-dimensional state space, yielding a codespace of dimension $2^{n-m}$.  Each logical qubit $i$ is associated with a pair of logical operators $X_i$, $Z_i \in \mathcal{N} \setminus \mathcal{M}$ such that $X_i$ and $Z_i$ commute with all of the stabilizers, but anti-commute with each other.  Logical operators on different logical qubits also commute. The situation is in direct correspondence with single-qubit Pauli operators on physical qubits.
A stabilizer code encoding $k$ logical qubits into $n$ physical qubits to a distance of $d$ is denoted as $[[n,k,d]]$.

\subsubsection{Stabilizer algebra
\label{sec:qec.qecc.stabilizer.algebra}
}

Given a set of generators and logical operators for a stabilizer code, it is possible to write out each of the codewords $\{\ket{\psi_i}\}$ explicitly, and therefore to calculate how the encoded quantum state evolves under unitary operations and measurements.  However, the stabilizer formalism offers an alternative which is usually more efficient and intuitive.  Consider the effect of applying a unitary $U$ to a codeword $\ket\psi$. We would like to understand how $U$ impacts the stabilizers and the logical operators of the code.  By definition, we have
\begin{equation}
U\ket\psi = U(M\ket\psi) = (UMU^\dagger)U\ket\psi
\end{equation}
for any stabilizer $M$.  Thus, a stabilizer $M$ of the original state is transformed by conjugation $UMU^\dagger$ to a stabilizer of the new state $U\ket\psi$.  The logical operators are similarly transformed by conjugation.

In this way, stabilizers offer an analog of the Heisenberg interpretation of quantum mechanics~\cite{Gottesman1998}.  Rather than tracking the evolution of the state $\ket\psi$, we may track the evolution of the stabilizers.  For a code on $n$ qubits, there are $2^n$ possible terms in the expansion of $\ket\psi$, but only at most $n$ stabilizer generators.  Thus expressing an encoded state in terms of its code stabilizers can be exponentially more efficient than the corresponding expression as a quantum state.

The effect of measurements on the stabilizers is slightly more complicated, but can still be calculated efficiently.  Consider a $Z$-basis measurement on the first qubit of an $n$-qubit codeword.  After the measurement, the state is stabilized by the operator $Z\otimes I^{\otimes{n-1}}$, up to a phase of $\pm 1$.  The definition of the stabilizer group implies that all stabilizers must commute.  Thus, the stabilizers of the state after the measurement must all commute with $Z\otimes I^{\otimes{n-1}}$.  Any operator that was a stabilizer before the measurement, but anti-commutes with $Z$ on the first qubit cannot be a stabilizer after the measurement.  Note however, that it is always possible to express the set of stabilizer generators so that at most one generator anti-commutes with the measurement.  If both $M_1$ and $M_2$ anti-commute with the measurement, then $M_2$ can be replaced by $M_1 M_2$, which does commute.  Thus the single anti-commuting generator is replaced by $Z\otimes I^{\otimes{n-1}}$ and all of the other generators remain unchanged. 

To make this more concrete, we illustrate with an example using the $[[7,1,3]]$ code due to Steane~\cite{Steane1996}.  The stabilizer generators of this code can be expressed as
\begin{equation}
\label{eq:713-stabilizers}
\begin{tabular}{c@{}c@{}c@{}c@{}c@{}c@{}c@{}cc@{}c@{}c@{}c@{}c@{}c@{}c@{}c}
      &$I$&$I$&$I$&$I$&$X$&$X$&$X$&
      &$I$&$I$&$I$&$I$&$Z$&$Z$&$Z$\\
      &$I$&$X$&$X$&$I$&$I$&$X$&$X$&
      &$I$&$Z$&$Z$&$I$&$I$&$Z$&$Z$\\
      &$X$&$I$&$X$&$I$&$X$&$I$&$X$&
      &$Z$&$I$&$Z$&$I$&$Z$&$I$&$Z$\\
\hline
$X_L=$&$X$&$X$&$X$&$X$&$X$&$X$&$X$&
$Z_L=$&$Z$&$Z$&$Z$&$Z$&$Z$&$Z$&$Z$
\end{tabular}
\enspace ,
\end{equation}
where $X_L$ and $Z_L$ are the $X$ and $Z$ logical operators, respectively, and for visual clarity the tensor product notation has been omitted.  Now consider the effect of applying the Hadamard operator to each qubit.  Hadamard swaps $X$ and $Z$ under conjugation; $HXH = Z$ and $HZH = X$. So the result of applying $H^{\otimes 7}$ is
\begin{equation}
\begin{tabular}{c@{}c@{}c@{}c@{}c@{}c@{}c@{}cc@{}c@{}c@{}c@{}c@{}c@{}c@{}c}
      &$I$&$I$&$I$&$I$&$Z$&$Z$&$Z$&
      &$I$&$I$&$I$&$I$&$X$&$X$&$X$\\
      &$I$&$Z$&$Z$&$I$&$I$&$Z$&$Z$&
      &$I$&$X$&$X$&$I$&$I$&$X$&$X$\\
      &$Z$&$I$&$Z$&$I$&$Z$&$I$&$Z$&
      &$X$&$I$&$X$&$I$&$X$&$I$&$X$\\
\hline
$X_L=$&$Z$&$Z$&$Z$&$Z$&$Z$&$Z$&$Z$&
$Z_L=$&$X$&$X$&$X$&$X$&$X$&$X$&$X$
\end{tabular}
\enspace .
\end{equation}
The stabilizers have been preserved, and the $X_L$ and $Z_L$ logical operators have been swapped.  The operator $H^{\otimes 7}$ therefore acts as a logical Hadamard on the code.

Now consider a $Z$-basis measurement on the first qubit.  All but the operators $XIXIXIX$ and $Z_L=XXXXXXX$ commute with the measurement.  However, the $Z_L$ operator may be multiplied by $XIXIXIX$ so that it commutes with $Z\otimes I^{\otimes 6}$. (Remember that multiplication by a stabilizer is equivalent to multiplying by the identity.)  The resulting stabilizers after measurement are
\begin{equation}
\begin{tabular}{c@{}c@{}c@{}c@{}c@{}c@{}c@{}cc@{}c@{}c@{}c@{}c@{}c@{}c@{}c}
      &$I$&$I$&$I$&$I$&$Z$&$Z$&$Z$&
      &$I$&$I$&$I$&$I$&$X$&$X$&$X$\\
      &$I$&$Z$&$Z$&$I$&$I$&$Z$&$Z$&
      &$I$&$X$&$X$&$I$&$I$&$X$&$X$\\
      &$Z$&$I$&$Z$&$I$&$Z$&$I$&$Z$&
$\quad\pm$&$\mathbf{Z}$&$\mathbf{I}$&$\mathbf{I}$&$\mathbf{I}$&$\mathbf{I}$&$\mathbf{I}$&$\mathbf{I}$\\
\hline
$X_L=$&$Z$&$Z$&$Z$&$Z$&$Z$&$Z$&$Z$&
$Z_L=$&$I$&$X$&$I$&$X$&$I$&$X$&$I$
\end{tabular}
\enspace ,
\end{equation}
where the new stabilizer is highlighted in bold and the $\pm1$ phase depends on the measurement outcome.

\subsubsection{Stabilizer states
\label{sec:qec.qecc.stabilizer.states}
}

Normally we are interested in codes that contain at least one logical qubit.  For a stabilizer code on $n$ qubits, this means that the number of stabilizer generators should be $(n-k)$ for some $k > 0$.  Then the set of codewords lives in a $2^k$-dimensional subspace representing $k$ logical qubits.  If $k=0$, however, then the set of codewords has dimension one, a single quantum state.  

An $n$-qubit state that is defined by a set of $n$ stabilizer generators is called a \emph{stabilizer state}.  In the seven-qubit code, for example, adding the $Z$ logical operator $Z_L$ to the set of stabilizers yields a stabilizer state.  By definition, this state is a $+1$-eigenstate of $Z_L$ and so this is the encoded state $\ket 0$, just as (physical) $\ket 0$ is the $+1$-eigenstate of $Z$.

Not all quantum states are stabilizer states.  Consider the effect of applying $T$ to the first qubit of the encoded $\ket 0$ state defined above. The conjugation relations for $T$ are
\begin{equation}\begin{split}
\label{eq:T-conjugation}
TZT^\dagger &= Z\\
TXT^\dagger &= (X+Y)/\sqrt{2}
\enspace . 
\end{split}\end{equation}
Therefore some of the resulting stabilizers are no longer tensor products of Paulis, but rather linear combinations of tensor products of Paulis.  The encoded state $T\otimes I^{\otimes 6}\lket 0$ is not a stabilizer state.

On the other hand, an inductive argument shows that the output of any circuit composed of Clifford gates, $\ket 0$ preparation and $Z$-basis measurement is a stabilizer state.  Conversely, the definition of the Clifford group implies that any stabilizer state can be expressed by such a circuit~\cite{Aaronson2004a}.  Stabilizer states and their corresponding circuits are a major component of fault-tolerant error correction, and are discussed in more detail in~\chapref{chap:ancilla}.
 
\subsection{CSS codes
\label{sec:qec.qecc.css}
}
A particularly useful subset of stabilizer codes can be constructed from classical linear codes.  The construction requires two linear codes $C_1=[n,k_1,d_1], C_2=[n,k_2,d_2]$ that are orthogonal, i.e., $C_1^\perp \subseteq C_2$. The parity checks of $C_1$ can be translated into tensor products of Pauli $X$ operators, and the parity checks of $C_2$ can be translated into tensor products of Pauli $Z$ operators.   Together these operators form the stabilizer generators of the quantum error-correcting code. The tensor products of $X$ are called \emph{$X$ stabilizers} and the tensor products of $Z$ are called \emph{$Z$ stabilizers}.

Codes based on this construction are known as CSS codes after Calderbank, Shor and Steane and include the most commonly known codes such as the Steane's $[[7,1,3]]$ code given by~\eqnref{eq:713-stabilizers}~\cite{Steane1996}, which is based on the seven-bit Hamming code~(\tabref{tab:hamming-codes}), and Shor's $[[9,1,3]]$ code~\cite{Calderbank1996}, which is based on the nine-bit repetition code.  CSS codes are ubiquitous in the study of fault-tolerant quantum computation; all of the codes considered in this thesis are CSS.

CSS codes have a couple of properties which make them nice for study and for use in fault-tolerance schemes.  First, the codewords of a CSS code follow the form
\begin{equation}
\label{eq:css-codeword}
\lket{x} = \frac{1}{\sqrt{\abs{C_1^\perp}}}\sum_{w\in C_1^\perp} \ket{x + w}
\enspace,
\end{equation}
where $x$ is the coset representative of an element of $C_2 / C_1^\perp$. Equation~\eqnref{eq:css-codeword} shows that each codeword $x$ can be interpreted as a superposition over each of the $X$ stabilizers. The code $C_1^\perp$ partitions $C_2$ into $\abs{C_2}/\abs{C_1^\perp}$ cosets and so there are $2^{k_2 - (n - k_1)}$ codewords. 
Second, CSS codes permit independent correction of $X$ errors and $Z$ errors.  The $X$ stabilizers defined by $C_1$ are used to correct $Z$ errors, and the $Z$ stabilizers defined by $C_2$ are used to separately correct $X$ errors.  Independent $X$ and $Z$ correction is exploited in~\chapref{chap:ancilla} and~\chapref{chap:threshold}. As a consequence of these two properties, the CSS construction yields a $[[n,k_1+k_2-n,\min\{d_1,d_2\}]]$ quantum code.

\subsection{Concatenated codes
\label{sec:qec.qecc.concatenated}
}
Stabilizer codes can be combined to form other larger stabilizer codes.
Given two stabilizer codes $C_1 = [[n_1,k_1,d_1]]$ and $C_2 = [[n_2,1,d_2]]$, a $[[n_1 n_2,k_1,d_1 d_2]]$ code is be obtained by encoding each physical qubit of $C_1$ in the code $C_2$~\cite{Knill1996}.  This construction is known as code \emph{concatenation}, and is a key element of many threshold theorems including the one in~\chapref{chap:threshold}.  In particular, concatenation can be performed repeatedly in order to obtain an arbitrarily large code distance.

Concatenation can also be accomplished when $C_2$ encodes multiple logical qubits, in which case the resulting code is $[[n_1 n_2, k_1 k_2, d_1 d_2]]$~(see, e.g., \cite{Gottesman1997a}).
Other methods for combining codes include \emph{pasting} to increase $k$~\cite{Gottesman1996}, and \emph{welding}~\cite{Michnicki2012}.  We focus only on concatenation in this thesis, however.

\subsection{Topological codes
\label{sec:qec.qecc.topological}
}
Another notable subset of stabilizer codes are so-called \emph{topological} codes.  These codes have the property that the stabilizer generators can be defined locally when qubits are laid out as a lattice on some manifold.  Prominent examples include the toric code~\cite{Kitaev1997a}, and the surface code~\cite{Bravyi1998}.  
%More complex topological codes have also been proposed~\cite{Bomb08,Bomb10b,Koenig2010}.

Each topological code is, in fact, a family of codes.  Notably, both the number of encoded qubits and the distance can be increased arbitrarily while maintaining locality of the stabilizer generators.  This permits fault-tolerance schemes which require only local interactions among qubits, a feature which is useful on a large number of proposed physical quantum computing architectures.  By contrast, concatenated codes require interactions between qubits which may be far apart.

\subsection{Non-stabilizer codes
\label{qec.qecc.non-stabilizer}
}
There are also quantum error correcting codes that do not conform to the stabilizer construction.  A variety of codes can be constructed by relaxing the stabilizer formalism in some way. Subsystem codes, for example, encode qubits as linear subsystems rather than two-dimensional subspaces~\cite{Baco06}.  Another relaxation of the stabilizer formalism can be used to construct approximate quantum error-correcting codes~\cite{Leung1997}. Codes can be used to protect qudits ($d$-dimensional quantum bits) rather than qubits~\cite{Knil96a}. Yet more codes are possible if the code block is entangled with an outside resource~\cite{Bowen2002}.
 
Stabilizer codes are generalized by so-called codeword stabilized codes~\cite{Looi2008,Cross2009a}.  A codeword stabilized code is characterized by a stabilizer state and a set of ``word operators'' that act as logical $X$ operators.  The structure of these codes is more complicated than for stabilizer codes.  The word operators need not commute with each other, for example. Codeword stabilized codes have not been widely studied in the context of fault-tolerant quantum computing.

% \section{Examples of quantum error-correcting codes \todo{}
% \label{sec:qec.examples}
% }
% 
% \subsection{[[15,1,3]]}
% 
% \subsection{The Golay code}

\section{Experimental realization
\label{sec:qec.experiment}
}
The protection offered by quantum error-correcting codes was demonstrated experimentally as early as $1998$, when the three-qubit phase-flip code was implemented in liquid state NMR~\cite{Cory1998}.  Since then, a number of small codes have been used in a variety of experimental setups including liquid state NMR~\cite{Leung1999,Knill2001,Boulant2002,Boulant2005,Zhang2011,Zhang2011b,Zhang2012}, trapped ions~\cite{Chiaverini2004,Schindler2011}, linear optics~\cite{Pittman2005,Yao2012a}, cavity-QED~\cite{Ottaviani2010}, solid state NMR~\cite{Moussa2011}, and superconducting circuit-QED~\cite{Reed2012}.
Most demonstrations have been quite basic; a single logical qubit is encoded using a small number of physical qubits, left idle for some period of time, and then decoded.  A few studies have demonstrated more complicated operations such encoded gates~\cite{Zhang2012} and state distillation~\cite{Souza2011} (see~\chapref{chap:fault}).

On the one hand, experimental demonstrations of quantum error correction provide proof-of-principle that codes can indeed suppress noise in realistic physical systems.  The initial $1998$ experiment showed that the gate error rate could be reduced from $p$ to roughly $p^2$.  A similar, but much more recent study shows even sharper improvements~\cite{Zhang2011}.  On the other hand, the limited scale of the experiments illustrate the need to improve threshold and resource overhead requirements.  Most experimental setups are large enough to encode only a single logical qubit, whereas quantum algorithms require hundreds or thousands of qubits.  Experimental capabilities will continue to improve, but so must the resource costs of error correction.

\section{Alternative methods of protection
\label{qec.alternatives}
}
Quantum error-correcting codes are not the only means by which to protect quantum information.  For completeness, we briefly outline some alternative techniques. 

\subsection{Decoherence-free subspaces and dynamical decoupling
\label{sec:qec.alternatives.dfs-dd}
}
Originally formalized for quantum information by~\cite{Palma1996} and~\cite{Duan1997} and later coined by~\cite{Lidar1998}, a \emph{decoherence-free subspace} (DFS) encodes data into states for which the effect of environmental noise is trivial.  As a toy example, consider a noise model in which only $Z$ errors occur, and when they do they occur simultaneously on all qubits in the system.  That is, for a system of $n$ qubits, the only possible error is $Z^{\otimes n}$.  Even this very simple noise model can cripple a quantum computer.  But this error acts trivially on certain states, for example,
\begin{equation}
Z\otimes Z (a\ket{00} + b\ket{11}) = a\ket{00} + b\ket{11}
\enspace .
\end{equation}
Thus, by encoding in the subspace $\{\ket{00},\ket{11}\}$ ($\ket{00}$ for logical $\ket 0$, and $\ket{11}$ for logical $\ket 1$) the logical qubit is completely immune to errors.  In this way, a DFS is equivalent to an error-correcting code for a very simple and specific noise model.

Decoherence-free subspaces enjoy several advantages over error-correcting codes.  First, they usually require only a very small number (two in the above example) of physical qubits per logical qubit.  Second, since errors act trivially, a DFS requires no active intervention in order to correct errors.  Furthermore, the strength of the noise can be very high, in contrast to error correcting codes which can tolerate only low levels of noise (see~\chapref{chap:fault}).  On the other hand, given a particular noise model, finding the symmetries required to construct a DFS, provided that they even exist, is difficult.  Indeed, DFS is known to be insufficient for some reasonable noise models~\cite{Lidar2001a}.

\emph{Dynamical decoupling} (DD) is another technique for suppressing errors for simple and well-characterized noise models~\cite{Viola1999,Ban1998}.  If noise causes the system to evolve in an uncontrolled but predictable way, then quick control pulses can be used to periodically ``reverse'' the noise and cause it to cancel out.  Again a toy example is helpful. Say that noise acts continuously on a qubit, and that for a fixed duration of time $t$ the effect is given by
\begin{equation}
E(t) = \begin{pmatrix} 1&0 \\ 0&e^{i\theta t}\end{pmatrix}
\enspace .
\end{equation}
By periodically applying Pauli $X$, the effect of the noise can be canceled since
\begin{equation}
E(t)XE(t)X = E(t)E^{\dagger}(t) = I
\enspace .
\end{equation}
In more practical examples the noise and the required control pulses are more complicated, but the idea is the same.

Dynamical decoupling has the advantage of requiring no additional qubits. Its disadvantage, is that it requires fast and accurate control.  Moreover, complicated pulse sequences can make data manipulation more difficult, and increase gate times.

DFS and DD are usually considered as complementary to quantum error-correcting codes.  A variety of authors have considered methods for using DFS, DD and quantum error-correcting codes in different combinations~\cite{Lidar1999,Lidar2001a,Ng2011,Paz-Silva2012}. DFS and DD can act as a ``first line of defense'' against errors, after which error-correction is applied to achieve arbitrary accuracy. In this thesis we focus only on fault-tolerance protocols based on codes.  It is likely, however, that the the best complete strategies for suppressing errors will involve elements from all three techniques.

\subsection{Topological quantum computation
\label{sec:qec.alternatives.topological}
}
A third, and dramatically different alternative to quantum error-correcting codes is \emph{topological} quantum computation.  In topological quantum computation, data is stored in exotic particles called \emph{anyons}~\cite{Kitaev2003}.  Consider a pair of particles which are placed side-by-side, and then exchanged; the particle on the left moves to the right, and the particle on the right moves to the left.  For typical physical particles, such as photons or electrons, the effect of this exchange is essentially trivial.  For anyons, however, this exchange induces a non-trivial phase akin to a diagonal unitary gate.  Sequences of exchanges, called ``braids'', can be composed in order to quantum compute~\cite{Freedman2000,Freedman2002}.

The novel feature of topological quantum computation is that, in principle, it is inherently robust against errors.  The computational states are degenerate ground states, which means that errors are suppressed naturally by the system.  So long as the anyons are kept far enough apart, no active error suppression is required.  Though promising, the existence and capability to produce anyons with the right properties is still largely speculative~\cite{DasSarma2005,Nayak2008,Lesanovsky2012a,Stern2013}. 

\chapter{Fault tolerance: making quantum computing error-free
\label{chap:fault}
}

The most straightforward use of quantum error-correcting codes is in transmitting quantum information over noisy channels.  In this case, the sender encodes his quantum state and sends it over the noisy channel to the receiver who then decodes.  Of course, in a realistic setting, errors can occur before the encoding process and after decoding, when the quantum information is unprotected.

In order to achieve reliable quantum computation, the data must be protected at all times.  In particular, unitary gates should be performed while the data is still encoded.  The typical procedure involves alternating rounds of encoded gates and error correction.  The encoded gate manipulates the data in the error correcting code, and error correction attempts to eliminate errors introduced by the encoded gate. See~\figref{fig:alternating-ec-ga}.

\begin{figure}
\centering
\includegraphics{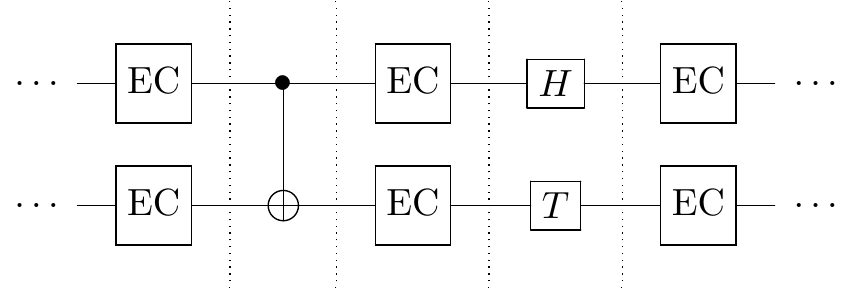}
\caption[Alternating gates and error correction.]{\label{fig:alternating-ec-ga} Typical fault-tolerant circuits are constructed by alternating rounds of error-correction with encoded gates.}
\end{figure}

The use of encoded gates alone is not enough.  Both the encoded gates and the error correction circuits should be \emph{fault tolerant}. Roughly, a quantum circuit is fault tolerant if the errors that occur during each step are small in number and can be kept well controlled.  Errors in a fault-tolerant circuit have very little chance of spreading or combining in order to cause data corruption.  In this chapter we make this concept precise, and examine techniques for constructing fault-tolerant quantum circuits.

\section{A brief history
\label{sec:fault.history}
}

% As with error-correcting codes, the concepts of fault tolerance for quantum computing have roots in classical computation. \todo{\cite{vonNeumann1956}}
% 
% \begin{itemize}
%   \item SAPO computer: early 1950's, used three redundant arithmetic units and majority voting. Citations? http://computer.org/computer-pioneers/svoboda.html
% \end{itemize}

Before delving into the details of fault-tolerant quantum computation, it is instructive to outline the path from its early beginnings to current state-of-the-art.  This history will show the successes and difficulties of the theory of fault-tolerant quantum computing, and provide motivation and context for the new results in subsequent chapters.

\subsection{Threshold proofs and calculations
\label{sec.fault.history.threshold}
}

The first proposal for fault tolerant quantum computation was posited by Shor in $1996$ \cite{Shor1997}.  Shor showed that his construction tolerates a noise rate that is logarithmic in the size of the computation (measured by the number of gates).  Roughly, Shor's error correction circuit contains a logarithmic number of gates, thus an error rate proportional to the inverse of that size is sufficient. 
Soon after, the first ``threshold theorems'' were proven independently by Aharonov and Ben-Or~\cite{Aharonov1996a}, Kitaev~\cite{Kitaev1997a}, and Knill, Laflamme and Zurek~\cite{Knill1996a} each of which permitted a constant noise rate per gate regardless of computation size.  Importantly, the amount of extra time and space resources required scales only as a polynomial in the logarithm of the computation size. 
\begin{theorem}[Constant noise threshold for quantum computation]
\label{thm:threshold}
Consider a quantum circuit $C$ of size $N$, a quantum computer with gates that fail independently with probability at most $p$, and target failure probability $\epsilon > 0$. There exists a different quantum circuit $C'$ of size at most
\begin{equation}
O\left(N\cdot \text{poly}\left(\log\frac{N}{\epsilon}\right)\right)
\end{equation}
that can be implemented on the quantum computer and simulates $C$ with probability of error at most $\epsilon$, provided that $p$ is below a constant \emph{threshold} value $p_{th}$.
\end{theorem}
The intuition is that an $[[n,k,d]]$ code yields encoded gates with a logical error rate at most $cp^{t+1}$, where $t = \lfloor (d-1)/2 \rfloor$, $c = \binom{A}{t+1}$ and $A$ is the number of physical gates contained in a single encoded gate plus error correction.  Concatenating the code with itself $j$ times requires $n^{j+1}$ qubits per block but an inductive argument yields a logical error rate per gate of
\begin{equation}
\label{eq:concatenated-logical-error-rate}
p_j \leq \frac{1}{c^{1/t}} \left(c^{1/t} p \right)^{(t+1)^{j+1}}
\enspace .
\end{equation}
That is, the size of the code scales exponentially, but so does the minimum distance.
The right-hand side of~\eqnref{eq:concatenated-logical-error-rate} converges so long as the physical error rate obeys
\begin{equation}
p < \frac{1}{c^{1/t}} = p_{\text{th}}
\enspace .
\end{equation}
Taking the logarithm of both sides of~\eqnref{eq:concatenated-logical-error-rate} twice, we see that achieving a target error rate per gate of $p_j \leq \epsilon/N$ only requires concatenation to level $j = O(\log\log N/\epsilon)$.  The total code size is then a polynomial in $\log N/\epsilon$.

Interestingly, the early threshold theorems hold only for quantum error-correcting codes of distance at least five.  Thresholds for distance-three codes were not known until $2006$, when they were discovered independently by Reichardt~\cite{Reichardt2005}, and Aliferis, Preskill and Gottesman~\cite{Aliferis2005}.  A novel fault-tolerance scheme using distance-two error-\emph{detecting} codes was proposed by Knill in $2005$, though without explicit proof of a threshold~\cite{Knill2004}. Rigorous proof of a threshold for distance-two schemes was proposed by Reichardt~\cite{Reichardt2007} (see also~\cite{Reichardt2006}), and later by Aliferis, Preskill and Gottesman~\cite{Aliferis2007c}.

Existence of a noise threshold permits arbitrary quantum computation for a constant amount of engineering cost per gate, at least in principle.  In practice, the value of the threshold matters since, while error rates near one percent are currently achievable in some small-scale experiments, e.g.,~\cite{Ladd2010,Monz11,Chow2012,Ghosh2013}, rates much lower than say $10^{-6}$ on a large-scale are perhaps impossible even in the long-term.  

The earliest estimate based on a rigorous threshold proof was calculated by Aharonov and Ben-Or to be an error rate per gate of about $10^{-6}$.  Later calculations based on~\cite{Reichardt2005} and~\cite{Aliferis2005} were similarly low at $6.75\times 10^{-6}$ and $2.73\times 10^{-5}$, respectively. Since then, rigorous threshold bounds have steadily improved.  As of $2011$, the highest lower bound was $1.25\times 10^{-3}$ by Aliferis and Preskill~\cite{Aliferis2009}.  In~\chapref{chap:threshold}, we adapt the technique of~\cite{Aliferis2005} to prove a threshold of $1.32\times 10^{-3}$.

Another popular technique is to \emph{estimate} the threshold using Monte Carlo simulation.  Threshold estimates, though not rigorous, paint a much more optimistic picture than lower bounds.  An initial estimate by Zalka placed the threshold at about $10^{-3}$~\cite{Zalk96}.  In $2004$ Knill estimated a threshold for his distance-two scheme as high as three percent.  Simulations for the surface code indicate a threshold of about one percent~\cite{Wang2011}. \figref{fig:threshold-plot} shows thresholds from a large number of studies and for a variety of error-correcting codes, noise models, and geometric constraints. 

\begin{figure}
\begin{minipage}{\linewidth}
\centering
\includegraphics[width=.8\textwidth]{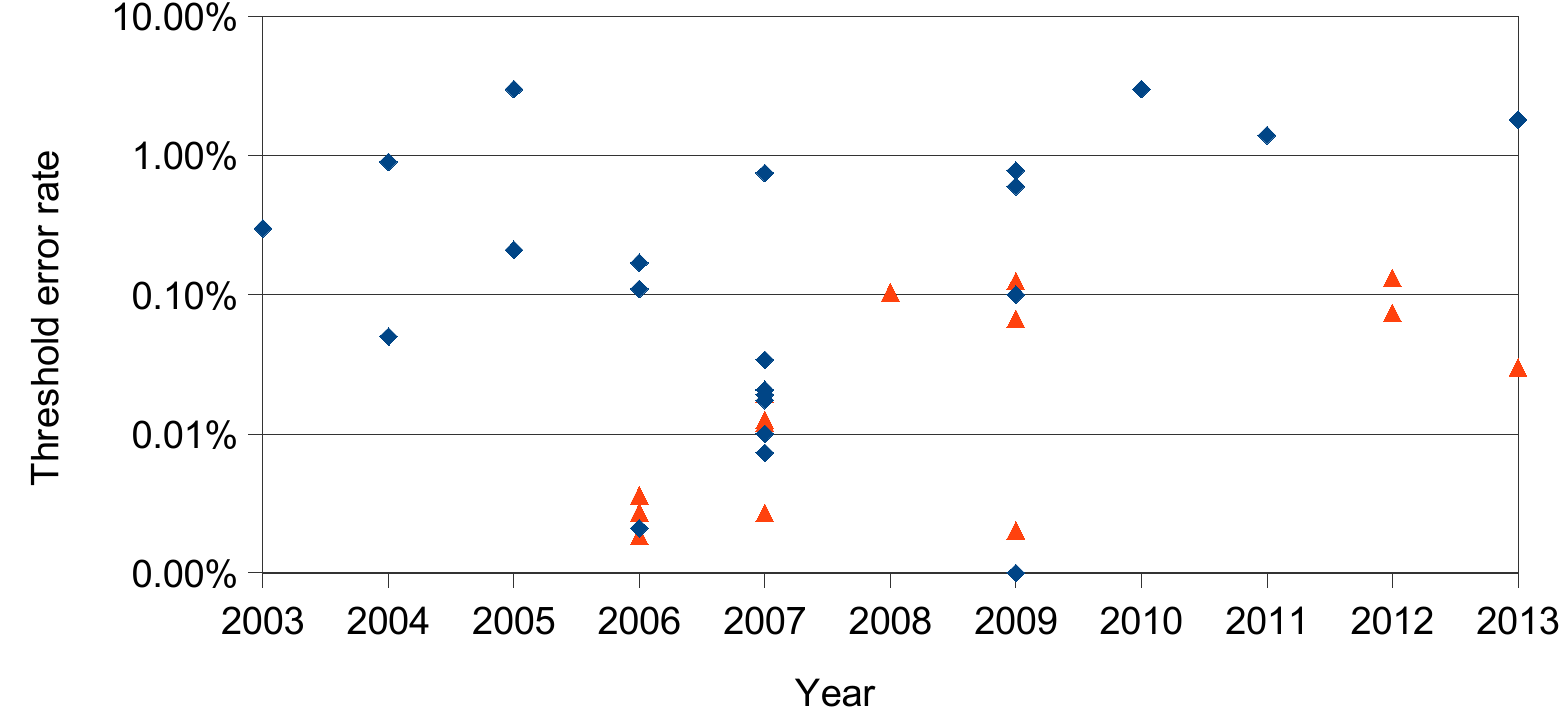}
\caption[Threshold bounds and estimates by year.]{\label{fig:threshold-plot}Threshold calculations since $2003$ arranged in chronological order by year.  ~Blue diamonds indicate estimates based on Monte Carlo sampling.\footnote{\cite{Steane2003,Reic04,Knill2004,Meto05,Szko04,Dawson2006,Svore2006b,Aliferis2007b,Raussendorf2007,Raussendorf2007a,Step09,Wang2009,Wang2009b,Fowler2009,Fujii2010a,Wang2011,Stephens2013}}  
Orange triangles indicate rigorous lower bounds (with varying assumptions).\footnote{\cite{Aliferis2005,Reichardt2005,Svore2006b,Aliferis2007c,Aliferis2007b,Spedalieri2009,Step07a,Aliferis2009,Paetznick2011,Fowler2012c,Lai2013a}}
}
\end{minipage}
\end{figure}

\subsection{Resource optimization
\label{sec:fault.history.resources}
}
Modern fault-tolerance schemes provide reasonable confidence that noise thresholds can be met using near-term technologies, at least for small numbers of qubits.  At the same time, the resource requirements for these schemes can be overwhelming. 
Knill, for example, estimates that his distance-two scheme would require a resource overhead ranging from one-thousand to one-billion fold, or more depending on computation size and gate error rate. Estimates for a cluster-state-based scheme due to Raussendorf, Harrington and Goyal are similarly large~\cite{Raussendorf2007a}.

Accordingly, the focus in quantum fault tolerance has shifted from threshold calculations to resource reduction and optimization.  In all schemes, particularly those based on concatenated codes, the dominant source of overhead is due to error correction.  Most encoded gates on an $n$-qubit code can be implemented using roughly $n$ gates.  Typical error correction procedures, meanwhile, require additional ancillary qubits and can require ten to one-hundred times as many gates~\cite{Shor1997,Steane1996,Knill2004}.

Steane has proposed an error correction method based on ancillary encoded stabilizer states~\cite{Steane1996} (see~\secref{sec:fault.qec.steane}), and in $2002$ showed a method for preparing such states fault-tolerantly~\cite{Stea02}. Steane's method uses a hierarchy of many encoded stabilizer states that can be used to verify the reliability of a single encoded state. Reichardt suggested a procedure for improving on Steane's method~\cite{Reichardt2006}, and in~\chapref{chap:ancilla} we examine additional improvements in detail.

Aliferis and Cross have demonstrated a fundamentally different approach to fault-tolerant error correction for the family of Bacon-Shor subsystem codes~\cite{Aliferis2007b}.  Their method eliminates the need for encoded ancillas and, instead, requires only nearest-neighbor two-qubit measurements which can be accomplished with just a single ``bare'' ancilla qubit. Similar bare-ancilla techniques are used for topological codes~\cite{Landahl2011,Fowler2012d}.

For many quantum error-correcting codes, Clifford operations can be implemented very efficiently.  In $2004$, Bravyi and Kitaev showed that universal fault-tolerant quantum computation is possible with only Clifford gates and special ``magic'' resource states~\cite{Bravyi2004}.  Specifically, a fault-tolerant $T$ gate can be obtained by progressively refining noisy magic states into fewer, but less noisy copies in a process known as \emph{state distillation}. See~\secref{sec:fault.gates.distillation}.

Unfortunately, state-distillation is usually very costly.  The cost of distilling a $T$ gate to fidelity $(1-\epsilon)$ scales as $O(\log^{2.47}(1/\epsilon))$, but again the numbers are large in absolute terms; usually thousands of magic states are required. Recently, though, a flurry of results have yielded significant improvements.  In $2012$, Meier, Eastin and Knill~\cite{Meier2012}, Bravyi and Haah~\cite{Bravyi2012a}, and Jones~\cite{Jones2012c} have each proposed new methods for $T$-gate distillation. The protocol of Jones comes arbitrarily close to $O(\log(1/\epsilon))$ in the number of magic states, and this is conjectured to be optimal. However, the total costs of the new protocols are more challenging to calculate, and so their practical benefits are less clear~\cite{Fowler2013,Jones2013b}.

Fowler and others have incorporated and optimized various distillation methods for use in the surface code~\cite{Fowler2012f,Fowler2013}, including a method for parallelization~\cite{Fowler2012g}.  Jones and Eastin have independently observed that distillation of so-called \emph{Toffoli states} can yield improvements compared to Toffoli gate constructions that use fault-tolerant Clifford and $T$ gates~\cite{Jones2012d,Eastin2012,Jones2013a,Jones2013b}.  In total, such optimizations can yield orders-of-magnitude improvements in the fault-tolerance resource overhead compared to naive methods~\cite{Jones2012d,Jones2013b}.

\subsection{Unitary decomposition
\label{fault.history.decomposition}
}
Fault-tolerance schemes provide universality through a small discrete set of encoded gates.  However, quantum algorithms are usually specified in terms of arbitrary unitaries.  Until recently, the standard method for decomposition into fault-tolerant gates has been the Solovay-Kitaev algorithm~\cite{Dawson2005}.  Once again, the decomposition cost of $O(\log^{3.97}(1/\epsilon))$ is asymptotically efficient, but often requires tens of thousands of fault-tolerant gates in absolute terms. 

In principle the decomposition cost is lower bounded by a more modest scaling of $O(\log(1/\epsilon))$~\cite{Kitaev1997a,Kita02}.
Fowler suggested an optimal approximation of single-qubit unitaries by optimized but exponential-time direct search~\cite{Fowl04c}.
In $2012$ Kliuchnikov, Maslov and Mosca (KMM) characterized the set of single-qubit unitaries that can be exactly decomposed with $\{\text{Clifford},T\}$ and gave an optimal and efficient algorithm for exact decomposition~\cite{Kliuchnikov2012}, and later an asymptotically optimal algorithm for approximate decomposition~\cite{Kliuchnikov2012a}.  Further improvements by Selinger~\cite{Selinger2012a}, and KMM~\cite{Kliuchnikov2012b} soon followed.

Several other methods for single-qubit unitary decomposition have been proposed.  One method involves preparing so-called \emph{Fourier states} and using phase kickback~\cite{Kita02}.  Using recent optimizations due to Jones~\cite{Jones2013}, this method is shown to be competitive with~\cite{Selinger2012} and~\cite{Kliuchnikov2012b} when using the surface code.  Bocharov and Svore have shown that decomposition into an alternative gate set $\{\text{Clifford},V_3=(I+2iZ)/\sqrt{5}\}$ can be up to six times better than~\cite{Kliuchnikov2012b}, but requires an implementation of $V_3$ which is more efficient than those currently known~\cite{Bocharov2013}.  In~\chapref{chap:repeat} we discuss a $V_3$ implementation that requires $5.26$ $T$ gates (in expectation), thus making~\cite{Bocharov2012} competitive with all of the methods above.  We also present a class of non-deterministic quantum circuits that can be used to approximate single-qubit unitaries for less than half the cost of existing methods.

\section{The noise threshold
\label{sec:fault.threshold}
}

Noise thresholds for quantum computation manifest in a variety of forms depending on physical noise and gate models, physical connectivity constraints, choice of error correcting code, method of error correction and the rigor with which the result is obtained.  In all cases though, the goal is the same: determine the conditions under which reliable large-scale implementation of a quantum algorithm is possible. We now discuss these various conditions, and outline techniques for calculating threshold values.

\subsection{Noise models
\label{sec:fault.threshold.noise}
}

In order for fault-tolerant techniques to be effective, the strength of the noise must be below a certain threshold value.  The way that strength is defined, and the methods for calculating the threshold depend on the way in which the noise is modeled.  Many different models can be considered and a broad categorization includes:
\begin{itemize}
  \item Stochastic - physical gates fail according to a probability distribution,
  \item Markovian - physical gates fail independently,
  \item Non-Markovian - gate failures may be correlated,
  \item Local - gate behavior is correlated to a constant number of other gates.
%  \item Unitary - errors come from a set of unitary operations,
%  \item Non-unitary - errors come from a set of general quantum operations.
\end{itemize}
Additional classifications are also possible.  For example, one can consider noise which acts unitarily only on the computer, and does not include the environment.

\subsubsection{Pauli and Clifford channels
\label{sec:fault.theshold.noise.pauli}
}
The simplest way to model noise is as a Pauli channel.  In this setting, each gate is specified by the ideal version of the unitary followed by either the identity, or some Pauli-group error according to a probability distribution.
The Pauli channel is an example of a stochastic and Markovian noise model in that errors occur independently at each gate according to a fixed probability distribution.  Specific cases include physically motivated noise such at the the depolarizing channel and the dephasing channel~\cite{Nielsen2000}. In the depolarizing channel, for example, a single-qubit gate may be followed by one of $\{X,Y,Z\}$ each with probability $p/3$, where the parameter $0 \leq p \leq 1$ specifies the strength of the noise.

Pauli channels can be generalized by enlarging the set of possible errors.  Clifford channels, for example, implement the ideal gate followed by an element of the Clifford group.  Indeed, since Clifford channels provide more parameters than Pauli channels, they can more closely characterize physical behavior in many cases~\cite{Magesan2012a}.
Pauli and Clifford channel noise models are useful because they allow for efficient Monte Carlo sampling and simulation. See~\secref{sec:fault.threshold.monte-carlo}.

\subsubsection{More general noise models
\label{sec:fault.threshold.noise.general}
}

Threshold calculations can be made for more general kinds of noise models, as well. Aliferis, Gottesman and Preskill (AGP)~\cite{Aliferis2005}, assume a local non-Markovian error model which is similar to a Pauli channel except that, when an error occurs, an adversary is allowed to choose the Pauli error. In this model, gates fail stochastically, but the adversary is allowed to coordinate the errors (in both time and space) among faulty gates in the circuit.

AGP also prove a threshold for a stronger non-stochastic model in which the behavior of a gate can depend on conditions of both the quantum computer and the environment at other points in space and time.  That is, gate failures are no longer independent but can be correlated by a kind of quantum memory.  Others have also considered non-stochastic models with varying restrictions on the type and strength of correlations~\cite{Terhal2005,Aharonov2006,Ng2009}.

Preskill has considered the most general noise model of all~\cite{Preskill2012a}.  In his model, the coupling between the environment and the computer is allowed to be completely arbitrary, assuming only that single qubits can be prepared with reasonable fidelity.  Preskill shows that a positive threshold exists so long as the strength of $k$-qubit interactions decays rapidly (i.e., exponentially) with $k$. 

% \begin{enumerate}
%   \item \cite{Terhal2005} - local bath (one per qubit), stochastic?, is this really less general than AGP?
%   \item \cite{Aliferis2005} - Two versions: adversarial stochastic, and coherent (with certain restrictions on the bath)
%   \item \cite{Ng2009} - coherent, uncoupled oscillators
%   \item \cite{Aharonov2006} - coherent, long range
%   \item \cite{Preskill2012a} - coherent, no bath restrictions!, only need preparation of single qubits with reasonable fidelity.  Threshold exists if $k$-qubit interactions decay sufficiently rapidly with $k$.
% \end{enumerate}

\subsection{Rectangles and gadgets
\label{sec:fault.threshold.rectangles}
}
Many threshold theorems consider fault-tolerant, noisy simulations constructed by compiling an ideal quantum circuit into a sequence of \emph{rectangles}, each of which contains an encoded operation ``gadget''~(Ga) and a trailing error correction gadget~(TEC). See~\figref{fig:exrec}.  The methods and notation here and in the remainder of the chapter follows~\cite{Aliferis2005}.
A gadget may contain many physical \emph{locations}, i.e., unitary gates and qubit preparations and measurements, each of which may be faulty (according to the prescribed noise model).  A gadget in which there are $n$ faulty locations is said to \emph{contain} $n$ faults.  

\begin{figure}
\centering
\includegraphics[scale=1.2]{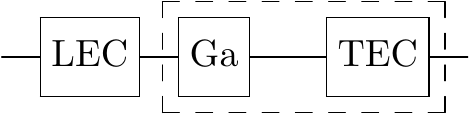}
\caption[An extended rectangle.]{\label{fig:exrec}
A rectangle, indicated here by the dotted line, includes a gate gadget (Ga) followed by a trailing error correction (TEC).  An extended rectangle (exRec) also includes the leading error correction (LEC).
}
\end{figure}

\begin{figure}
\centering
\includegraphics[height=1.3cm]{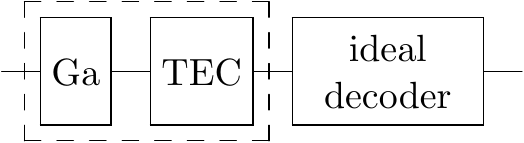}
\raisebox{.48cm}{\Large{~~$\equiv$~~}}
\raisebox{.15cm}{\includegraphics[height=1cm]{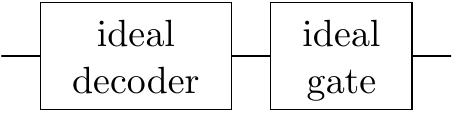}}
\caption[A correct rectangle.]{\label{fig:correct-rectangle}
A rectangle is correct if the rectangle followed by an ideal decoder is equivalent to an ideal decoder followed by the ideal gate.
}
\end{figure}

For simplicity, we will assume that data is encoded into a quantum error-correcting code that encodes just a single qubit.  That is, each logical qubit belongs to its own code block.  We will also assume that the same error-correcting code is used throughout.

A \emph{decoder} is gadget that maps an encoded logical state, possibly containing errors, to the corresponding single qubit state.  We can use the decoder gadget in order to reason about the relationship between a rectangle and the intended logical gate.
\begin{definition}[Rectangle correctness]
\label{def:rectangle-correctness}
A rectangle is \emph{correct} if the output of the rectangle followed by an ideal decoder (a decoder containing no faults) is equivalent to the output of an ideal decoder followed by an ideal implementation of the corresponding gate. See~\figref{fig:correct-rectangle}.
If a rectangle is not correct then it is \emph{incorrect}.
\end{definition}
In other words, a correct rectangle effectively acts as an encoded version of the intended gate.  If all rectangles are correct then a simple inductive argument shows that the compiled, noisy circuit successfully simulates the original ideal circuit.  By ``simulates'' we mean that the probability distribution obtained by measuring the outputs of ideal circuit is equivalent to the probability distribution obtained by measuring the outputs of the noisy fault-tolerant circuit.
We should emphasize here that the decoder gadget, ideal or otherwise, is conceptual only.  It is not actually used in the fault tolerant simulation.

For a fixed stochastic noise model and a fixed quantum error-correcting code, the probability that a rectangle is correct is a constant and therefore the probability that all rectangles are correct will generally be exponentially small in the number of gates in the circuit being simulated.  To achieve a constant success probability, code concatenation (see~\secref{sec:qec.qecc.concatenated}) is often used. In a concatenated fault tolerant simulation, each gate is first compiled into a rectangle, called a level-one rectangle ($1$-Rec), as described above.  Then, a level-two rectangle ($2$-Rec) is constructed by compiling each physical gate of the $1$-Rec into a rectangle.  This process is repeated as many times as desired, resulting in a circuit composed of a hierarchy of rectangles. 

\subsubsection{Strict fault tolerance}
\defref{def:rectangle-correctness} says nothing about the conditions under which we can expect the rectangle to be correct.  Of course, we should expect that a rectangle is correct when it contains zero faults.  It will be helpful to impose some additional constraints on each gadget, however. A gadget which satisfies these constraints will be called \emph{strictly fault tolerant}.

Informally, strict fault tolerance requires that a gadget must $1$) faithfully perform its encoded function (either correction of errors or data manipulation) and $2$) control the propagation of errors.  The roles of gate and error correction gadgets are distinct, and we define strict fault-tolerance separately for each.

In the definitions below, let $t = \lfloor (d-1)/2 \rfloor$, where $d$ is the minimum distance of the error-correcting code in use.
\begin{definition}[Strict fault tolerance: Ga]
\label{def:strict-fault-tolerance-gate}
Consider a Ga that contains $r$ faults and for which the input contains an error of weight $s$ such that $r+s\leq t$. Then the Ga is \emph{strictly fault tolerant} if and only if:
\begin{enumerate}
  \item the effect of perfectly decoding the output of the Ga is the same as first perfectly decoding the input to the Ga and then performing the corresponding ideal gate, and
  \item the weight of the error at the output of the Ga is at most $r+s$.
\end{enumerate}
\end{definition}

\begin{definition}[Strict fault tolerance: EC]
\label{def:strict-fault-tolerance-ec}
Similarly consider an EC that contains $r$ faults and has an input with a weight $s$ error. The EC is \emph{strictly fault tolerant} if and only if:
\begin{enumerate}
  \item for $r+s\leq t$ the state obtained by decoding the output of the EC is the same as the state obtained by \emph{removing} the EC and (ideally) decoding the input, and
  \item the output of the EC contains an error of weight at most $r$ for all $r\leq t$, regardless of $s$.
\end{enumerate}
\end{definition}

In the above definitions, the input $\ket\psi$ to the gadget is some quantum state on $n$ qubits.  The input is said to \emph{contain} an error of weight-$k$ if $\ket\psi$ is equal to a codeword multiplied by some Pauli error of weight $k$, modulo the stabilizers and the logical operators.

\subsubsection{The extended rectangle}
Rectangles do not overlap, but the output of a rectangle is the input of another and so rectangles do not act independently when errors occur.  An error on the output of one rectangle could combine with an error that occurs in the subsequent rectangle to cause a logical error.
In order to circumvent this problem, the preceding (or leading) error correction gadget~(LEC) of a rectangle can be included to form an \emph{extended rectangle} (exRec). ExRecs \emph{do} overlap, but under certain reasonable assumptions, the behavior of an exRec is independent of the errors on its inputs.

In particular, if the correction applied by the LEC is deterministic for all possible input errors, then it can be shown that the syndrome on the output of the LEC is independent of the input~\cite{Cross2009}.  The correctness of the enclosed rectangle, therefore, can be determined by analyzing the exRec in isolation.  This observation is a key element of the malignant set counting technique discussed in~\secref{sec:fault.threshold.malignant}.

\subsection{Level reduction
\label{sec:fault.threshold.level-reduction}
}
If all rectangles at all levels of concatenation are correct, then the fault tolerant simulation reproduces the results of the corresponding ideal quantum circuit.  \emph{Level reduction} is a conceptual technique for coping with incorrect rectangles in order to maintain a faithful simulation result. The idea of level reduction is to incrementally replace each rectangle at the lowest level with either an ideal location (when the rectangle is correct), or a faulty location (when the rectangle is incorrect).  Repeating the process for each level of concatenation yields a quantum circuit that directly reflects the original circuit, and hopefully contains no faulty locations.

Level reduction begins by placing ideal decoders at the outputs of the rightmost $1$-Recs.  If a $1$-Rec is correct, then by definition, the behavior of the simulation is unchanged by moving the decoder to the left and replacing the rectangle with the corresponding ideal location.  If a $1$-Rec is incorrect, however, then the decoder is stuck and cannot be moved to the left.  Instead, an ideal decoder-encoder pair is placed to the left of the LEC of the corresponding $1$-\emph{exRec}.  The result is a $1$-exRec flanked by an ideal encoder and decoder that can be represented at level-two by a faulty location. See~\figref{fig:incorrect-exrec}.

\begin{figure}
\centering
\includegraphics{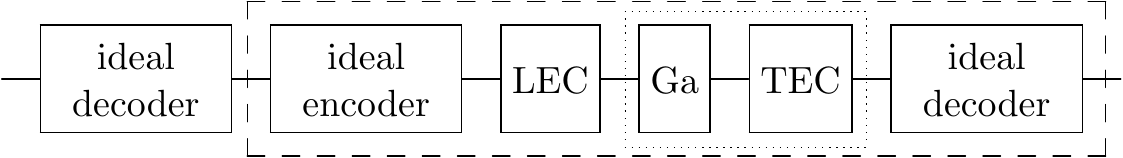}
\caption[Level-reduction on an incorrect rectangle.]{\label{fig:incorrect-exrec}
If a rectangle (indicated by the dotted line) is found to be incorrect, then the ideal decoder cannot be moved through to the left.  Instead, an ideal decoder-encoder pair is placed to the left of the LEC so that the entire exRec is flanked.  The new ideal decoder can now proceed to the left as normal.  The encoder-exRec-decoder sequence in the dashed box is replaced by a single faulty location in the next level of concatenation.
}
\end{figure}

By repeating the process for each $1$-Rec, the ideal decoders gradually sweep from right to left, across the entire simulation.  A decoder is free to move to the left until encountering an incorrect $1$-Rec at which point a new decoder-encoder pair is created to take its place.  The result is a level-$(k-1)$ simulation with faulty locations at previously incorrect $1$-Recs.
In this way, level reduction allows the level-$(k+1)$ analysis to proceed by treating each $k$-Rec as a single \emph{independent} location.  The probability that a ``location'' fails in the level-$(k+1)$ simulation is upper bounded by the probability that the corresponding $k$-Rec is incorrect.

As a concrete example, consider the circuit shown in~\figref{fig:overlapping-exrecs-with-decoders}.  The level-reduction procedure proceeds as follows.
\begin{enumerate}
  \item Examine exRec $2$.  If the enclosed rectangle is incorrect then replace the entire \emph{exRec} with a faulty version of the associated (level-zero) gate.  Otherwise, replace the \emph{rectangle} with an ideal version of the associated gate.  
  \item Examine exRec $3$.  Follow the same procedure as for exRec $2$.  
  \item Examine exRec $1$.  Depending on the outcomes of exRec $2$ and exRec $3$, one or both of the TECs may have been removed.  The enclosed rectangle now consists of the encoded CNOT and any remaining TECs.  If the remains of rectangle $1$ are incorrect, exRec $1$ is replaced with a faulty level-zero gate.  Otherwise, the rectangle is replaced with an ideal level-zero gate.  
\end{enumerate}

\begin{figure}
\centering
\includegraphics[scale=1]{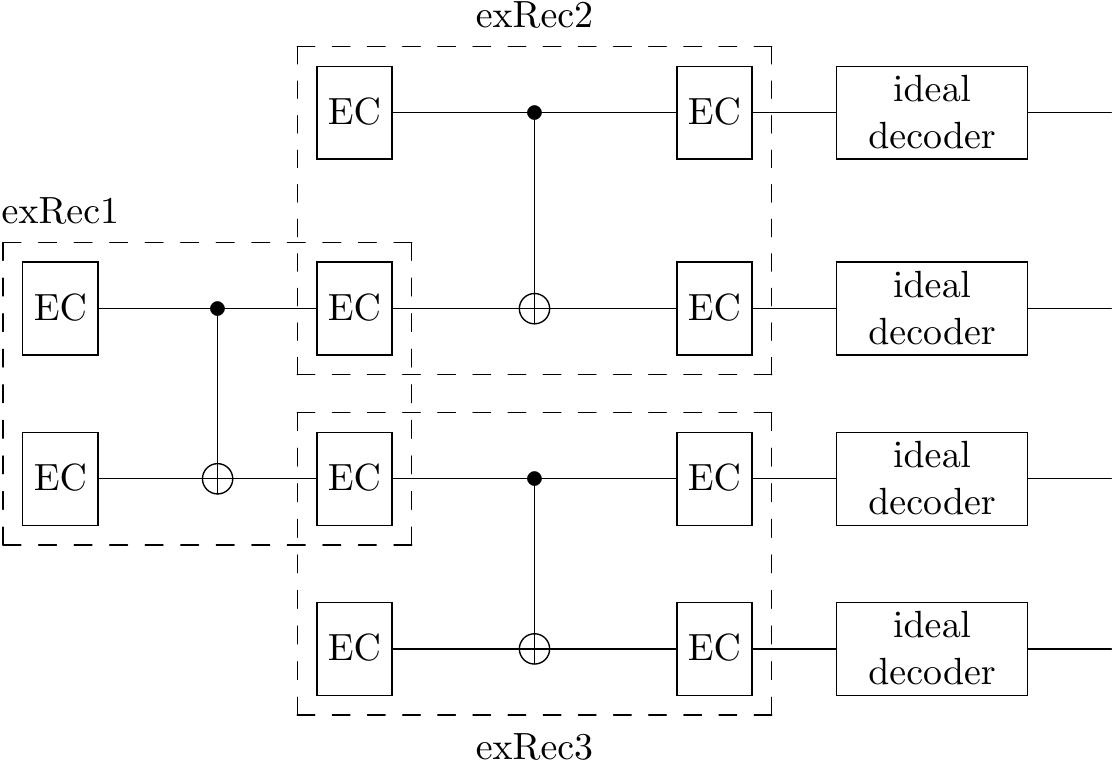}
\caption[Overlapping exRecs.]{An example of a fault tolerant simulation with three overlapping exRecs.  Level-reduction starts by sweeping the decoders back through exRec 2 and exRec 3, and then moving on to exRec 1.}
\label{fig:overlapping-exrecs-with-decoders}
\end{figure}

There are two technicalities in the level-reduction process that must be addressed.  First, when an incorrect rectangle is encountered, the newly created decoder-encoder splits the preceding rectangle, effectively removing the TEC from the rectangle.  This problem can be readily fixed by defining correctness for the partial rectangle in a straightforward way.  Second, flanking the incorrect exRec with an encoder and decoder allows us to treat it as a faulty location at level-two.  But the definition of incorrectness is insufficient to identify which error actually occurred.  AGP solve this problem by using an adversarial noise model in which the worst-case error is always assumed.  In~\chapref{chap:threshold}, we will see that other noise models can be accommodated by more carefully characterizing correctness and incorrectness.

\subsection{Malignant set counting
\label{sec:fault.threshold.malignant}
}
At each level $k$ of concatenation, the probability that the $k$-Rec is correct increases relative to level $k-1$ so long as the strength of the noise is below a certain value, i.e., the threshold. The threshold is calculated by upper bounding the probability that each type of rectangle is incorrect. But, as discussed above, rectangle behavior is dependent on its inputs. AGP therefore obtain an upper bound by instead analyzing the exRec.

Consider a code that corrects errors up to weight $t$, and assume that the gadgets in the exRec are strictly fault tolerant.  Then the enclosed rectangle is guaranteed to be correct if it contains no more than $t$ faults, and the probability of incorrectness $p_1$ can be naively upper-bounded as
\begin{equation}
\label{eq:p1-naive}
p_1 \leq \binom{n}{t+1}p^{t+1}
\enspace,
\end{equation}
where $n$ is the number of locations in the exRec and $p$ is an upper bound on the probability that a location is faulty.  An inductive argument shows that the threshold is then lower bounded by
\begin{equation}
\label{eq:pth-naive}
p_\text{th} \geq \binom{n}{t+1}^{-1/t}
\enspace .
\end{equation}

Equation~\eqnref{eq:p1-naive} (and therefore~\eqnref{eq:pth-naive}) can be improved by noting that, though the code can only correct errors up to weight $t$, an exRec that contains more than $t$ faults need not be incorrect. Say, for example, that two faults occur, one in the LEC and one in the TEC and that the code can correct a single error---i.e., $t=1$.  If the TEC fault occurs early on, then it is likely that the two faults combine to cause an uncorrectable error.  But if the TEC fault occurs \emph{after} the error from the LEC has been corrected, then the rectangle will still be correct.  

Malignant set counting is the process of enumerating subsets of faulty locations in the exRec, and counting only those that can actually cause incorrectness. A set of locations is considered \emph{malignant} if there exists some fixed combination of nontrivial Pauli errors acting on that set of locations that causes the enclosed rectangle to be incorrect.  Let $M_k$ be the number of malignant sets of size $k$.  Then by counting all of the malignant sets of size at most $K$, we may use the bound
\begin{equation}
p_1 \leq \binom{n}{K+1}p^{K+1} + \sum_{k=t+1}^{K} M_k p^{k}
\enspace,
\end{equation}
which can be substantially better than~\eqnref{eq:p1-naive}.

Malignant set counting is both conceptually simple, and highly flexible.  As a concrete example, AGP used malignant set counting to prove a threshold of $2.73\times 10^{-5}$ for a deterministic scheme based on the $[[7,1,3]]$ code.  Later, they extended the technique to accommodate non-deterministic gadgets~\cite{Aliferis2009}.  Malignant set counting can also be used to analyze schemes that are geometrically constrained~\cite{Svore2006b,Lai2013a}.  Furthermore, since malignant set counting yields concrete polynomials, it is easy to calculate effective noise strengths, even for very low physical error rates.

A significant drawback of malignant set counting is that the number of possible subsets grows exponentially with $K$.  It is usually feasible to count subsets only up to some small fixed size.  In~\chapref{chap:threshold} we discuss a solution that eliminates many subsets of locations which are unlikely to be simultaneously faulty, thereby permitting much larger values of $K$. 

\subsection{Alternative proof techniques
\label{sec:fault.threshold.alternatives}
}
Malignant set counting and related techniques are effective for proving threshold lower bounds for schemes based on concatenated codes. 
For other codes, and for topological codes in particular, the arguments made by level-reduction no longer apply, since there are no ``levels'' so-to-speak.  Alternative proof techniques are available, however.  One popular method is to map errors in topological codes onto models based on statistical physics~\cite{Dennis2001a,Harrington2004}. With these models, it is possible to prove thresholds in the range $1$-$10$ percent.  However, these high thresholds are obtained by assuming the ability to measure stabilizer generators without creating correlated errors and the ability to classically compute corrections based global information about the syndromes.  Recently, Fowler used a combinatorial argument to prove a lower bound of $7.4\times 10^{-4}$ for the surface code~\cite{Fowler2012c}.  His model includes explicit circuits used to measure stabilizer generators (so that measurements can introduce correlations) and requires only locally-bounded classical syndrome processing. 

\subsection{Monte Carlo simulation
\label{sec:fault.threshold.monte-carlo}
}

An alternative solution to the malignant set counting complexity problem is to randomly sample rather than exhaust over all possible subsets.  Any stochastic error model induces a probability distribution of faulty locations, which can be sampled using the Monte Carlo method.  The result is an estimate of the threshold to within some statistical confidence interval. Aliferis and Cross have used this technique to calculate thresholds for a variety of codes~\cite{Aliferis2007b}.
Steane~\cite{Steane2003} and Knill~\cite{Knill2004} have used Monte Carlo sampling to directly simulate depolarizing noise on sequences of rectangles and calculate the probability of correctness.
Svore and others have used a more limited simulation of a single level-one exRec~\cite{Svore2005,Svore2006a,Cross2009}.  Their simulations yield a value called the \emph{pseudo-threshold}, which is a rough estimate of the threshold rather than a statistical bound, but is easier to calculate.

Monte Carlo simulation has been used extensively to estimate thresholds for schemes based on topological error-correcting codes, which do not conform to the rectangle and gadget paradigm outlined in~\secref{sec:fault.threshold.rectangles}~\cite{RAUSSENDORF2006,Raussendorf2007,Raussendorf2007a,Fowler2009,Wang2009b,Fujii2010a,Wang2011,Stephens2013}.  In these cases, a small patch of the code is simulated many times over a range of physical error rates and for progressively larger code distances.  Plotting the results by code distance yields a ``waterfall'' shape in which the intersection of the curves converges to a point which is deemed the threshold. \figref{fig:waterfall} shows simulation results for the surface code~\cite{Fowler2013d}.

\begin{figure}
\centering
\includegraphics[width=.7\textwidth]{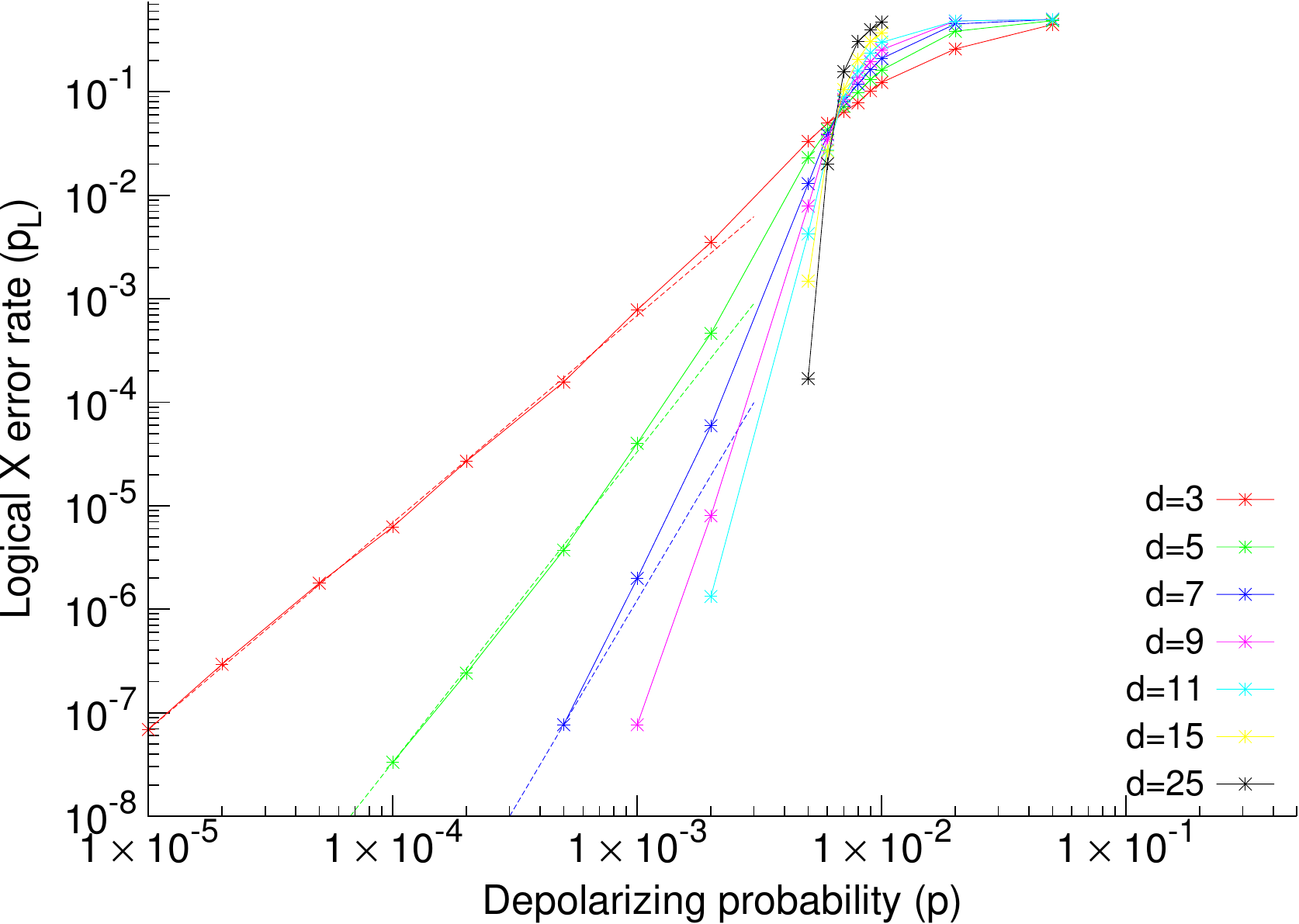}
\caption[Monte Carlo waterfall plot for the surface code.]{\label{fig:waterfall}
An example of Monte Carlo simulation for the surface code.  A patch of the surface code is simulated for a variety of depolarizing noise strengths $p$ and code distances $d$.  The threshold corresponds to the intersection point. Reproduced, with permission, from~\cite{Fowler2013d}.
}
\end{figure}

\subsection{Limitations
\label{sec:fault.threshold.limitations}
}

Threshold theorems describe the circumstances for which reliable quantum computation is possible.  Under which circumstances is quantum computation \emph{not} possible?  In other words, what are the upper bounds on the noise threshold?

Harrow and Nielsen showed that two-qubit gates are incapable of generating entanglement when subject to depolarizing noise with strength $0.74$, or $0.50$ for more general noise~\cite{Harrow2003}.  This result for depolarizing noise was sharpened to $0.67$ by~\cite{Virmani2004}.  Another way to upper bound the threshold is to allow perfect stabilizer operations and then determine the noise rate at which $\{\Clifford,T\}$ circuits can be simulated classically. (Recall from~\secref{sec:mechanics.universality.clifford} that Clifford circuits can be simulated classically.)  Using this technique~\cite{Virmani2004} show that classical simulation is possible for dephasing noise with strength $0.3$, or about $0.15$ for worst-case noise.  This result was extended to depolarizing noise with strength $0.45$ by~\cite{Buhrman2006b}.  Both~\cite{Virmani2004} and~\cite{Buhrman2006b} were later shown to be tight in the sense that magic state distillation (see~\secref{sec:fault.gates.distillation}) permits universal quantum computation if the noise strength on the $T$ gate is below the bound~\cite{Reic05}. More recently, it was shown that the results are tight for all single-qubit non-Clifford gates, not just $T$~\cite{VanDam2009}.

On the other hand, \cite{Kempe2008} have considered the case in which single-qubit gates are perfect, but $k$-qubit gates are subject to depolarizing noise.  For the case $k=2$, they show that the output of the circuit is independent of the input when the noise strength is $0.357$.
Plenio and Virmani give perhaps the most comprehensive set of upper bounds, using both noisy Clifford and non-Clifford operations for a variety of noise models and schemes~\cite{Plenio2010}.  In particular, they give a depolarizing noise upper bound of $0.26$ without restriction on the protocol for non-Clifford gates.  The bounds in other cases are as low as $0.03$.

The upper bounds above assume limits on the correlations present in the noise.  Depolarizing noise, for example, assumes that errors on distinct gates are independent.  We may also ask what kinds of correlations can be tolerated. Preskill's result shows that correlations that decay exponentially with the number of qubits can be tolerated~\cite{Preskill2012a}.  However, the correlations cannot be unlimited. Ben-Aroya and Ta-Shma show that controlled-phase flips cannot be corrected, even approximately~\cite{Ben-Aroya2009}.  Kalai has speculated that the types of errors afflicting highly entangled codewords will be strongly correlated across large numbers of qubits~\cite{Kalai2011}.  Such strong correlations, if they exist, would quickly thwart fault-tolerance schemes.  Experimental demonstrations of quantum error-correcting codes (see~\secref{sec:qec.experiment}) suggest, however, that these speculations are overly pessimistic.
A more complete discussion of the impact of correlated noise can be found in~\cite{Staudt} and~\cite{Preskill2012a}.

\section{Encoded computing
\label{sec:fault.gates}
}
Fundamentally, fault tolerance is the practice of simulating an ideal computation by carefully manipulating encoded data.  In particular, we should like that encoded operations meet the conditions given by~\defref{def:strict-fault-tolerance-gate}, namely that they faithfully execute the intended logical operations, and that they prevent errors from spreading among physical qubits.  

The logical operations permitted within a fixed error-correcting code are limited, however.  For a stabilizer code, the set of available operations corresponds exactly with the normalizer, i.e., operators that commute with all of the stabilizers~\secref{sec:qec.qecc.stabilizer}.  Most unitary operations that can be performed on a code block do not actually realize a unitary operation on the encoded qubits.

Furthermore, proposals for quantum computing architectures usually provide a small set of physical one- and two-qubit operations (see, e.g.,~\cite{Ladd2010}).
Translating these physical gates into an arbitrary fault-tolerant encoded operation is a challenging task and is not possible to do exactly in general.  Instead, fault-tolerance schemes provide a small set of universal logical operations constructed from elementary physical gates.
One popular choice of universal gate set is $\{H,T,\text{CNOT}\}$, though there are others. In~\chapref{chap:transversal} we will use $\{H,\CCZ\}$, and in~\chapref{chap:repeat} we will discuss another alternative gate set.
  
\subsection{Transversal gates
\label{sec:fault.gates.transversal}
}

The simplest and most well-behaved class of encoded operations is called \emph{transversal}.  A circuit is transversal if each physical gate acts on at most one qubit in the encoded block.  In the case of multi-qubit operations, the circuit is transversal if each gate acts on at most one qubit in each of the encoded blocks, and no qubit is involved in more than one gate. See~\figref{fig:transversal-gate}.

\begin{figure}
\centering
\includegraphics[height=4cm]{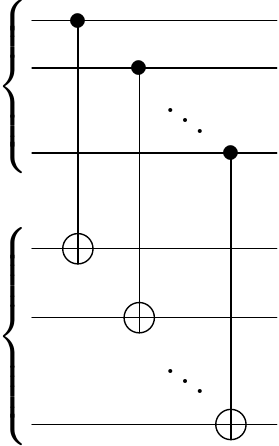}
\caption[A transversal CNOT.]{\label{fig:transversal-gate} 
Transversal implementation of an encoded CNOT.  Each gate touches exactly one qubit per block, and no qubit is involved in more than one gate.}
\end{figure}

Transversal circuits are automatically (strictly) fault tolerant.  A single faulty gate can produce only a single error on a given block.  Thus the maximum weight of an error on any block after application of a transversal circuit is at most $r + s$, where $r$ is the maximum weight of an existing error on any block and $s$ is the number of faulty gates.

The set of encoded gates that can be implemented transversally depends on the error-correcting code. The single-qubit Pauli operators are transversal for any stabilizer code, and CNOT is transversal for any CSS code (a consequence of independent $X$ and $Z$ stabilizers). Specific codes may admit transversal implementations of other operations.  The $[[7,1,3]]$ code, for example, admits transversal implementation of $H$ and $S$, in addition to $\{X,Y,Z,\text{CNOT}\}$.

As noted in~\secref{sec:mechanics.universality.clifford}, the gate set $\{H,S,\text{CNOT}\}$ generates the Clifford group which, though useful, is insufficient for universal quantum computation. Indeed, no quantum error-correcting (or error-detecting) code admits transversal implementation of a universal set of gates~\cite{Eastin2009a}.  In \chapref{chap:transversal}, however, we will see a scheme which effectively circumvents this limitation by incorporating error correction.

\subsection{State distillation
\label{sec:fault.gates.distillation}
}

Given the Clifford group, universality can be achieved by adding a single non-Clifford gate (see, e.g., ~\cite{Campbell2012} Appendix D).  Fault-tolerant implementation of the non-Clifford gate is usually accomplished by preparing many noisy copies of a special resource state, and ``distilling'' them into a single high-fidelity copy. The high-fidelity state can then be used to effect the desired gate using gate-teleportation (\secref{sec:mechanics.teleportation}).  

Importantly, the distillation and gate teleportation circuits for certain resource states can be accomplished using only Clifford gates and $Z$-basis measurement. For example, the state $\ket{A} = T\ket{+}$ can be distilled and teleported to implement $T$ using only CNOT, $H$ and $S$~\cite{Bravyi2004}. See~\figref{fig:tgate-implementation}.  

\begin{figure}
\centering
\begin{subfigure}[b]{\textwidth}
\centering
\includegraphics[width=.35\textwidth]{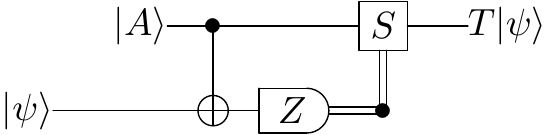}
\caption{\label{fig:T-gate-teleportation}}
\end{subfigure}
\begin{subfigure}[b]{.45\textwidth}
\centering
\includegraphics[width=.95\textwidth]{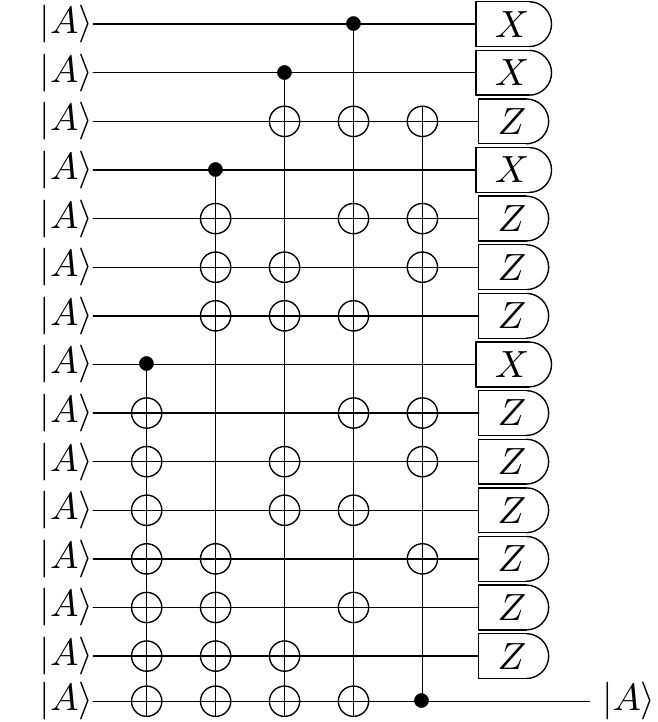}
\caption{}
\end{subfigure}
\hfill
\begin{subfigure}[b]{.45\textwidth}
\includegraphics[width=\textwidth]{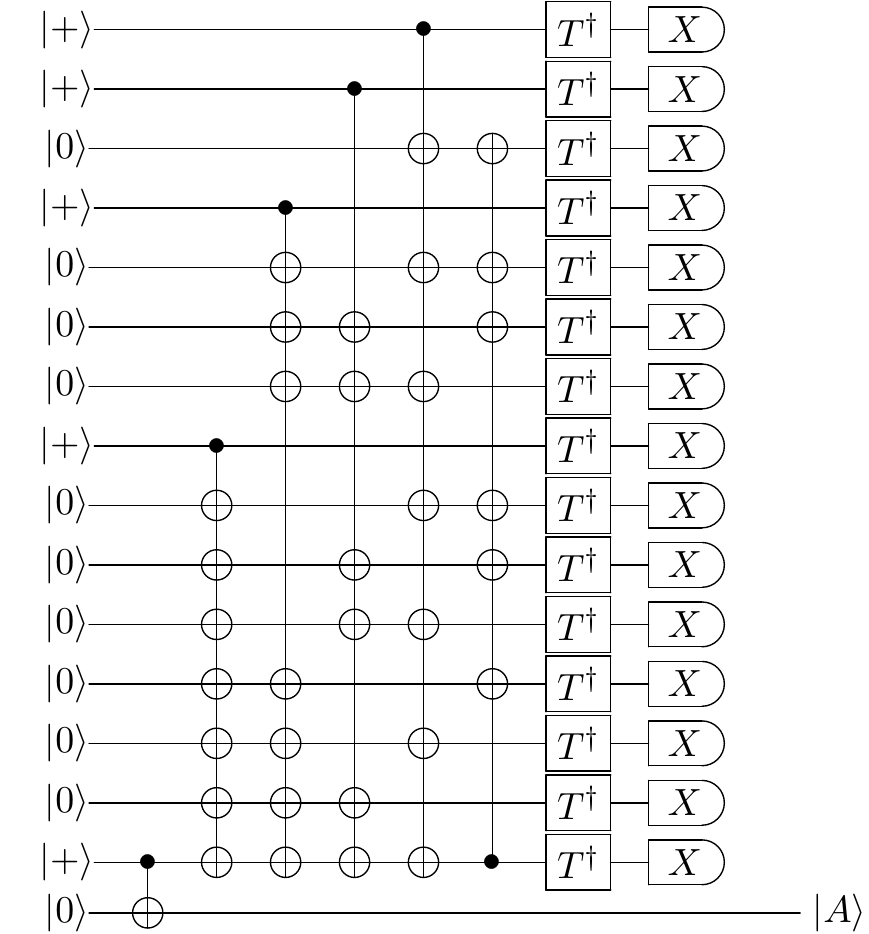}
\caption{\label{fig:T-distillation-Bell-pair}}
\end{subfigure}
\caption[State distillation for the $T$ gate.]{\label{fig:tgate-implementation} 
State distillation of the $T$ gate.
(a) For many quantum error-correcting codes, the $T$ gate is implemented by preparing the resource state $\ket A = \frac{1}{\sqrt 2}(\ket 0 + e^{i\pi/4}\ket 1)$ and using gate-teleportation. Conditioned on the measurement outcome, an $S$ correction may be required. (b)  $15$ noisy $\ket A$ states can be used to prepare a single high-fidelity $\ket A$ state, conditioned on a $+1$ outcome for each measurement~\cite{Bravyi2004}. The circuit is based on the decoding circuit for the $[[15,1,3]]$ code.
(c) Alternatively, transversal $T$ may be applied to one half of a Bell-pair that is encoded into the $[[15,1,3]]$ code. A logical $X$-basis measurement then teleports the $T$ gate onto the other half of the Bell-pair, again conditioned on $+1$ results for each measurement~\cite{Raussendorf2007a}.  An abstract version of this circuit is shown in~\figref{fig:distillation-by-gate-teleportation}.
}
\end{figure}

The distillation circuit shown in~\figref{fig:T-distillation-Bell-pair} can be understood as a novel kind of gate teleportation circuit.  Consider the circuit in~\figref{fig:distillation-by-gate-teleportation}.  This circuit implements teleportation of the state $T\ket+$.  In this case, however, the $T$ gate has been commuted to the right side of the CNOT (rather than the left).  Before performing the $T$ gate, the top ancilla is encoded into a quantum error correcting code that supports transversal (or otherwise robust) implementation of $T$. In this case we use the $[[15,1,3]]$ code which is based on the $15$-bit Hamming code and supports transversal $T$~\cite{Knill1996a}. An encoded $X$-basis measurement then completes the circuit.  Usually, the entire circuit is already encoded in the code that we are using to implement Clifford gates.  The encoding gate in~\figref{fig:distillation-by-gate-teleportation} then concatenates this ``base'' code with the $[[15,1,3]]$ code for the purpose of robustly implementing the $T$ gate.

\begin{figure}
\centering
\includegraphics{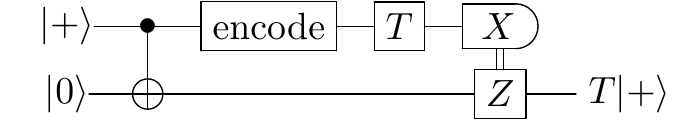}
\caption[Distillation by encoded gate teleportation.]{\label{fig:distillation-by-gate-teleportation}
This circuit outputs $\ket{A} = T\ket+$ by gate teleportation.  Before performing the $T$ gate, the top qubit is encoded into an error correcting code.  The $T$ gate and $X$-basis measurements are performed logically on the code.
}
\end{figure}

Variations on~\figref{fig:tgate-implementation} and~\figref{fig:distillation-by-gate-teleportation} also work.
Recently, Bravyi and Haah showed how to construct a wide class of quantum codes that admit efficient implementation of the encoded $T$ gate~\cite{Bravyi2012a} and can similarly be used for distillation.  Others have developed protocols based on codes that admit transversal Hadamard~\cite{Meier2012,Jones2012c}.
Toffoli gates can be obtained using a similar distillation and teleportation procedure~\cite{Eastin2012,Jones2012d}.

Early proposals for fault tolerant implementations of non-Clifford gates differed somewhat from the protocol described above.  Shor proposed a procedure for implementing the Toffoli gate based on fault-tolerant construction of a cat state plus other transversal gates~\cite{Shor1997}.  Knill, Laflamme and Zurek proposed the use of the $[[15,1,3]]$ code, for which $T$ is transversal, but $H$ is not~\cite{Knill1996a}.  They construct fault-tolerant $H$ using preparation of an encoded $\ket +$ state and a teleportation-like circuit.  These methods are further discussed in~\chapref{chap:transversal}.

Topological codes offer a qualitatively different way to perform fault-tolerant encoded gates.  Many topological codes are also stabilizer codes, and for those codes the same concept of transversality still applies.  However, it can be more productive to implement encoded gates by instead deforming the surface on which the code is supported.  In the surface code, for example, encoded qubits are defined by introducing logical ``defects'' into the lattice of physical qubits.  Encoded gates are then performed by moving defects around each other and fault-tolerance is ensured by keeping the defects sufficiently far apart (see, e.g.,~\cite{Fowler2012d}).  Code deformation is not universal on its own, though.  State distillation is typically used for topological codes, as well.

\section{Robust error correction
\label{sec:fault.qec}
}
Fault-tolerant encoded gates are carefully designed to prevent errors from spreading between qubits.  Even so, errors must be periodically identified and flushed away by measuring error syndromes and making corrections.  There is a very simple circuit that measures the error syndrome. \figref{fig:naive-syndrome-measurement} shows an example for a weight-four stabilizer.  However, this circuit is not fault tolerant.  An error on the ancilla qubit can spread to many of the data qubits, possibly causing a logical error. More complicated error correction circuits are usually required in order to limit the spread of errors.

\begin{figure}
\centering
\includegraphics{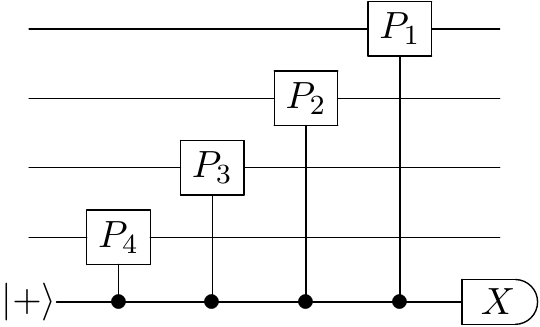}
\caption[Naive syndrome measurement.]{\label{fig:naive-syndrome-measurement}
This circuit measures the the four-qubit stabilizer $P_1\otimes P_2\otimes P_3\otimes P_4$ where the $P_i$ are Pauli operators.  Each Pauli operator is applied, controlled on the ancilla qubit $\ket +$.  The measurement outcome corresponds to the eigenvalue of the stabilizer.  This circuit is not fault tolerant since an error on the ancilla can spread to the other qubits through the controlled-$P_i$ gates.
}
\end{figure}

A variety of error-correction techniques have been studied, and three broad categories are so-called Shor-type~\cite{Shor1997}, Steane-type~\cite{Steane1996} and Knill-type~\cite{Knill2004a} error correction.  This is only a rough categorization, and it leaves significant room for introducing new ideas and optimization within or beyond these categories; see, e.g.,~\cite{Reic04,DiVi06,Aliferis2007b}. 

Common to each of these types of error correction is the use of ancillary qubits to extract error information from the data blocks.  Before interacting with the data, the ancilla qubits need to be prepared in an entangled state.  Error information is transferred by coupling this state with the data.  Finally, measurements are used to obtain syndrome information.  The methods differ mainly in the type of entangled states that are required.

\subsection{Steane error correction
\label{sec:fault.qec.steane}
}   
Steane-type error correction is based on the circuit shown in~\figref{fig:steane-ec}.  $X$ errors are corrected by preparing an encoded $\ket +$ state, performing CNOT from the data to the ancilla and then measuring the ancilla in the $Z$ basis.  $Z$ errors are independently corrected by instead preparing encoded $\ket 0$, performing CNOT from the ancilla to the data, and measuring the ancilla in the $X$ basis.  Note that, under ideal conditions, neither of the circuits have any effect on the encoded data. The state $\ket +$ is the $+1$-eigenstate of $X$ and is therefore invariant as the target of a CNOT.  Likewise, a CNOT does not activate when its control qubit is in state $\ket 0$.

Steane error correction requires that $X$ and $Z$ errors can be corrected independently, and therefore applies only to CSS codes (see~\secref{sec:qec.qecc.css}).  Transversal measurement of the ancilla effectively measures all of the stabilizer generators of a particular type (either $X$ or $Z$) in parallel. Thus, it is typically more efficient for large codes than Shor-type error correction, which measures each generator individually.  Its conceptual simplicity has also made it a popular choice for threshold studies, e.g.,~\cite{Aliferis2005,Reichardt2005,Cross2009}.

The drawback of Steane error correction is that preparation of sufficiently robust encoded $\ket 0$ and $\ket +$ states can be complicated.  Systematic techniques for preparing such encoded stabilizer states exist~\cite{Stea02,Patel2003}, but errors can occur during preparation.  The ancilla state must therefore be ``verified'' before being coupled to the data~\cite{Stea02,Reic04,Reichardt2006}.  Several techniques for improving stabilizer state preparation and verification are discussed in detail in~\chapref{chap:ancilla}.

\begin{figure}
\centering
\includegraphics{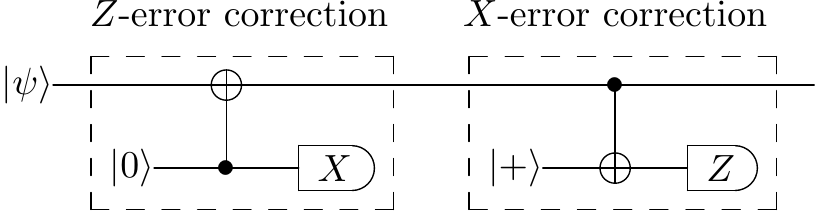}
\caption[Steane-style error correction.]{\label{fig:steane-ec} In Steane-style error-correction, $Z$ and $X$ errors are corrected separately. $Z$ errors are corrected by preparing an encoded $\ket 0$, performing transversal CNOT, and then transversally measuring in the $X$-basis.  Similarly, $X$ errors are corrected by preparing encoded $\ket +$, performing transversal CNOT, and measuring transversally in the $Z$-basis.} 
\end{figure}

\subsection{Knill error correction
\label{sec:fault.qec.knill}
}
Knill-type error correction, like Steane-type, uses encoded ancillary states.  In this case, however, the required states are more complicated.  Knill error correction is based on gate teleportation. See~\figref{fig:knill-ec}.  First an ancillary Bell state is prepared, followed by application of the desired unitary. Then a Bell measurement serves simultaneously to teleport the data, and measure the error syndrome.

\begin{figure}
\centering
\includegraphics{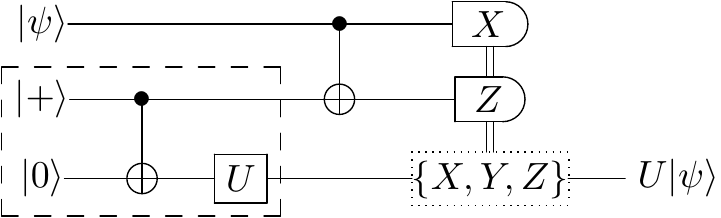}
\caption[Knill-style error correction.]{\label{fig:knill-ec}
In Knill-style error correction, syndrome measurements and an encoded gate are accomplished simultaneously.  The circuit above implements the single-qubit encoded unitary $U$ and corrects both $X$ and $Z$ errors.  The resource state in the dashed box can be prepared and verified offline.  Conditioned on the (logical) measurement outcomes, a Pauli correction may be required (as in teleportation).  In most cases this Pauli correction can be noted classically and need not actually be applied.  
}
\end{figure}

The advantage of preparing a more complicated resource state is that the preparation can be done ``offline''.  The bulk of the work of both error correction and encoded gates can be completed before ever touching the data.  As a result, it is possible to use error-\emph{detection}, throwing away ancillary states that exhibit errors.  Since codes can detect far more errors than they can correct, this method can offer substantially higher thresholds than Steane-type or Shor-type error correction~\cite{Knill2004}.  The concept has strong similarities with state distillation, which also uses error detection.

High-levels of non-determinism, however, can be very resource intensive and can lead to poor threshold performance in certain circumstances~\cite{Lai2013a}.  Rigorous threshold analysis is also more complicated~\cite{Reichardt2007,Aliferis2007c}.

\subsection{Shor error correction
\label{sec:fault.qec.shor}
}
Shor-type error correction is the only of the three types that does not use encoded ancillary states to extract syndrome information.  Instead, each syndrome is measured by preparing a so-called GHZ, or ``cat'' state as in~\figref{fig:shor-ec-circuits}.  Like the Steane and Knill methods, the ancillary state must first be verified using error detection to make sure that errors do not spread back to the data. See~\figref{fig:shor-ec-circuits}. In order to ensure reliable results, each syndrome measurement is repeated a number of times that is proportional to the distance of the code.

\begin{figure}
\centering
\includegraphics{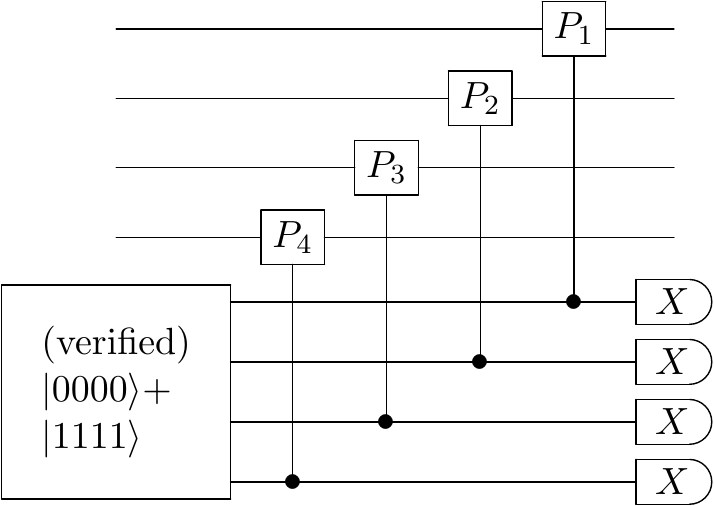}
\caption[Shor-style syndrome measurement.]{\label{fig:shor-ec-circuits}
A weight-four Shor-style syndrome measurement.  This circuit differs from~\figref{fig:naive-syndrome-measurement} in that each qubit of the cat-state ancilla interacts with at most one qubit of the encoded data.  The cat state $\frac{1}{\sqrt{2}}(\ket{0000} + \ket{1111})$ must be checked for errors (verified) before it can be used.
}
\end{figure}

The size of a stabilizer measurement circuit corresponds directly with the weight of the stabilizer.  Consequently, Shor-type error correction is most useful for codes that have low-weight stabilizers. In some cases, syndrome measurements can be implemented using a single ``bare'' ancilla qubit, without creating a cat state.  Bare ancillas are usually used in surface code schemes, for example~\cite{Fowler2012d}.

\section{Resource requirements
\label{sec:fault.resources}
}
The resource requirements for fault-tolerant quantum computation can be specified in a variety of ways including: circuit size (number of gates), circuit depth (time), circuit width (number of physical qubits), or the number of a particular type of gate ($T$ gates, for example).  Often it is possible to trade one type of resource for another.  One common example is to trade circuit width and circuit depth.  Fowler, for example, has shown how to minimize computational depth in the surface code at the expense of a larger qubit lattice~\cite{Fowler2012g}.  Therefore, it is often sensible to express resource requirements in terms of circuit area (depth$\times$width) or volume (space$\times$time).

Threshold theorems show that both the time and space resources required for reliable quantum computation scale efficiently with respect to the size the original noisy computation, in the asymptotic sense.  Given a noisy circuit of size $n$, it is possible to construct a fault-tolerant simulation that takes time and space $n \cdot \text{polylog}(n)$.  The constants in the polynomial can be overwhelmingly high, however, and some examples were noted in~\secref{sec:motivation.overhead}.

We now have a clearer picture for why the overhead is so large, and how the various parts of fault-tolerance schemes contribute to the overhead.  The most obvious sources are error correction and state distillation, both of which involve multiple rounds of error checking in order to produce resource states of suitable fidelity. But there are other less obvious sources, too.
For example, the resource overhead increases rapidly as the gate error-rate approaches the threshold. 
From~\eqnref{eq:concatenated-logical-error-rate} we see that a physical error rate of $p = p_\text{th}/\alpha$ induces a multiplicative factor in the overhead of $1/\log\alpha$, which increases exponentially near the threshold $p_\text{th}$.
%Thus, one way to decrease overhead requirements is to improve threshold calculations, which is the main goal of~\chapref{chap:threshold}.
Additional overhead is incurred from decomposing unitaries from the quantum algorithm into the limited set of fault-tolerant gates.

In the remaining chapters, we will examine each of these sources of overhead, in turn.  In most cases, our focus will be on optimizing size and width requirements, though some optimizations will also improve circuit depth.

\section{Architectural considerations
\label{sec:fault.architecture}
}

In addition to suppressing noise, fault tolerant constructions must also satisfy other hardware constraints.  This can mean, for example, accounting for more complicated noise models such as those with qubit leakage, but may also involve limitations on the set of available gates, or the placement of and interactions among qubits.

One of the most significant limitations of proposed quantum computing architectures is qubit geometry.  Many such proposals involve a lattice of qubits in a limited number of spatial dimensions~(see references contained in~\cite{Step09,Fowler2012d}).  Qubits in the lattice are allowed to interact only with a small number of nearby qubits, usually only nearest neighbors.  A variety of studies have considered lattices in one-dimension~\cite{Gott00,Shafaei2013,Devitt04,Fowl04,Step09,Saeedi2010}, two-dimensions~\cite{Svore2005,Svore2006b,Fowler2012d}, and three-dimensions~\cite{Bomb07b,Haah2011,Michnicki2012,Bravyi2012,Kim2012}.  Geometric connectivity constraints can significantly impact the performance of a fault-tolerance scheme, particularly for those based on concatenated codes~\cite{Svore2006b,Lai2013a}. Topological codes, however, are each tailored to a specific geometry and suffer little when the computer geometry is similar to the intended topology of the code.

Other limitations have also been considered.  Gate execution times can vary depending on the gate.  Measurements often take longer than unitary gates, though it is possible to overcome this limitation~\cite{DiVi06,Paz-Silva2010a}.  Fault-tolerance can also be achieved when control of individual qubits is limited~\cite{Bririd2003,Kay2005,Kay2007,Fitzsimons2007,Fitzsimons2009,Paz-Silva2010,Paz-Silva2011}.  Production of qubit lattices on physical substrates will likely include some number of defective qubits.  With some care fault-tolerance protocols can be adapted to avoid defective regions, even subject to geometric locality constraints~\cite{Nagayama}.
Practically speaking, it is easier to manufacture many small regions rather than one monolithic lattice.  Several authors have considered fault-tolerant quantum computation in which qubits are distributed among many small nodes~\cite{Devitt2008,VanM09,Kim2009,Devitt2009a,Hors11}.

An assumption that is almost ubiquitous in analysis of fault-tolerance schemes is that perfect and arbitrarily fast classical control logic is available.  In reality, though, classical computers have limitations, and connecting classical and quantum logic requires physical space.  Decoding and interpreting measurement results is efficient for concatenated codes, and can be made similarly efficient for topological codes~\cite{Devitt2010,Fowl11b,Fowler2013a}.  Even so, low-latency high-performance classical logic is desirable and may be necessary for architectures with small quantum gate times.  One attractive option is to use low-power superconducting technology~\cite{Herr11,Mukhanov2011,Volkmann2013,Holmes2013,Herr2013}, which could be placed nearby or on the same substrate as the qubits.  Developing and optimizing the necessary classical control algorithms is a worthy topic of future research.

\chapter[Universality with transversal gates]{Fault-tolerant universal computation with transversal gates
\label{chap:transversal}
}

This chapter is based on material that appears in~\cite{Paetznick2013a}.
\vspace{1cm}

At the highest level, fault-tolerant quantum computation involves only two steps: encoded computation, and error correction.  Thus reducing resource overhead requires simplification of either or both of these steps.  In this chapter we address the former, encoded computation.  In particular, we show that a universal set of fault-tolerant gates can be implemented using only the simplest of constructions, transversal gates.

Recall from~\secref{sec:fault.gates.transversal} that a transversal gate is the application of physical gates transversally across the codewords, usually meaning that the $j$th gate is applied to the $j$th qubits of the codewords, for every~$j$.  Transversal gates are highly desirable because they are both extremely simple and automatically fault tolerant, according to~\defref{def:strict-fault-tolerance-gate}.
Depending on the gate, a transversal implementation may or may not preserve the codespace and execute a valid encoded operation.  Consider the $[[7,1,3]]$ code, for example.  Transversal application of Hadamard preserves the set of stabilizers, and exchanges the $X$ and $Z$ logical operators and so transversal Hadamard implements logical Hadamard.  On the other hand, transversal $T$ is \emph{not} a logical operation on this code; it corrupts the $X$ logical operator.

Until $2007$, an important open question in quantum information theory was whether or not there exist codes that admit transversal implementation of a universal set of gates.  Due to the inability to find one, it was conjectured that no such code existed.  Zeng, Cross and Chuang confirmed this conjecture for stabilizer codes on qubits~\cite{Zeng2007}, and then along with Chen and Chung extended the result to qudits~\cite{Chen2008a}.  Soon after, Eastin and Knill showed that the conjecture is true for any nontrivial quantum code~\cite{Eastin2009a}.
\begin{theorem}[Transversal universality is impossible~\cite{Eastin2009a}]
\label{thm:eastin-knill}
For any quantum code capable of detecting an error on any physical subsystem, the set of transversal logical operations is not universal.
\end{theorem}  

\thmref{thm:eastin-knill} is unfortunate because the traditional method for completing a universal set of fault-tolerant gates is state distillation, a procedure which is highly costly compared to transversal gates. See~\secref{sec:fault.gates.distillation}.  Indeed state distillation dominates the resource overhead for fault-tolerant quantum computation~\cite{Raussendorf2007a,Fowler2013}.

In this chapter we propose a way of implementing a universal set of quantum gates transversally,
up to a correction that can be made by the standard error-correction procedure. The inclusion of error correction means that~\thmref{thm:eastin-knill} is preserved. However, since error correction is required anyway, our protocol effectively shows that the no-go theorems~\cite{Zeng2007,Chen2008a,Eastin2009a} can be circumvented without adding any new machinery.
Separate injection and distillation procedures are not required.  

% The construction works only 
% Therefore, implementing our construction directly may not reduce the overhead compared to using state distillation with more efficient codes that can tolerate higher noise rates.  However, based on our construction, we derive a state-distillation procedure that, with realistic error parameters, reduces the overhead compared to previous state-of-the-art state distillation methods~\cite{Eastin2012,Jones2012d}.  

Our construction is based on two main insights for the class of ``triorthogonal" quantum stabilizer codes, introduced recently by Bravyi and Haah~\cite{Bravyi2012a}.  First, we observe that the controlled-controlled-$Z$ operation (defined by $\CCZ \ket{a,b,c} = (-1)^{abc} \ket{a,b,c}$ for bits $a,b,c$) can be implemented transversally for any triorthogonal quantum code.  Second, we show that Hadamard can be implemented by transversal $H$ gates followed by stabilizer measurements and Pauli $X$ corrections.  Together, $H$ and $\CCZ$ are universal for quantum computation~\cite{Shi02,Aharonov2003}.

% In addition to allowing for universal quantum computation directly, our result also permits more efficient state injection and distillation when non-triorthogonal codes are used for computation.  Traditional injection and distillation procedures are used for fault-tolerantly implementing the $T$ gate.
% Recently, Eastin~\cite{Eastin2012} and Jones~\cite{Jones2012d,Jones2013a} (see also~\cite{Jones2013b}) have shown that the cost of implementing a Toffoli gate can be reduced by distilling so-called Toffoli \emph{states} rather than using fault-tolerant $T$ gates.  We observe that the distillation procedure of~\cite{Bravyi2012a} for $T$ gates can be adapted to Toffoli states in order to improve on the procedures of Eastin and Jones.

\section{Triorthogonal quantum codes}
\label{sec:transversal.triorthogonal}
Let us begin by specifying the construction of stabilizer codes based on triorthogonal matrices. For two binary vectors $f, g \in \{0,1\}^n$, let $f\cdot g \in \{0,1\}^n$ be their entry-wise product, and let $\abs f$ denote the Hamming weight of~$f$.
\begin{definition}[Triorthogonal matrix~\cite{Bravyi2012a}]
\label{def:triorthogonal-matrix}
An $m\times n$ binary matrix $G$, with rows $f_1,\ldots,f_m \in \{0,1\}^n$, is \emph{triorthogonal} if
\begin{align*}
|f_i\cdot f_j| &= 0 \!\!\!\pmod 2 &\text{and}\quad |f_i\cdot f_j\cdot f_k| &= 0 \!\!\!\pmod 2
\end{align*}
for all pairs $(i,j)$ and triples $(i,j,k)$ of distinct indices. 
\end{definition}

An $m \times n$ triorthogonal matrix $G$ can be used to construct an $n$-qubit, ``triorthogonal," CSS code as follows. 
\begin{definition}[Triorthogonal code~\cite{Bravyi2012a}]
\label{def:triorthogonal-code}
For each even-weight row of a triorthogonal matrix $G$, add an $X$ stabilizer generator by mapping non-zero entries to $X$ operators, e.g., $(1,0,1) \mapsto X \otimes I \otimes X$. Similarly add a $Z$ stabilizer for each row of the orthogonal complement $G^\perp = \{g ~:~ \abs{g \cdot f} = 0 \mod 2,  \forall f\in G\}$. The logical $X$ and $Z$ operators are then given by mapping non-zero entries of the odd-weight rows of $G$ to $X$ and $Z$, respectively.
\end{definition}

For example, the $[[15,1,3]]$ code is triorthogonal~\cite{Knill1996a}, and is discussed in detail in~\secref{sec:transversal.example}.
Bravyi and Haah have constructed a $[[49,1,5]]$ triorthogonal error-correcting code and a family of $[[3k+8, k, 2]]$ triorthogonal error-detecting codes~\cite{Bravyi2012a}.

\subsection{Triply-even codes
\label{sec:transversal.triorthogonal.triply-even}
} 
A special subset of triorthogonal codes admit transversal implementation of the single-qubit $T$ gate.  The $[[15,1,3]]$ code is a well-known example.
These codes, in addition to the conditions in~\defref{def:triorthogonal-matrix} satisfy the slightly more restrictive condition
\begin{equation}
\label{eq:triply-even}
|f_i \cdot f_j| = 0\mod 4
\enspace ,
\end{equation}
for all distinct pairs of even-weight rows $(f_i, f_j)$.  This condition implies that all of the stabilizers of the code have weight $0\mod 8$.  Codes that satisfy~\eqnref{eq:triply-even} are called \emph{triply even}~\cite{Betsumiya2012}.
In general, $T$ is transversal for triorthogonal codes only up to (non-transversal) Clifford corrections~\cite{Bravyi2012a}.

\section{Transversal CCZ}
\label{sec:transversal.ccz}
\def\fodd {f_{\star}} %{f_1} {f_{\text{odd}}

We next construct a fault-tolerant $\CCZ$ gate for a triorthogonal code.
%Together with the Hadamard gate, $\CCZ$ can be used to implement the Toffoli gate as shown in~\figref{fig:toffoli-ccz-equivalence}. 
We claim that for any triorthogonal code, transversal application of $\CCZ$ gates realizes $\CCZ$ gates on the encoded qubits. 

\begin{theorem}[Transversal $\CCZ$ for triorthogonal codes]
\label{thm:transversal-ccz}
Let $C$ be a triorthogonal code based on a triorthogonal matrix $G$.  Then transversal application of $\CCZ$ implements logical $\CCZ$ transversally on each of the encoded qubits of $C$.
\end{theorem}
\begin{proof}
For simplicity consider first the case of a triorthogonal code with a single encoded qubit, i.e., based on a triorthogonal matrix $G$ with a single odd-weight row $f_{\star}$. Let $\mathcal{G}_0 \subseteq \{0,1\}^n$ be the linear span of all the even-weight rows of $G$ and let $\mathcal{G}_1$ be the coset $\{\fodd + g : g\in \mathcal{G}_0\}$. Then the encoding of $\ket a$, for $a \in \{0,1\}$, is given by the uniform superposition over $\mathcal{G}_a$: $\ket{\logical{a}} = \frac{1}{\sqrt{\abs{\mathcal{G}_a}}} \sum_{g \in \mathcal{G}_a} \ket g$.

The action of transversal $\CCZ$ on an encoded basis state $\ket{\logical{a,b,c}}$, for $a, b, c \in \{0,1\}$, is therefore given by 
\begin{equation}
\begin{split}
\CCZ^{\otimes n} \ket{\logical{a,b,c}} 
 &=\sum_{g \in \mathcal{G}_a, h \in \mathcal{G}_b, i \in \mathcal{G}_c} \CCZ^{\otimes n} \ket{g,h,i} \\
 &= \sum_{g \in \mathcal{G}_a, h \in \mathcal{G}_b, i \in \mathcal{G}_c} (-1)^{|g\cdot h\cdot i|}\ket{g,h,i}
\enspace.
\end{split}
\label{eq:transversal-ccz-on-basis-state}
\end{equation}
Now $g \cdot h \cdot i$ can be expanded as $(a \fodd + g')\cdot (b \fodd + h')\cdot (c \fodd + i')
$,
where $g', h', i' \in \mathcal{G}_0$. Expanding further gives one term $abc (\fodd \cdot \fodd \cdot \fodd) = abc \fodd$, plus other triple product terms in which $\fodd$ appears at most twice. Since $G$ is triorthogonal, these other terms necessarily have even weight. The term $abc \fodd$ has odd weight if and only if $a = b = c = 1$. Substituting back into~\eqref{eq:transversal-ccz-on-basis-state}, as desired, 
\begin{equation}
\CCZ^{\otimes n} \ket{\logical{a,b,c}} = (-1)^{abc} \ket{\logical{a,b,c}} %= \logical{\CCZ} \ket{\logical{abc}}
\enspace.
\end{equation}

In the case that $G$ has some number $k > 1$ of odd-weight rows $\{\fodd^{(1)}, \fodd^{(2)},\ldots,\fodd^{(k)}\}$ we may define $2^k$ cosets, one for each codeword.
Let $\mathbf{a}$ be a length-$k$ binary vector where each element $a_i$ represents a logical qubit of the code, and let
\begin{equation}
\label{eq:coset-vector}
\mathcal{G}_\mathbf{a} := \{g' + \sum_{i=1}^k a_i \fodd^{(i)} ~|~ g' \in \mathcal{G}_0\}
\enspace .
\end{equation}
From~\eqnref{eq:coset-vector}, we can see that the expansion of $g\cdot h\cdot i$ will contain even-weight terms plus $k$ terms of the form $a_ib_ic_i(\fodd^{(i)}\cdot \fodd^{(i)}\cdot \fodd^{(i)})$, each of which is odd if and only if $a_i = b_i = c_i = 1$.  Again substituting back into~\eqnref{eq:transversal-ccz-on-basis-state} we obtain
\begin{equation}
\CCZ^{\otimes n} \ket{\logical{\mathbf{a},\mathbf{b},\mathbf{c}}} = \prod_{i=1}^k(-1)^{a_ib_ic_i} \ket{\logical{\mathbf{a},\mathbf{b},\mathbf{c}}}
\enspace .
\end{equation}
Thus transversal $\CCZ$ implements logical CCZ transversally across each of the encoded qubits.
\end{proof}

We note that transversality of $\CCZ$ for the the subset of triply-even codes follows trivially from the fact that $\CCZ$ can be expressed as a sequence of gates from $\{T,\text{CNOT}\}$~\cite{Nielsen2000}. \thmref{thm:transversal-ccz} extends this result to all triorthogonal codes.
In a sense,~\thmref{thm:transversal-ccz} shows that $\CCZ$ is more ``natural'' than $T$ for triorthogonal codes, since Clifford corrections may be required for $T$~\cite{Bravyi2012a}, but are never required for $\CCZ$.

If the orthogonality conditions on the matrix $G$ are increased, then additional types of diagonal operations are transversal.  If $G$ satisfies the condition that all $j$-tuple products have weight $(0\mod 2)$ for all $2 \leq j \leq h$, then the $h$-fold controlled-$Z$ gate is transversal in the corresponding stabilizer code.
This observation is similar to a result of Landahl and Cesare, who demonstrated that codes satisfying increasingly stringent conditions on weights of the codewords admit transversal $Z$-axis rotations of increasing powers of $1/2^k$~\cite{Landahl2013}.

\section{Transversal Hadamard}
\label{sec:transversal.hadamard}
To achieve universality, we also require a fault-tolerant implementation of the Hadamard gate. For Hadamard to be transversal, the code must be self-dual, i.e., $\mathcal{G}_0 = G^\perp$. Unfortunately, no triorthogonal code is self-dual. Indeed, otherwise, since $\CCZ$ is transversal it would be possible obtain a transversal implementation of Toffoli and $H$ for the same code. See~\figref{fig:toffoli-ccz-equivalence}.  However, Toffoli and $H$ together are universal~\cite{Shi02,Aharonov2003} and so transversal implementations of both would violate~\thmref{thm:eastin-knill}.

\begin{figure}
\centering
\includegraphics[height=2cm]{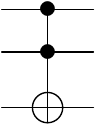}
\raisebox{1cm}{~~\large{=}~~}
\includegraphics[height=2cm]{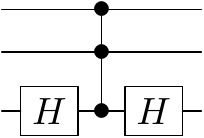}
\caption[Toffoli and CCZ equivalence.]{\label{fig:toffoli-ccz-equivalence}
The Toffoli gate is equivalent to a $\CCZ$ gate in which the target qubit is conjugated by Hadamard gates.
}
\end{figure}

Nonetheless, fault-tolerant and effectively transversal implementations of logical $H$ are still possible. 
\begin{theorem}[Transversal $H$ for triorthogonal codes]
\label{thm:transversal-hadamard}
Let $C$ be a triorthogonal code based on a triorthogonal matrix $G$.  Then the encoded Hadamard gate on each of the encoded qubits of $C$ can be implemented fault-tolerantly using transversal $H$, fault-tolerant syndrome measurement and classically-controlled transversal $X$ gates.
\end{theorem}
\begin{proof}
When transversal $H$ is performed on a triorthogonal code, the logical operators are transformed properly: logical $X$ maps to logical $Z$ and vice versa. A subset of the stabilizers is preserved; observe that $\mathcal{G}_0 \subset G^\perp$, and thus each element of $\mathcal{G}_0$ corresponds to both $X$ and $Z$ stabilizers, which transversal $H$ swaps.  Transversal $H$ does not preserve the $Z$ stabilizers corresponding to $G^\perp \setminus \mathcal{G}_0$, so these must be restored by measuring and correcting them.  
%In the $[[15,7,3]]$ example above, this involved measuring the six gauge qubits' logical $Z$ operators.

Consider the effect of measuring one of the $Z$ stabilizer generators $\zeta$ corresponding to $G^\perp \setminus \mathcal{G}_0$.  The measurement projects the code block onto either the $+1$ or $-1$ eigenspace of $\zeta$ according to the measurement outcome.  Let $\chi$ be a tensor product of $I$ and $X$ operators such that $\chi$ anticommutes with $\zeta$ and commutes with all other $Z$ stabilizer generators and $Z$ logical operators.  Such an operator always exists since $\zeta$ is neither an element of the (current) stabilizer nor an element of the normalizer.  If the measurement outcome is $-1$, then applying $\chi$ restores the code block to the $+1$ eigenspace of $\zeta$.

Importantly, even with additional $X$ corrections to fix the $Z$ stabilizers of $G^\perp \setminus \mathcal{G}_0$, the procedure is fault tolerant. That is, $k$ gate failures can lead to a data error of weight at most~$k$, for $k$ less than half the code's distance~$d$. Let $d_Z$ be the code's distance against $Z$ errors, as determined by the $X$ stabilizers of $\mathcal{G}_0$. Likewise, let $d_X$ be the distance against $X$ errors, as determined by the $Z$ stabilizers of $G^\perp$.  The minimum distance of the code (against arbitrary Pauli errors) is then $d=\min\{d_X,d_Z\}$.  But $\mathcal{G}_0 \subset G^\perp$ implies that $d_Z < d_X$ and, therefore, the code's minimum distance is determined solely by $\mathcal{G}_0$.  Since both the $X$ and $Z$ stabilizers of $\mathcal{G}_0$ are preserved, a minimum distance of $d$ is maintained throughout.  So long as the stabilizer measurements are performed fault-tolerantly, and since the other operations are transversal, the entire procedure is fault-tolerant.
\end{proof}

In fact, the Hadamard construction of~\thmref{thm:transversal-hadamard} holds for any CSS code in which the $X$ and $Z$ logical operators have identical supports and transversal Hadamard conjugates the $X$ stabilizers to a subset of the $Z$ stabilizers.  The triorthogonality condition (\defref{def:triorthogonal-matrix}) is not strictly necessary.  Rather it is the symmetry of the $X$ and $Z$ stabilizers in the triorthogonal code construction that is important.

Informally,~\thmref{thm:transversal-hadamard} takes advantage of the fact that the $X$ and $Z$ stabilizers have an asymmetry which is required in order to provide triorthogonality (and therefore transversal $\CCZ$), but which is otherwise unnecessary.  In principle, the extra $X$-error distance provided by the $Z$ stabilizers could be used to improve performance for biased noise~\cite{Aliferis2008,Brooks2012}.  But it can be difficult to properly exploit this asymmetry in practice. For example, direct application of transversal $T$ is not allowed because it splits $X$ errors into both $X$ and $Z$ errors (see~\eqnref{eq:T-conjugation}).  We choose, instead, to use the asymmetry to reduce the complexity of the Hadamard gate.

The stabilizer measurements required by~\thmref{thm:transversal-hadamard} can be incorporated into the normal fault-tolerant error-correction procedure.
Steane's procedure~\cite{Steane1996}, for example, involves a transversal CNOT from the data to an encoded $\ket +$ ancilla state.
%---implementing an encoded CNOT---followed by transversal $Z$-basis measurements of the ancilla.
%The transversal CNOT implements encoded CNOT and has the effect of copying the unwanted $X$ stabilizers onto the ancilla, and copying the desired $Z$ stabilizers from the ancilla to the data.  
Transversal $Z$-basis measurements of the ancilla then permit correcting $X$ errors on the data, while simultaneously restoring the stabilizer group.
%(Once again, in the $[[15,7,3]]$ example, the ancilla's gauge qubits are prepared in $\ket{0^6}$. Each is measured as $0$ or $1$, with a correction required in the latter case.)  
See~\figref{fig:steane-hadamard}. (See also~\secref{sec:fault.qec.steane}.)  
Alternatively, Knill-style or Shor-style error correction could be used. 
In any case, the required stabilizers can be measured and corrected using $H$, $X$, CNOT, $\ket 0$ preparation and $Z$-basis measurements.  By using $\CCZ$ gates to simulate CNOT and $X$, universality can be achieved using only $\ket 0$ preparation, $Z$-basis measurement, and $H$ and $\CCZ$ gates. 

\begin{figure}
\centering
\raisebox{3.5cm}{\includegraphics[height=3.5cm]{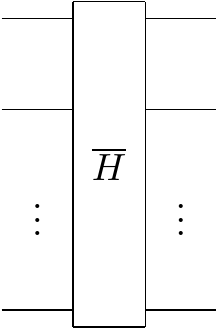}}
\raisebox{5cm}{~~~\Huge{=}~~~}
\includegraphics[height=7cm]{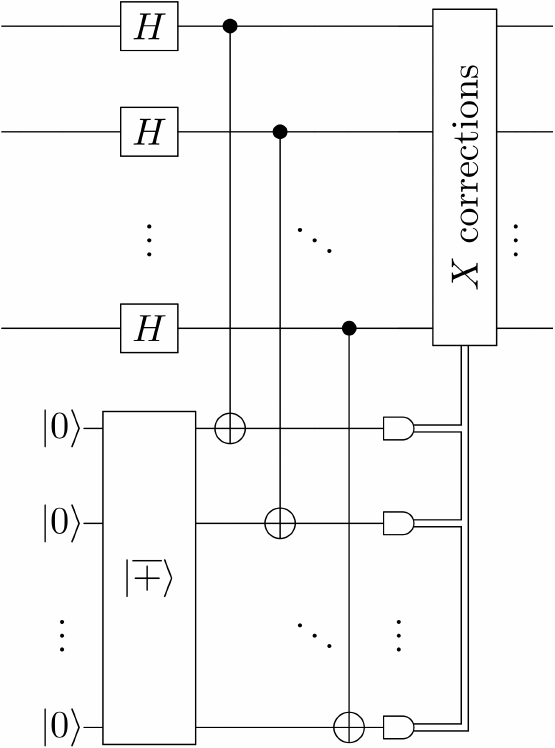}
\caption[Transversal $H$ plus error correction.]{\label{fig:steane-hadamard}
An implementation of the logical Hadamard operation in a triorthogonal code, using Steane's method for error correction.  Transversal Hadamard gates are applied to the data block.  In order to restore the data to the codespace, and also correct any $X$ errors, an encoded $\ket +$ state is prepared, coupled to the data with transversal CNOT gates and measured. $X$ corrections are applied as necessary.
}
\end{figure}

\section{Example: 15-qubit codes
\label{sec:transversal.example}
}

In order to make our universal construction concrete, we now walk through an example based on the $15$-qubit code. We present the example in two equivalent ways.  First with the $[[15,1,3]]$ code, and then with the $[[15,7,3]]$ code.

The $[[15,1,3]]$ code is based on the triorthogonal matrix in~\tabref{tab:15-bit-Hamming-parity-checks}.
The stabilizer generators can be presented as:
\begin{center}
\label{eq:1513-stabilizers}
\begin{tabular}{c@{$\,\!\!$}c@{$\,\!\!$}c@{$\,\!\!$}c@{$\,\!\!$}c@{$\,\!\!$}c@{$\,\!\!$}c@{$\,\!\!$}c@{$\,\!\!$}c@{$\,\!\!$}c@{$\,\!\!$}c@{$\,\!\!$}c@{$\,\!\!$}c@{$\,\!\!$}c@{$\,\!\!$}c@{$\quad$}c@{$\,\!$}c@{$\,\!$}c@{$\,\!$}c@{$\,\!$}c@{$\,\!$}c@{$\,\!$}c@{$\,\!$}c@{$\,\!$}c@{$\,\!$}c@{$\,\!$}c@{$\,\!$}c@{$\,\!$}c@{$\,\!$}c@{$\,\!$}c@{,}}
   &   &   &   &   &   &   &   &   &   &   &   &   &   &    & $Z$&$\cdot$&$\cdot$&$\cdot$&$\cdot$&$\cdot$&$Z$&$\cdot$&$\cdot$&$Z$&$\cdot$&$Z$&$\cdot$&$\cdot$&$\cdot$\\
   &   &   &   &   &   &   &   &   &   &   &   &   &   &    & $\cdot$&$Z$&$\cdot$&$\cdot$&$\cdot$&$\cdot$&$Z$&$\cdot$&$Z$&$\cdot$&$\cdot$&$Z$&$\cdot$&$\cdot$&$\cdot$\\
   &   &   &   &   &   &   &   &   &   &   &   &   &   &    & $\cdot$&$\cdot$&$Z$&$\cdot$&$\cdot$&$\cdot$&$Z$&$Z$&$\cdot$&$\cdot$&$\cdot$&$Z$&$\cdot$&$\cdot$&$\cdot$\\
   &   &   &   &   &   &   &   &   &   &   &   &   &   &    & $\cdot$&$\cdot$&$\cdot$&$Z$&$\cdot$&$\cdot$&$Z$&$\cdot$&$Z$&$Z$&$\cdot$&$\cdot$&$\cdot$&$\cdot$&$\cdot$\\
   &   &   &   &   &   &   &   &   &   &   &   &   &   &    & $\cdot$&$\cdot$&$\cdot$&$\cdot$&$Z$&$\cdot$&$Z$&$Z$&$\cdot$&$Z$&$\cdot$&$\cdot$&$\cdot$&$\cdot$&$\cdot$\\
   &   &   &   &   &   &   &   &   &   &   &   &   &   &    & $\cdot$&$\cdot$&$\cdot$&$\cdot$&$\cdot$&$Z$&$Z$&$Z$&$Z$&$\cdot$&$\cdot$&$\cdot$&$\cdot$&$\cdot$&$\cdot$\\
$\cdot$&$\cdot$&$\cdot$&$\cdot$&$\cdot$&$\cdot$&$\cdot$&$X$&$X$&$X$&$X$&$X$&$X$&$X$&$X,$ & $\cdot$&$\cdot$&$\cdot$&$\cdot$&$\cdot$&$\cdot$&$\cdot$&$Z$&$Z$&$Z$&$Z$&$Z$&$Z$&$Z$&$Z$ \\
$\cdot$&$\cdot$&$\cdot$&$X$&$X$&$X$&$X$&$\cdot$&$\cdot$&$\cdot$&$\cdot$&$X$&$X$&$X$&$X,$ & $\cdot$&$\cdot$&$\cdot$&$Z$&$Z$&$Z$&$Z$&$\cdot$&$\cdot$&$\cdot$&$\cdot$&$Z$&$Z$&$Z$&$Z$ \\
$\cdot$&$X$&$X$&$\cdot$&$\cdot$&$X$&$X$&$\cdot$&$\cdot$&$X$&$X$&$\cdot$&$\cdot$&$X$&$X,$ & $\cdot$&$Z$&$Z$&$\cdot$&$\cdot$&$Z$&$Z$&$\cdot$&$\cdot$&$Z$&$Z$&$\cdot$&$\cdot$&$Z$&$Z$ \\
$X$&$\cdot$&$X$&$\cdot$&$X$&$\cdot$&$X$&$\cdot$&$X$&$\cdot$&$X$&$\cdot$&$X$&$\cdot$&$X,$ & $Z$&$\cdot$&$Z$&$\cdot$&$Z$&$\cdot$&$Z$&$\cdot$&$Z$&$\cdot$&$Z$&$\cdot$&$Z$&$\cdot$&$Z$
\end{tabular}
\end{center}
where the $X$ stabilizers come directly from~\tabref{tab:15-bit-Hamming-parity-checks} and the $Z$ stabilizers come from the orthogonal complement.  For visual clarity, identity operators are indicated by dots.  The logical $X$ and $Z$ operators correspond to transversal $X$ and $Z$, respectively.  By construction, this code is triorthogonal according to~\defref{def:triorthogonal-code}.  The four $X$ stabilizers provide distance three protection against $Z$ errors and the $11$ $Z$ stabilizers provide distance seven protection against $X$ errors.

Transversal Hadamard swaps the $X$ and $Z$ logical operators.  The $X$ and $Z$ stabilizers are also swapped.  The bottom four generators (both $X$ and $Z$) are preserved, since they are symmetric.  The remaining six $Z$ generators have now become $X$ generators.  Now the code provides distance three protection against $X$ errors and distance seven protection against $Z$ errors; it is the dual of the original code.  The $[[15,1,3]]$ code is restored by measuring each of the top six $Z$ generators.  For each $-1$ outcome, the $X$ correction corresponds to one of the six asymmetric $X$ generators of the dual code. 

There is an alternative way to interpret this example by using the $[[15,7,3]]$ code.  The $[[15,7,3]]$ code has the same $X$ generators as the $[[15,1,3]]$ code, but uses only the bottom four $Z$ generators.  It encodes seven logical qubits.  The logical $Z$ operators correspond to the top six generators of the $[[15,1,3]]$ code, plus transversal $Z$.  However, the code, as given, is not triorthogonal.

In order to induce triorthogonality, we will treat six of the encoded qubits as ``gauge qubits''.  That is, we will not use them to store computational data.  Instead, we will require that they are always prepared as encoded $\ket{0^6}$, so that the logical $Z$ operators are now stabilizers. If we choose the six gauge qubits so that the remaining computation qubit is the one with transversal logical $Z$, then we recover the $[[15,1,3]]$ code, which is triorthogonal.  Now, transversal $H$ can be interpreted as implementing logical $H$, except that the gauge qubits have been corrupted because they are no longer in state $\ket{0^6}$. As for the $[[15,1,3]]$ code, the gauge can be restored by measuring the six corresponding logical $Z$ operators.

\section{Computation with triorthogonal codes}
\label{sec:transversal.computation}
The simplest way to use the $\CCZ$ and Hadamard constructions above is with a concatenated triorthogonal code.  The relation shown in~\figref{fig:toffoli-ccz-equivalence} implies that a universal set of fault-tolerant operations can be constructed from only $\CCZ$ and $H$ gates.  Thus using triorthogonal codes for computation could be advantageous for circuits that contain large numbers of Toffoli gates.  One could also imagine using multiple codes for computation by, for example, teleporting into the code best suited for each logical operation.  In this setting, a triorthogonal code could be used to implement efficiently the $\CCZ$ operation.  

%Our construction allows for quantum computation to arbitrary accuracy, so long as the error rate per physical gate is below a constant ``threshold'' value~\cite{Aharonov1996a,Aliferis2005}.
Threshold error rates for triorthogonal codes are largely unknown, though one estimate for the $[[15,1,3]]$ code is roughly $0.01$ percent per gate for depolarizing noise~\cite{Cross2009}. Toffoli- and $\CCZ$-type gates have been demonstrated in a number of experimental settings,  with fidelities ranging from $68$ to $98$ percent~\cite{Monz2009,Mariantoni2011,Fedorov2012,Reed2012}.
If, however, the $\CCZ$ operation is constructed from a sequence of one- and two-qubit gates then the threshold is likely lower. 
Since resource overhead increases rapidly as the physical noise rate approaches threshold, our construction is likely to be outperformed by schemes based on other codes, for which the threshold can be nearly one percent or higher~(see, e.g., \cite{Knill2004,Raussendorf2007,Wang2011,Stephens2013}). The existence of high-performing triorthogonal codes is not out of the question, however, and could be a fruitful area of research.
%\todo{Color codes, for example\ldots}

\section{Toffoli state distillation
\label{sec:transversal.distillation}
}
Ironically, while the original motivation for implementing $\CCZ$ transversally was to eliminate state distillation, \thmref{thm:transversal-ccz} also implies an alternative protocol for distillation.  Bravyi and Haah have proposed distillation procedures using triorthogonal codes that permit fault-tolerant implementation of the $T$ gate~\cite{Bravyi2012a}.  We show that a similar procedure can be used to implement Toffoli gates.

\subsection{A recursive distillation protocol}
The Toffoli state is defined by the output of the Toffoli gate on input $\ket{+,+,0}$, where the third qubit is the target. The circuit in~\figref{fig:toffoli-distillation} uses a $[[3k+8,k,2]]$ triorthogonal code and  $3k+8$ noisy $\CCZ$ gates to produce $k$ Toffoli states with higher fidelity. 
Note that the Hadamard gates are performed after decoding and thus the circuit in~\figref{fig:steane-hadamard} is not required.

\begin{figure}
\centering
\includegraphics{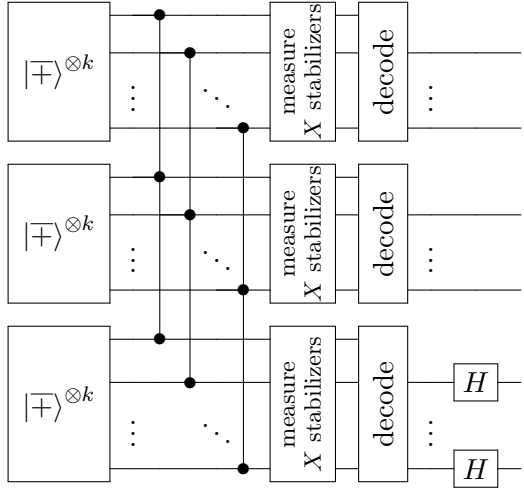}
\caption[Toffoli state distillation using CCZ.]{\label{fig:toffoli-distillation}
A Toffoli state distillation circuit using a triorthogonal code encoding $k$ qubits.  Three separate blocks are prepared in the encoded state $\ket{\overline{+}}^{\otimes k}$ and then transversal $\CCZ$ gates are applied.  Conditioned on detecting no errors, each block is decoded and Hadamard gates are applied to each of the target qubits, yielding $k$ Toffoli states.}
\end{figure}

To simplify the analysis we will assume that all Clifford gates can be implemented perfectly.  This assumption is justified by the fact that many quantum error-correcting codes admit simple (e.g., transversal) implementation of the Clifford group.  Using such a code, we can then arbitrarily reduce the logical error rate per Clifford gate using fault tolerant protocols for that code.  Error-free Clifford gates are conventional for analysis of state-distillation protocols, though some studies have considered a more complete error model~\cite{Jochym-O'Connor2012,Brooks2013a}.

The circuit in~\figref{fig:toffoli-distillation} is directly adapted from the $T$-gate distillation protocol of Bravyi and Haah.  Their protocol involved only a single code block, but the error analysis can be re-used here directly.  Consider a $[[3k+8,k,2]]$ triorthogonal code for some $k > 2$. Suppose that each qubit is independently subjected to a $Z$ error with probability $p$, after which the $X$ stabilizers are (perfectly) measured and the code block is decoded. Bravyi and Haah show that, conditioned on detecting no errors during stabilizer measurement, the probability of an error on a (logical) qubit after decoding is given by $(3k+1)p^2$ to leading order in $p$.  The scaling in the error comes from counting the number of weight-two logical $Z$ operators that have support on a particular logical qubit, which is equal to $(3k+1)$.

For Bravyi and Haah, the independent $Z$ errors originate from $T$ gates.  In this case, the independent $Z$ errors instead originate from $\CCZ$ gates.  However, the error analysis for a given code block is precisely the same.
Given access to $\CCZ$ gates that contain $Z$ errors independently with probability $p$, and conditioned on detecting no errors during stabilizer measurement, the circuit in~\figref{fig:toffoli-distillation} produces $k$ Toffoli states with error rate $(3k+1)p^2$ per state, to leading order in $p$.

Perhaps the most obvious way to obtain $\CCZ$ gates containing $Z$ errors with probability $p$ is to use a recursive protocol.  At the lowest level of recursion, we may choose to use physical $\CCZ$ gates if they are available, or an equivalent circuit composed of Clifford and $T$ gates.  These physical $\CCZ$ gates may also contain $X$ errors.  But it is possible to eliminate $X$ errors by probabilistically applying Clifford gates, a process known as ``twirling''~\cite{DiVincenzo2002}.  Alternatively, $X$ errors can be eliminated by using gate teleportation. See~\figref{fig:ccz-teleportation}.  Since $\CCZ$ is diagonal, it can be commuted through the control of a CNOT gate.  $X$ errors on the $\CCZ$ have no impact on the $X$-basis measurements, and $Z$ errors on the $\CCZ$ can lead only to $Z$ errors on the output.  The concept here is similar to that of~\figref{fig:T-distillation-Bell-pair} for $T$ distillation, except with a three-qubit gate.

\begin{figure}
\centering
\includegraphics{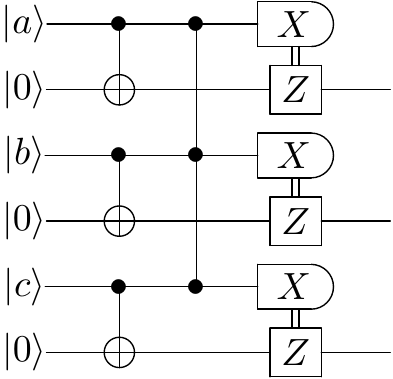}
\caption[CCZ implementation from Toffoli gate teleportation.]{\label{fig:ccz-teleportation}
This circuit implements a $\CCZ$ gate on the input $\ket{abc}$.  Each input qubit is individually teleported onto an ancilla.  The $\CCZ$ gate commutes through the CNOT controls and can therefore be performed after the CNOT gates.  Assuming perfect Clifford operations, the output contains only $Z$ errors.
}
\end{figure}

Another issue with the recursive protocol is that errors on the $k$ Toffoli states of the output of~\figref{fig:toffoli-distillation} are not independent.  Therefore, two Toffoli states from the same distillation circuit cannot be used together as inputs to a distillation circuit at the next level up.  When many Toffoli states are required, as is expected in large-scale quantum algorithms, then Toffoli states can be routed appropriately without any waste.

\subsection{A bootstrap distillation protocol}
We find, however, that a more efficient method is to use the Toffoli distillation protocol due to Eastin~\cite{Eastin2012} and Jones~\cite{Jones2012d} to implement $\CCZ$ gates, and use the triorthogonal protocol only at the top level.
A Toffoli state can be used to implement the $\CCZ$ gate with the help of classically controlled Clifford gates as shown in~\figref{fig:ccz-gate-teleportation}.  To see how~\figref{fig:toffoli-distillation} can be combined with the protocol of Eastin and Jones, we give the following illustrative example.

\begin{figure}
\centering
\raisebox{.75cm}{\includegraphics[height=3cm]{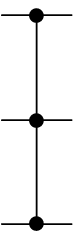}}
\raisebox{2cm}{~~~\Huge{=}~~~}
\includegraphics[height=4.5cm]{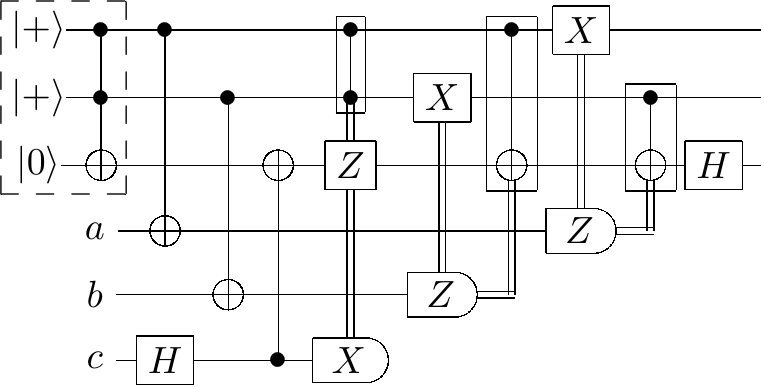}
\caption[CCZ gate teleportation.]{\label{fig:ccz-gate-teleportation}
A $\CCZ$ gate can be implemented by consuming a single Toffoli state~\cite{Nielsen2000}.  The input qubits are teleported into the Toffoli state (enclosed by the dashed line) with Clifford corrections conditioned on the measurement outcomes.}
\end{figure}

Suppose we wish to implement a Toffoli gate with error below $10^{-13}$.  The procedure of~\cite{Jones2012d} consumes eight $T$ gates with error rate $p$ to produce a Toffoli state with error rate $28p^2$. See~\figref{fig:jones-error-detecting-toffoli}.  The $T$ gates can be implemented using a combination of protocols; Table I of~\cite{Jones2012c} lists optimal protocol combinations for a large range of target error rates. If physical $T$ gates can be performed with error at most $10^{-2}$, then using the Toffoli construction of~\cite{Jones2012d}, as given, requires on average $540.16$ $T$ gates.  

\begin{figure}
\centering
\includegraphics{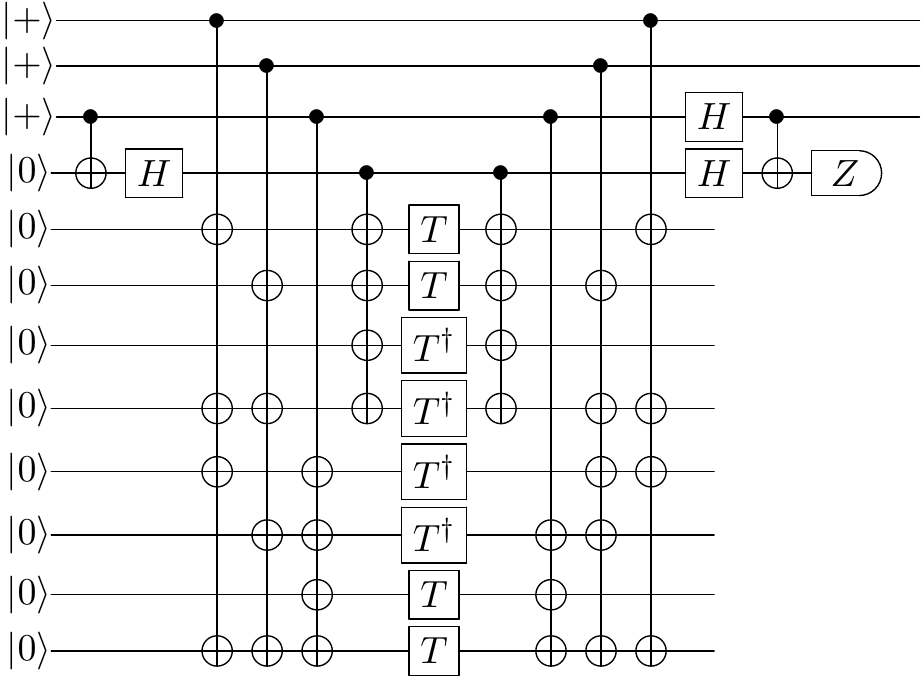}
\caption[An error detecting Toffoli gate.]{\label{fig:jones-error-detecting-toffoli}
This circuit prepares a Toffoli state on the top three qubits~\cite{Jones2012d}.  Assume that each $T$ gate fails with probability $p$ and the Clifford gates are perfect.  Then  conditioned on a $Z$-basis measurement outcome of zero, the probability of an error on the output is $28p^2$, to leading order in $p$.  The bottom eight qubits can be discarded.
}
\end{figure}

Alternatively, we could use a $[[3k+8,k,2]]$ triorthogonal code and~\figref{fig:toffoli-distillation} for distillation at the top level, and construct Toffoli states using~\figref{fig:jones-error-detecting-toffoli} as input to implement the $\CCZ$ gates. The distillation circuit fails to detect a faulty Toffoli state input only if the number of errors on \emph{each} triorthogonal code block is even. To leading order, this occurs only if a pair of input Toffoli states contain identical errors.  There are seven possible errors on the output of states from~\cite{Jones2012d}, each of which are equally likely. Thus, if the input Toffoli states have error $p_1$, then to leading order the failure probability of the triorthogonal protocol is given by $7(3k+1)(p_1/7)^2$ per output Toffoli state. For $k=100$, this yields an average $T$-gate cost of $428.7$, a savings of $25$\% over~\cite{Jones2012d} alone. Calculations for a range of target error rates are shown in~\figref{fig:t-cost-plot}.

\begin{figure}
\centering
\includegraphics[width=.7\textwidth]{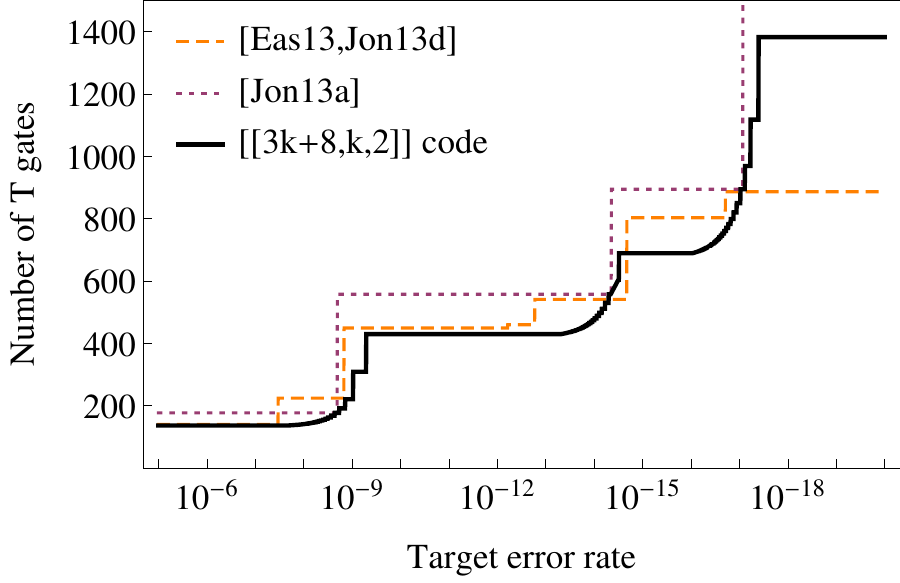}
\caption[$T$ count for Toffoli distillation protocols.]{\label{fig:t-cost-plot}
The average number of physical $T$ gates required for three different Toffoli state distillation protocols.  For the previous protocols of~\cite{Eastin2012,Jones2012d} and~\cite{Jones2013a}, input $T$ gates are first distilled to the appropriate fidelity according to Table I of~\cite{Jones2012c}.  The solid black line shows the cost of our protocol for $[[3k+8,k,2]]$ triorthogonal codes where an even integer $2 \leq k \leq 100$ has been optimally selected at each target error rate. Input $\CCZ$ gates to the triorthogonal protocol are produced using~\cite{Jones2012d}. Physical $T$ gates are assumed to have error at most~$10^{-2}$.}
\end{figure}

The $T$ gate count alone is an incomplete measure of the overhead.  Indeed,~\figref{fig:t-cost-plot} shows that the double error-detecting protocol of~\cite{Jones2013a} usually has higher $T$ gate cost than the single error-detecting protocol. However, the double error-detecting protocol can still yield savings since smaller code distances may be used for Clifford gates in intermediate distillation levels~\cite{Fowler2013,Jones2012d,Jones2013a,Jones2013b}.  Our protocol similarly allows for reduced Clifford gate costs and offers the flexibility to be used recursively or on top of any other Toffoli state distillation protocol, including~\cite{Eastin2012,Jones2012d} and~\cite{Jones2013a}.  Complete overhead calculations depend on architectural considerations. 

Jones has performed detailed optimizations and resource calculations of various Toffoli constructions for the surface code~\cite{Jones2013b}, though the protocol of~\figref{fig:toffoli-distillation} is not among them.  He finds that the single-error detecting circuit of~\cite{Eastin2012,Jones2012d} usually requires the smallest total space-time volume.  Given the results in~\figref{fig:t-cost-plot}, we expect that triorthogonal distillation performs similarly well in the surface code.  The corresponding optimization and volume calculations have not been performed here, however.

\section{Alternative methods for universality
\label{sec:transversal.alternatives}
}
Although state distillation is the most widely used protocol, other methods for achieving universality exist for certain codes. Shor's original proposal used Toffoli states and teleportation to implement Toffoli gates for the class of ``doubly-even" codes~\cite{Shor1997}.  However, each Toffoli state was prepared using a verified cat state and a particular four-qubit transversal gate rather than distillation, which was developed afterwards.  Shor's approach was later extended by Gottesman to accommodate any stabilizer code~\cite{Gottesman1998a}.

Knill, Laflamme and Zurek showed that $T$ and CNOT are transversal for the $[[15,1,3]]$ code~\cite{Knill1996a}.  For the Hadamard they proposed the circuit shown in~\figref{fig:1513-H-teleportation}.  Except for preparation of $\ket +$, each of the gates in this circuit can be performed transversally.  This circuit bears a striking resemblance to the gate teleportation circuit used for state distillation in~\figref{fig:T-gate-teleportation}.  Indeed, the most costly part of~\figref{fig:1513-H-teleportation} is the fault-tolerant preparation of the ``resource state'' $\ket +$.  One difference in this case, though, is that $\ket +$ is a stabilizer state, and can be prepared with the methods discussed in~\chapref{chap:ancilla}.  This method for achieving universality has also been used by Bombin and others in topological color codes~\cite{Bomb07b,Bombin2013}.

\begin{figure}
\centering
\includegraphics{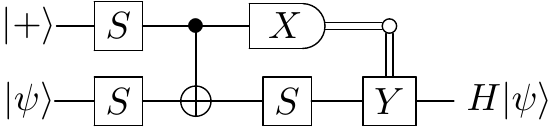}
\caption[Gate teleportation of $H$ in the {$[[15,1,3]]$} code.]{\label{fig:1513-H-teleportation}
This circuit implements $H$ (up to a global phase) with the help of $\ket +$ and $X$-basis measurement.}
\end{figure}

Another alternative has been employed to implement a fault-tolerant $T$ gate in the $[[7,1,3]]$ code.  Shor's cat-state method can be used multiple times to measure the operator $TXT^\dagger = SX$ of which $\ket A = T\ket+$ is the $+1$-eigenstate. Conditioned on the outcomes of these measurements, an ancilla state is projected onto encoded $\ket A$ with high fidelity~\cite{Aliferis2005}.

Recently, Jochym-O'Connor and Laflamme have proposed a different protocol for universality~\cite{Jochym-O'Connor2013}.  They concatenate two different codes and use the incomplete set of transversal gates from each one in order to obtain a universal set overall. Their method uses only transversal gates (in a certain sense), but whereas the distance of a concatenated code is typically given by the product $d_1d_2$ of the two code distances, they achieve a minimum distance of only $\min\{d_1,d_2\}$.  Thus, while their protocol is conceptually interesting, it is less efficient than ours.

These methods for achieving universality suggest several possible categorizations.
\begin{description}
\item[Distillation and teleportation]~

This category includes traditional $T$~\cite{Bravyi2004,Meier2012,Bravyi2012a,Jones2012c} and Toffoli~\cite{Eastin2012,Jones2012d} distillation protocols, as well as the $[[15,1,3]]$ protocol shown in \figref{fig:1513-H-teleportation} \cite{Knill1996a}.
\item[Cat state projection]~

Protocols in this category use cat states and transversal gates in order to measure a particular operator of which the desired state is an eigenstate.  This includes~\cite{Shor1997,Gottesman1998a} and~\cite{Aliferis2005}.
\item[Transversal gates and error correction]~ 

This category includes the protocol described in this chapter, and potentially~\cite{Jochym-O'Connor2013}.
\end{description}

Each protocol, regardless of the category requires preparation of some sort of ancillary state.  Even the Hadamard described in~\secref{sec:transversal.hadamard} requires an ancilla in order to measure the stabilizer generators.  Another way to partition universality techniques, therefore, is based on the type of ancilla state that is prepared.  One obvious choice is to group the protocols that require only stabilizer states such as $\ket +$ or cat states, and those that require non-stabilizer states such as $\ket A$ or Toffoli states.

Regardless of categorization, though, the most important property of each protocol is the amount of resources required.  Transversal gates plus error correction is the simplest of all protocols. But the uncertainty regarding thresholds for triorthogonal codes prevents more thorough analysis.  High thresholds and minimal connectivity requirements suggest that the $T$ or Toffoli distillation in the surface code may require fewer resources overall.

% \section{Outlook and extensions \todo{}}
% 
% \todo{Are there codes that admit transversal Toffoli directly? If a code was both additive and multiplicative, then it would.  But it seems that any such code would have trivial distance.  Need to go back through my notes on this.}
% 
% Other things to discuss here:
% \begin{enumerate}
%   \item Complications for threshold analysis (i.e., direct $T$/$\CCZ$ gates)
%   \item Direct application of $X$ corrections, vs. Pauli frame
%   \item Transversality for other codes?
%   \item Connections to $H$-codes?
%   \item What does this do for resource reduction?
%   \item Result shows that distillation costs can be shifted entirely into error correction.
% \end{enumerate}

\chapter{Reducing the overhead of error-correction
\label{chap:ancilla}
}

This chapter is based on material that appears in~\cite{Paetznick2011}.
\vspace{1cm}

We have seen in Chapters~\ref{chap:fault} and~\ref{chap:transversal} that error correction circuits are much more complicated than transversal gates.  Furthermore, since error correction is also required in distillation circuits, it is the dominant factor in determining a scheme's resource overhead, and is usually the major bottleneck in determining the noise threshold.  In particular, the details of how error correction is implemented are more important than the properties of the underlying quantum error-correcting~code.  

For example, with the nine-qubit Bacon-Shor code, a fault-tolerant logical CNOT gate between two code blocks can be implemented using nine physical CNOT gates, whereas an optimized error-correction method uses $24$ physical CNOT gates~\cite{Aliferis2007b}.  For larger quantum error-correcting codes, the asymmetry between computation and error correction is greater still.   

Larger quantum error-correcting codes, with higher distances and possibly higher rates, can still outperform smaller codes.  Separate numerical studies by Steane~\cite{Steane2003} (see also~\cite{Stea07}) and Cross, DiVincenzo and Terhal~\cite{Cross2009} have each compared fault-tolerance schemes based on a variety of codes.  They identify larger codes that, compared to the $[[7,1,3]]$ code and the nine-qubit Bacon-Shor code, can tolerate higher noise rates with comparable resource requirements.  In particular, their estimates single out the $23$-qubit Golay code as a top performer.  

The method most commonly used for error-correction in large codes is due to Steane \cite{Steane1996}.  In this method, encoded ancilla states $\ket 0$ and $\ket +$ are prepared and used to detect errors on the data.  The complexity of ancilla preparation grows quickly with the size of the code, however, and dominates the overall cost of error-correction. 

In this chapter we present a variety of methods for reducing the cost of ancilla state preparation for CSS codes.
Our derivation is based on two main ideas.  First, we simplify Steane's Latin-rectangle-based scheme for preparing encoded $\ket 0$ states~\cite{Stea02}, by taking advantage of overlaps among the code's stabilizers. Second, we reduce the overall number of encoded ancilla states required for error correction by carefully tracking the exact propagation of errors. 

To demonstrate the utility of our approach, we give an optimized fault-tolerant error-correction procedure for the Golay code that uses only $640$ CNOT gates (compared to $1177$ for a more naive procedure), while also being highly parallelizable. All of our methods are generally applicable to other large quantum error-correcting codes.

\section{Preparation of encoded stabilizer states
\label{sec:ancilla.stabilizer-states}
}

Robust preparation of stabilizer states is a key ingredient of both Steane- and Knill-style error correction protocols. Indeed, preparation of stabilizer states is required for any fault-tolerance scheme based on stabilizer codes in order to prepare logical qubits for computation. 

One way to prepare a stabilizer state for an $n$-qubit code is to prepare any state on $n$ qubits, say $\ket{0^n}$.  Then by measuring each of the stabilizer generators (including the corresponding logical operator) the state is projected onto the one-dimensional subspace that defines the stabilizer state.  Steane has proposed an alternative method for CSS codes, which is more compact~\cite{Stea02}.

\subsection{Steane's Latin rectangle method}
Steane's method involves constructing and solving a partial Latin rectangle based on the stabilizer generators.  For simplicity, consider a $[[n,1,d]]$ CSS code.  Let $n_X$ be the number of $X$ stabilizer generators. Then the $X$ stabilizer generators form a $n_X\times n$ binary matrix in which the $X$ operators in the tensor product are represented as $1$s.  Each column represents a (physical) qubit in the code, and each row represents one stabilizer generator.  To prepare encoded $\ket 0$, Gaussian elimination is performed until the matrix is of the form 
\vspace{.25cm}
\begin{equation}
\label{eq:stabilizer-standard-form}
{\scriptstyle{n_X}} \left\{ \left(
\begin{array}{c|c}
\raisebox{0ex}[1.5ex]{\mbox{\ensuremath{\overbrace{I}^{n_X}}}} & 
\raisebox{0ex}[1.5ex]{\mbox{\ensuremath{\overbrace{A}^{n-n_X}}}}
\end{array}
\right)\right.
\end{equation}
The first $n_X$ qubits, called ``control" qubits, are prepared as $\ket{+}$, and the remaining ``target" qubits are prepared as $\ket{0}$.  The matrix $A$ is called the \emph{redundancy matrix}, and represents a partial Latin rectangle, the solution to which is used to schedule rounds of CNOT gates from control to target qubits.

For example, by swapping qubits three and four, the $X$ stabilizers of the $[[7,1,3]]$ code (\tabref{tab:7-bit-Hamming-parity-checks}) are of the form~\eqnref{eq:stabilizer-standard-form}.  A schedule of three rounds of CNOT gates and corresponding quantum circuit are shown in~\figref{fig:latin-rectangle-steane-code}.

To see that this procedure indeed prepares encoded $\ket 0$, notice that the stabilizer generators of the initial state of the control and target qubits are described by the binary matrix
\begin{equation}
\left(
\begin{array}{c|c}
I^{n_X} & 0\\
\hline
0 & I^{n_Z+1}
\end{array}
\right)
\enspace,
\end{equation}
where the first $n_X$ rows are weight-one $X$ generators and the last $n_Z + 1 = n - n_X$ rows are weight-one $Z$ generators. 
The first $n_X$ qubits are controls and the remaining $n_Z+1$ qubits are targets.  Let $S_i$ be the operator corresponding to row $i$, let $U$ be the unitary operation corresponding to the CNOT schedule, and let $\ket{\psi_0}$ be the initial state.  Then $U$ performs the transformation
\begin{equation}
\begin{split}
S_i &\mapsto US_iU^\dagger\\
\ket{\psi_0} &\mapsto U\ket{\psi_0}
\enspace .
\end{split}
\end{equation}
The operators $US_iU^\dagger$ form an independent set of stabilizers of $U\ket{\psi_0}$, the first $n_X$ of which are the $X$ stabilizer generators of the code, by construction.  The remaining $n_Z$ operators are also independent stabilizers of $U\ket{\psi_0}$.  Indeed, they form a basis for the $(n-n_X)$-dimensional subspace orthogonal to the $X$ stabilizers and are therefore equivalent to the $Z$ stabilizer generators and the $Z$ logical operator of the code.  Therefore $U\ket{\psi_0} = \lket{0}$.

The procedure for encoded $\ket +$ is entirely analogous, except that the $Z$ stabilizers are used in place of the $X$ stabilizers, and the roles of control and target are swapped.  The procedure can also be generalized to CSS codes that encode multiple qubits.

\begin{figure}
\centering
\begin{subfigure}[b]{0.28\textwidth}
$$\begin{array}{ccc|cccc}
1&2&4&3&5&6&7\\
\hline
X& & &3&2& &1\\
 &X& &1& &2&3\\
 & &X& &3&1&2
\end{array}$$
\caption{\label{fig:713-latin-schedule}}
\end{subfigure}
\hfill
\begin{subfigure}[b]{0.32\textwidth}
\centering
\includegraphics[scale=.65]{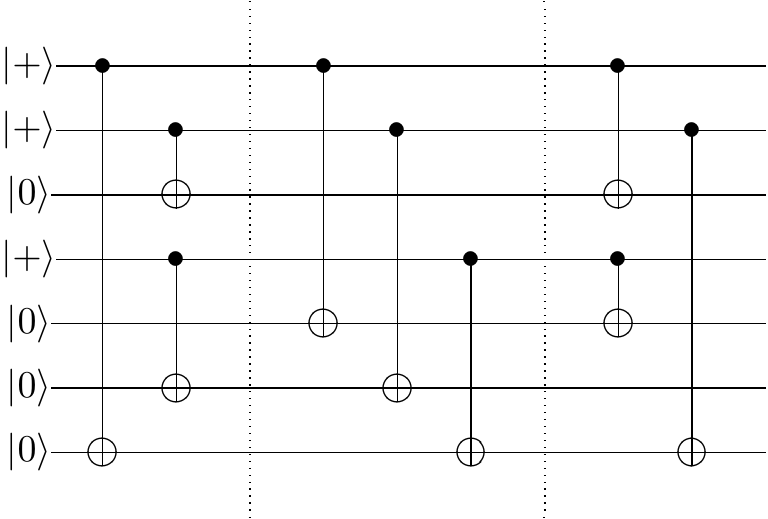}
\caption{\label{fig:steane-prep}}
\end{subfigure}
\hfill
\begin{subfigure}[b]{0.3\textwidth}
\includegraphics[scale=.65]{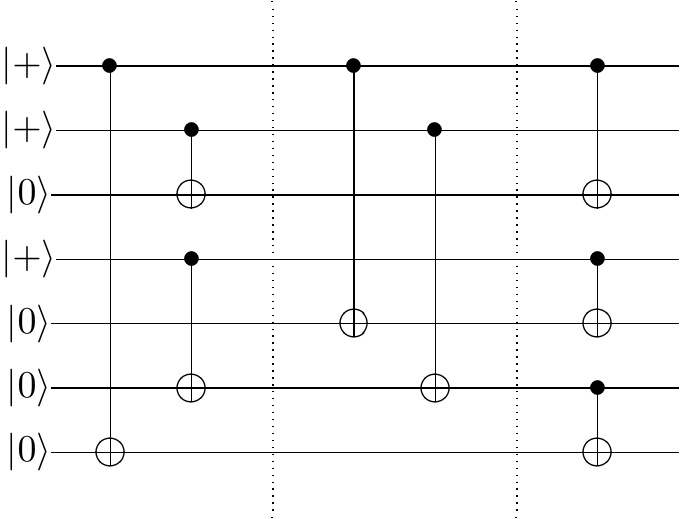}
\caption{\label{fig:steane-overlap}
}
\end{subfigure}
\caption[Preparation of encoded $\ket 0$ for the {$[[7,1,3]]$} code.]{\label{fig:latin-rectangle-steane-code}
Preparation of encoded $\ket 0$ for the $[[7,1,3]]$ code.
(a) Swapping columns three and four of~\tabref{tab:7-bit-Hamming-parity-checks} yields a $3\times 4$ matrix of the form given by~\eqnref{eq:stabilizer-standard-form}.  The four rightmost columns define a partial Latin rectangle, a solution to which is shown. (b) The Latin rectangle solution defines a schedule of CNOT gates. A nonzero value $t$ of entry $(r,c)$ specifies a CNOT on qubits $r$ and $c$ controlled by $r$ in timestep $t$.
(c)
An alternative circuit for preparing encoded $\ket 0$ using one fewer CNOT gate. The new CNOT gate has the same effect as the two removed gates.  See~\secref{sec:ancilla.stabilizer-states.overlap}.
}
\end{figure}

\subsection{Exploiting stabilizer overlap
\label{sec:ancilla.stabilizer-states.overlap}
}
Steane's Latin rectangle method treats each stabilizer generator independently. However, by taking advantage of similarities between stabilizer generators it is possible to significantly reduce the number of CNOT gates.

To explain the optimization, consider once again the $[[7,1,3]]$ code.  The Latin rectangle-based encoding schedule, shown in \figref{fig:713-latin-schedule}, needs nine CNOT gates.  An equivalent circuit requiring only eight CNOT gates is shown in \figref{fig:steane-overlap}. This circuit removes two of the CNOTs for which qubit seven is a target and replaces them with a single CNOT from qubit six to qubit seven in round three.  This works because in~\ref{fig:steane-prep} qubits six and seven are both the targets of CNOTs from qubits two and four; the corresponding stabilizer generators overlap on qubits six and seven.

The phenomenon of overlapping stabilizers generalizes to any CSS code, and the savings for larger codes is substantially greater.  However, larger codes are harder to analyze by hand.  We now describe a systematic method for optimizing stabilizer state preparation.

Our method for exploiting overlaps in large codes identifies the amount of overlap between each pair of stabilizer generators and uses the largest overlaps first.  The amount of overlap between each pair of stabilizers can be calculated by $\transpose{A}A=O$, where $A$ is the redundancy matrix of the $X$ stabilizer generators when expressed in form~\eqnref{eq:stabilizer-standard-form}, and $\transpose{A}$ is the transpose of $A$.  Entry $O(i,j)$ of this matrix corresponds to the number of non-zero entries shared by stabilizers $i$ and $j$.

The algorithm proceeds by selecting the set of disjoint pairs $\{(i_k,j_k)\}$ that yields the largest sum $\sum_k O(i_k, j_k)$, for some $k \leq n_X/2$.  The overlap between each pair of columns $(i,j)$ is then removed from column $j$ of $A$, and the process is repeated until no overlaps remain.  The schedule of CNOT gates is then obtained from the chosen pairs, and the remaining $1$s in $A$, while also accounting for the time-ordering required by the overlap CNOTs.

For example, swapping columns three and four of the $[[7,1,3]]$ code as before, we obtain an overlap matrix
\begin{equation}
\begin{tabular}[b]{c|cccc}
 &3&5&6&7\\
\hline
3&2&1&1&2\\
5& &2&1&2\\
6& & &2&2\\
7& & & &3
\end{tabular}
\enspace,
\end{equation}
where the lower triangular entries have been omitted because the matrix is symmetric.  Each diagonal entry indicates the weight of the corresponding column, and the off-diagonal entries indicate the overlap between pairs of columns.  In this example we see that column seven has overlap two with each of the other three columns.  In~\figref{fig:steane-overlap} we have chosen to use the overlap between columns six and seven.  Alternatively we could have chosen to use the overlap between columns three and seven or columns five and seven.  

In this case, column seven is the only choice that yields improvement over Steane's method.  Overlaps of one yield no net gain.  In~\secref{sec:ancilla.examples} we will examine larger codes for which there are more overlaps.

In the asymptotic setting, for arbitrarily large circuits of CNOT gates, the overlap-based method bares resemblance to the algorithm presented in~\cite{Patel2003}.  Both methods exploit similarities across columns (or rows) of a matrix to eliminate CNOT gates.  Our method differs in that we use only the redundancy matrix rather than the full $n\times n$ linear transformation, and we exploit similarities between columns without first using Gaussian elimination to make the columns identical.  This way, and by making the optimizations by hand, we are usually able to preserve circuit depth.

\subsection{Benefits of optimized preparation circuits}
The most obvious benefit of this method is the reduction in the size of the encoding circuit. For the $[[7,1,3]]$ code the number of CNOT gates was reduced from nine to eight.
In other cases, the depth of the circuit can also be decreased.
A less obvious, but more important benefit for fault-tolerance is that the number of correlated errors that can occur during the encoding circuit is also reduced.

\begin{definition}[Correlated error]
\label{def:correlated-error}
Consider an encoding circuit $C$ for a code with distance $d$.  An error $e$ caused by a set of $k \leq \lfloor (d-1)/2 \rfloor$ faulty locations in $C$ is \emph{correlated} if $e$ propagates through $C$ to an error $f$ such that $|f| > k$.  An error that is not correlated is said to be \emph{uncorrelated}.
\end{definition}
Informally, an error is correlated if its weight, modulo the stabilizers, is larger than the number of faulty locations that combined to cause the error.  This definition is motivated by the desire for strict fault tolerance (\defref{def:strict-fault-tolerance-ec}).  If each location in the circuit fails with probability $p$, then an uncorrelated error of weight $k$ occurs with probability at most $p^k$.
Preparation of stabilizer states with small numbers of correlated errors is highly desirable for fault-tolerant error correction, as we shall see in~\secref{sec:ancilla.verification}.

For the $[[7,1,3]]$ code, the number of correlated errors can be counted by hand.  Modulo the stabilizers, the only weight-two errors that can occur due to a single fault in~\figref{fig:steane-prep} are $\{X_1X_3, X_2X_7, X_3X_4\}$, and there are no weight-three errors.  Here the notation $X_i$ indicates an $X$ error on qubit $i$.  In~\figref{fig:steane-overlap}, however, there are only two possible correlated errors $\{X_1X_3,X_4X_5\}$.  A correlated $XX$ error could occur on the final CNOT between qubits six and seven.  However, $X_6X_7$ is equivalent to $X_4X_5$ modulo the stabilizer $X_4X_5X_6X_7$.  The reduction in the number of correlated errors is fairly modest for this code, but can be substantially larger for other codes.

\section{Extension to non-stabilizer states
\label{sec:ancilla.stabilizer-states.extension}
}
Stabilizer state preparation can be extended in order to encode an arbitrary state $\ket\psi$.  One way to prepare an arbitrary state is to use a teleportation protocol due to Knill~\cite{Knill2004a}.  The idea here is to prepare an encoded Bell pair and then teleport the (physical) input state $\ket\psi$ into the encoding. See~\figref{fig:encoded-psi-by-teleportation}.  The circuit requires two encoded stabilizer states $\ket 0$ and $\ket +$ plus some additional Clifford operations.

\begin{figure}
\centering
\includegraphics{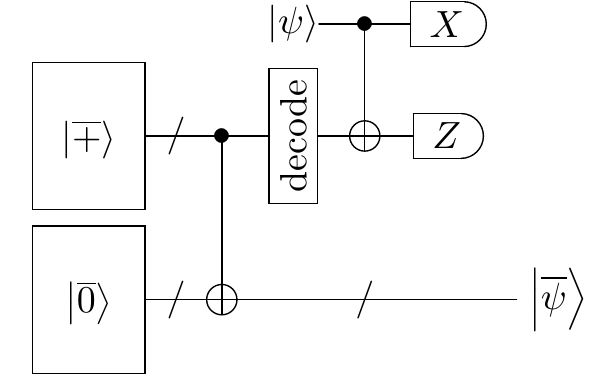}
\caption[Encoding an arbitrary state using teleportation.]{\label{fig:encoded-psi-by-teleportation}
Encoding of an arbitrary state $\ket\psi$ by teleportation~\cite{Knill2004a}. An encoded Bell pair is constructed by preparing stabilizer states $\lket 0$ and $\lket +$ and coupling with CNOT.  One half of the Bell pair is then decoded.  The decoded half is then used in a Bell measurement to teleport the input state $\ket\psi$ onto the encoded half of the Bell pair. 
}
\end{figure}

A more efficient alternative, however, is to use just the encoded $\ket 0$ preparation circuit and a controlled version of the logical $X$ operator, as shown in~\figref{fig:encoded-psi-by-ctrlX}. Let $U$ be the unitary operation implemented by the $\ket 0$ encoding circuit, and consider the operator
\begin{equation}
\label{eq:propagated-XL} 
X'_L = U^\dagger X_LU
\end{equation} 
obtained by propagating logical $X$ from the output through $U$ to the input. 
In~\figref{fig:encoded-psi-by-ctrlX} we take one of the input $\ket 0$ qubits with support on the logical operator $X'_L$, and replace it with $\ket\psi$.  Let $\tilde{X}_L$ be the part $X'_L$ that does not have support on this qubit.  Then we perform $\tilde{X}_L$, controlled by $\ket\psi$. Assuming that the encoding circuit contains only CNOT gates, $X'_L$ is a tensor product of $X$ and $I$ and so the controlled operation can be accomplished with CNOT gates.
Finally, implementing the circuit for encoded $\ket 0$ (using either Steane's method or by exploiting overlaps) outputs an encoded version $\lket\psi$ of the input state.  Here we have assumed a single-qubit state $\ket\psi$, though the procedure can be adapted to multiple qubits.

\begin{figure}
\centering
\includegraphics{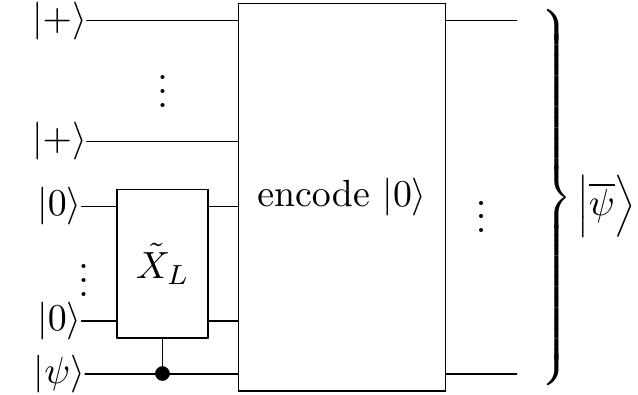}
\caption[Encoding an arbitrary state without teleportation.]{\label{fig:encoded-psi-by-ctrlX}
Encoding of an arbitrary single-qubit state $\ket\psi$ without resorting to teleportation.  Qubits are prepared as in the encoding procedure for encoded $\ket 0$, except that one of the $\ket 0$ inputs is replaced by $\ket\psi$.  Controlled on $\ket\psi$, the $X$ logical operator is conditionally applied.  Here $\tilde{X}_L$ indicates the part of logical $X$ with support disjoint from $\ket\psi$.  Implementing the encoding circuit for $\ket 0$ then yields the encoded state $\lket\psi$.
}
\end{figure}

To see that this works, we examine the effect of the circuit on each of the basis states of $\ket\psi = a\ket0 + b\ket1$. Let $\Lambda_X$ be the controlled $\tilde{X}_L$ operation, $U$ be the unitary corresponding to the encoding circuit, and $\ket\phi$ be the state of the $(n-1)$ qubits other than $\ket\psi$. We need to show that
\begin{equation}
\label{eq:encode-psi}
U \Lambda_X \ket\phi (a\ket0 + b\ket1) = a\lket{0} + b\lket{1}
\enspace .
\end{equation}

We will examine the two basis states $\ket 0$ and $\ket 1$ separately.  The result will then follow by linearity.
The case in which $\ket\psi = \ket 0$ is obvious.  In this case, the controlled $\tilde{X}_L$ gate does not activate, and we obtain $U\ket\phi\ket0 = \lket 0$, by construction.

Now consider the case $\ket\psi = \ket 1$.  Since the control activates in this case, the circuit is equivalent to setting $\ket\psi$ to $\ket 0$, applying $X'_L$ and then applying $U$.  That is
\begin{equation}
U\Lambda_X \ket\phi\ket1 = UX'_L\ket\phi\ket0
\enspace .
\end{equation}
Using~\eqnref{eq:propagated-XL}, we then obtain
\begin{equation}
UX'_L\ket\phi\ket0 = X_LU\ket\phi\ket0 = X_L\lket0 = \lket1
\enspace .
\end{equation}

Note that when using Steane's Latin rectangle construction, $X'_L = X_L$ since all of the qubits on which $X_L$ has support are targets of CNOT gates.  Any $X$ operator of $X_L$ on a control qubit can be removed by multiplying by a stabilizer. Pauli $X$ commutes through a CNOT target, and therefore $X_L$ commutes through a Latin rectangle circuit. For overlap-based circuits, the operator $X'_L$ may be somewhat different, but will still be a tensor product of $X$ and $I$.

Circuits of the form given by~\figref{fig:encoded-psi-by-ctrlX} are typically used for state distillation in topological codes~\cite{Fowler2012d}, where the code in question is the $[[7,1,3]]$ code (for the $S$ gate) or the $[[15,1,3]]$ code (for the $T$ gate).  The overlap-based optimizations given in~\secref{sec:ancilla.stabilizer-states.overlap} therefore suggest that such distillation circuits could be improved, particularly~\figref{fig:T-distillation-Bell-pair}.  In~\secref{sec:ancilla.examples.1513} we show this optimization explicitly.

\section{Examples
\label{sec:ancilla.examples}
}
The $[[7,1,3]]$ code is useful as a toy example for demonstrating the overlap optimization technique.  However, the actual resource savings are somewhat undramatic.  We now illustrate larger resource savings that can be obtained when using larger codes.

\subsection{[[15,1,3]] code
\label{sec:ancilla.examples.1513}
}
For our first example, we examine the $[[15,1,3]]$ code.  The stabilizers of this code are given by~\eqnref{eq:1513-stabilizers}.  There are four $X$ stabilizer generators, each of which have weight eight. A Latin rectangle encoding circuit for encoded $\ket 0$, therefore has size $28$ and depth seven.  There are ten $Z$ stabilizer generators. In~\eqnref{eq:1513-stabilizers} there are four generators of weight eight and six of weight four.  However, using Gaussian elimination we can obtain the following presentation in which each generator has weight four:
\begin{equation}
\label{eq:1513-reduced-Z-stabilizers}
\begin{tabular}[b]{c@{~}c@{~}c@{~}c@{~}c@{~}c@{~}c@{~}c@{~}c@{~}c@{~}c@{~}c@{~}c@{~}c@{~}c}
Z&$\cdot$&$\cdot$&$\cdot$&$\cdot$&$\cdot$&Z&$\cdot$&$\cdot$&$\cdot$&Z&$\cdot$&Z&$\cdot$&$\cdot$\\
$\cdot$&Z&$\cdot$&$\cdot$&$\cdot$&$\cdot$&Z&$\cdot$&$\cdot$&$\cdot$&Z&$\cdot$&$\cdot$&Z&$\cdot$\\
$\cdot$&$\cdot$&Z&$\cdot$&$\cdot$&$\cdot$&Z&$\cdot$&$\cdot$&$\cdot$&Z&$\cdot$&$\cdot$&$\cdot$&Z\\
$\cdot$&$\cdot$&$\cdot$&Z&$\cdot$&$\cdot$&Z&$\cdot$&$\cdot$&$\cdot$&$\cdot$&$\cdot$&Z&Z&$\cdot$\\
$\cdot$&$\cdot$&$\cdot$&$\cdot$&Z&$\cdot$&Z&$\cdot$&$\cdot$&$\cdot$&$\cdot$&$\cdot$&Z&$\cdot$&Z\\
$\cdot$&$\cdot$&$\cdot$&$\cdot$&$\cdot$&Z&Z&$\cdot$&$\cdot$&$\cdot$&$\cdot$&$\cdot$&$\cdot$&Z&Z\\
$\cdot$&$\cdot$&$\cdot$&$\cdot$&$\cdot$&$\cdot$&$\cdot$&Z&$\cdot$&$\cdot$&Z&$\cdot$&Z&Z&$\cdot$\\
$\cdot$&$\cdot$&$\cdot$&$\cdot$&$\cdot$&$\cdot$&$\cdot$&$\cdot$&Z&$\cdot$&Z&$\cdot$&Z&$\cdot$&Z\\
$\cdot$&$\cdot$&$\cdot$&$\cdot$&$\cdot$&$\cdot$&$\cdot$&$\cdot$&$\cdot$&Z&Z&$\cdot$&$\cdot$&Z&Z\\
$\cdot$&$\cdot$&$\cdot$&$\cdot$&$\cdot$&$\cdot$&$\cdot$&$\cdot$&$\cdot$&$\cdot$&$\cdot$&Z&Z&Z&Z
\end{tabular}
\enspace .
\end{equation}
The corresponding Latin rectangle circuit for encoded $\ket+$ then has size $30$.  The depth is at least six, the maximum weight of a column of~\eqnref{eq:1513-reduced-Z-stabilizers}.

By exploiting overlaps between pairs of generators, as described above, we construct the circuits shown in~\figref{fig:1513-overlap-circuits}.  The circuit for encoded $\ket 0$ has size $22$ and the circuit for $\ket +$ has size $25$, a size decrease by roughly $27\%$ and $20\%$, respectively.  The depth of both circuits is seven.  The depth for the $\ket +$ circuit has actually increased relative to the Latin rectangle circuit.  The extra timestep is necessary to exploit overlaps between two weight-six columns. 

\begin{figure}
\centering
\begin{subfigure}[b]{.45\textwidth}
\centering
\includegraphics[scale=1.15]{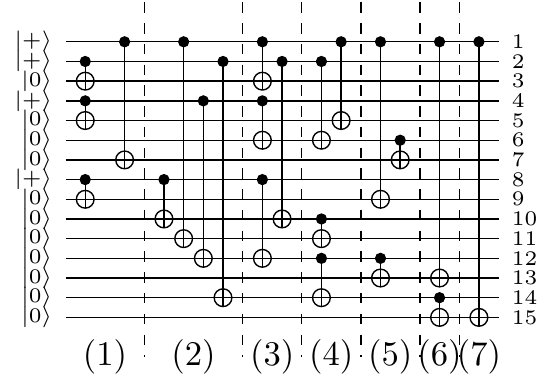}
\caption{\label{fig:1513-encode0}
Encoded $\ket 0$
}
\end{subfigure}
\begin{subfigure}[b]{.45\textwidth}
\centering
\includegraphics[scale=1.15]{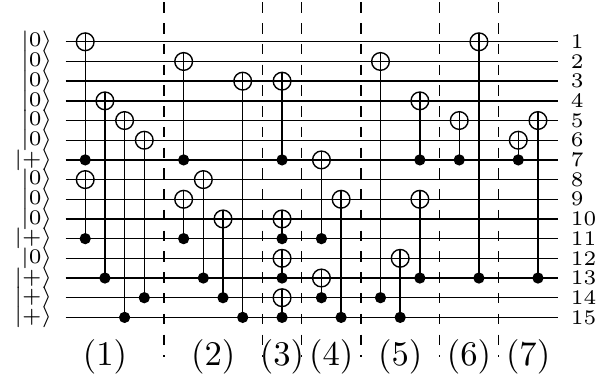}
\caption{\label{fig:1513-encode+}
Encoded $\ket +$
}
\end{subfigure}
\caption[Optimized encoding circuits for the {$[[15,1,3]]$} code.]{\label{fig:1513-overlap-circuits}
Optimized encoding circuits for the $[[15,1,3]]$ code. (a) An encoding circuit for $\ket 0$ requires $22$ CNOT gates and seven rounds. (b)  An encoding circuit for $\ket +$ requires $25$ CNOT gates and seven rounds.  Gates in the same round are applied in parallel.
}
\end{figure}

As an immediate consequence of~\figref{fig:1513-encode0}, the gate cost of state distillation can be decreased.  This circuit can be substituted for the bulk of the CNOT gates in~\figref{fig:tgate-implementation} using the protocol discussed in~\secref{sec:ancilla.stabilizer-states.extension}.  Additional savings can be obtained by noting that not all of the qubits need to be prepared at the beginning of the circuit.  For example, qubit $15$ is not required until timestep six.

Thorough analysis of the resource savings requires choosing another error-correcting code for computation and specifying any geometric constraints.  The standard $[[15,1,3]]$ encoding circuit has been heavily optimized by hand for use in the surface code, for example~\cite{Fowler2012f}.  Similar hand optimization of~\figref{fig:1513-encode0} could yield improved results, though we do not perform the required analysis here.

\subsection{Bacon-Shor codes
\label{sec:ancilla.examples.bacon-shor}
}
Next we consider the family of Bacon-Shor codes~\cite{Baco06}. For a fixed $n$, this code family  uses $n^2$ physical qubits to encode one logical qubit to a distance of $n$ and $(n-1)^2$ logical qubits to a distance of two.  Usually, only the single distance-$n$ qubit is used and the state of the remaining  ``gauge'' qubits is ignored.  In this case the code is treated as $[[n^2,1,n]]$.
 
The qubits of this code can be laid out as an $n\times n$ square lattice.  In this geometry, the stabilizer generators can be expressed in a particularly simple form.  The $X$ stabilizer generators correspond to neighboring pairs of rows, and the $Z$ stabilizers correspond neighboring to pairs of columns.  Following~\cite{Aliferis2007b}, for each row $j$ let $X_{j,*}$ be the operator that acts as a tensor product of Pauli $X$ on the qubits of row $j$ and acts trivially elsewhere.  Similarly, for each column $j$ let $Z_{*,j}$ be the operator that acts as a tensor product of Pauli $Z$ on column $j$.  Then the stabilizer generators of the code are given by
\begin{equation}
\{X_{j,*}X_{j+1,*}; Z_{*,j}Z_{*,j+1} ~|~ j\in [n-1] \}
\enspace .
\end{equation}

When presented in this way, we immediately see that each generator has weight $2n$ and, except for $X_{1,*}X_{2,*}$ and $X_{n-1,*}X_{n,*}$ (and respectively, $Z_{*,1}Z_{*,2}$ and $Z_{*,n-1}Z_{*,n}$) overlaps with two other generators in on exactly $n$ qubits.  In order to see how to take advantage of these overlaps, however, we will prefer to present the generators in a different way.  Consider the product of the last two $X$ generators $X_{n-2} := (X_{n-2,*}X_{n-1,*})(X_{n-1,*}X_{n,*}) = X_{n-2,*}X_{n,*}$.  This operator has support on rows $(n-2)$ and $n$.  We may similarly define operators $X_j$ using the recursion relation
\begin{equation}
X_j := (X_{j,*}X_{j+1,*})X_{j+1} = X_{j,*}X_{n,*}
\enspace .
\end{equation}
The set $\{X_j ~|~ j\in [n-1]\}$ forms an alternate basis of $X$ stabilizer generators for the code.  Each column of the generator matrix has weight one, except for the last $n$ columns which each have weight $n-1$.  For example, the $X$ generators for the case $n=4$ are given by
\begin{equation}
\begin{tabular}[b]{c@{~}c@{~}c@{~}c@{~}c@{~}c@{~}c@{~}c@{~}c@{~}c@{~}c@{~}c@{~}c@{~}c@{~}c@{~}c}
$X$&$X$&$X$&$X$&$\cdot$&$\cdot$&$\cdot$&$\cdot$&$\cdot$&$\cdot$&$\cdot$&$\cdot$&$X$&$X$&$X$&$X$\\
$\cdot$&$\cdot$&$\cdot$&$\cdot$&$X$&$X$&$X$&$X$&$\cdot$&$\cdot$&$\cdot$&$\cdot$&$X$&$X$&$X$&$X$\\
$\cdot$&$\cdot$&$\cdot$&$\cdot$&$\cdot$&$\cdot$&$\cdot$&$\cdot$&$X$&$X$&$X$&$X$&$X$&$X$&$X$&$X$\\
\end{tabular}
\enspace .
\end{equation}
The weight-one columns can be filled in using a total of $(n-1)^2$ CNOT gates, and column $n$ can be filled in using $(n-1)$ additional CNOTs.  Then the remaining block of $(n-1)^2$ $X$s can be filled, using overlaps, with $(n-1)$ CNOTs. The corresponding circuit prepares logical $\ket 0$ on each of the encoded qubits (including the gauge qubits) using $(n-1)(n+1)$ CNOTs. See~\figref{fig:bacon-shor-overlap-circuit}. By obtaining a similar presentation of the $Z$ generators, encoded $\ket +$ can be prepared across all logical qubits for the same cost.

\begin{figure}
\centering
\includegraphics{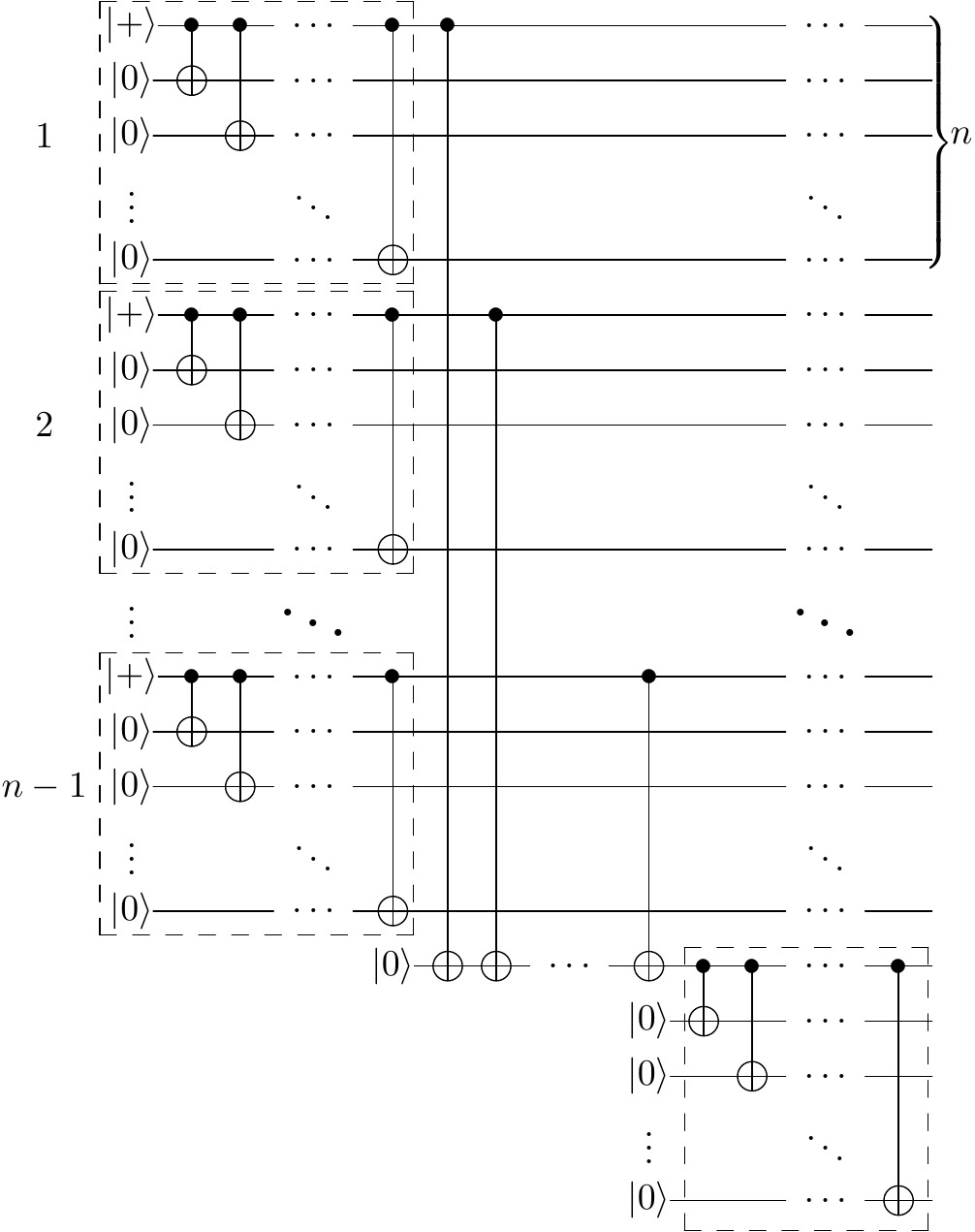}
\caption[Overlap circuit for Bacon-Shor encoded $\ket 0$.]{\label{fig:bacon-shor-overlap-circuit}
This circuit prepares logical $\ket 0$ on each of the qubits (including gauge qubits) of an $n$-qubit Bacon-Shor code.  For visual clarity, each of the $n$ boxed subcircuits use CNOTs from the same control qubit. Alternate but equivalent subcircuits can be implemented in depth $\lceil \log_2(n)\rceil$.
}
\end{figure}

From~\figref{fig:bacon-shor-overlap-circuit} we see that the circuit consists of $(n-1)$ cat state preparations, plus another circuit that also resembles a cat state.  A cat state can be prepared in depth $\lceil \log_2(n) \rceil$ using a tree-like sequence of CNOT gates, and so the entire circuit can be implemented in depth $n + 2\lceil \log_2(n) \rceil - 1$.

Indeed, Aliferis and Cross have observed that by preparing each of the gauge qubits in logical $\ket +$ rather than $\ket 0$, the encoded $\ket 0$ state (on the distance $n$ qubit) can be expressed as a tensor product of $n$ cat states $(\ket{0^n} + \ket{1^n})/\sqrt{2}$, breaking the coupling required in~\figref{fig:bacon-shor-overlap-circuit}.  Thus, if we are unconcerned with the state of the gauge qubits, then encoded $\ket 0$ can be prepared using only $n(n-1)$ CNOTs and $\lceil \log_2(n) \rceil$ timesteps.

Both~\figref{fig:bacon-shor-overlap-circuit} and the cat state method of~\cite{Aliferis2007b} compare favorably to the Latin rectangle method. The Latin rectangle method requires each of the rows to be filled separately, yielding $(n-1)(2n-1)$ CNOT gates and a depth of $2n-1$.  The overlap and cat state circuits beat this by roughly a factor of two in size.
Statistics for all three encoding methods are shown in~\tabref{tbl:bacon-shor-encoding-stats}.

\begin{table}
\centering 
\begin{tabular}{c|l@{\qquad}l}
\hline\hline
Method & Size (CNOTs) & Depth\tabularnewline
\hline
Latin rect. & $(n-1)(2n-1)$ & $2n-1$\tabularnewline
Cat state   & $n(n-1)$      & $\lceil \log_2(n)\rceil$\tabularnewline
Overlap     & $(n-1)(n+1)$  & $n + 2\lceil \log_2(n)\rceil - 1$\tabularnewline
\hline\hline
\end{tabular}
\caption[Circuit statistics for Bacon-Shor encoded $\ket 0$.]{\label{tbl:bacon-shor-encoding-stats}
Circuit statistics for encoding $\ket 0$ or $\ket +$ for a $[[n^2,1,n]]$ Bacon-Shor code.  Column one shows the Latin rectangle method due to~\cite{Stea02}, column two shows the cat state method due to~\cite{Aliferis2007b} and column three shows the overlap method from~\secref{sec:ancilla.stabilizer-states.overlap}.  The cat state method of~\cite{Aliferis2007b} prepares $\ket 0$ on the distance-$n$ logical qubit and $\ket +$ on each of the gauge qubits, whereas the overlap method prepares $\ket 0$ on each of the gauge qubits.
} 
\end{table}

This example also illustrates why exploiting stabilizer overlaps reduces the number of correlated errors produced by the encoding circuit when compared to the Latin rectangle method.  Reichardt has observed that the correlated errors in a Latin rectangle circuit can be characterized in a systematic way~\cite{Reichardt2006}.  Consider a single $X$ stabilizer generator of weight $m$.  Ignoring the qubits on which this generator acts trivially, the circuit for this generator is of the form
\begin{equation}
\label{eq:cat-state-circuit}
\includegraphics{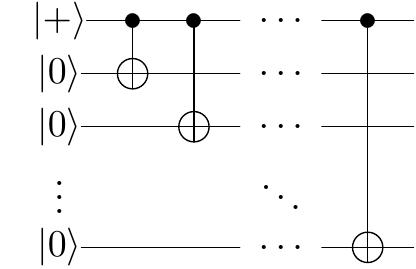}
\enspace .
\end{equation}
Next consider the $X$ errors that can occur as a result of a single faulty gate in the circuit.
Pauli $X$ errors on target qubits do not propagate and are uncorrelated.  Any correlated $X$ error must have support on the first qubit and some consecutive sequence of qubits $\{j,\ldots,m\}$ for $j > 1$.  Up to multiplication by the stabilizer, $X_1\ldots X_m$ is trivial and $X_1X_3\ldots X_m$ has weight one.  Thus, there are exactly $(m-2)$ unique correlated errors that occur with first-order probability.

Of course, the stabilizer generators of the entire code are not disjoint, and so the total number of first-order correlated errors is more complicated to compute.  However, in the case of the Bacon-Shor code the intersections between stabilizers are particularly simple, and do not affect the analysis.  An $n$-qubit Bacon-Shor code has $(n-1)$ $X$ generators each of weight $2n$, and so a Latin rectangle encoding circuit will contain $(n-1)(2n-2)$ correlated $X$ errors to first-order.

The situation for the overlap-based circuit is somewhat different. From~\figref{fig:bacon-shor-overlap-circuit} we see that the encoding circuit contains $n$ subcircuits of the same form as~\eqnref{eq:cat-state-circuit}.  Each of these subcircuits can produce $(n-2)$ correlated $X$ errors from a first-order fault.  The extra CNOT gates that span the circuit add another $(n-1)$ such correlated errors.  Thus the entire circuit can produce $n(n-1) - 1$ order-one correlated $X$ errors, roughly half of the number of correlated $X$ errors produced by a corresponding Latin rectangle circuit.

\subsection{Golay code
\label{sec:ancilla.examples.golay}
}

In our final example, we construct circuits for encoding $\ket 0$ in the $23$-qubit Golay code.  The Golay code has $11$ $X$ stabilizer generators, each of weight eight:
\begin{equation} \label{eq:golay-stabilizers}
\begin{tabular}{c@{}c@{}c@{}c@{}c@{}c@{}c@{}c@{}c@{}c@{}c@{}c@{}c@{}c@{}c@{}c@{}c@{}c@{}c@{}c@{}c@{}c@{}c@{}}
$\cdot$ &$X$ &$\cdot$ &$\cdot$ &$X$ &$\cdot$ &$\cdot$ &$X$ &$X$ &$X$ &$X$ &$X$ &$\cdot$ &$\cdot$ &$\cdot$ &$\cdot$ &$\cdot$ &$\cdot$ &$\cdot$ &$\cdot$ &$\cdot$ &$\cdot$ &$X$\\
$X$ &$\cdot$ &$\cdot$ &$X$ &$\cdot$ &$\cdot$ &$X$ &$X$ &$X$ &$X$ &$X$ &$\cdot$ &$\cdot$ &$\cdot$ &$\cdot$ &$\cdot$ &$\cdot$ &$\cdot$ &$\cdot$ &$\cdot$ &$\cdot$ &$X$ &$\cdot$\\
$\cdot$ &$X$ &$X$ &$\cdot$ &$X$ &$X$ &$X$ &$\cdot$ &$\cdot$ &$\cdot$ &$X$ &$X$ &$\cdot$ &$\cdot$ &$\cdot$ &$\cdot$ &$\cdot$ &$\cdot$ &$\cdot$ &$\cdot$ &$X$ &$\cdot$ &$\cdot$\\
$X$ &$X$ &$\cdot$ &$X$ &$X$ &$X$ &$\cdot$ &$\cdot$ &$\cdot$ &$X$ &$X$ &$\cdot$ &$\cdot$ &$\cdot$ &$\cdot$ &$\cdot$ &$\cdot$ &$\cdot$ &$\cdot$ &$X$ &$\cdot$ &$\cdot$ &$\cdot$\\
$X$ &$X$ &$X$ &$X$ &$\cdot$ &$\cdot$ &$\cdot$ &$X$ &$\cdot$ &$\cdot$ &$X$ &$X$ &$\cdot$ &$\cdot$ &$\cdot$ &$\cdot$ &$\cdot$ &$\cdot$ &$X$ &$\cdot$ &$\cdot$ &$\cdot$ &$\cdot$\\
$X$ &$\cdot$ &$X$ &$\cdot$ &$X$ &$\cdot$ &$X$ &$X$ &$X$ &$\cdot$ &$\cdot$ &$X$ &$\cdot$ &$\cdot$ &$\cdot$ &$\cdot$ &$\cdot$ &$X$ &$\cdot$ &$\cdot$ &$\cdot$ &$\cdot$ &$\cdot$\\
$\cdot$ &$\cdot$ &$\cdot$ &$X$ &$X$ &$X$ &$X$ &$\cdot$ &$X$ &$X$ &$\cdot$ &$X$ &$\cdot$ &$\cdot$ &$\cdot$ &$\cdot$ &$X$ &$\cdot$ &$\cdot$ &$\cdot$ &$\cdot$ &$\cdot$ &$\cdot$\\
$\cdot$ &$\cdot$ &$X$ &$X$ &$X$ &$X$ &$\cdot$ &$X$ &$X$ &$\cdot$ &$X$ &$\cdot$ &$\cdot$ &$\cdot$ &$\cdot$ &$X$ &$\cdot$ &$\cdot$ &$\cdot$ &$\cdot$ &$\cdot$ &$\cdot$ &$\cdot$\\
$\cdot$ &$X$ &$X$ &$X$ &$X$ &$\cdot$ &$X$ &$X$ &$\cdot$ &$X$ &$\cdot$ &$\cdot$ &$\cdot$ &$\cdot$ &$X$ &$\cdot$ &$\cdot$ &$\cdot$ &$\cdot$ &$\cdot$ &$\cdot$ &$\cdot$ &$\cdot$\\
$X$ &$X$ &$X$ &$X$ &$\cdot$ &$X$ &$X$ &$\cdot$ &$X$ &$\cdot$ &$\cdot$ &$\cdot$ &$\cdot$ &$X$ &$\cdot$ &$\cdot$ &$\cdot$ &$\cdot$ &$\cdot$ &$\cdot$ &$\cdot$ &$\cdot$ &$\cdot$\\
$X$ &$\cdot$ &$X$ &$\cdot$ &$\cdot$ &$X$ &$\cdot$ &$\cdot$ &$X$ &$X$ &$X$ &$X$ &$X$ &$\cdot$ &$\cdot$ &$\cdot$ &$\cdot$ &$\cdot$ &$\cdot$ &$\cdot$ &$\cdot$ &$\cdot$ &$\cdot$
\end{tabular}
\end{equation}
The $Z$ stabilizers are entirely symmetric (the code is self-dual).  The logical $X$ and $Z$ operators correspond to transversal $X$ and transversal $Z$, respectively.

Latin rectangle circuits for $\ket 0$ use $77$ CNOT gates and seven time steps.  The overlap optimized circuit for $\ket 0$ also has depth $7$ but uses only $57$ CNOT gates, a savings of about $35\%$. See~\figref{fig:Golay-overlap-57}.  Since the $X$ and $Z$ stabilizers of the Golay code are symmetric, $\ket +$ can be prepared from the circuit for $\ket 0$ by taking the dual circuit in the standard way.

\begin{figure}
\centering
\includegraphics[scale=1]{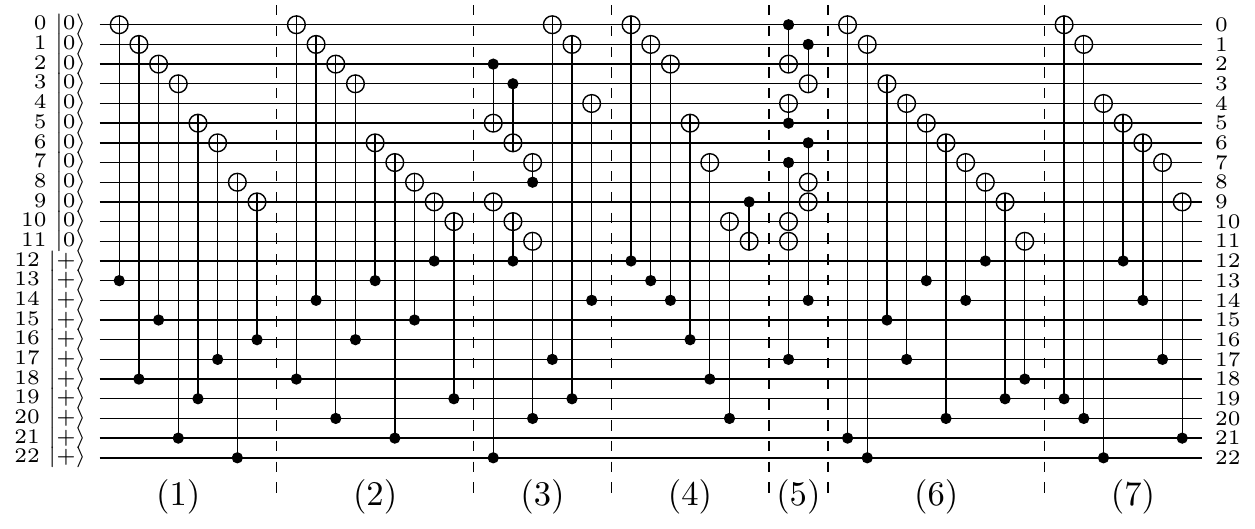}
\caption[A $57$ CNOT circuit for encoded $\ket 0$ in the Golay code.]{\label{fig:Golay-overlap-57}
An optimized circuit for preparing $\ket 0$ encoded in the Golay code uses $57$ CNOT gates applied in seven rounds.  Gates in the same round are applied in parallel.}
\end{figure}

By reducing the number of CNOT gates, this circuit also reduces the number of correlated errors.  For example, a single failure in the Latin rectangle encoded circuits can cause up to $22$ weight-two errors, but a single failure in~\figref{fig:Golay-overlap-57} can only cause up to $16$ weight-two errors.  The contrast for second-order faults is even larger.  The improvement for the overlap optimized circuit is roughly a factor of two. The correlated error counts for first and second order are shown in \tabref{tbl:overlap-errorWeights}.

\begin{table}
\centering
\begin{subtable}[b]{.47\textwidth}
\begin{tabular}{c|l@{~~}l@{~~}l@{~~}l@{~~}l@{~~}l@{~~}l|}
\hline \hline 
$X$-error weight: & 
2 & 3 & 4 & 5 & 6 & 7 \tabularnewline
\hline
Order 1: &
16 & 14 & 4 & 0 & 0 & 0 \tabularnewline
Order 2: & 
- & 493 & 400 & 35 & 2 & 0 \tabularnewline
\hline \hline 
\end{tabular}
\caption{\label{tbl:overlap-errorWeights}
Overlap
} 
\end{subtable}
\hfill
\begin{subtable}[b]{.47\textwidth}
\centering 
\begin{tabular}{c|l@{~~}l@{~~}l@{~~}l@{~~}l@{~~}l@{~~}l|}
\hline \hline 
$X$-error weight: & 
2 & 3 & 4 & 5 & 6 & 7 \tabularnewline
\hline
Order 1: &
22&22&11&0&0&0 \tabularnewline
Order 2: & 
- &848&718&73&3&0 \tabularnewline
\hline \hline 
\end{tabular}
\caption{\label{tbl:Golay-Latin-correlated-X-errors}
Latin rectangle
} 
\end{subtable}
\caption[Correlated $X$ error counts for Golay encoded $\ket 0$.]{
Correlated $X$ error counts for circuits encoding $\ket 0$ in the Golay code.
(a) Correlated $X$ error counts for the overlap optimized circuit in \figref{fig:Golay-overlap-57}. (b) Correlated $X$ error counts for a Latin rectangle encoding circuit (not shown).
}
\end{table}

\begin{table}
\end{table}

We briefly note that the overlap method, and the circuit in~\figref{fig:Golay-overlap-57} in particular, may not be optimal.  Indeed there are equivalent circuits with fewer CNOT gates. However, \figref{fig:Golay-overlap-57} is the smallest circuit we found that also preserves \emph{depth}. 

\section{Encoded ancilla verification
\label{sec:ancilla.verification}
}
None of the stabilizer state preparation circuits shown thus far are fault tolerant.
A single physical fault may lead to errors on multiple qubits. For example, an $XX$ error on the final CNOT of~\figref{fig:steane-overlap} leaves the weight-two error $X_6X_7$.  The code is limited by its distance and cannot necessarily protect against such correlated errors. As a result, the ancilla states themselves must be checked for errors.  
The primary task of fault-tolerant ancilla preparation then, is to prevent errors in the preparation circuit from spreading through the ancilla block. 

\subsection{Steane-style verification}
One way to check for errors which is particularly useful for large CSS codes is to use a Steane-style error-detection circuit.  To check for $X$ errors, a second encoded ancilla is prepared as $\ket +$ and a transversal CNOT is used to copy errors from the first ancilla to the second, as shown in~\figref{fig:verify-X}. If the $Z$-basis measurement implies the presence of an error, then the ancilla is discarded and the process begins again. To check for $Z$ errors, we instead prepare encoded $\ket 0$ and swap the control and target of the CNOT.  However, correlated $X$ errors that occur during preparation of $\ket 0$ can propagate through the CNOT to the original ancilla.  To prevent this we first check the $\ket 0$ state for $X$ errors, and then proceed to use it for $Z$ error detection, as in~\figref{fig:verify-Z}.  Again, if an error is detected, the the ancilla is discarded.

\begin{figure}
\centering
\begin{subfigure}[b]{.4\textwidth}
\centering
\includegraphics{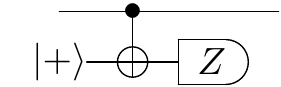}
\caption{\label{fig:verify-X}
$X$-error verification
}
\end{subfigure}
\begin{subfigure}[b]{.4\textwidth}
\centering
\includegraphics{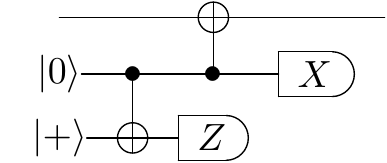}
\caption{\label{fig:verify-Z}
$Z$-error verification
}
\end{subfigure}
\caption[First order verification circuits.]{\label{fig:verify-first-order}
First-order verification circuits. (a) $X$ errors are copied onto the encoded $\ket +$ ancilla and then detected by the $Z$-basis measurement. (b) An encoded $\ket 0$ ancilla is first checked for $X$ errors in order to prevent $X$ errors from spreading to the top qubit.  Then $Z$ errors are copied from the top qubit and detected by the $X$-basis measurement.
}
\end{figure}

The circuits in~\figref{fig:verify-first-order} are sufficient to detect correlated errors up to first order.  But for high distance codes we desire verification up to order $t = \lfloor (d-1)/2 \rfloor$.  Higher-order verification can be accomplished by using additional and more complex hierarchical error detection circuits.  In general, $(t+1)t+1$ encoded ancillas are sufficient to produce a single ancilla verified to order $t$.  For example, use $t$ $X$-error verifications, followed by $t$ $Z$-error verifications in which each encoded $\ket 0$ ancilla has been verified using an additional $t$ $X$-error verifications.
The total overhead required to prepare a fault-tolerant ancilla depends also on the probability that any errors are detected.

To maximize efficiency, preparation and verification circuits may be constructed using a pipeline architecture in which part of the computer is dedicated to preparing many ancillas in parallel. Even so, ancilla production constitutes the majority of the space requirement for a fault-tolerant quantum circuit. In~\cite{Isailovic2008a}, for example, the ancilla pipeline is estimated to take up to $68$ percent of the entire circuit footprint.

% \tabref{tbl:verify-compare} shows estimates of the probability that all of the verification stages accept along with the corresponding expected resource requirements for the different verification circuits. For depolarizing noise rates near $p = 10^{-3}$, our circuits reduce both the expected number of qubits and the expected number of CNOT gates by roughly a factor of four over the twelve-ancilla circuit.  A more detailed analysis of the acceptance probability and overhead is given in \secref{sec:Overhead}.  
  
% The circuit in \figref{fig:overlap-prep-ckt} is clearly not strictly fault-tolerant because, for example, with first order probability a faulty CNOT gate may produce an error of weight two when an $X$-error occurs on both its control qubit and its target qubit.  We call this error, and any other for which the weight of the resulting error is larger than the probability order with which it occurs, ``correlated''.  
% 
% To achieve strict fault tolerance, additional encoded ancillas are prepared.  Errors from the original ancilla are copied onto the additional ancillas which are then measured.  If no errors are detected, the ancilla is accepted and may be used for error correction on the data.  Otherwise, the ancilla is rejected, all of the prepared ancillas are discarded, and the procedure is restarted.  In \figref{fig:FourAncillaVerifyCkt}, two pairs of ancillas are prepared.  One of the ancillas from each pair is checked for $X$ errors.  If neither check detects an error, then one of the two remaining ancillas is used to check the other for $Z$ errors.  

One of the reasons that a hierarchical verification structure is required is because identically prepared stabilizer states produce identical sets of correlated errors.  For example, say that two encoded ancillas are identically prepared.  Assume that a single failure occurs in the first ancilla and propagates through the preparation circuit to produce a weight three error.  Then the same single failure in the other ancilla will produce the \emph{same} weight three error.  When the error from the first ancilla is copied to the second, the two errors will cancel each other and no error will be detected.  This is a second-order event that results in a weight-three error.

However, DiVincenzo and Aliferis~\cite{DiVi06} have observed that different preparation circuits exhibit different error propagation behavior, and this can be exploited.  Intuitively, if the sets of errors produced by two different preparation circuits are sufficiently different, only a small number of errors will cancel out at each verification, and fewer verifications steps will be required overall.
Therefore, we seek to prepare encoded ancillas that produce different correlated error sets.  In the next section we analyze the correlated errors produced by preparation circuits for the Golay code, and randomized methods for finding ancillas with different correlated error sets.

\subsection{Optimization by counting correlated errors in the Golay code}
\label{sec:ancilla.verify.golay}
Since the circuits and therefore the correlated errors differ depending on the employed error-correcting code, the verification circuits that can be obtained by mixing preparation circuits will also differ.  The most concrete way to show the benefits of this technique are with an example.  In this section we consider the $23$-qubit Golay code.  The Golay code is an illustrative example because it has relatively large distance, but is small enough for manual inspection.  Furthermore, estimates show that the Golay code has a fairly high threshold.  The examples discussed here will also be used in~\chapref{chap:threshold} to prove a lower bound on the threshold for the Golay code.

For the Golay code, the standard recursive verification technique requires twelve encoded ancillas and at least $1177$ CNOT gates.  One such circuit is shown in \figref{fig:TwelveAncillaVerifyCkt}. Variants of this circuit have been used in previous studies of the Golay code, including in~\cite{Steane2003} and~\cite{Cross2009}.  
By considering many different preparation circuits, we find that the number of ancillas can be significantly reduced. We now outline two methods that produce circuits of the form shown in~\figref{fig:FourAncillaVerifyCkt}, requiring only four encoded $\ket{0}$ ancillas and as few as $297$ CNOT gates.

%One of these circuits is specified by Figures~\ref{fig:FourAncillaVerifyCkt} and~\ref{fig:overlap-prep-ckt}, and \tabref{tbl:overlap-permutations}.  

\begin{figure}
\centering
\includegraphics[width=14.2cm]{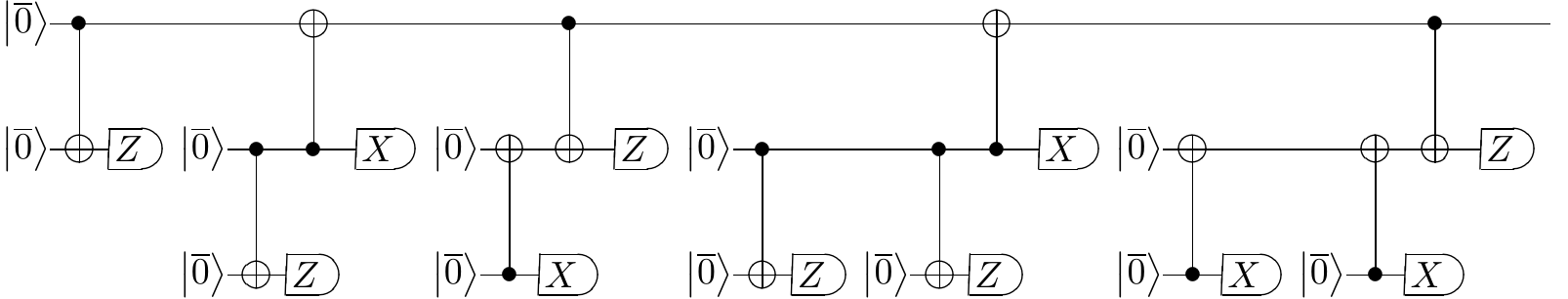}
\caption[Twelve ancilla verification circuit for the Golay code.]{\label{fig:TwelveAncillaVerifyCkt}
This circuit produces a single Golay encoded $\ket 0$ state that is ready to be used in fault-tolerant error correction.  Each of the twelve encoded $\ket{0}$ ancillas, denoted $\lket{0}$, is identically prepared using the Steane Latin rectangle method (see \secref{sec:RandomSteane}).  The wires represent $23$-qubit code blocks and the indicated CNOT and measurement operations are transversal.
} 
\end{figure}

\begin{figure}
\centering
\includegraphics[scale=1]{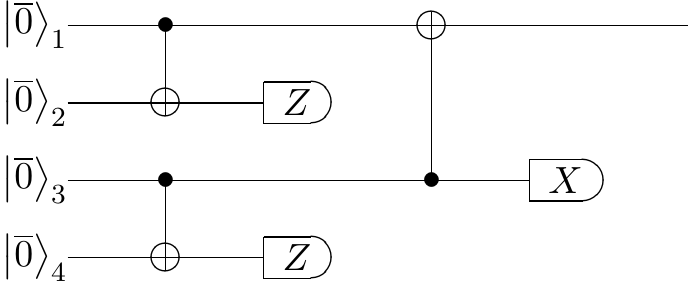}
\caption[Four ancilla verification circuit for the Golay code.]{Our simplified ancilla preparation and verification circuit uses only four encoded $\ket 0$ ancillas.  The ancillas are prepared using different encoding circuits, shown in \figref{fig:Golay-overlap-57} and \tabref{tbl:overlap-permutations}, and also in~\tabref{tab:randomizedSchedules}.}  
\label{fig:FourAncillaVerifyCkt}
\end{figure}

\subsubsection{Randomized method for preparing encoded \texorpdfstring{$\ket 0$}{0}} \label{sec:RandomSteane}

An $X$ error in the preparation circuit can propagate to other qubits only if it occurs on a control qubit, and then only through the $X$ stabilizer being created from that control qubit.  Thus single faults can create up to $22$ weight-two errors (for each of the eleven $X$ stabilizers, either $IIIIIIXX$ or $IIXXXXXX \sim XXIIIIII$), $22$ weight-three errors and eleven weight-four errors ($IIIIXXXX$ for each stabilizer).  

A single $X$ fault, i.e., a fault resulting in an $X$ error, cannot break the verification circuit in \figref{fig:FourAncillaVerifyCkt}.  If it creates a correlated error on the first ancilla, that error will be detected on the second ancilla, and both will be discarded.  Four or more $X$ faults also cannot break the verification circuit because we only seek fault tolerance up to order three.  

Two $X$ faults can break the verification circuit only if there is one failure in each ancilla preparation that propagates to an error of weight at least three---necessarily the same error so that it is undetected.  To obtain a crude estimate for how likely this is to occur, consider a circuit obtained by sampling uniformly at random over all possible circuits that prepare encoded $\ket 0$. (Several methods for approximating such a sample are discussed below.) Pretend that the correlated errors created by such a circuit are uniformly distributed among all errors of the same weights.  The number of errors on encoded $\ket{0}$ for each weight are given in \tabref{tbl:ketZero-errorWeights}.  Then the probability that two preparation circuits share no such correlated errors is estimated as 
\begin{equation*}
\frac{\binom{1771-22}{22}}{\binom{1771}{22}} \cdot \frac{\binom{1771-11}{11}}{\binom{1771}{11}} \approx 0.71
 \enspace .
\end{equation*}
Here, $\binom{1771-22}{22}$ is the number of ways to select $22$ weight-three $X$ errors on the second ancilla such that none of them correspond to the $22$ weight-three errors on the first ancilla.  Similarly $\binom{1771-11}{11}$ is the number of ways to select $11$ weight-four $X$ errors on the second ancilla.

\begin{table}
\centering 
\begin{tabular}{c|cccccccc}
\hline \hline 
Weight: & 
0 & 1 & 2 & 3 & 4 & 5 & 6 & 7 \tabularnewline
\hline
Number of $X$ errors: &
1 & 23 & 253 & 1771 & 1771 & 253 & 23 & 1 \tabularnewline
Number of $Z$ errors: & 
1 & 23 & 253 & 1771 & 0 & 0 & 0 & 0 \tabularnewline
\hline \hline 
\end{tabular}
\caption[Distribution of errors for Golay encoded $\ket 0$.]{The number of errors on Golay encoded $\ket{0}$ by Hamming weight.  All $Z$ errors are correctable so there are no $Z$ errors of weight greater than three.}
\label{tbl:ketZero-errorWeights}
\end{table}

Three $X$ errors can break the circuit if they lead to an undetected error of weight four or greater on the first ancilla.  Consider the case that there are two failures while preparing the first ancilla and one failure while preparing the second ancilla.  The number of different weight-four errors created with second-order probability (i.e., excluding those created with first-order probability) depends on the circuit. For ten random circuits, the smallest count we obtained was $688$ and the largest $735$, with an average of $711$.  Using this average value, we estimate that the probability of a random circuit succeeding against three $X$ errors is roughly $[\binom{1771-711}{11} / \binom{1771}{11}]^2 \approx 1.2 \cdot 10^{-5}$.  (Here the square is because we want the circuit to work against both the case of two failures in the first ancilla, one failure in the second, and vice versa.) Overall, we expect to have to try about $1.2 \cdot 10^5$ random pairs of preparation circuits before we find one that gives fully fault-tolerant $X$-error verification.

The result of $X$-error verification is a single ancilla free of correlated $X$ errors up to weight-three, but possibly containing correlated $Z$ errors.  The $Z$-error propagation can be analyzed in a manner similar to that used for $X$ errors.  A single failure in an $X$-error verified ancilla can produce roughly $60$ $Z$ errors of weight three. Again assuming a uniform distribution, the probability of finding two $X$-error verified ancillas that share no correlated $Z$ errors of weight three is $\binom{1771-60}{60} / \binom{1771}{60} \approx 0.12$.  In total, we expect to try about five $X$-error fault-tolerant pairs in order to find two pairs that are fully fault-tolerant for both $X$-error and $Z$-error verification, as $\binom{5}{2} = 10$.  

To find fault-tolerant verification circuits in this way, one needs to be able to generate sufficiently random preparation circuits.  As the Latin rectangle procedure for finding encoding circuits is fully algorithmic, it can be randomized by starting with a random presentation of the Golay code.  Alternatively, one can begin with a fixed encoding circuit and randomly permute the seven rounds of CNOT gates (all of the CNOTs commute).  
The Golay code is preserved by qubit permutations in a symmetry group known as the Mathieu group $M_{23}$.\footnote{This symmetry is inherited from the classical $23$-bit Golay code.  See, e.g.,~\cite{Huffman1998} pp.\ 1411.  Generators for this group can be obtained at~\cite{Gang1999}.} Therefore another option is to permute encoding circuits based on random elements of $M_{23}$.  By trying roughly $10^5$ random pairs, we found $14$ pairs of ancillas that were fully fault-tolerant against $X$ errors.  Of the $\binom{14}{2}$ combinations, six were also fully fault-tolerant against $Z$ errors.  \tabref{tab:randomizedSchedules} presents one such set.  

\begin{table}
\newcommand{\myfontsize}{\fontsize{10}{11}\selectfont}
\myfontsize
\def\colspace{{$\;\,$}}
\centering
\begin{subtable}[b]{.45\textwidth}
\centering
\begin{tabular}{c|r@{\colspace}r@{\colspace}r@{\colspace}r@{\colspace}r@{\colspace}r@{\colspace}r}
  & 1& 2& 3& 4& 5& 6& 7\\ \hline
 2& 0&22& 7&11& 4& 8&19\\
 3& 9&19& 4& 8& 7& 1& 6\\
10& 5& 1& 0& 6&14& 7& 9\\
12& 1& 0&14& 5&22&11& 4\\
13& 6& 8&22& 9& 0& 4& 5\\
15& 4& 5& 9&14&19&22& 7\\
16&14& 7& 5& 4&11& 6& 8\\
17& 8&11& 6&19& 5& 0& 1\\
18& 7& 9& 1&22& 8& 5&11\\
20&19& 6&11& 7& 1&14&22\\
21&11& 4&19& 0& 6& 9&14
\end{tabular}
\caption{Ancilla 1}
\vspace{.5cm}
\end{subtable}
\hfill
\begin{subtable}[b]{.45\textwidth}
\centering
\begin{tabular}{c|r@{\colspace}r@{\colspace}r@{\colspace}r@{\colspace}r@{\colspace}r@{\colspace}r}
  & 1& 2& 3& 4& 5& 6& 7\\ \hline
 0& 5&16&17&22& 1&15& 9\\
 3&15& 2& 6& 5&17&16&11\\
 7& 1&22& 4&17& 2& 5& 6\\
 8& 6&13&16& 1&15& 4&17\\
10&22&11& 5&13&16& 6& 1\\
12& 9&17&13& 2& 6&22&16\\
14& 4& 6&11&15&13& 2&22\\
18&16& 1&15&11& 9&13& 2\\
19&17& 4& 1& 9&22&11&13\\
20&11&15& 9& 6& 4& 1& 5\\
21& 2& 5&22&16&11& 9& 4
\end{tabular}
\caption{Ancilla 2}
\vspace{.5cm}
\end{subtable}
\begin{subtable}{.45\textwidth}
\centering
\begin{tabular}{c|r@{\colspace}r@{\colspace}r@{\colspace}r@{\colspace}r@{\colspace}r@{\colspace}r}
  & 1& 2& 3& 4& 5& 6& 7\\ \hline
 1&21&16& 7&13&10&15& 0\\
 2&16& 7&12& 0&18&19&13\\
 3&13& 0&15&12&19&10&20\\
 4&12&21&18&20& 7&13&10\\
 5& 6&13&21&10& 0&18&19\\
 8&18&19&13&21&15&20&16\\
 9&19& 6&10&15&20& 7&21\\
11&20&12& 6& 7&13&16&15\\
14& 7&18&20&16&21& 0& 6\\
17& 0&15&19& 6&16&21&12\\
22&10&20&16&19& 6&12&18
\end{tabular}
\caption{Ancilla 3}
\end{subtable}
\hfill
\begin{subtable}{.45\textwidth}
\centering
\begin{tabular}{c|r@{\colspace}r@{\colspace}r@{\colspace}r@{\colspace}r@{\colspace}r@{\colspace}r}
  & 1& 2& 3& 4& 5& 6& 7\\ \hline
 0& 1&16& 3&12&17&13&11\\
 2&22&18&14& 3&20&17& 6\\
 4& 3&20& 6& 1&12&22&13\\
 5& 6&14&16&20& 1&12&17\\
 7&20&22&17&13&16&18& 1\\
 8&16& 6&18&11& 3& 1&20\\
 9&18&12&13&16&14&20& 3\\
10&14&17&20&22&13&11&12\\
15&12&11& 1&17&18& 6&22\\
19&17&13&11&18& 6&16&14\\
21&11& 3&12& 6&22&14&16
\end{tabular}
\caption{Ancilla 4}
\end{subtable}
\caption[Random ancilla preparation schedules for Golay encoded $\ket 0$.]{\label{tab:randomizedSchedules}
Four seven-round ancilla-preparation schedules.  In each table, the entry in row $i$, column $j$ specifies the target qubit of a CNOT gate with control qubit $i$ applied in round~$j$.  Using these schedules in the verification circuit of \figref{fig:FourAncillaVerifyCkt}, the output encoded $\ket 0$ state is fully fault-tolerant against both $X$ and $Z$~errors.  
} 
\end{table}

\subsubsection{Overlap method for preparing encoded \texorpdfstring{$\ket 0$}{0}} \label{sec:Overlap}

Ideally, though, we could use preparation circuits based on the overlap optimization of~\secref{sec:ancilla.stabilizer-states.overlap}. 
The smaller number of correlated errors produced by~\figref{fig:Golay-overlap-57} means that it should be easier to find fault-tolerant circuits.  However, unlike Latin rectangle schedules the overlap-based schedule depends on a fixed code presentation and on a fixed round ordering, since the CNOT gates do not commute.

To obtain randomized overlap method encoding circuits, we use the qubit permutation symmetry of the Golay code and permute the qubits of \figref{fig:Golay-overlap-57} according to a pseudo-random element of the symmetry group~$M_{23}$.  By analyzing the correlated error sets of randomly permuted circuits, we have found many sets of fault-tolerant four-ancilla preparation circuits.  In fact, we have even found sets for which the fault order required for a weight-$k$ error to pass verification is at least $k+1$ (rather than $k$) for all $k \leq 2$. This reduces, for example, the probability of accumulating an uncorrectable error on the data block by first a weight-two error in $Z$-error correction and then another weight-two error in $X$-error correction.  One such set of four permutations is given in \tabref{tbl:overlap-permutations}.  

\begin{table}
\centering
\small
\begin{tabular}{c|c}
\hline\hline
Ancilla & Qubit permutation \tabularnewline
\hline
$\lket{0}_2$ &  (0, 20, 13, 7, 12, 14, 1)(2, 11)(3, 19, 5, 4, 8, 22, 6, 15, 10, 16, 9, 18, 21, 17) %(20, 0, 11, 19, 8, 4, 15, 12, 22, 18, 16, 2, 14, 7, 1, 10, 9, 3, 21, 5, 13, 17, 6) 
\tabularnewline
$\lket{0}_3$ &  (0, 14, 6, 12, 16, 2, 11, 22, 17, 21, 9, 20, 5, 7, 3, 13, 18, 4, 15, 1, 10, 8, 19) %(14, 10, 11, 13, 15, 7, 12, 3, 19, 20, 8, 22, 16, 18, 6, 1, 2, 21, 4, 0, 5, 9, 17) 
\tabularnewline
$\lket{0}_4$ & (0, 12, 4, 17, 9, 6, 1)(2, 10, 18, 22, 21, 16, 13)(3, 11, 20, 15, 7, 19, 5)(8)(14) %(12, 0, 10, 11, 17, 3, 1, 19, 8, 6, 18, 20, 4, 2, 14, 7, 13, 9, 22, 5, 15, 16, 21)
\\ \hline \hline
\end{tabular}
\caption[Golay code permutations for ancilla verification.]{The first ancilla in \figref{fig:FourAncillaVerifyCkt} is prepared using the circuit of \figref{fig:Golay-overlap-57}.  The other three ancillas are prepared in the same way, except with the qubits rearranged according to the above permutations.  %\comment{What is the notation?}
} \label{tbl:overlap-permutations}
\end{table}

\subsection{Resource overhead}\label{sec:Overhead}
To evaluate the practical importance of our optimizations, we now analyze the resource requirements of Steane-style error correction circuits based on ancillas prepared by~\figref{fig:FourAncillaVerifyCkt}.  We use Monte Carlo simulation to compare overhead of our ancilla preparation and verification circuits for the Golay code to that of standard circuits.

One natural measure for the overhead is the number of CNOT gates used to ready an ancilla.  Another overhead measure, important given the difficulty of scaling quantum computers, is the space complexity, i.e., the number of qubits that must be dedicated to ancilla preparation in a pipeline so that an ancilla is always ready in time for error correction.  We consider both measures.  

\begin{figure}
\centering
\begin{subfigure}[b]{.49\textwidth}
\includegraphics[width=\textwidth]{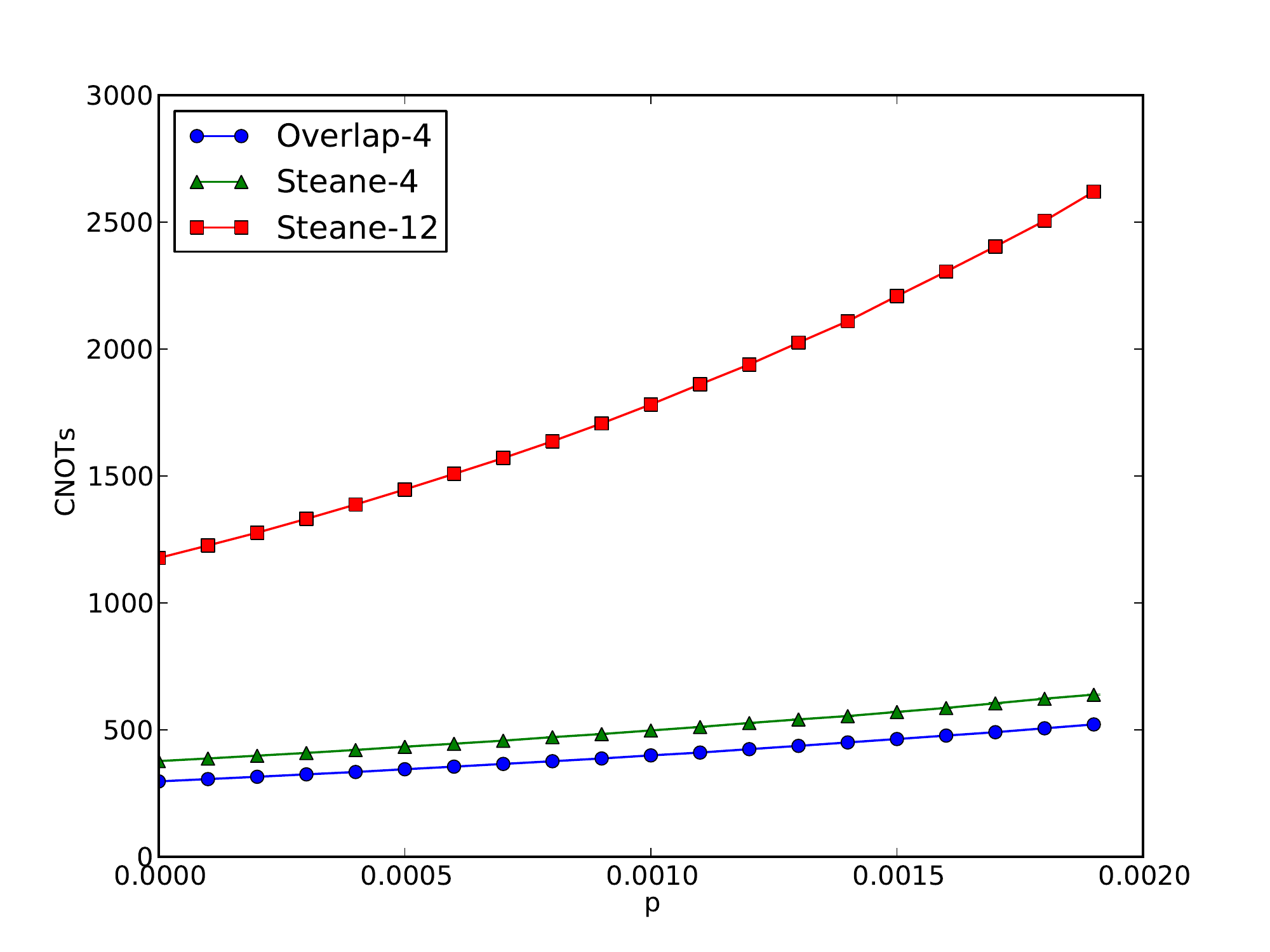}
\caption{\label{fig:sim-cnot-overhead}}
\end{subfigure}
\begin{subfigure}[b]{.49\textwidth}
\includegraphics[width=\textwidth]{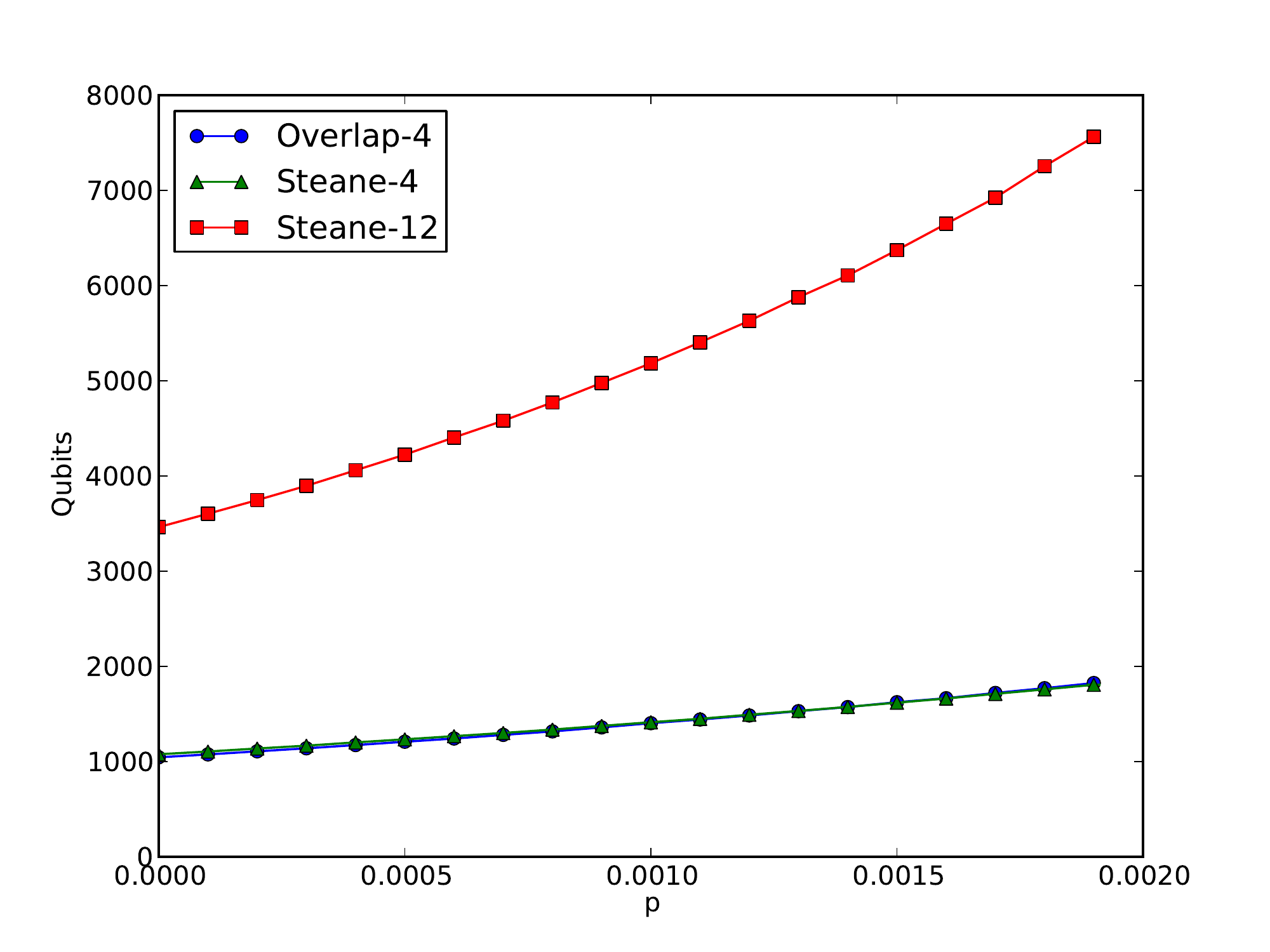}
\caption{\label{fig:sim-qubit-overhead}}
\end{subfigure}
\caption[Overhead estimates for Golay code error correction.]{Overhead estimates for the twelve-ancilla ancilla preparation and verification circuit and for each of our optimized circuits. The Steane-$4$ circuit is based on ancillas prepared according to~\tabref{tab:randomizedSchedules}.  Overlap-$4$ is based on ancillas prepared according to~\figref{fig:Golay-overlap-57} and~\tabref{tbl:overlap-permutations}. (a) Expected number of CNOT gates required to produce a verified encoded $\ket 0$.  (b) Number of qubits required to produce one verified encoded $\ket{0}$, in expectation, at every time step.  
Standard error intervals are too small to be seen here.}
\end{figure}

As listed in the third column of \tabref{tbl:verify-compare}, the overlap-based four-ancilla preparation and verification circuit involves roughly a factor of four fewer CNOT gates than the standard twelve-ancilla circuit.  In fact, this understates the improvement.  The overhead also depends on the acceptance rates of each verification test.  For an ancilla to leave the twelve-ancilla circuit, it must pass eleven tests, compared to only three tests for the four-ancilla circuit.  The probability of passing all tests should be significantly higher for the optimized circuit, and so one expects the ratio between the \emph{expected} numbers of CNOT gates used by the two circuits to be greater than four.  

To estimate the expected overhead, each circuit was modeled and subjected to depolarizing noise in a Monte Carlo computer simulation.  We assumed that test results are available soon enough that a failed verification circuit can be immediately aborted; later test failures are therefore the most costly.  This assumption impacts the twelve-ancilla circuits the most, since there are many ways to construct the hierarchy of verifications.  The circuit shown in~\figref{fig:TwelveAncillaVerifyCkt} is a reasonable choice here because only six of the verification tests depend on results of previous tests.  Other circuits---see, e.g., \cite[Sec.~2.3.2]{Reichardt2006}---may contain as many as nine dependent tests.

Estimates of the expected number of CNOT gates required for each circuit are given in the last column of \tabref{tbl:verify-compare} for the CNOT depolarization rate $p = 10^{-3}$, and are plotted versus~$p$ in \figref{fig:sim-cnot-overhead}.  At $p = 10^{-3}$, the overlap method reduces the expected number of CNOT gates by roughly a factor of~$4.5$, compared to the twelve-ancilla circuit, and the improvement for our optimized Latin rectangle scheme is a factor of~$3.6$.  At lower error rates, the improvement is less.  To investigate the effects of different error parameters, we also considered setting the rest error rate to zero; in this case, the expected number of CNOT gates used in the overlap circuit further decreases by about~$11$ percent, compared to less than four percent for our other four-ancilla circuit and less than two percent for the twelve-ancilla circuit.  The larger improvement for the overlap circuit is due primarily to the fact that the overlap preparation method replaces many CNOT gates with rest locations. 

To evaluate the space overhead, we plot in \figref{fig:sim-qubit-overhead} the number of qubits required to produce a single verified encoded $\ket 0$, in expectation, per time step, for each of the preparation and verification circuits. Thus, for example, the space overhead for a pipeline to produce a single \emph{unverified} ancilla state is $8 \cdot 23 = 184$ qubits; at any given time step, one $23$-qubit block is initialized, and CNOT gates are applied to seven other blocks---one per round in, e.g., \figref{fig:Golay-overlap-57}---so that one ancilla is prepared. (In fact, the overhead is slightly less than this since some of the qubits in the block can be prepared during rounds one and two.)  Estimates are calculated recursively by computing E[qubits] = (E[qubits]$_1$ + E[qubits]$_2$)/Pr[accept] for each verification step where the numerator is the expected number of qubits required to prepare the two states used in that verification step and Pr[accept] is the probability that the verification measurement detects no errors.  The results at $p = 10^{-3}$ are given in the second column of \tabref{tbl:verify-compare}.  Both of our optimized schemes reduce the required space by a factor of~$3.6$ at~$p = 10^{-3}$.

\begin{table}
\centering 
\begin{tabular}{c|cccccccc}
\hline \hline 
Verification & $\Pr$[accept] & E[\# qubits] & min \# CNOTs & E[\# CNOTs] \tabularnewline
\hline
Steane-$12$ & $0.419 \pm 0.001$ & $5183 \pm 14.2$ & $1177$ & $1782 \pm 4.9$ \tabularnewline
Steane-$4$ & $0.648 \pm 0.002$ & $1413 \pm  3.7$ & $377$ & $497.6 \pm 1.3$ \tabularnewline
Overlap-$4$ & $0.633 \pm 0.002$ & $1399 \pm 3.8$ & $297$ & $399.4 \pm 1.1$ \tabularnewline
\hline \hline 
\end{tabular}
\caption[Acceptance probabilities for Golay code ancilla verification.]{\label{tbl:verify-compare}
Estimates of the acceptance probability and overhead for the twelve-ancilla fault-tolerant ancilla preparation circuit and our two optimized circuits, at a depolarizing noise rate of $p = 10^{-3}$. The column labeled Pr[accept] gives the probability that all auxiliary ancilla measurements in the verification circuit detect no errors.  The next column, E[qubits], gives the expected number of physical qubits required to produce one verified encoded $\ket{0}$.  This is calculated recursively, by computing the expected number of qubits needed to pass each verification step.  The last two columns specify, respectively, the minimum number of CNOT gates and the expected number of CNOT gates required to produce a single verified ancilla.  
} 
\end{table}

To judge the significance of these results, recall that the ancilla production pipeline can consume the majority of resources in a fault-tolerant quantum computer.  In the case of~\cite{Isailovic2008a}, physical ancilla production space is proportional to the number of CNOT gates in the pipeline.  A factor of~$4.5$ reduction in the CNOT overhead for ancilla preparation should give, very roughly, about a $50$~percent improvement in the total footprint of the quantum computer.

\def\threshOverlap{1.32 \times 10^{-3}}
\def\threshSteane{1.24 \times 10^{-3}}

\chapter{Improving threshold lower bounds
\label{chap:threshold}
}

This chapter is based on material that appears in~\cite{Paetznick2011}.
\vspace{1cm}

The malignant set counting technique discussed in~\chapref{chap:fault} provides a simple way to calculate lower bounds on the the noise threshold, particularly for low-distance codes. However, it suffers from two limitations.  First, the number of faulty gate sets of size $k$ scales exponentially with $k$. A large fraction of faulty sets may be harmless, but counting all of them is computationally intractable.
Second, the assumed noise model is adversarial and, while more general than the model of independent Pauli channels, is probably overly pessimistic.

The first limitation is particularly troublesome if we wish to prove high thresholds for large codes which can correct many more sets of errors than smaller codes.  Large codes can be more efficient than small codes because they require fewer levels of concatenation in order to achieve the same level of error protection.  Using large codes could, therefore, lead to significant reduction in resource overhead. 

Instead of exhaustively counting all subsets of locations, Aliferis and Cross have used Monte Carlo sampling in order to estimate the fraction of malignant subsets to within prescribed confidence intervals~\cite{Aliferis2007b}.  Despite this improvement, the scaling of the population size is still exponential, and so the ability to count large subsets is limited.  

In this chapter, we show how malignant set counting can be adapted to prove good thresholds for large codes while simultaneously removing the requirement for an adversarial noise model.  The adaptation is based on two main ideas.  First, when errors occur independently, it is possible to partition the error correction circuit into small subcircuits.  Malignant subsets within each subcircuit can be counted separately, and then recombined in an efficient way.  By combining information from each subcircuit, we can effectively count very large sets.

The second main idea involves the way that error rates are calculated for each level of code concatenation.  Standard malignant set counting calculates the probability that \emph{any} uncorrectable error occurs during execution of the encoded gate.  This error rate can then be re-used to calculate similar probabilities at increasing levels of concatenation.  We instead keep track of the probability of \emph{each} type of uncorrectable error that can occur.  This can significantly improve the accuracy of the effective noise model for higher levels of concatenation.  

For example, say that the probability that an encoded gate introduces a logical $Z$ error is $0.01$ and that the probability of a logical $X$ error is the same.  In standard malignant set counting, this would be treated as a total error probability of $0.02$ at the next level of concatenation.  Using our method, error rates are reported separately, potentially saving a factor of two in this example. 

By combining these two ideas and including error-correction optimizations from~\chapref{chap:ancilla}, we can calculate rigorous lower bounds on the noise threshold for relatively large codes.  As a concrete example we calculate an error-rate bound of $0.00132$ per gate for the $23$-qubit Golay code.  This bound is the best known for any code and is an order of magnitude improvement over the best previous lower bound for the Golay code~\cite{Aliferis2007b}, based on an adversarial noise model.

\section{Requirements and assumptions}
\label{sec:threshold.requirements}
Before describing the adapted malignant set counting procedure in detail, it is worthwhile to examine the requirements that will be imposed on the noise model and fault-tolerance scheme.  There are essentially only two requirements: 
\begin{enumerate}
  \item errors must occur independently at each circuit location, and
  \item error-correction and gate gadgets must be strictly fault-tolerant.
\end{enumerate}
Roughly, the strict fault-tolerance requirement means that for a code that corrects up to $t$ errors, the probability that the circuit causes a weight-$k$ error on the data is no more than $O(p^k)$ for all $k\leq t$ and gate error rate $p$.  This requirement was described in~\secref{sec:fault.threshold.rectangles}.
We begin, instead, with the noise model.

\subsection{Noise model}
\label{sec:threshold.requirements.noise}
An important requirement of the modified malignant set counting technique is that errors occur independently at each physical circuit location.  Indeed, one primary motivation for modifying the malignant set counting procedure was to move away from the adversarial noise model in which circuit locations \emph{fail} independently, but the \emph{errors} at the failing locations are correlated.

We study noisy circuits constructed from the following physical operations: $\ket{0}$ and $\ket{+}$ initialization, a CNOT gate, and single-qubit measurement in the $Z$ and $X$ eigenbases.  Every qubit in the computer can be involved in at most one operation per discrete time step.  CNOT gates are allowed between arbitrary qubits, without geometry constraints.  Resting qubits are also subject to noise.

\begin{definition}[Independent Pauli noise with parameter $\gamma$]
\label{def:independent-pauli-noise}

Choose weights $w_{ab} \in [0,1]$ for all $a,b \in \{I,X,Y,Z\}$ such that
\begin{equation}
\label{eq:cnot-weight-restriction}
\sum_{a,b:ab\neq II} w_{ab} = 15
\enspace .
\end{equation}
Additionally, choose weights $w_{\ket 0}, w_{\ket +}, w_{\text{m}X}, w_{\text{m}Z}, w_{\text{r}X}, w_{\text{r}Y}, w_{\text{r}Z} \in [0,1/\gamma]$, such that $(w_{\text{r}X} + w_{\text{r}Y} + w_{\text{r}Z})\gamma \leq 1$.

Then noisy operations are modeled by: 
\begin{enumerate}
\item
A noisy CNOT gate is a perfect CNOT gate followed by, with probability $15 \gamma$, a non-trivial two-qubit Pauli error drawn from $\{I,X,Y,Z \}^{\otimes 2} \setminus \{I \otimes I\}$ according to $\{w_{ab}/15\}$.    
\item
Noisy preparation of a $\ket 0$ state is modeled as ideal preparation of $\ket 0$, followed by application of an $X$ error with probability $w_{\ket 0} \gamma$.  Similarly, noisy preparation of $\ket +$ is modeled as ideal preparation of $\ket +$ with probability $1-w_{\ket +}\gamma$ and of $\ket - = Z \ket +$ with probability $w_{\ket +}\gamma$.
\item
Noisy $Z$-basis ($\ket 0, \ket 1$) measurement is modeled by applying an $X$ error with probability $w_{\text{m}X} \gamma$, followed by ideal $Z$-basis measurement.  Similarly, noisy $X$-basis ($\ket +, \ket -$) measurement is modeled as ideal measurement except preceded by a $Z$ error with probability $w_{\text{m}Z} \gamma$.
\item
A noisy rest operation is modeled as applying either the identity gate, with probability $1-(w_{\text{r}X}+w_{\text{r}Y}+w_{\text{r}Z})\gamma$, or a Pauli error $a \in \{X,Y,Z\}$ with probability $w_{\text{r}a}\gamma$.
\end{enumerate}
All locations fail independently of each other.  
\end{definition}

Informally, this noise model works by modeling each physical location as an ideal operation, possibly followed (or preceded) by an error on the corresponding qubits.  When an error occurs, it is selected from a probability distribution defined by the weights for that location. \defref{def:independent-pauli-noise} defines weights only for CNOT, qubit preparation and measurement in the $X$ and $Z$ bases, and rest locations.  This set of locations is sufficient for the fault-tolerance schemes considered in this chapter.  However, additional locations (e.g., Hadamard) can be added as necessary. The counting procedure and threshold calculations of this section can be extended to accommodate any number of location types.

The condition imposed by \eqnref{eq:cnot-weight-restriction} is for convenience and concreteness, only.  A sum of $15$ was chosen to correspond nicely with a depolarizing noise model in which $w_{ab}=1$ for all $a,b$.

The noise model described by~\defref{def:independent-pauli-noise} is quite flexible and greatly improves our ability to analyze fault-tolerant quantum circuits when compared to an adversarial noise model.  However, it is weaker than adversarial noise and may seem artificial compared to even more general, or more physically realistic noise models described in~\chapref{chap:fault}.  

We justify~\defref{def:independent-pauli-noise} in two ways.  First, as a special case, this noise model describes independent depolarizing noise, which is commonly used in Monte Carlo threshold estimates~\cite{Zalk96,Steane2003,Reic04,Knill2004,Dawson2006,Svore2006b,Cross2009,Lai2013a}.  Therefore, our adapted malignant set counting technique can be used to obtain rigorous threshold lower bounds that can be more fairly compared with Monte Carlo threshold estimates.
Second, although physical noise may be complicated, methods for rigorously replacing realistic physical noise with simpler models do exist.  For example, Magesan et al. have shown how to replace an arbitrary single-qubit channel with a Pauli channel that approximates the original channel as closely as possible without underestimating the error strength~\cite{Magesan2012a}.

During error counting, $X$ and $Z$ errors are usually considered separately and the error probability is computed by omitting the $Z$ or $X$ part of each error, respectively.  For example, when considering only $X$, error $XY$ is equivalent to $XX$, $XZ$ is equivalent to $XI$ and so on.  Thus, the marginal distribution of $X$ errors for a CNOT is:
\begin{equation}\begin{split}
\Pr[IX] &= w_{IX}+w_{IY} + w_{ZX}+w_{ZY},\\
\Pr[XI] &= w_{XI}+w_{YI} + w_{XZ}+w_{YZ},\\
\Pr[XX] &= w_{XX}+w_{XY} + w_{YX}+w_{YY}
\enspace .
\end{split}
\end{equation}
The $Z$ error distribution for CNOT, and the $X$ and $Z$ error distributions for rest locations are calculated similarly.  When preparing $\ket 0$ or measuring in the $Z$ basis, no $Z$ errors are possible, and similarly no $X$ errors are possible when preparing $\ket +$ or measuring in the $X$ basis.

For computer analysis, it is convenient to choose integer-valued weights for each location.  Any noise model that satisfies~\defref{def:independent-pauli-noise} can be approximated to arbitrary precision with integer weights by relaxing~\eqnref{eq:cnot-weight-restriction} and rescaling $\gamma$.

\subsection{Additional assumptions}
\label{sec:threshold.requirements.assumptions}
In order to both reduce the time-complexity of the counting procedure, and to simplify its analysis we will make a few additional assumptions.  First, we assume that the quantum error-correcting code (or codes) in use are CSS codes.  Specifically, when $X$ and $Z$ errors can be corrected independently, as is the case for CSS codes, the number of errors that must be counted is significantly reduced.  This optimization is described in~\secref{sec:threshold.split}.  

The second simplifying assumption is that quantum gates are not geometrically constrained.  That is, multi-qubit gates can act on any set of qubits of appropriate size, and the properties of a quantum gate do not depend on the qubits on which the gate acts or the position of the gate within the circuit.

The unconstrained geometry assumption is common to many threshold calculations, including the AGP method of malignant set counting.  AGP do not require use of CSS codes.  However, nearly all fault-tolerance schemes that have been studied use CSS codes. (Some exceptions include~\cite{DiVi97,Gottesman1998a}.)

Finally, we will assume some level of determinism in the error-correction gadgets.  Specifically, syndrome measurements and corresponding corrections must be deterministic, though offline procedures such as ancilla preparation and verification may still be non-deterministic.  In particular, verification procedures such as those described in~\chapref{chap:ancilla} are allowed.

\section{Splitting up the extended rectangle}
\label{sec:threshold.split}
Perhaps the biggest drawback of malignant set counting for high-distance codes is that obtaining an accurate threshold value requires counting large subsets, but the counting complexity scales poorly with subset size.  The number of subsets of size $k$ in an exRec with $n$ locations scales as $\binom{n}{k}$, which is exponential in $k$.

Monte Carlo simulations of circuits using the $23$-qubit Golay code~\cite{Steane2003,Dawson2006,Cross2009} indicate that the depolarizing noise threshold should be on the order of $p = 10^{-3}$.  Unfortunately, it is not straightforward to prove such a high threshold using malignant set counting.  For example, say that we check for malignancy all location subsets of size up to $\kGood$, and we assume that all larger subsets are malignant.  Then the estimate we obtain for the probability of an incorrect rectangle is at least
\begin{equation}
\sum_{k=\kGood+1}^n \binom{n}{k} p^k (1-p)^{n-k}
\enspace .
\end{equation} 
Using optimized circuits from~\chapref{chap:ancilla}, the size of CNOT exRec for the Golay code is $n = 5439$. For this size and $p = 10^{-3}$, probability of incorrectness drops below $10^{-3}$ only for $\kGood \geq 14$.  However, there are more than $10^{41}$ subsets of size at most $14$, so checking them one at a time is computationally intractable.  
  
Instead of checking each set for malignancy, one can sample random sets of locations in order to estimate the fraction that are malignant.  This technique, called malignant set sampling, can provide threshold estimates with statistical confidence intervals.  However, both malignant set counting and sampling techniques study the threshold for worst-case adversarial noise, and may be overly conservative for a more physically realistic, non-adversarial noise model such as depolarizing noise.  For example, malignant set sampling results from~\cite{Aliferis2007b} estimate a threshold of only $p \approx 10^{-4}$ for the Golay code.  

On the other hand, when a large number of errors occur, it is relatively unlikely that all of the errors occur in the same region.  Rather, we expect errors to be distributed roughly evenly throughout the exRec.
  We therefore choose to divide the exRec into a hierarchy of components and sub-components.  We then compute an upper bound on the probability of each error that a component may produce, by counting location sets up to a certain small size.  At the exRec level, we synthesize the component error bounds into upper bounds on the probability that the rectangle is incorrect.  The resulting error probabilities are treated as an effective transformed noise model for the encoded gate.  With some care, the transformed noise model can be fed recursively back into the procedure to determine an effective noise model for the next level of encoding, and so on. See~\secref{sec:threshold.asymptotic.preserve-noise}. 

Effectively, dividing the exRec into components allows us to account efficiently for even very large location subsets.  Most large sets will be roughly evenly divided between the components, with only a small number of locations in each component.  
The remainder of this section outlines the exRec component structure.

\subsection{Circuit components} 
\label{sec:threshold.split.components}

We will divide the exRec into its encoded operation and its error corrections.  The error corrections will each divide into $X$-error correction and $Z$-error correction, and further recursive divisions will continue until reaching the physical location~level.  

\begin{figure}
\centering
\includegraphics[scale=1]{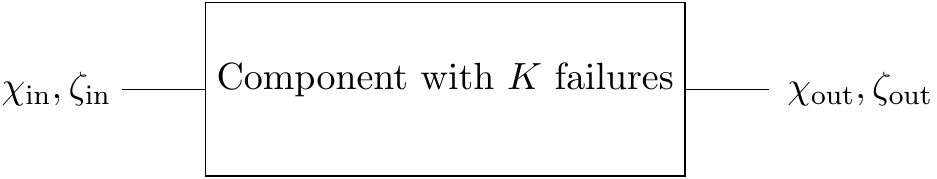}
\caption[A circuit component.]{A circuit component with input error $(\chiIn, \zetaIn)$ and output error $(\chiOut, \zetaOut)$} \label{fig:component}
\end{figure}

\def\xIn{x_\text{in}}
\def\zIn{z_\text{in}}
\def\xOut{x_\text{out}}
\def\zOut{z_\text{out}}

Each component in the hierarchy has input error $(\chiIn,\zetaIn)$, some number of internal failures~$K$, and output error $(\chiOut,\zetaOut)$ which depends on the internal failures and on the input error (see \figref{fig:component}).
Here, the notation $(\chi, \zeta)$ indicates an error equal to the product $\chi\zeta$ where~$\chi$ is a tensor product of $X$ and $I$ operators and $\zeta$ is a tensor product of $Z$ and $I$ operators.
For every error equivalence class on the inputs and outputs and for every~$k \in \N$, we would like to compute
\begin{equation}
  \label{eq:pr-component}
  \Pr\big[(\chiOut, \zetaOut)=(\xOut,\zOut), K=k \,\vert\, (\chiIn, \zetaIn)=(\xIn,\zIn)\big] \enspace ,
\end{equation}
the probability that there are exactly $k$ failures and the output error is $(\xOut, \zOut)$ conditioned on the input error $(\xIn, \zIn)$.    

For components that are physical gate locations the probability in~\eqnref{eq:pr-component} is defined by the appropriate Pauli-channel noise model (\defref{def:independent-pauli-noise}).  Larger components are analyzed by first analyzing each enclosed sub-component.  At the exRec level the LEC, transversal Ga and TEC components provide all of the information necessary to determine the probability that the enclosed rectangle is incorrect.  Indeed, we shall see in \secref{sec:threshold.split.exrec} that they contain enough information to compute the probability for each \emph{way} that the rectangle can be incorrect.  

There are, however, two logistical problems.  First, on each $n$-qubit code block, there $2^{n+1}$ inequivalent Pauli errors in total (assuming a single encoded qubit per block). For a component involving two code blocks, this means we should compute for each~$k$ up to $(2^{n+1})^4$ quantities, one for each combination of input and output errors.  Second, since there are $\binom{n}{k}$ size-$k$ subsets of $n$ locations and since each CNOT gate has $15$ different ways to fail, a computation that accounts for all possibilities scales roughly as $\binom{n}{k} 15^k$.  Such a computation is feasible only for small~$k$ and small~$n$.  

The first problem can be solved by observing that $X$ errors and $Z$ errors can be corrected independently for CSS codes.  Furthermore, error correction can be accomplished without using gates that mix $X$ and $Z$, so $X$ and $Z$ errors mostly propagate independently.  There are cases, such as ancilla verification, in which $X$ and $Z$ errors cannot be treated independently entirely. A specific example of this issue is discussed in~\secref{sec:threshold.golay.cnot.xverify}.
 Still, for most components, the $X$-error part of the output of a component depends only on the $X$-error part of the input and the $X$ failures that occur inside the component.  A similar observation holds for $Z$ errors.  Thus, expression~\eqnref{eq:pr-component} may be split into separate $X$ and $Z$ parts: 
\begin{subequations}\begin{align}
  \label{eq:prX-component}
  \Pr[\chiOut=\xOut, K_X &=k \vert \chiIn=\xIn] \\
  \label{eq:prZ-component}
  \Pr[\zetaOut=\zOut, K_Z &=k \vert \zetaIn=\zIn] \enspace .
\end{align}\end{subequations}
Here, the random variable $K_X$ is the number of failures inside the component that contain an $X$ when decomposed into a tensor product of Pauli operators.  The value $K_Z$ is similarly defined for~$Z$.  When considering $X$ and~$Z$ errors separately, the input and output of a two-block component contain at most roughly $2^{n}$ inequivalent errors, for codes that protect evenly against $X$ and $Z$ errors, and the worst case combination is a large but more manageable $2^{2n}$ cases.

The second problem is eliminated by noting that, for a fixed $k$, the probability of an order-$k$ fault decreases rapidly as the size of the component decreases.  For example, for $p = 10^{-3}$, the probability of an order-ten fault in an exRec of size $5000$ is about $0.018$.  However, the probability that all ten failures are located in a subset of $1000$ locations is less than $10^{-7}$. Thus there is little gain in counting errors of order-ten or higher in components of small size.  

In general, the probability that a component contains a fault of order greater than $\kGood$ can be bounded according to
\begin{equation}
\label{eq:prBad-bound}
  \Pr[K > \kGood] \leq \sum_{k=\kGood+1}^n \binom{n}{k} (1-p)^{n-k} p^k 
  \enspace ,
\end{equation}
where $p$ is an upper bound on the probability of a physical gate failure.
(A tighter bound can be achieved by considering separate $k$ for each location type.  See \cite{Paetznick2011} Appendix A.)  We will choose a value of $\kGood$ for each component and then pessimistically assume that all faults of order greater than $\kGood$ within the component cause the rectangle to be incorrect. For large enough values of~$\kGood$ the overall impact on the threshold is negligible.  There is a tradeoff here between running time and accuracy.  A larger value of $\kGood$ yields a more accurate bound on the probability that the rectangle is incorrect.  A smaller value of $\kGood$ is easier to compute.  We must choose for each component a suitable $\kGood$ that balances the two.  

In the end we are left with two sets of faults for each component, those of order at most~$\kGood$ and those of order greater than~$\kGood$.  Each fault in the first set is counted to obtain accurate estimates of \eqnref{eq:prX-component} and \eqnref{eq:prZ-component}.  When a fault from this set occurs we call it a \emph{good} event.  Faults in the second set are not counted and are instead bounded using \eqnref{eq:prBad-bound} and pessimistically added to the final incorrectness probability bounds for the rectangle.  When a fault from this set occurs we call it a \emph{bad} event.  The probability that the rectangle is incorrect is then upper-bounded by 
\begin{equation}
  \Pr[\incorrect] \leq \Pr[\incorrect, \good] + \Pr[\bad]
  \enspace .
\end{equation}

In general, there are four quantities we need to upper bound for each component:
%$\Pr[\chiOut=\xOut, K_X=k, \goodX \vert \chiIn]$, $\Pr[\zetaOut=\zOut, K_Z=k, \goodZ \vert \zetaIn]$, $\Pr[\badX]$, and $\Pr[\badZ]$.
\begin{subequations}
\label{eq:component-quantities}
\begin{align}
  &\Pr[\chiOut=\xOut, K_X=k, \goodX \vert \chiIn], \\
  &\Pr[\zetaOut=\zOut, K_Z=k, \goodZ \vert \zetaIn],\\
  &\Pr[\badX], \\
  &\Pr[\badZ] \enspace .
\end{align}
\end{subequations}
The event $\goodX \equiv \neg \badX$ occurs when there is a set of $X$-error failures in the component that we choose to count.  It will usually depend only on~$\kGood$ in which case $\goodX \Leftrightarrow (K_X \leq \kGood)$. In some cases $\goodX$ may depend on a vector~$\vec{k}$ representing the number of $X$-error failures across multiple sub-components.  The event $\goodZ \equiv \neg \badZ$ is similarly defined for $Z$.  

Finally, it is assumed that most components operate deterministically.  Non-deterministic components can be accommodated, however.  If, for example, the output errors of a component are dependent on a ``successful'' measurement outcome, then the component must also report the probability of success.  Then, the component output probabilities can be bounded using Bayes's rule
\begin{equation}
\Pr[\text{output}\vert \eventFont{success}] 
  = \frac{\Pr[\text{output},\eventFont{success}]}{\Pr[\eventFont{success}]}
  \leq \frac{\Pr[\text{output}]}{\Pr[\eventFont{success}]}
  \enspace .
\end{equation}

In the remainder of this section we outline the procedure for computing the above quantities for the error-correction and exRec components.  Details of lower level components, such as ancilla preparation and verification, depend on the choice of error-correcting code.

\subsection{The error-correction component} 
\label{sec:threshold.split.ec}

\begin{figure}
\centering
\begin{subfigure}[b]{.45\textwidth}
\includegraphics[width=\textwidth]{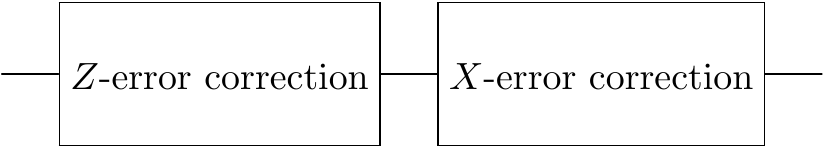}
\vspace{.25cm}
\caption{\label{fig:component.css-ec} CSS error-correction component}
\end{subfigure}
\hfill
\begin{subfigure}[b]{.45\textwidth}
\includegraphics[width=\textwidth]{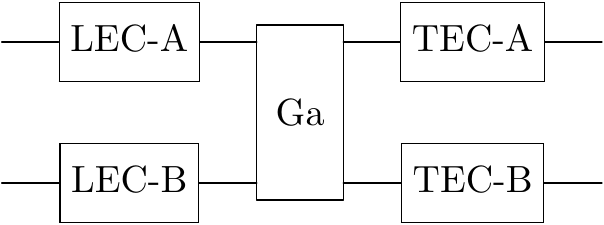}
\caption{\label{fig:component.two-qubit-exrec} Two-qubit exRec component}
\end{subfigure}
\caption[EC and exRec components.]{(a) The error-correction component for a CSS code consists of independent $Z$-error and $X$-error corrections.  Here, we have chosen an arbitrary convention that $X$-error correction follows $Z$-error correction. (b) The (encoded) two-qubit exRec consists of two leading error-correction (LEC) components, a gate gadget (Ga) component and two trailing error-correction (TEC) components.
}
\end{figure}

An error-correction component consists of $Z$-error correction and $X$-error correction, as shown in \figref{fig:component.css-ec}. (Recall that CSS codes admit independent correction of $X$ and $Z$ errors.) After extracting the error syndrome, the lowest-weight correction is computed.  The correction itself can be applied classically, and therefore without error, by a change in the qubit's Pauli frame~\cite{Knill2004}.

There are two types of error correction components: leading error correction (LEC) and trailing error correction (TEC).  For the LEC, we may assume that the input errors $\chiIn$ and $\zetaIn$ are both zero.  This is because we have assumed that syndrome measurement and correction are deterministic. The probability that the rectangle is incorrect depends only on the syndrome of the output of the LEC and that syndrome depends only on the errors inside of the LEC~\cite{Cross2009}.  

To be more precise, consider the two errors $X_1$ and $X_1 X_L$, where $X_1 = X\otimes I^{n-1}$ and $X_L$ is the logical $X$ operator of the code.  These two errors yield the same syndrome, but they are inequivalent since $X_1 X_L$ flips the logical state of the encoded qubit, and $X_1$ does not.  But correctness of the rectangle that follows is independent of the logical state of the input. The rectangle is not accountable for a logical error that occurred prior to its execution.
Accordingly, we may treat $X_1$ and $X_1 X_L$ as equivalent errors in this case.  More generally, we may assume that all of the errors at the output of the LEC are correctable, since the relationship with the logical operator is irrelevant.  This reduces the number of inequivalent errors at the output of each LEC by a factor of two, and therefore reduces the counting complexity by the same amount.

For trailing error correction, we care only about the result of applying a logical decoder to the output.  In other words, we only need to know whether the output errors $\chiOut$ and $\zetaOut$ represent correctable errors or not.  The four relevant quantities are:
\begin{center}
\begin{tabular}{cc}
\underline{LEC} & \underline{TEC}\\
  $\Pr[\chiOut=\xOut, K_X=k, \good \vert \chiIn=0]$,& $\Pr[D(\chiOut)=d, K_X=k, \good \vert \chiIn=\xIn]$, \\
  $\Pr[\zetaOut=\zOut, K_Z=k, \good \vert \zetaIn=0]$,& $\Pr[D(\zetaOut)=d, K_Z=k, \good \vert \zetaIn=\zIn]$, 
\end{tabular}
\end{center}
where $d \in \{0, 1\}$ and $D(e)$ identifies whether $e$ is a correctable error ($0$) or an uncorrectable error ($1$).  That is, $D(e)=1$ if and only if $e$ decodes to a nontrivial Pauli error. The details of $D$ depend on the choice of error-correcting code. 

\subsection{The exRec component} 
\label{sec:threshold.split.exrec}

A two-qubit exRec, shown in \figref{fig:component.two-qubit-exrec}, is divided into five components: two leading error corrections, gate gadget, and two trailing error corrections.  At this level, we are interested in \emph{malignant} events---the events for which the rectangle is incorrect. Furthermore, when a malignant event occurs we would like to know \emph{how} the rectangle is incorrect.  

Let $\ket{\psi_1}$ be the two-qubit state obtained by applying ideal decoders on the two blocks of the Ga immediately following the LECs.  Similarly let $\ket{\psi_2}$ be the state obtained by applying ideal decoders immediately following the TECs. Then define $\malig_{IX}$ as the event that $(I\otimes X) U_{\text{Ga}} \ket{\psi_1} = \ket{\psi_2}$, where $U_{\text{Ga}}$ is the two-qubit unitary corresponding to the ideal Ga gate.  Similarly define the events $\malig_{XI}$, $\malig_{XX}$, $\malig_{IZ}$, $\malig_{ZI}$, $\malig_{ZZ}$.  The event $\malig_E$ can be informally interpreted as the event in which the rectangle introduces a logical error $E$.

The relevant quantities are
\begin{subequations}
\label{eq:exrec-quantities}
\begin{align}
&\Pr[M_X, K_X=k, \good]\text{, and}\\
&\Pr[M_Z, K_Z=k, \good],
\end{align}
\end{subequations} 
for $M_X \in \{ \malig_{IX}, \malig_{XI}, \malig_{XX} \}$ and $M_Z \in \{ \malig_{IZ}, \malig_{ZI}, \malig_{ZZ} \}$. Each of the malignant events can be determined by propagating errors from the output of the LECs and Ga through the TECs.  For example, let $x_1$ and $x_2$ be the $X$ errors on the outputs of the first and second LECs, respectively.  Let $x'_1$ and $x'_2$ be the $X$ result of propagating $x_1$ and $x_2$ to the input of the TECs and combining with $X$ error $x_3$ of the Ga.  Then the probability of the malignant $IX$ even is given by
\begin{equation}
  \Pr[\malig_{IX}\vert x_1,x_2,x_3] =  \Pr[D(\chi_\out)=0\vert \chiIn=x'_1] \cdot
                                       \Pr[D(\chi_\out)=1\vert \chiIn=x'_2] 
  \enspace ,
\end{equation}
where as before, $D(x)$ determines whether $x$ is a correctable error ($0$) or not ($1$).  The quantities on the right-hand side can be readily obtained from the TEC components.  Recall from~\secref{sec:threshold.split.ec} that the errors $x_1$, $x_2$ are assumed to be correctable errors.  Therefore, $D(\chiOut)=0$ corresponds to a logical identity operator and $D(\chiOut)=1$ corresponds to a logical $X$. Probabilities of the other malignant events can be similarly calculated.

When counting $X$ and $Z$ errors separately, it is not possible to compute logical $Y$ error quantities and the analysis will therefore double-count $Y$ errors.  Intuitively this is not a great loss, because the correlations between $X$ and $Z$ are much smaller at this level than they are in the original noise model.  In \secref{sec:threshold.asymptotic} we show how to use~\eqnref{eq:exrec-quantities} to compute a lower bound on the threshold.

\section{Provisions for computer analysis}  
\label{sec:threshold.computer}

The component quantities~\eqnref{eq:component-quantities} are conceptually straightforward and easy to compute numerically for a fixed $\gamma$.  However, we would like to compute {exact} bounds that hold for a range of~$\gamma$.  In this section we discuss a few of the implementation details that allow for maintaining the bounds as polynomials with integer coefficients.  

The ultimate goal is to compute upper bounds on the probabilities of malignant events at the outermost layer of the exRec.  That is, we want to compute Equations~\eqnref{eq:component-quantities} and combine them to get, for example, 
\begin{equation}
\label{eq:malig-event-bound}
  \Pr[\malig_{IX}(\vec \chi) \vert \accept] \leq \Pr[\malig_{IX}(\vec{\chi}), \goodX \vert \accept] + \Pr[\badX \vert \accept] \enspace .
\end{equation}
Here, $\accept$ is the event that any and all non-deterministic sub-components (ancilla verification, for example) accept or succeed.
The right-hand side of this inequality decomposes into sums of individual component quantities of the form
\begin{equation}\begin{split} 
  \Pr[\chi=x,K_X=k]
  & = \sum_{\vec{\abs{k}}=k}
      \Pr[\chi=x, \vec{K}_X=\vec{k}] \enspace ,
\end{split}
\label{eq:ComponentQuantity}
\end{equation}
where $\vec{k} = (k_1, k_2, k_3, k_4)$ expresses the number of failing CNOT, rest, $\ket 0$ preparation and $Z$-basis measurements, respectively.  

For each term in the sum, the number of failures for each type of location is fixed, but the particular locations on which those failures occur are not fixed, nor are the errors that occur at those locations.  Let $L(\vec k) := \{ \vec{l} : (\vec{\abs{l_1}}, \vec{\abs{l_2}}, \vec{\abs{l_3}}, \vec{\abs{l_4}}) = \vec{k} \}$ be the set of all possible tuples of failing locations consistent with $\vec k$. Also, let $E(\vec{l})$ be the set of all possible tuples of $X$ errors consistent with failures at all locations $\vec l$.  To fix the locations and the errors, use
\begin{equation}\begin{split}
  \Pr[\chi=x, \vec{K}_X=\vec{k}]
  & = \sum_{\substack{ \vec{l} \in L(\vec{k}),
                       \vec{e} \in E(\vec{l})
                    }}
      \Pr[\chi=x, \vec{E}=\vec{e}] \\
  & = \sum_{\substack{ \vec{l} \in L(\vec{k}),
                       \vec{e} \in E(\vec{l})
                    }}
      \mathcal{I}(x, \vec{e}) \Pr[\vec{E}=\vec{e}]
\end{split}
\label{eq:LocationSum}
\end{equation}
where in the second line we have made the substitution $\mathcal{I}(x, \vec{e}) = \Pr[\chi=x \vert \vec{E}=\vec{e}]$.

The indicator function $\indicator(x, \vec{e})$ takes value one if the component produces the error $x$ for a given ``configuration'' of errors $\vec{e}$ and value zero otherwise.  The error configuration $\vec{e}$ fully specifies the locations that have failed and the error at each failing location. Let $\vec{n}=(n_1, n_2, n_3, n_4)$ be the total number of CNOT, rest, $\ket{0}$ preparations and $Z$-basis measurements in the component, respectively. Let $W_1=w_{IX}+w_{IY}+w_{XI}+w_{YI}+w_{XX}+w_{YY}$ be the sum of all of the CNOT $X$-error weights, let $W_2=w_{\text{r}X}+w_{\text{r}Y}$, $W_3=w_{\ket 0}$, $W_4=w_{\text{m}X}$ and $W:=\max\{W_1,W_2,W_3,W_4\}$.  For simplicity, assume also that $w_{IX}=w_{IY}=w_{XI}=w_{YI}=w_{XX}=w_{YY}=:w_1$, $w_{\text{r}X}=w_{\text{r}Y}=:w_2$ and let $w_{\ket 0}=:w_3$. Then from the marginal noise model discussed in \secref{sec:threshold.requirements.noise} and a configuration of $X$ errors $\vec{e}$ we have
\begin{equation}
\begin{split}
  \Pr[\vec{E} = \vec{e}] &= \prod_{j=1}^4 (1-W_j\gamma)^{n_j} \left( \frac{w_j\gamma}{1-W_j\gamma} \right)^{k_j} \\
  &\leq A_{\vec{n}} \left( \frac{\gamma}{1-W\gamma} \right)^{k} \prod_{j=1}^4 w_j^{k_j}
   \enspace ,
\end{split}
\label{eq:PrConfig}
\end{equation}
where $A_{\vec{n}} := \prod_{j=1}^4(1-W_j\gamma)^{n_j}$.  This inequality is a reasonable approximation for small $\gamma$.  It allows us to move $\gamma$ into a prefactor in front of the sum of \eqnref{eq:ComponentQuantity} and, assuming integer weights $\{w_j\}$, permits an integer representation in the computer analysis.  Indeed, substituting back into equation \eqnref{eq:ComponentQuantity} gives 
\begin{equation}
  \label{eq:Component-expression-X}
  \Pr[\chi=x,K_X=k]
  \leq  A_{\vec{n}} \left( \frac{\gamma}{1-W\gamma} \right)^{k}
    \sum_{\substack{ \vec{\abs{k}} = k \\
                     \vec{l} \in L(\vec{k}),
                     \vec{e} \in E(\vec{l})
         }}
    \indicator(x, \vec{e})
    \prod_{j=1}^4 w_j^{k_j}
  \enspace .
\end{equation}

Another advantage of counting component probabilities in this way, is that the counts compose nicely.  If we apply \eqnref{eq:ComponentQuantity} to itself and combine with \eqnref{eq:Component-expression-X}, we end up with
\begin{equation}\begin{split}
  \label{eq:prX-recursive}
  \Pr[\chi=x,K_X=k]
  &= \sum_{\substack{ \vec{\abs{k}} = k \\
                     \vec{x} \in \out(x)
          }}
     \prod_i \Pr[\chi_j=x_i, K_{X,i}=k_i] \\
  &\leq A_{\vec{n}}
    \left( \frac{\gamma}{1-W\gamma} \right)^{k}
    \Bigg[
    \sum_{\substack{ \vec{\abs{k}} = k \\
                     \vec{x} \in \out(x)
          }}
    \prod_i
    \sum_{\substack{ \vec{\abs{k_i}} = k_i \\
                     \vec{l} \in L(\vec{k}_i), \vec{e} \in E(\vec{l})
         }}
    \indicator(x_j, \vec{e})        
    \prod_{j=1}^4 w_j^{k_j}
    \Bigg] \enspace .
\end{split}\end{equation}
The substitution made in the first line can be applied successively for each sub-component $i$.  Once the lowest level component is reached, we use \eqnref{eq:Component-expression-X} to push dependence on~$\gamma$ outside of the sum.  The integer value inside of the brackets is the discrete convolution of weighted counts from the sub-components summed over all possible failure partitions $\vec k$ of size $k$.  It is a weighted count of all possible ways to produce error $x$ with an order $k$ fault.  

A similar formula holds for the general case in which each of the weights may be unique (i.e., $w_{IX}\neq w_{IY}\neq w_{XI}\ldots$, etc.).  In general, the product of weights $\prod_{j=1}^4 w_j^{k_j}$ is more complicated and may depend on the error configuration $\vec{e}$.

The primary task of the computer analysis is to compute $\indicator$ for each (good) error configuration, starting with the lowest level component, and to store the resulting weighted sums 
\begin{equation}
  \sum_{\substack{ \vec{\abs{k}} = k \\
                     \vec{l} \in L(\vec{k}), \vec{e} \in E(\vec{l})
         }}
    \indicator(x, \vec{e})
    \prod_{j=1}^4 w_j^{k_j}
\end{equation}
(or equivalent) for use in the counting of larger components.  At each level, counts for the sub-components are convolved to generate new counts.  The prefactor $A_{\vec{n}} \left( \frac{\gamma}{1-W\gamma} \right)^{k}$ need only be computed at the end, when calculating the threshold.

\section{Calculating the pseudo-threshold}
\label{sec:threshold.pseudo}

One quantity that can be immediately calculated from our counts is the so-called pseudo-threshold~\cite{Svore2006a} for the CNOT location. The pseudo-threshold for location $l$ is defined as the solution to the equation $p = p_l^{(1)}$, where $p$ is the probability that the physical (level-$0$) location fails, and $p_l^{(1)}$ is the probability that the $1$-Rec for location $l$ is incorrect. We may compute a lower bound on the pseudo-threshold for CNOT by upper bounding 
\begin{equation}
  p_\text{cnot}^{(1)} \leq \Pr[\bad \vert \accept] + \sum_{k} \big( \Pr[\malig_X, K_X=k, \good] + \Pr[\malig_Z, K_Z=k, \good] \big) \enspace ,
\end{equation}
where $\malig_X \equiv (\malig_{IX} \vee \malig_{XI} \vee \malig_{XX})$, $\malig_Z \equiv (\malig_{IZ} \vee \malig_{ZI} \vee \malig_{ZZ})$.  
  
The pseudo-threshold is of practical interest for cases in which a finite failure probability is acceptable and only a few levels of concatenation are desired.  For example, when the physical failure rate is sufficiently below the pseudo-threshold, a large code code could be used to bootstrap into other codes with lower overhead.

The pseudo-threshold is useful to us for two reasons.  First, pseudo-threshold estimates have been calculated for a variety of fault-tolerant quantum circuits and codes~\cite{Cross2009}, and therefore serve as a reference for our counting results.  Second, it was conjectured by~\cite{Svore2006a} that the pseudo-threshold is an upper bound on the asymptotic threshold.  It thus provides a reasonable target for our calculation of the asymptotic threshold lower bound, which requires a noise strength maximum to be specified.
% (see \appref{sec:bounding-polys}).

\section{Calculating the asymptotic threshold}
\label{sec:threshold.asymptotic}
Traditionally, malignant sets are those for which any combination of Pauli errors at the corresponding locations combine to cause the enclosed rectangle to be incorrect.
Our malignant sets are different.  We count subsets of faulty locations, but the counted information is synthesized into error probability upper bounds based on a particular noise model and error correction scheme.

In this section we outline an alternative method for rigorously lower bounding the noise threshold that is tailored specifically to the information obtained by our counting procedure.  The basic idea is to treat each level-one rectangle in the level-two simulation as a single ``location'' with a transformed noise model based on the malignant event upper bounds obtained in~\secref{sec:threshold.split}.  In particular, we show how to treat each level-one exRec independently while maintaining valid upper bounds on the error probabilities.

The asymptotic noise threshold is defined as the largest value $\gammaTH$ such that, for all $\gamma < \gammaTH$, the probability that the fault-tolerant simulation succeeds can be made arbitrarily close to one by using sufficiently many levels of code concatenation.  To prove a lower bound on the threshold we must show, in particular, that the probability of an incorrect CNOT $k$-Rec decreases monotonically with~$k$ for all $\gamma < \gammaTH$.  Our counting technique gives an upper bound on the probability that a CNOT $1$-Rec is incorrect.  We now show how to upper bound incorrectness for level-two and higher and therefore lower bound $\gammaTH$.

\subsection{Preserving independent Pauli noise under level reduction}
\label{sec:threshold.asymptotic.preserve-noise}
Consider an isolated level-one CNOT exRec. Let $\Pr[\malig_E]$ be the probability that the malignant event $\malig_E$ occurs.  For this event, the enclosed $1$-Rec behaves as an encoded CNOT gate followed by a two-block error that, when ideally decoded, leaves a two-qubit error~$E$ on the decoded state. Then our counting technique provides upper bounds on $\Pr[\malig_E]$ for $E \in \{IX, XI, XX, IZ, ZI, ZZ\}$.  These upper bounds can be viewed as an error model for the CNOT $1$-Rec in which the correlations between $X$ and $Z$ errors are unknown.

We would now like to analyze the level-two CNOT exRec.  Ideally, we could treat each $1$-Rec in the level-two simulation as a single ``location'' and use the error model obtained from level-one to describe the probability of failure.  Then level-two analysis could proceed by feeding this ``transformed'' error model back into the counting procedure in order to compute $\Pr[\malig_E]$ for the CNOT $2$-Rec.

However, the transformed error model is based on analysis of an isolated level-one CNOT exRec.  A typical level-one simulation will contain many exRecs, and adjacent exRecs may share error corrections at which point they can no longer be considered independently.  

The reason that level reduction works when counting sets of malignant locations is because exRecs with incorrect rectangles are replaced with faulty gates in the same way regardless of the malignant event that actually occurs. The quantity used to bound incorrectness probability is strictly non-increasing as locations (i.e., TECs) are removed.  To see this, consider sets of exRec locations of size $k$ and denote the set of all such sets by~$S_k$. Let $M \subseteq S_k$ be those sets for which some combination of nontrivial errors at the $k$ locations causes the rectangle to be incorrect (i.e., the malignant sets).  The probability that the rectangle is incorrect due to failures at exactly $k$ locations is then no more than $\abs{M} p^k$.  If an error correction is removed from the exRec, some of the sets in $M$ now contain fewer than $k$ exRec locations. The remaining sets with $k$ exRec locations are those that do not contain a location in the removed error correction.  The number of such sets is at most $\abs{M}$ and so the original bound on the incorrectness probability still holds.  

\def\X{\eventFont{X}}
\def\I{\eventFont{I}}

The disadvantage to this approach for non-adversarial noise models is that it fails to consider all of the available information. In particular, for a fixed set of malignant locations it assumes the worst-case error for each location. The probability that a given set of $k$ locations is actually malignant can be significantly less than $p^k$.  To obtain a more accurate analysis of the second level, we would like to replace each incorrect $1$-Rec according to the malignant event that has actually occurred. 

Our transformed noise model of an isolated CNOT exRec provides upper bounds on the probability of each type of malignant event, but we must show that the bounds still hold when exRecs overlap.  Unfortunately, the bounds almost certainly will \emph{not} hold.  Consider, for example, the control block of the CNOT exRec, shown in \figref{fig:cnot-block}. Assume that the error immediately preceding the transversal CNOT is correctable (the error itself is not important). Let $\X$ be the event that an uncorrectable $X$ error exists on the output of the TEC and $\I$ be the event that the error on the output is correctable.  In other words $\X \equiv (\malig_{XI} \lor \malig_{XX})$ and $\I \equiv \lnot \X$.  Then define $\X' \equiv \lnot \I'$ as the event that an uncorrectable $X$ error exists on the block following the transversal CNOT but before error correction.  $\Pr[\malig_{XI}]$ will be non-increasing when removing the trailing error correction only if $\Pr[\X'] \leq \Pr[\X]$.  On the other hand, $\Pr[\malig_{IX}]$ will be non-increasing only if $\Pr[\I'] \leq \Pr[\I]$.  Since $\Pr[\X]+\Pr[\I]=\Pr[\X']+\Pr[\I']=1$, both conditions are satisfied only if $\Pr[\X]=\Pr[\X']$ and $\Pr[\I]=\Pr[\I']$, which of course is highly unlikely.  

\begin{figure}
\centering
\includegraphics[scale=1]{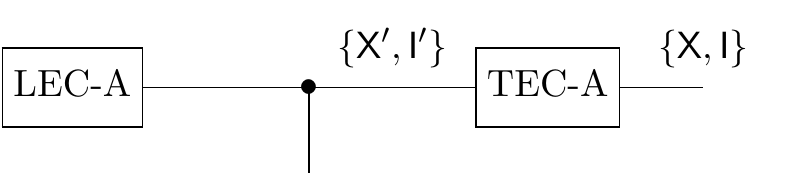}
\caption[Upper block of the CNOT exRec.]{Upper block of the CNOT exRec.  The error at the output of the TEC is either correctable ($\I$), or not ($\X$).  Similarly the error immediately preceding the TEC is either correctable ($\I'$) or not ($\X'$).}
\label{fig:cnot-block}
\end{figure}

In order to ensure a proper upper bound on each of the malignant event probabilities, we must calculate upper bounds for the complete exRec and for incomplete exRecs in which one or more trailing error corrections have been removed.  Calculations for the complete exRec were discussed in \secref{sec:threshold.split.exrec}.  Calculations for the incomplete exRecs are the same except that some of the TEC components are not considered.  Bounding the malignant event probability is a matter of finding a polynomial that bounds all four cases. Details of the bounding polynomial can be found in Appendix D of~\cite{Paetznick2011}.

Once proper bounds on the level-one malignant event probabilities are determined, we would like to plug the transformed error model into our counting procedure in order to determine the level-two error probabilities.  There are a few things to consider before doing so.  First, part of the counting strategy, such as ancilla verification, may rely on using the correlations between $X$ and $Z$ errors in order to avoid over-counting that occurs during postselection (for example, see~\secref{sec:threshold.golay.cnot.xverify}).  The transformed error model, however, contains no such correlation information, so the counting strategy must be altered accordingly.  Second, the CNOT malignant event upper bounds do not contain information about rest, preparation or measurement locations.  Level-one error models for these locations can be computed using the same counting strategy as the CNOT, but with an appropriately modified exRec.\footnote{Alternatively, they can be incorporated into the CNOT exRecs~\cite{Aliferis2007b}.}  

Finally, in the Pauli-channel noise model, the error probabilities of each location are constant multiples of the noise strength~$\gamma$.  Our upper bounds on the malignant event probabilities, however, need not have any scalar relationship. For computer analysis, error probabilities must be re-normalized in terms of $\gamma$ and error weights recalculated as follows. Let $\mathcal{P}^{(1)}_E$ be our upper bound on the level-one malignant event $\malig_E$.  Then construct a polynomial $\Gamma^{(1)}$ and choose constants $\alpha_E$ such that 
\begin{equation}
  \mathcal{P}^{(1)}_E(\gamma) \leq \alpha_E \Gamma^{(1)}(\gamma)
\end{equation}
for all $E$.  The polynomial $\Gamma^{(1)}$ can be viewed as an effective noise strength ``reference'' for level-one.  $\Gamma^{(1)}(\gamma)$ is a function of $\gamma$, but we will usually denote it as $\Gamma^{(1)}$ for convenience of notation. Together with weights $\alpha_E$, $\Gamma^{(1)}$ defines a new independent Pauli channel noise model. Again, see Appendix D of~\cite{Paetznick2011} for details of the construction.  

Now the new error model is input into the counting procedure and upper bounds on the level-two error rates are computed.  Let $\mathcal{P}^{(2)}_E(\Gamma)$ be the upper bound computed for $\malig_E$ at level-two.  Then we have the following conditions on the level-one and level-two malignant event probabilities: 
\begin{equation}\begin{split}
\label{eq:transformed-noise-polys-level-1-2}
  \Pr[\malig_E^{(1)}] &\leq \mathcal{P}^{(1)}_E(\gamma) \leq \alpha_E \Gamma^{(1)} \\
  \Pr[\malig_E^{(2)}] &\leq \mathcal{P}^{(2)}_E(\Gamma^{(1)})
  \enspace .
\end{split}\end{equation}

\subsection{Proving an asymptotic threshold}
\label{sec:threshold.asymptotic.proof}
The transformed noise model provides a means for computing malignant event probabilities at level-two based on the malignant event probabilities of level-one.  In principle, it is possible to repeat that procedure to calculate malignant event probabilities up to any desired level of concatenation.

To prove a noise threshold, we could continue to concatenate until the transformed noise strength is sufficiently low, and then use schemes for which a threshold is known.  For example, Aliferis and Preskill prove a threshold for depolarizing noise of $1.25\times 10^{-3}$ for a scheme based on the $[[4,2,2]]$ error-detecting code~\cite{Aliferis2009}.

In order to take full advantage of noise suppression of the large-distance code, though, we should prefer to prove a threshold directly.  To do so, consider again~\eqnref{eq:transformed-noise-polys-level-1-2}.
We claim that $\mathcal{P}^{(2)}_E$ obeys the following property: 

\begin{claim}
  For $0 \leq \epsilon \leq 1$, $\mathcal{P}^{(2)}_E (\epsilon \Gamma^{(1)}(\gamma)) \leq \epsilon^{t+1} \mathcal{P}^{(2)}_E (\Gamma^{(1)}(\gamma))$, where $t=\lfloor (d-1)/2 \rfloor$ and $d$ is the minimum distance of the (unconcatenated) code.
  \label{clm:P2-multiplicative}
\end{claim}
In other words, the level-two malignant event polynomials decrease with $\gamma$ at a rate that corresponds with the distance of the code.  This is just the kind of behavior that we should expect from a strictly fault-tolerant scheme.
Proof of this claim is based on the form of the polynomials constructed by our counting technique and the fact that our circuits are strictly fault-tolerant.  Details of the proof are delegated to \appref{app:counting.scaling-proof}.

We are now in a position to establish conditions for a noise threshold, i.e., the conditions under which the probability of a successful simulation can be made arbitrarily close to one. 

\begin{theorem}
 Let $M$ be the set of all level-one CNOT, preparation, measurement and rest malignant events consisting of: $\malig_{IX}$, $\malig_{XI}$, $ \malig_{XX}$, $\malig_{IZ}$, $\malig_{ZI}$, $\malig_{ZZ}$, $\malig_{X}^{\text{prep}}$, $\malig_{Z}^{\text{prep}}$,$\malig_{X}^{\text{meas}}$, $\malig_{Z}^{\text{meas}}$, $\malig_{X}^{\text{rest}}$ and $\malig_{Z}^{\text{rest}}$. Also let $\mathcal{P}^{(1)}_E$, $\mathcal{P}^{(2)}_E$ and $\Gamma^{(1)}$ be polynomials and $\alpha_E$ constants as discussed above.
 Then the tolerable noise threshold for depolarizing noise is lower bounded by the largest value $\gammaTH$ such that
 \begin{equation}
   \mathcal{P}^{(2)}_E(\Gamma^{(1)}(\gammaTH)) \leq \alpha_E \Gamma^{(1)}(\gammaTH)
 \end{equation}
 for all $\malig_E \in M$.  
 \label{thm:threshold}   
\end{theorem}

\begin{proof}
Assume that $\mathcal{P}^{(2)}_E(\Gamma^{(1)}) < \alpha_E \Gamma^{(1)}$, for all $\malig_E$ and $\gamma \in (0,\gammaTH)$.  Then, for a fixed $\gamma \in [0,\gammaTH)$, there exists some positive $\epsilon < 1$ such that, for all malignant events $\malig_E$, $\mathcal{P}^{(2)}_E(\Gamma^{(1)}) \leq \epsilon \alpha_E \Gamma^{(1)}$.  

By choosing $\Gamma^{(2)} := \epsilon \Gamma^{(1)}$ we obtain an effective noise model for level two in which the weights~$\alpha_E$ are unchanged.  Since our counting method depends only on the error weights, the polynomials that upper bound the level-three malignant events will be the same as the polynomials that upper bound the level-two malignant events. That is, $\mathcal{P}^{(k)}_E(\Gamma) = \mathcal{P}^{(2)}_E(\Gamma)$ for $k \geq 2$.
Thus,
\begin{align}
  \Pr[\malig^{(3)}_E] \leq \mathcal{P}^{(3)}_E(\Gamma^{(2)}) 
  = \mathcal{P}^{(2)}_E(\epsilon \Gamma^{(1)})
  \leq \epsilon^{t+2} \alpha_E \Gamma^{(1)}
  \enspace ,
\end{align}
where the last inequality follows from \claimref{clm:P2-multiplicative}.
 Defining $\Gamma^{(3)} := \epsilon^{t+1} \Gamma^{(2)}$ and repeating this process $k$ times yields
\begin{equation}
\label{eq:PrE-recurrence}
  \Pr[\malig_E^{(k+1)}] 
  \leq \mathcal{P}^{(k+1)}_E(\Gamma^{(k)})
  \leq \epsilon^{(k-1)(t+1)+1} \alpha_E \Gamma^{(1)}
  \enspace ,
\end{equation}
which approaches zero in the limit of large $k$.  
\end{proof}

Testing of the assumption $\mathcal{P}^{(2)}_E(\Gamma^{(1)}) < \alpha_E \Gamma^{(1)}$ over a fixed interval $(0,\gammaTH)$ is straightforward if all of the malignant event polynomials (including $\Gamma^{(1)}$) are monotone non-decreasing up to sufficiently large values of $\gamma$.  Monotonicity is highly plausible for values of $\gamma$ surrounding or below threshold, but must be checked explicitly based on the weighted counts obtained from malignant set counting. Appendix C of~\cite{Paetznick2011} provides an explicit procedure for checking monotonicity. 

\section{Summary of the modified malignant set counting procedure}
\label{sec:threshold.summary}

The entire malignant set counting procedure is somewhat lengthy.  For convenience, we now summarize each of the steps.

\begin{enumerate}
  \item Choose a CSS code, error correction scheme, and an independent Pauli noise model.  Construct the corresponding extended rectangle that satisfies Definitions~\ref{def:strict-fault-tolerance-gate} and \ref{def:strict-fault-tolerance-ec}, for each encoded gate type.
  \item Partition each exRec into a hierarchy of small components.
  \item For each lowest-level component choose a small integer $k_\good$, count all of the errors that occur with up to $k_\good$ faulty locations, according to the weights of the selected noise model.  Also compute $\Pr[\bad]$, the probability that more than $k_\good$ locations are faulty. If necessary, compute $\Pr[\accept]$ that the component is accepted.
  \item For higher level components, again choose a $k_\good$, and count errors by convolving results from lower level components up to $k_\good$.  Calculate $\Pr[\bad]$ and $\Pr[\accept]$ as necessary.
  \item For each exRec, compute $\Pr[E]$ the probability of the logical error $E$ for each $X$ and $Z$ error. Construct the corresponding transformed Pauli noise model.
  \item  Either repeat the procedure (if parts of the exRec are non-deterministic), or bound the threshold analytically using~\thmref{thm:threshold}.
\end{enumerate}

\section{Example: a depolarizing noise threshold for the Golay code}
\label{sec:threshold.golay}

In order to quantify the efficacy of our adapted malignant set counting technique, we use it to calculate the depolarizing threshold of the $23$-qubit Golay code.  The Golay code is ideal for this task for a variety of reasons.  First, with distance seven, it is substantially larger than typically studied codes which usually have distance three. Still, it is small enough so that the number of possible errors on a single block is quite manageable. Second, numerical estimates place the Golay code as one of the top performers, with depolarizing threshold estimates on the order of $10^{-3}$~\cite{Steane2003,Dawson2006,Cross2009}.  On the other hand, malignant set sampling has yielded statistical lower bounds for adversarial noise of just $10^{-4}$, leaving ample room for improvement.

In this section, we prove a depolarizing noise threshold lower bound of $1.32\times 10^{-3}$ for the Golay code, which essentially matches numerical estimates and is the highest known rigorous lower bound for any code.  Furthermore, we show that the resource overhead for our scheme is usually substantially lower than the $[[4,2,2]]$ Fibonacci scheme for which the next best threshold lower bound is known~\cite{Aliferis2009}. 

Threshold results were obtained by implementing our counting technique as a collection of modules written in Python and C; the source code is available at~\cite{Paetznick}.  We calculated thresholds for error correction circuits based on the four-ancilla protocols described in~\secref{sec:ancilla.verify.golay}. Results are given in~\tabref{tab:threshold-results}.  The main program takes as input the four-ancilla preparation circuits, the noise model, and the good and bad event settings.  It outputs, for each type of exRec and each malignant event, a polynomial representing an upper bound on the event probability.  See \figref{fig:malig-events}.  These polynomials are either evaluated directly to calculate the pseudo-threshold, or processed into a transformed error model and fed back into the program.

The most time-consuming part of the computation involved the CNOT exRec component.  Computing weighted counts for this component required a custom convolution with nearly four trillion combinations.  Running the entire program to completion for a fixed ancilla preparation and verification schedule on $31$ cores in parallel took about four days.

\subsection{The depolarizing noise model}
\label{sec:threshold.golay.noise}
The depolarizing noise model is particularly easy to define in terms of the weights prescribed by~\defref{def:independent-pauli-noise}.  For the CNOT gate, choose $w_{ab}=1$ for all $a,b \in \{I,X,Y,Z\}$.  The rest location weights are chosen based on the one-qubit marginals of the CNOT.  Use $w_{\text{r}a} = \sum_{b \in \{I,X,Y,Z\}} w_{ab} = 4$ for $a \in \{X,Y,Z\}$.  For preparation and measurement locations use $w_{\ket 0}=w_{\ket +}=w_{\text{m}X}=w_{\text{m}Z}=4$.  The preparation and measurement weights are lower than the one-qubit marginals (which would imply values of eight) because any higher noise rate could be reduced to $4\gamma + O(\gamma^2)$ by repeating the preparation or measurement using two qubits coupled by a CNOT.

\subsection{The CNOT exRec}
\label{sec:threshold.golay.cnot}

The threshold calculation is most limited by the exRec with the largest number of locations.  The Golay code admits transversal implementations of encoded Clifford group unitaries.  Universality can be achieved by state distillation. Therefore the largest exRec in our case is for the encoded CNOT gate, an exRec that consists of four Steane-type error corrections plus $23$ CNOT gates (see \figref{fig:cnot-exrec-components}).  \tabref{tab:MarginalCounts} gives a breakdown of the number of locations for our preparation circuits, and the total number of locations in the CNOT exRec.

\begin{figure}
\centering
% \begin{subfigure}[b]{5cm}
% \includegraphics[width=5cm]{images/cnotExRec}
% \caption{\label{fig:exRec}}
% \end{subfigure}
\hspace{1cm}
\begin{subfigure}[b]{6.5cm}
\includegraphics[width=6.5cm]{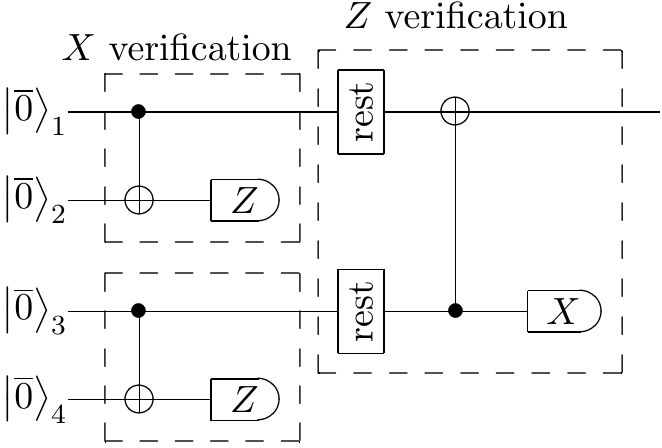}
\caption{\label{fig:ancilla-verification}}
\end{subfigure}
\begin{subfigure}[b]{14cm}
\includegraphics[width=13.5cm]{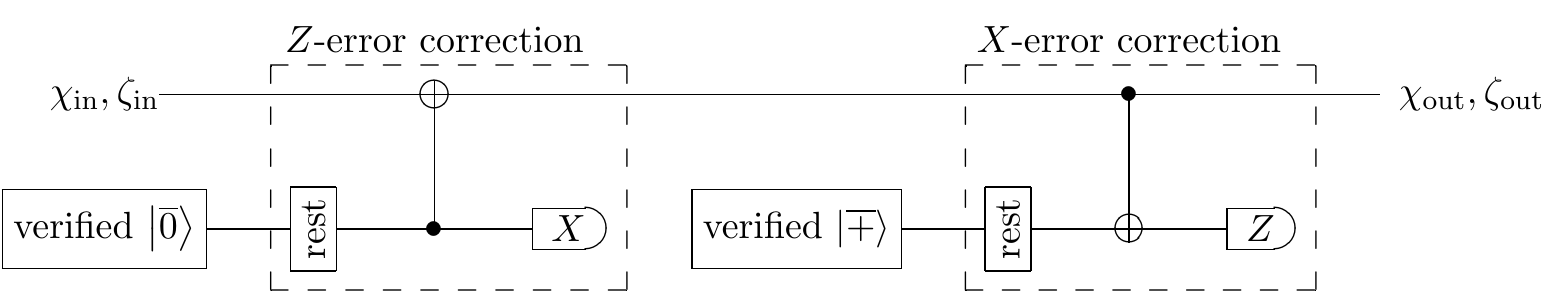}
\caption{\label{fig:ECclean}}
\end{subfigure}
\caption[CNOT exRec components.]{\label{fig:cnot-exrec-components}
Organization of a CNOT exRec, for the Golay code. The CNOT exRec includes four error corrections and a transversal CNOT gate as illustrated in~\figref{fig:component.two-qubit-exrec}.  (a) Each error-correction component consists of separate $Z$ and $X$ error corrections.  $Z$-error correction requires an encoded $\ket 0$ state ($\lket{0}$) that has been verified against errors, and $X$-error correction requires a verified $\lket{+}$ ancilla state.  (b) A verified $\lket 0$ state is prepared by checking two pairs of prepared $\lket 0$ states against each other for $X$ errors, then, conditioned on no $X$ errors being detected, checking the results against each other for $Z$ errors.  Verified $\lket +$ is prepared by taking the dual of the $\lket 0$ circuit. See~\chapref{chap:ancilla}.  
} 
\end{figure}

\begin{table}
\centering
\begin{tabular}{c|cccc|c|c}
\hline \hline
$\lket 0$ preparation & \multicolumn{4}{c|}{Location type} &  &CNOT exRec \\ 
circuit & CNOT & Prep. & Meas. & Rest & Total & total \\
\hline
Steane & 77 & 23 & 0 & 6 & 106 & 5439 \\
Overlap & 57 & 23 & 0 & 38 & 118 & 5823 \\
\hline \hline
\end{tabular}
\caption[Location counts for preparing $\ket0$ in the Golay code.]{Location counts for preparing encoded $\ket 0$ in the Golay code. Encoded $\ket{0}$ ancillas are prepared with either the pseudorandomly constructed Steane preparation circuits (\tabref{tab:randomizedSchedules}), or the overlap preparation circuits (\figref{fig:Golay-overlap-57} and \tabref{tbl:overlap-permutations}).  The last column shows the total number of locations inside the CNOT exRec shown in~\figref{fig:cnot-exrec-components}, including the transversal CNOT operation and four error corrections.}
\label{tab:MarginalCounts}
\end{table}

\begin{figure}
\centering
\begin{subfigure}[t]{.49\textwidth}
\includegraphics[width=\textwidth]{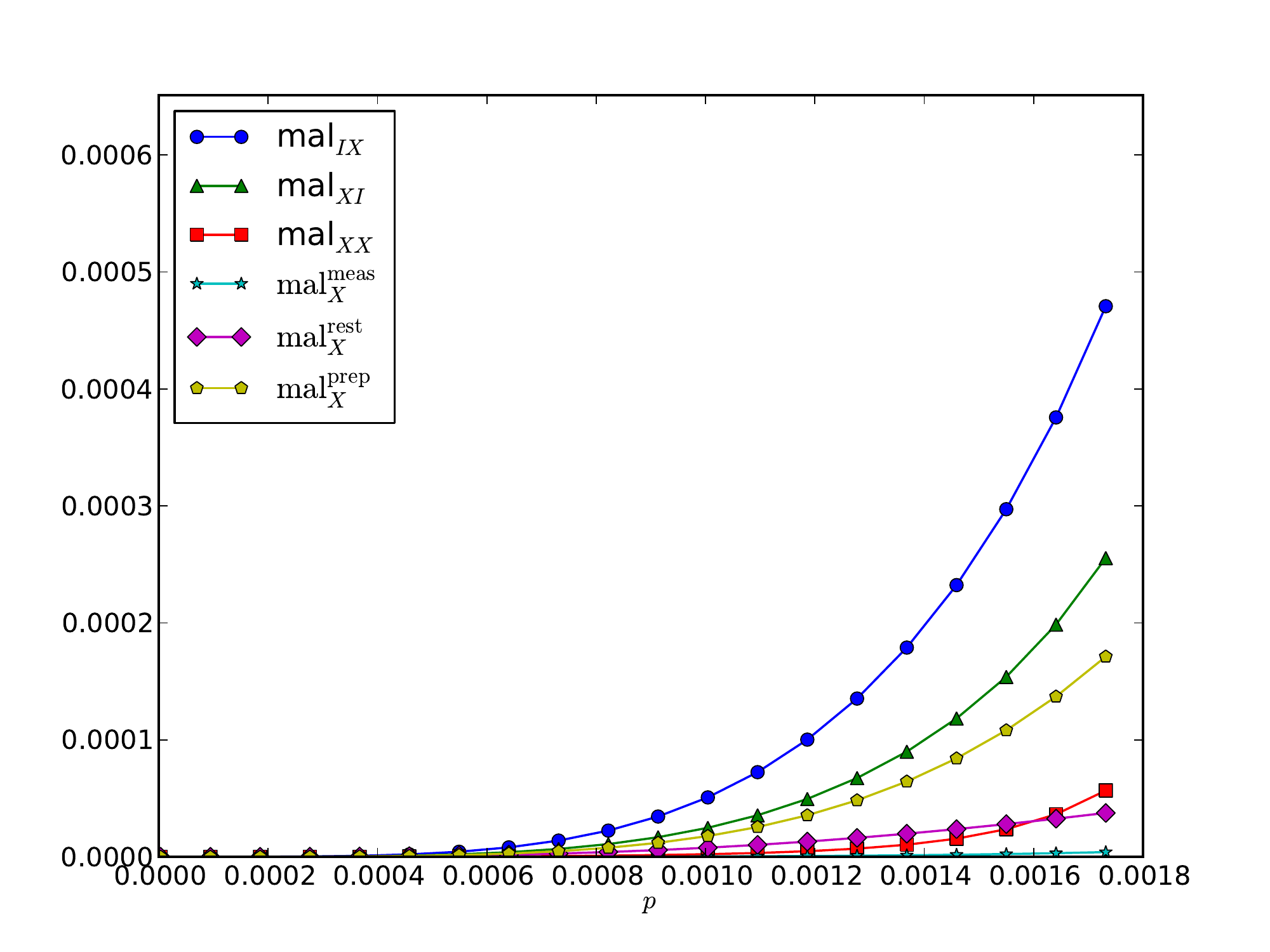}
\caption{\label{fig:malig-events-x}
$X$-error malignant events}
\end{subfigure}
\begin{subfigure}[t]{.49\textwidth}
\includegraphics[width=\textwidth]{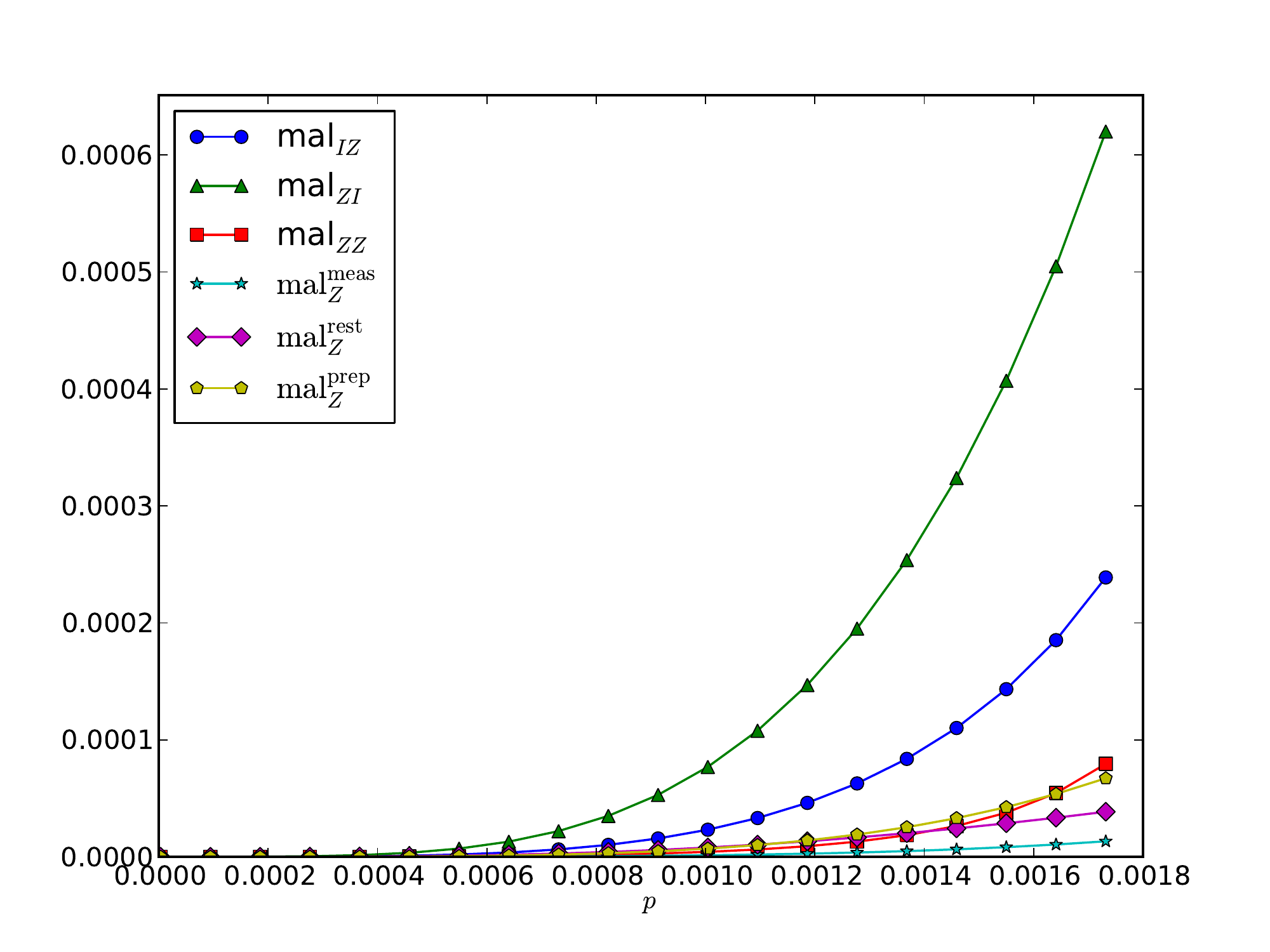}
\caption{\label{fig:malig-events-z}
$Z$-error malignant events}
\end{subfigure}
\caption[Malignant event probabilities for Golay code CNOT.]{These plots show upper bounds on probability of malignant events for the different level-one exRecs.  The $\malig_{IX}$, $\malig_{XI}$, $\malig_{XX}$, $\malig_{IZ}$, $\malig_{ZI}$ and $\malig_{ZZ}$ events all pertain to the CNOT exRec; the $\malig^{\text{prep}}_X$ and $\malig^{\text{prep}}_Z$ events correspond to the $\ket 0$ and $\ket +$ preparation exRecs, respectively; $\malig^{\text{meas}}_X$ and $\malig^{\text{meas}}_Z$ correspond to $Z$-basis and $X$-basis measurement exRecs; $\malig^{\text{rest}}_X$ and $\malig^{\text{rest}}_Z$ pertain to the rest exRecs.  Note that the upper bound on $\malig_{ZI}$ is significantly higher than that of its dual counterpart $\malig_{IX}$.  This is due largely to the arbitrary choice in error correction to correct $Z$ errors first and $X$ errors second.}
\label{fig:malig-events}
\end{figure}

\begin{table}
\centering
\begin{tabular}{c|c|c}
\hline\hline
Verification schedule & CNOT Pseudothreshold & Threshold \tabularnewline
\hline
Steane-$4$  & $1.72 \times 10^{-3}$ & $\threshSteane$ \tabularnewline
Overlap-$4$ & $1.73 \times 10^{-3}$ & $\threshOverlap$ \\
\hline \hline
\end{tabular}
\caption[Depolarizing noise threshold lower bounds for the Golay code.]{Threshold lower bounds for circuits based on our four-ancilla preparation and verification schedules for the Golay code, based on~\figref{fig:FourAncillaVerifyCkt}. Thresholds are given with respect to $p$ the probability that a physical CNOT gate fails, according to the depolarizing noise model defined in \secref{sec:threshold.golay.noise}.}
\label{tab:threshold-results}
\end{table}

\subsubsection{\texorpdfstring{$X$}{X}-error verification} 
\label{sec:threshold.golay.cnot.xverify}

$X$-error verification requires two encoded $\ket 0$ states. The first is verified against the second for $X$ errors by applying transversal CNOT gates between the two code blocks and then measuring each qubit of the second block in the $Z$ eigenbasis ($\ket 0, \ket 1$ basis).  Conditioned on no $X$ errors being detected, the first code block is accepted.  See \figref{fig:ancilla-verification}.  

Letting $\acceptX$ denote the event that no $X$ errors are detected, we use Bayes's rule 
\begin{equation}
  \Pr[\text{event} \vert \acceptX]=\frac{\Pr[\text{event}, \acceptX]}{\Pr[\acceptX]}
\end{equation} 
to compute the conditional probabilities of different error events.  For an event $\chi$ involving only~$X$ errors, this calculation is straightforward.  

However, if the event is a $Z$ error $\zeta$, then the numerator $\Pr[\zeta=z, \acceptX]$ is difficult to compute as it mixes $X$ and $Z$ errors.  The obvious bound, $\Pr[\zeta=z, \acceptX] \leq \Pr[\zeta=z]$, is quite pessimistic because in the depolarizing noise model we expect $X$ errors to occur with $Z$ errors roughly half of the time, and so $X$-error verification should remove many $Z$ errors.  It is important to obtain an accurate count of $Z$ errors since they strongly influence the acceptance rate of the upcoming $Z$-error verification.  Therefore, we also count $X$ and $Z$ errors \emph{together} for very low-order faults and apply a correction to the $Z$-only counts.  

Specifically, when counting $X$ and $Z$ errors together, we keep track of the errors that are \emph{rejected} rather than those that are accepted.  Since the $Z$-only counts contain all errors, we may subtract off the rejected error counts while maintaining proper counts for the accepted errors.
Details of are worked out in~\cite{Paetznick2011}.

The improvement obtained by counting $X$ and $Z$ errors simultaneously is twofold.  First, the reduction in $Z$ errors directly reduces the probability of a $Z$-error malignant event.  Indeed, we find that the correction cuts the number of $Z$ errors roughly in half, as expected.  More importantly, though, a smaller number of $Z$ errors means an increased acceptance probability during the upcoming $Z$-error verification.
We see from~\figref{fig:pr-accept-lower-bounds} that the lower bound on $Z$-error verification acceptance at $p=10^{-3}$ is about $0.84$.  We crudely estimate a lower bound \emph{without} the correction of about $0.63$, a decrease by a factor of~$1.3$.  There are four $Z$-error verifications of encoded $\ket 0$ in the (full) exRec and four similar $X$-error verifications of encoded $\ket +$.  Thus, in the normalization factor alone, the correction reduces upper bounds on the malignant event probabilities by roughly a factor of $1.3^8 \approx 8$.  The savings is less, of course, as $p$ decreases.

\subsubsection{exRec}
\label{sec:threshold.golay.cnot.exrec}

Counting of the exRec component was discussed in~\secref{sec:threshold.split.exrec}.  However, there are a few items of note for our example based on the Golay code.
First, the ancilla verification components are non-deterministic.  Accordingly, all of the malignant event probabilities must be conditioned on acceptance of all of the verification stages.  Since the counts reported by the ancilla verification stages assume successful verification already, calculating the conditional probability is simply a matter of dividing by the product of all of the acceptance probabilities.

Second, we seek to combine large subsets of the sub-component counts.  However, due to the block-size of the Golay code and size of the sub-components in the CNOT exRec, taking all possible convolutions of the sub-component error counts is impractical.  Instead,
the $\badX$ event for the exRec (and analogously the $\badZ$ event) occurs when any of the following are true: 
\begin{itemize}
  \item any of the sub-components are $\badX$,  
  \item there are more than $25$ $X$ failures in the exRec,
  \item there is more than one $X$ failure in the transversal CNOT \emph{and} there are more than than three $X$ failures in each of the two leading ECs.  
\end{itemize}

The last condition eliminates faults that are particularly difficult to count. The time required to count an exRec fault is proportional to the product of the number of unique syndromes that can result at the output of the two leading ECs and the transversal CNOT.  The number of unique syndromes that can result from the transversal CNOT with two $X$ failures is $\binom{23}{2} 3^2 = 2277$, while the number of unique syndromes with one $X$ failure is $23 \cdot 3 = 69$.  The numbers of unique syndromes at the output of the leading ECs are $24$, $277$ and $2048$ for one, two, and three $X$ failures respectively.  So, for example, the event $K_{X,1}=2, K_{X,2}=3, K_{X,3}=1$ ($277 \cdot 2048 \cdot 69 \approx 4 \cdot 10^7$) requires far less time than the event $K_{X,1}=2, K_{X,2}=3, K_{X,3}=2$ ($277 \cdot 2048 \cdot 2277 \approx 1 \cdot 10^{9}$).  In particular, we would like to avoid counting faults for which $K_{X,3} = 2$.  

Calculations for each of the $\badX$ terms are plotted in \figref{fig:prBad-cnot-upper-bounds}.
Label each of the exRec sub-components with numbers, starting with the LECs ($1,2$), then the CNOT ($3$), and then the TECs ($4,5$).  
%Na{\"i}vely one might expect bounds for partial exRecs---those for which one or more TECs have been removed---to be lower than bound for the full exRec by as much as a factor of two.
The overall probability is generally dominated by either the transversal CNOT ($\Pr[\bad_X^{(3)}]$) or the condition involving the transversal CNOT and the two LECs ($\Pr[K_{X,3}>1] \prod_{j=1}^2 \Pr[K_{X,j} > 3 \vert \accept^{(j)}]$).  
%Thus removing the TECs has little impact on the probability that the exRec is bad.  

\begin{figure}
\begin{subfigure}[b]{.495\textwidth}
\includegraphics[width=\textwidth]{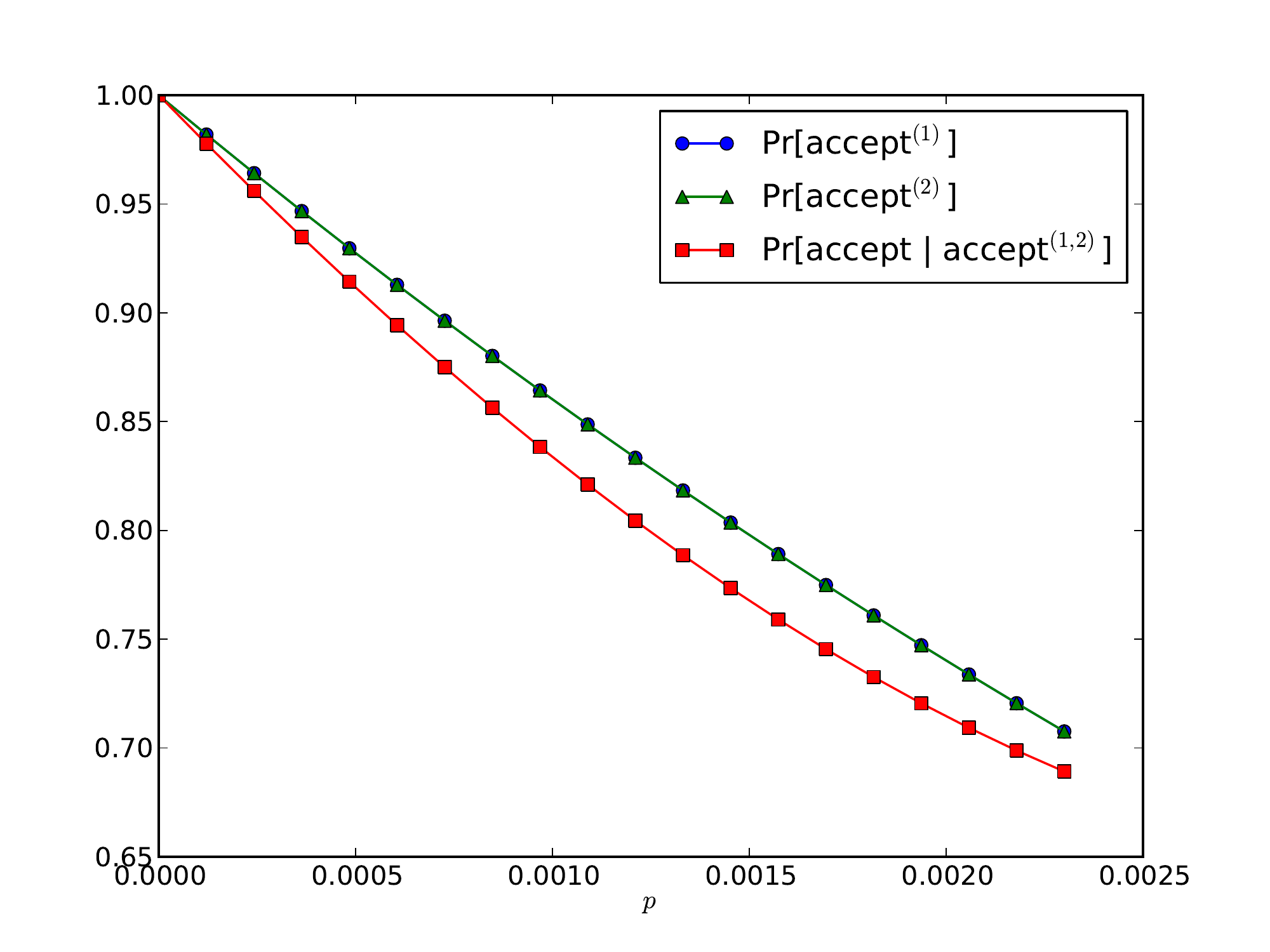}
\caption{\label{fig:pr-accept-lower-bounds}}
\end{subfigure}
\begin{subfigure}[b]{.495\textwidth}
\includegraphics[width=\textwidth]{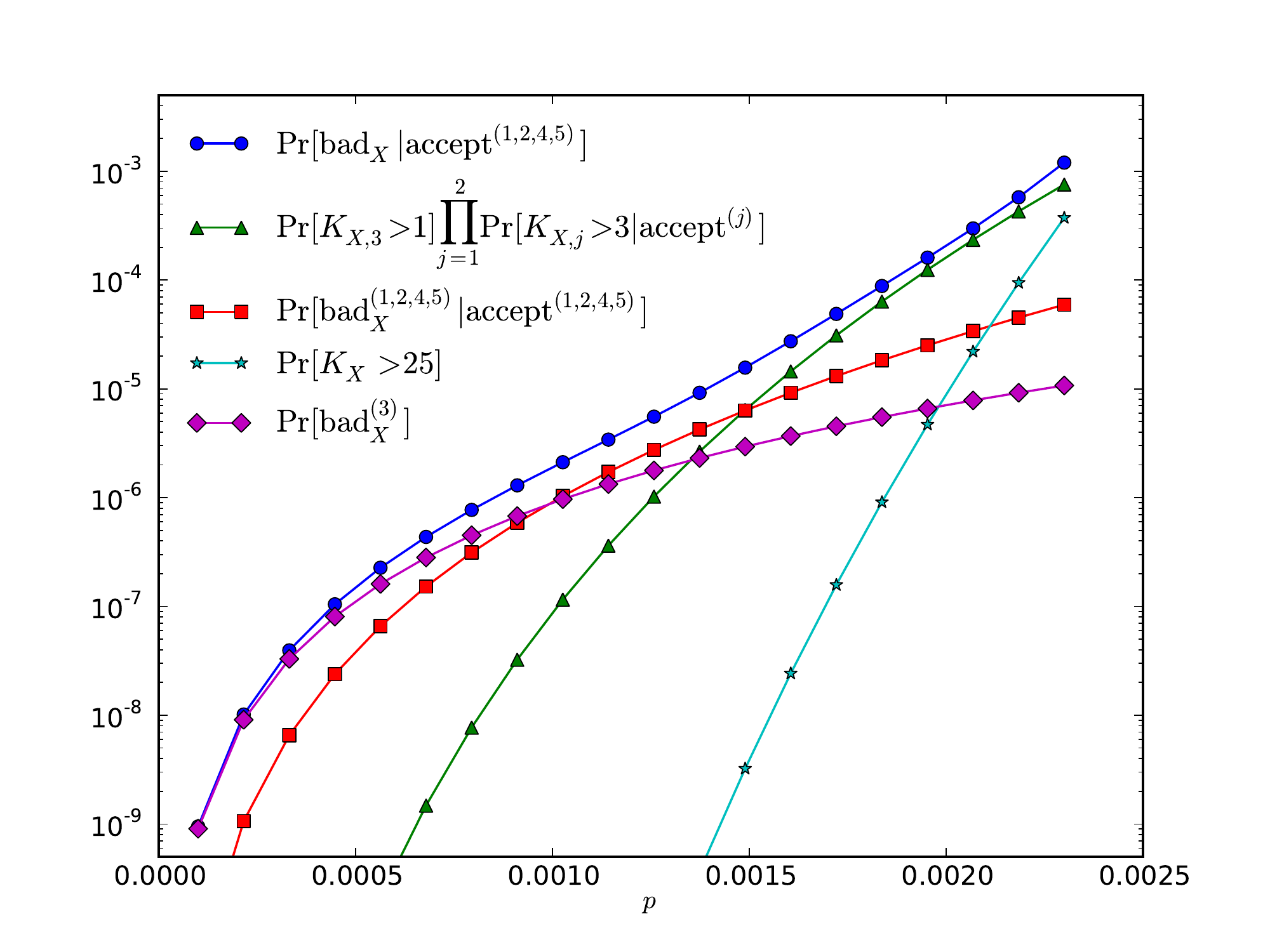}
\caption{\label{fig:prBad-cnot-upper-bounds}}
\end{subfigure}
\caption[Acceptance probabilities and {$\Pr[\bad]$}.]{Plotted in (a) are lower bounds on the Overlap-$4$ acceptance probabilities for the two $X$-error verifications ($\accept^{(1)}$ and $\accept^{(2)}$) and for the $Z$-error verification ($\accept$) conditioned on success of the $X$-error verifications.  The plot in~(b) shows upper bounds on conditions that lead to a $\badX$ event in the CNOT exRec.
}
\end{figure}

\subsection{Threshold analysis}
\label{sec:threshold.golay.threshold}
Our thresholds compare favorably to threshold results for similar circuits.  For a six-ancilla preparation and verification circuit, Aliferis and Cross~\cite{Aliferis2007b} give a threshold estimate based on malignant set sampling of $p \approx 1 \times 10^{-4}$ for adversarial noise.  Our results beat this by an order of magnitude and provide strong evidence that our counting technique is an improvement over malignant set sampling and malignant set counting for the case of depolarizing noise.  
Our results also essentially close the gap with other analytical and Monte Carlo threshold estimates for depolarizing noise.  Using a closed form analysis, Steane~\cite{Steane2003} estimated a threshold on the order of $10^{-3}$ for the Golay code with similar noise parameters.  Dawson, Haselgrove and Nielsen calculated a higher estimate of just under $3\times 10^{-3}$, and Cross et al.~\cite{Cross2009} estimated a pseudo-threshold of $2.25 \times 10^{-3}$ based on Monte Carlo simulations of a twelve-ancilla preparation and verification circuit.

Beyond circuits based on the Golay code, our results are apparently the highest rigorous threshold lower bounds known.  Aliferis and Preskill~\cite{Aliferis2009} prove a lower bound of $p \geq 1.25 \times 10^{-3}$. Their analysis applies to teleportation-based gates due to Knill~\cite{Knill2004} in which Bell pairs encoded into an error correcting code $C_2$ are prepared by first encoding each qubit of the $C_2$ block into an error \emph{detecting} code $C_1$ and performing error detection and postselection after each step of the $C_2$ encoding. Our best threshold is only about $5$ percent better, but applies to circuits that usually require far less overhead (see~\secref{sec:threshold.golay.resources}).  This implies only that in the depolarizing noise model our analysis is more accurate, and not that our schemes tolerate more noise.

The limiting factor on the threshold value is the event $\malig_{ZI}$.  That is, $\malig_{ZI}$ is the event $E$ for which $\Pr[\malig_E^{(2)}] = \Pr[\malig_E^{(1)}]$ takes the smallest value of $p$.  In fact, the corresponding threshold values for nearly all $Z$-error malignant events are lower than threshold values for \emph{any} of the $X$-error events.  This asymmetry is due to the arbitrary order with which we perform error correction---$Z$ first, then $X$.  Some $X$ errors resulting from the leading $Z$-error correction will be corrected by the $X$-error correction that follows.  However, $Z$ errors resulting from the $X$-error correction may propagate through the encoded operation before arriving at the $Z$-error correction on the trailing end.  As a result, it is more likely for $Z$ errors on individual blocks to be combined by the CNOT gate and create an uncorrectable error.  Evidence of this effect can be seen in the level-one malignant event probabilities shown in \figref{fig:malig-events}.  

It should be possible to reduce such lopsided event probabilities by customizing the error correction order for each EC based on the specifics of the ancilla preparation circuits.  However, analyzing such a scheme would require consideration of up to $36$ different full or partial CNOT exRecs (two choices for each EC) instead of four and is likely to yield only a small improvement in the threshold.  Note that other small improvements could be made by, for example, eliminating measurement or rest exRecs at level-two.  For simplicity, these optimizations were not considered.

\subsection{Resource analysis}
\label{sec:threshold.golay.resources}

The threshold provides a target accuracy for quantum computing hardware, but it does not produce a complete picture on its own.  In particular, we would also like to understand how the resource overhead for our scheme scales as the physical error rate drops below threshold.  Ultimately, the resource scaling will determine how small physical error rates must be in order to keep space and time resources to a manageable level.
In this section we calculate upper bounds on the number of physical gates and the number of physical qubits required to implement a single logical gate with a given effective error rate.

Our threshold analysis assumes that an infinite supply of ancilla qubits is available for use in error correction.  In order to bound the resource overhead we instead assume that some finite number of ancillas are available to each $k$-EC. Error correction proceeds normally unless all ancilla verifications fail. If the number of available ancillas is high enough, then the probability that all verifications fail will be small and the impact on the logical errors will be similarly small.

More precisely, our approach is as follows.  The ancilla verification circuit (\figref{fig:ancilla-verification}) is considered as a single unit.  Each level-$k$ $Z$-error correction consists of $m_k$ $\ket 0$ verifications performed in parallel plus a transversal rest, CNOT and $X$-basis measurement.  If all of the $m_k$ verifications fail, then $Z$-error correction is aborted and the data is left idle.  Level-$k$ $X$-error correction is similar.  For simplicity, if any of the error corrections are aborted, then we consider the entire top-level logical gate to have failed. 

Let $p_\text{target}$ be overall target error rate per logical gate, $\Pevent^{(k)} := \max_i \Pevent_i^{(k)}$, and let $K$ be the minimum level of concatenation that achieves $\Pevent^{(k)} < p_\text{target}$ assuming an unbounded number of ancilla.
We may then calculate a bound on the number of ancilla verifications $m_k$ for every $k \leq K$.
Setting $\delta^{(k)} = p_\text{target} - \Pevent^{(k)}$, the total gate overhead $g(k)$ for a CNOT $k$-Rec can be computed recursively by $g(k) \leq (2 m_k A_\text{EC} + 23) \cdot g(k-1)$, where $A_{\text{EC}}$ is the number of locations in the error-correction component. Details are provided in~\cite{Paetznick2011}.

\begin{figure}
\label{fig:gate-overhead}
\centering
\begin{subfigure}[t]{.49\textwidth}
\includegraphics[width=\textwidth]{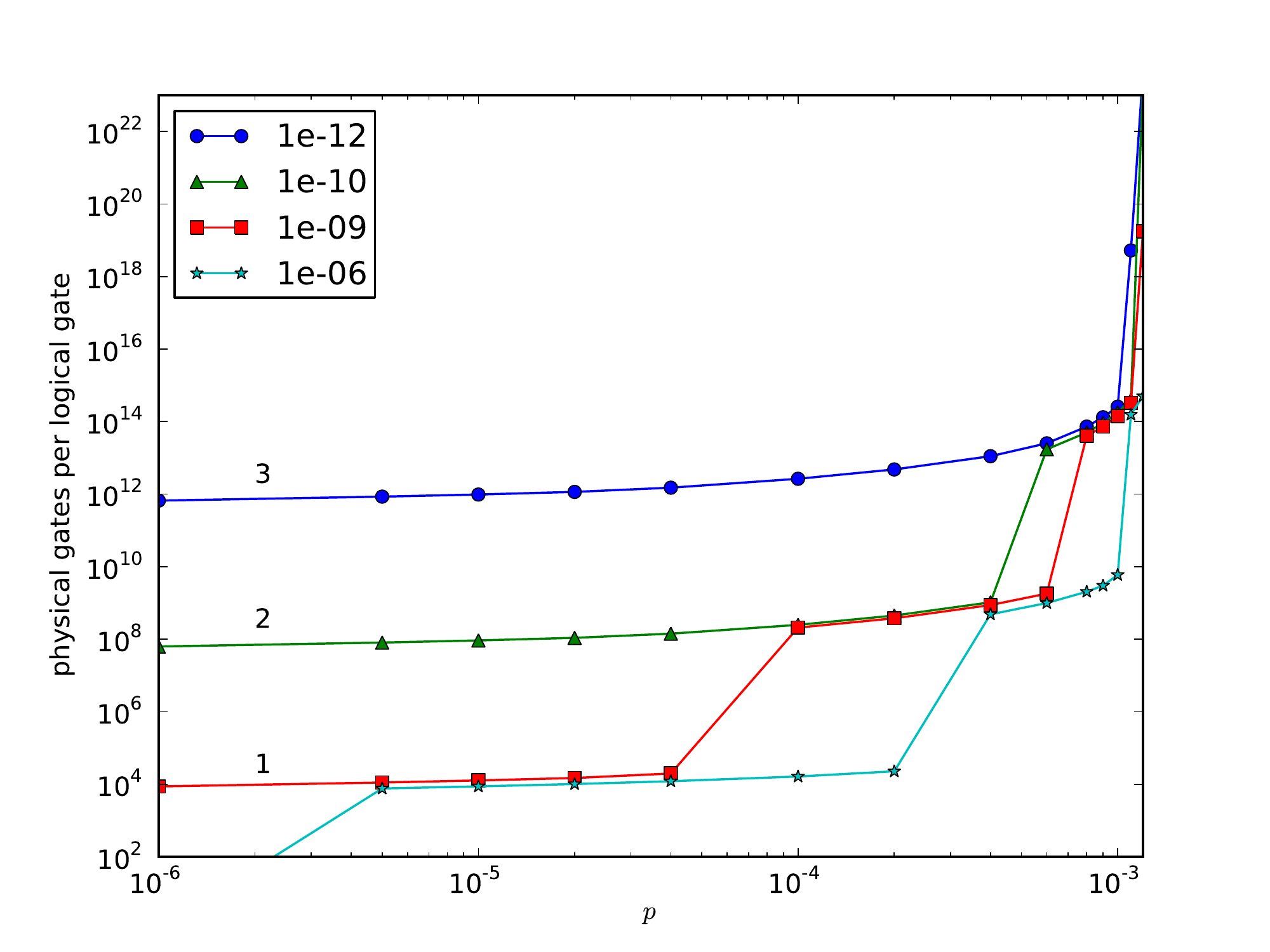}
\caption{\label{fig:gate-overhead-overlap}
Golay scheme with Overlap-$4$ preparation}
\end{subfigure}
\begin{subfigure}[t]{.49\textwidth}
\includegraphics[width=\textwidth]{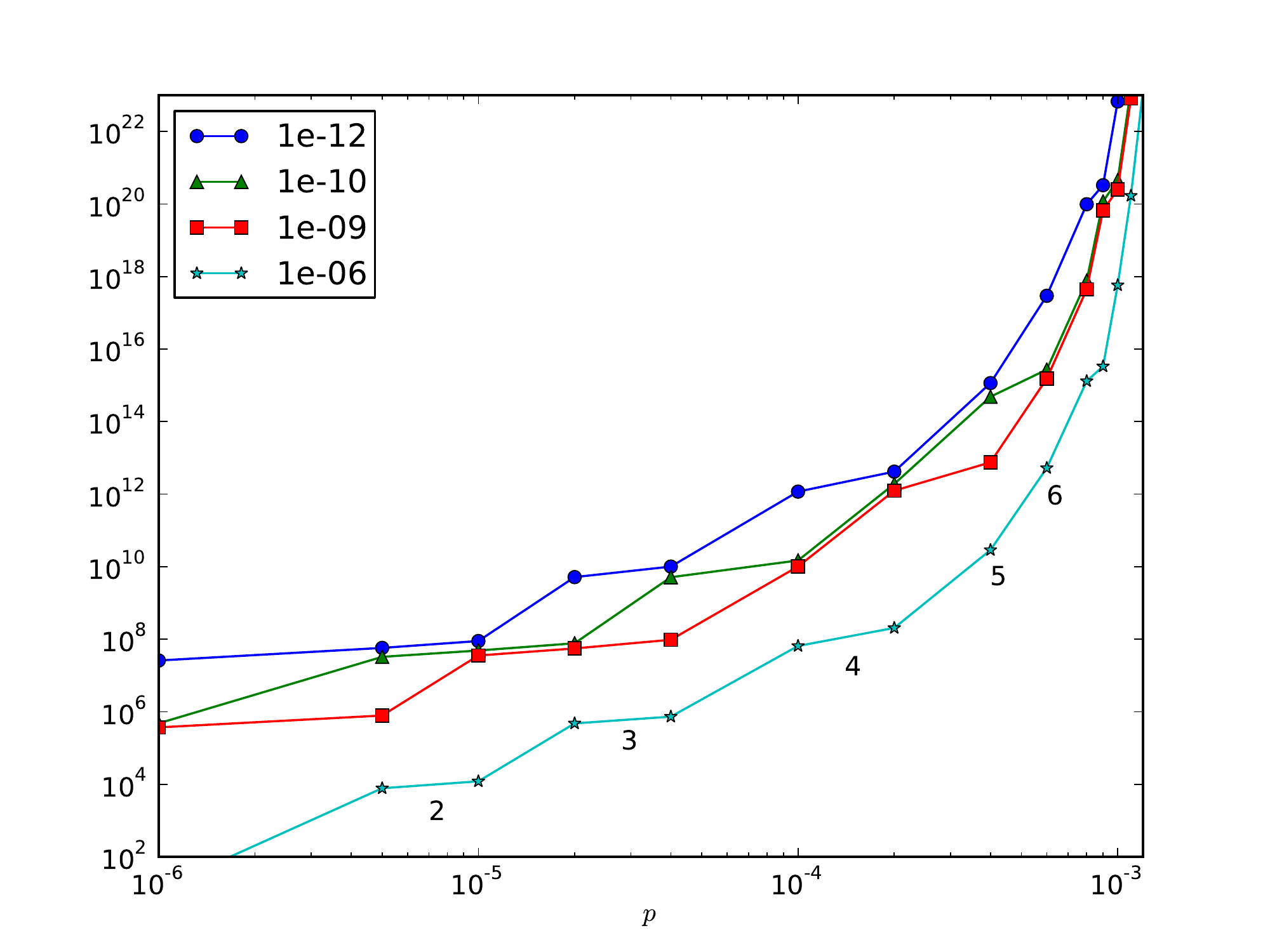}
\caption{\label{fig:gate-overhead-fibonacci}
{$[[4,2,2]]$} Fibonacci scheme
}
\end{subfigure}
\caption[Gate overhead for Golay and Fibonacci schemes.]{\label{fig:gate-overhead-plot}
Gate overhead upper bounds for (a) our Golay scheme with overlap ancilla preparation and (b) the Fibonacci scheme presented in~\cite{Aliferis2009}. Each plot shows the number of physical gates required to implement a logical gate with target error rates $p_\text{target} \in \{ 10^{-12}, 10^{-10}, 10^{-9}, 10^{-6} \}$. Black text labels indicate the required level of concatenation and colored lines are a guide for the eye.}
\end{figure}

\begin{figure}
\centering
\begin{subfigure}[t]{.49\textwidth}
\includegraphics[width=\textwidth]{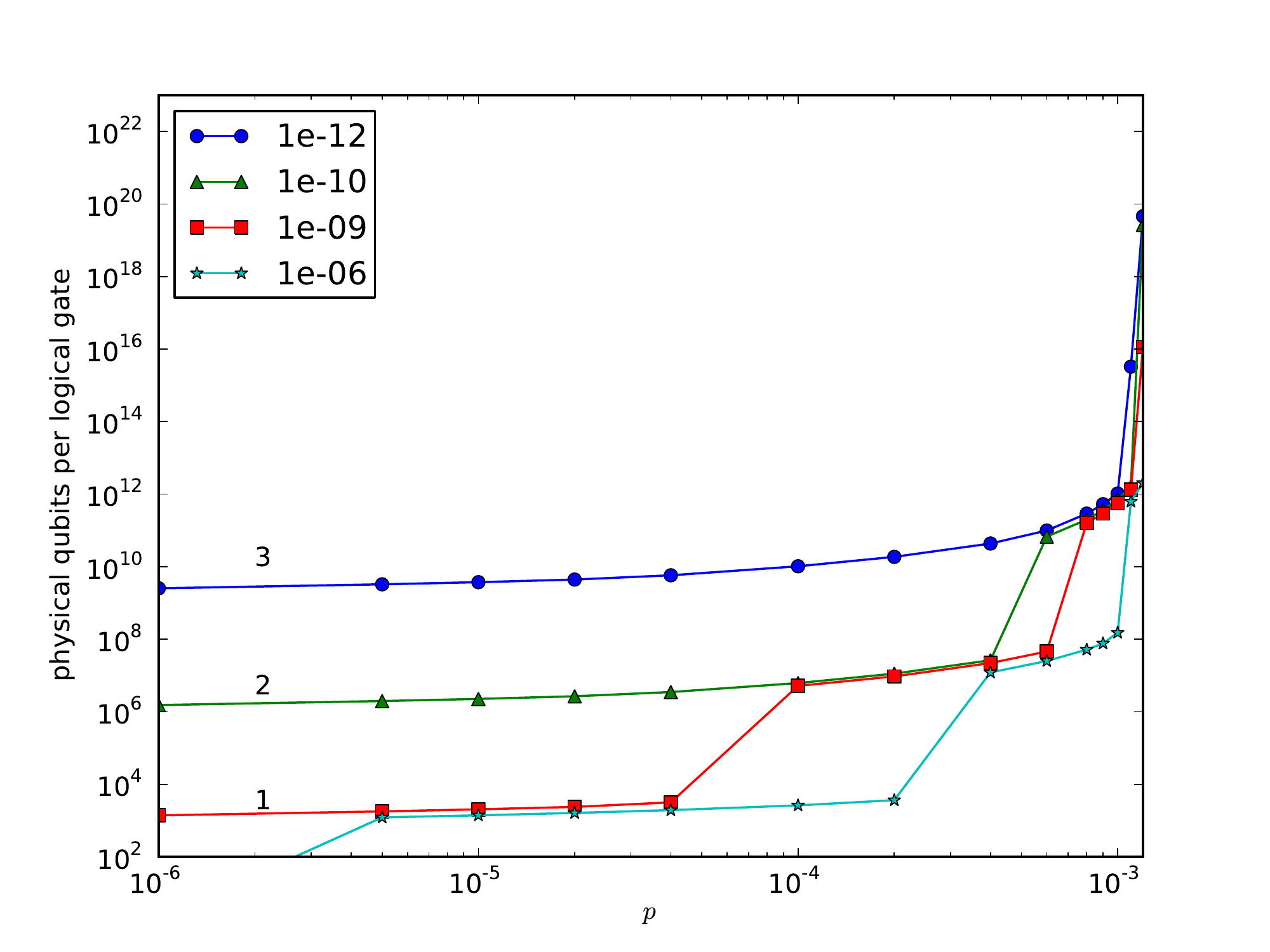}
\caption{
Golay scheme with Overlap-$4$ preparation
}
\end{subfigure}
\begin{subfigure}[t]{.49\textwidth}
\includegraphics[width=\textwidth]{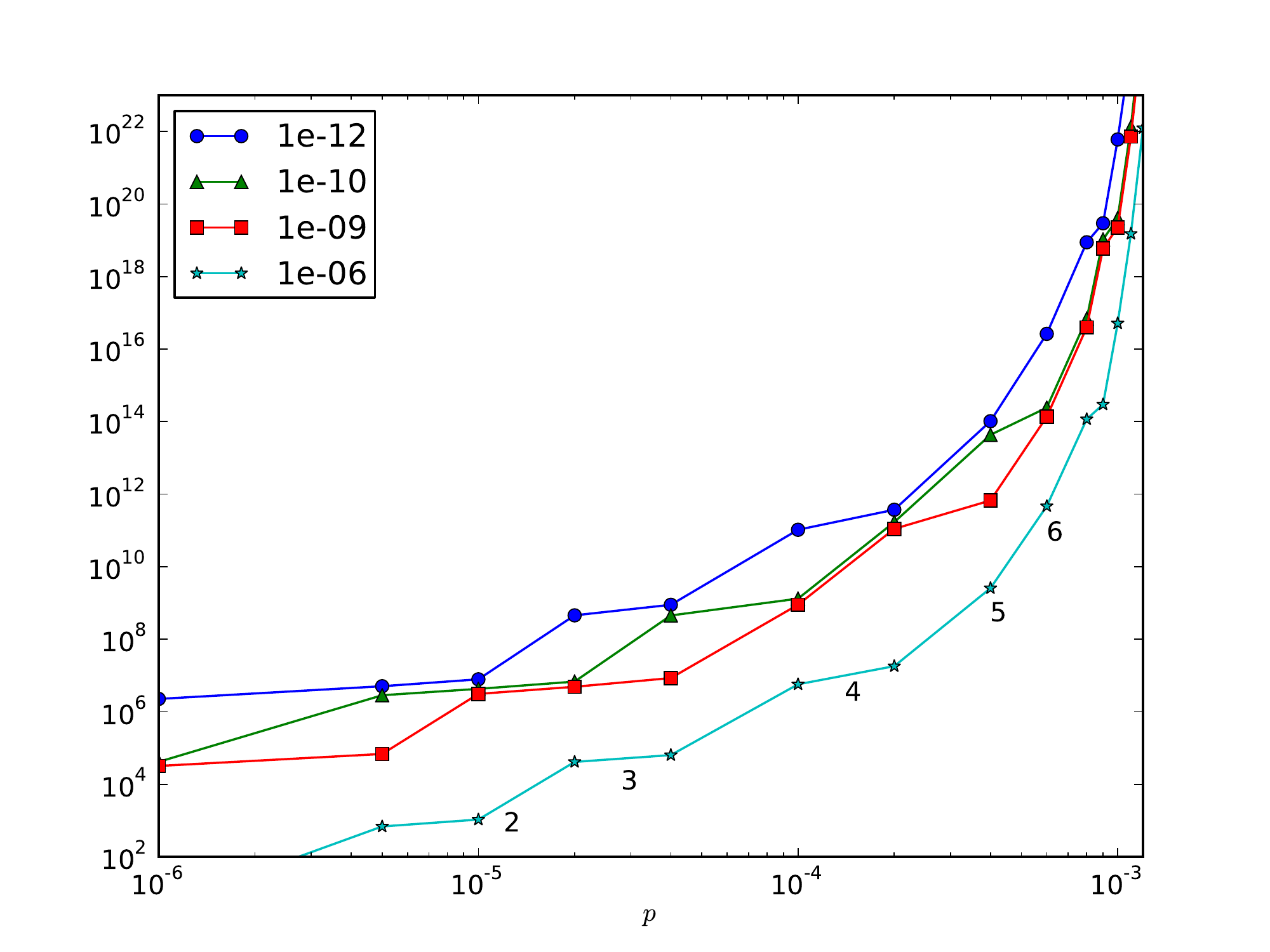}
\caption{
{$[[4,2,2]]$} Fibonacci scheme
}
\end{subfigure}
\caption[Qubit overhead for Golay and Fibonacci schemes.]{\label{fig:qubit-overhead-plot}
Qubit overhead upper bounds.  Plots are formatted identically to~\figref{fig:gate-overhead-plot}.
}
\end{figure}

Gate overhead upper bounds for the overlap-based scheme are shown in~\figref{fig:gate-overhead-overlap}.  The overhead increases dramatically as the target logical error rate decreases.  However, compared to similar upper bounds for the Fibonacci scheme---which has a similar threshold lower bound~\cite{Aliferis2009}---our scheme is better for a wide range of error rates often by several orders of magnitude.  One reason for the improved overhead is that our scheme is based on a code with higher distance than the Fibonacci scheme which uses the $[[4,2,2]]$ error detecting code. The logical error rate for our Golay scheme falls faster and thus requires fewer levels of concatentation.

Bounds on qubit overhead may be obtained from the gate overhead.  Our threshold analysis requires that all ancillas  be ready on-demand without delay---i.e., each $k$-Rec has depth three, independent of $k$.  We, therefore, pessimistically assume that once a qubit is measured it cannot be re-used within the same rectangle.  The qubit overhead then depends only on the gate overhead and the qubit-gate ratio for $\ket 0$ verification.  Using a ratio of $8 \cdot 23 / (A_\text{EC} - 46)$ we obtain $q(k) \leq 23^k + 0.15^k g(k)$ Therefore, the level-$k$ qubit overhead is roughly $k$ orders of magnitude lower than the level-$k$ gate overhead.

The qubit-gate ratio for Bell-state preparation in the Fibonacci scheme is relatively large ($\approx 0.6$ for levels three and above).  Therefore, similar to gate overhead, qubit overhead for the Golay scheme compares favorably to the Fibonacci scheme for a wide range of noise parameters. See~\figref{fig:qubit-overhead-plot}.

The drawback of using a larger code is that the increase in overhead from one level of concatenation to the next is much larger.  This makes it harder to ``tune'' the overhead parameters to some specific error rates.  For example, for $p_\text{target}=10^{-10}$ and $p = 10^{-6}$ our scheme requires two levels of concatenation and about $10^8$ physical gates per logical gate.  For the same error rates, the Fibonacci scheme requires three levels of concatenation, but fewer than $10^6$ gates.

Finally, note that bounds for our scheme when $p_\text{target}=10^{-12}$ are a bit loose due to a constant offset that is added during the transformed noise model construction.  In our computer analysis, these offsets were on the order of $\epsilon \approx 10^{-13}$.  In principle, this offset does not affect the actual error rates; rather it is an artifact of our construction.

\section{Discussion}
\label{sec:threshold.discussion}

Our explicit calculations for the Golay code show the power of the modified malignant set counting technique.  Compared to standard malignant set counting we are able to count much larger sets of faulty locations, and obtain a bound on the threshold which is about an order of magnitude larger than previous attempts. Intuitively, this is because we efficiently ignore subsets of faulty locations which are unlikely to occur. Use of the independent Pauli noise model permits fair comparisons of our bounds with Monte Carlo estimates.  In the case of the Golay code, our rigorous lower bound roughly matches numerical estimates due to~\cite{Steane2003,Dawson2006,Cross2009}.

The technique is quite general, and can be applied to any CSS code.  However, there are still several drawbacks to our approach.  First, we count errors in terms of equivalence classes based on the stabilizers of the code, but the number of unique errors per block is still exponential in the block size.  For the Golay code, this meant keeping track of $2^{12}$ $X$ errors and $2^{12}$ $Z$ errors.  For two blocks the total number of errors was $2\cdot 2^{24}$.  This number of errors is manageable, but numbers for larger codes may become unwieldy.

Another drawback is that we have assumed arbitrary qubit interactions, ignoring any physical geometric locality constraints.  This simplifies the analysis greatly, but artificially inflates the threshold and underestimates resource requirements in the case that geometric constraints are actually required.  Therefore, our results are not directly comparable to thresholds for topological codes including the surface code, for example.  Of course, our technique can be adapted to account for geometric constraints by, for example, inserting swap gates, if necessary. We have not considered such adaptations here.

\chapter{Decomposition of single-qubit unitaries into fault-tolerant gates
\label{chap:repeat}
}

This chapter is based on material that appears in~\cite{Paetznick2013b}.
\vspace{1cm}

The mapping of a quantum algorithm into its equivalent fault-tolerant circuit representation requires a choice of universal basis, most commonly consisting of CNOT and single-qubit gates. (See~\secref{sec:mechanics.universality}.)
Traditional methods for single-qubit unitary decomposition take as input a unitary $U$ and a distance parameter $\epsilon$, and output a sequence of gates $W = G_1\ldots G_k$ such that $\norm{U-W} \leq \epsilon$, for $G_1,\ldots,G_k$ in the chosen gate set, and some choice of norm $\norm{\cdot}$.  The operation $W$ is said to approximate $U$ to within a distance $\epsilon$.  This approach is justified by the fact that when $\norm{U-W}$ is small, the output distribution of a circuit containing $U$ is close to the output distribution obtained by substituting $W$.

The set of single-qubit unitaries that can be implemented fault tolerantly is predominantly dictated by the existence of resource-efficient fault-tolerance protocols. See~\secref{sec:fault.gates}.
A common universal, single-qubit basis is $\{H,S,T\}$, since $H$ and $S$ can often be implemented transversally, and $T$ can be achieved through state distillation.
The cost of a $\{H,S,T\}$ circuit is usually defined to be the number of $T$ gates, since the resource cost of a fault-tolerant $T$ gate is up to an order of magnitude larger than the resource cost of a fault-tolerant $H$ gate~\cite{Raussendorf2007a,Fowler2013}.\footnote{The inclusion of $S$ is a direct consequence of the choice of cost function. The $S$ gate is otherwise redundant since $S = T^2$.}

Typically, the approximation $W$ involves no measurements, and is therefore deterministic (at the logical level).  In this chapter we will show that by allowing a small number of ancilla qubits and measurements, non-deterministic circuits can outperform deterministic circuits which are otherwise optimal.  The circuits that we consider can be used to approximate a single-qubit unitary with roughly one-third to one-fourth the cost of traditional decomposition methods.

As an example, consider the circuit shown in~\figref{fig:v3-toffoli}, which performs the single-qubit unitary $V_3 = (I + 2iZ)/\sqrt{5}$.
This circuit involves two measurements in the $X$-basis.  If both measurement outcomes are zero, then the output is equivalent to $V_3\ket\psi$.  If any other outcome occurs, then the output is $I\ket\psi = \ket\psi$.  Thus, the circuit may be repeated until obtaining the all zeros outcome, and the number of repetitions will vary according to a geometric probability distribution. (In this case the probability of getting both zeros is $5/8$.)  Upon measuring all zeros, the unitary $V_3$ is implemented $\emph{exactly}$, even though the overall circuit is non-deterministic.  Each Toffoli gate can be implemented using four $T$ gates, and so the overall expected cost is $12.8$.  By contrast, an approximation of $V_3$ to within $\epsilon = 10^{-6}$ using the deterministic algorithm of~\cite{Kliuchnikov2012b} requires $67$ $T$ gates.

\begin{figure}
\centering
\begin{subfigure}[b]{.3\textwidth}
\centering
\includegraphics[scale=1]{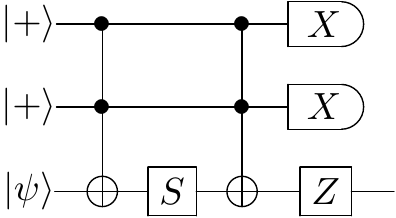}
\caption{\label{fig:v3-toffoli}
$\text{Exp}[$T$]=12.8$}
\end{subfigure}
\hfill
\begin{subfigure}[b]{.3\textwidth}
\centering
\includegraphics[scale=1]{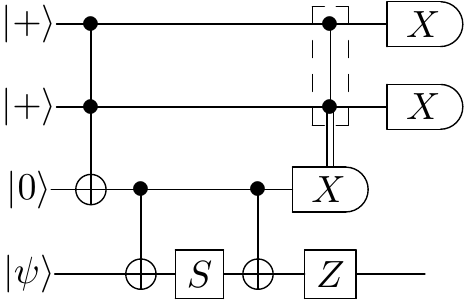}
\caption{\label{fig:v3-jones}
$\text{Exp}[$T$]=6.4$}
\end{subfigure}
\hfill
\begin{subfigure}[b]{.3\textwidth}
\centering
\includegraphics[scale=1]{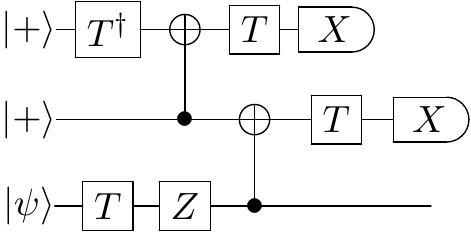}
\caption{\label{fig:v3-5.26}
$\text{Exp}[$T$] < 5.26$}
\end{subfigure}
\caption[Repeat-until-success circuits for $V_3$]{
\label{fig:v3-circuits}
Repeat-until-success circuits for $V_3 = (I + 2iZ)/\sqrt{5}$.
Each of the circuits above implements $V_3$ conditioned on an $X$-basis measurement outcome of zero on each of the top two ancilla qubits. If any other measurement outcome occurs, then each circuit implements the identity. The probability of measuring $00$ is $5/8$ for each circuit.  Repeating the circuit until success yields an expectation value for the number of $T$ gates, as indicated.
(a) A slight modification of the circuit presented in~\cite{Nielsen2000} pp. $198$. Each Toffoli gate can be implemented with four $T$ gates (see~\cite{Jones2012d}). (b) A circuit proposed by Jones that requires just a single Toffoli gate~\cite{Jones2013b}.
(c) An alternative circuit found by our computer search. Measurement of the first qubit can be performed before interaction with the data qubit.  Thus the top-left part of the circuit can be repeated until measuring zero.  The probability of measuring zero on the first qubit is $3/4$.  The probability of measuring zero on the second qubit, conditioned on zero outcome of the first qubit, is $5/6$.  The $T$ gate applied directly to $\ket\psi$ can be freely commuted through the CNOT.  In the case that an even number of attempts are required, the $T$ gates can be combined yield the Clifford gate $T^2 = S$.
}
\end{figure}

We call a circuit of the form of~\figref{fig:v3-toffoli}, which may be repeated until obtaining some desired outcome, a \emph{repeat-until-success} circuit or RUS circuit for short.
%RUS circuits have been used previously for unitary decomposition by Wiebe and Kliuchnikov~\cite{Wiebe2013}, who propose a family of tree-like, hierarchical RUS circuits that yield small $T$ counts for small-angle $Z$-axis rotations.
%Our results show that the utility of RUS circuits applies to arbitrary $Z$-axis rotations, not just small ones.
Through the use of an optimized direct-search algorithm, we present thousands of RUS circuits which exactly implement select unitary rotations at extremely low $T$-count.
By explicitly computing the circuit sequences, we construct a large database of single-qubit unitaries which is sufficiently large to approximate an arbitrary $Z$-axis rotation with within $\epsilon \geq 10^{-6}$.  
% The database is sufficiently dense to approximate an arbitrary single-qubit unitary to within distance $\epsilon = 10^{-6}$. 
Using this database, the expected number of $T$ gates required to approximate a random $Z$-axis rotation $R_Z(\theta) = \left(\begin{smallmatrix} 1&0\\0&e^{i\theta} \end{smallmatrix}\right)$ scales as
\begin{equation}
\label{eq:axial-tcount-scaling}
\text{Exp}_Z[T] = 1.26\log_2(1/\epsilon) - 3.53
\enspace.
\end{equation}
While existing algorithmic decomposition methods are capable of approximations to smaller distances, our techniques provide approximations with extremely low $T$ counts. Furthermore, approximations to within $10^{-6}$ are sufficient for many quantum algorithms, including Shor's factoring algorithm~\cite{Fowl03b}, and quantum chemistry algorithms~\cite{Jones2012}.

An arbitrary single-qubit unitary $U$ can be approximated by first expressing it as a product of three $Z$-axis rotations
\begin{equation}
U = R_Z(\theta_1)H R_Z(\theta_2)H R_Z(\theta_3)
\enspace .
\end{equation}
Each rotation can then be decomposed individually.  However, RUS circuits can also be used to approximate arbitrary single-qubit unitaries directly, without resorting to $Z$-axis rotations.  Our results indicate a $T$-count scaling of
\begin{equation}
\label{eq:nonaxial-tcount-scaling}
\text{Exp}_U[T] = 2.4\log_2(1/\epsilon) - 3.3
\enspace ,
\end{equation}
roughly another $50$ percent better than using~\eqnref{eq:axial-tcount-scaling} and up to four-fold better than traditional deterministic decomposition of three $Z$-axis rotations.
Constructing a database of RUS circuits for arbitrary unitaries is significantly more challenging than for the $Z$-axis case, however.  We have computed approximations only up to $\epsilon = 8 \times 10^{-3}$.

\section{Deterministic decomposition methods
\label{sec:repeat.deterministic}
}
By the Solovay-Kitaev theorem~\cite{Kitaev1997a,Kita02}, a single-qubit unitary operation can be efficiently approximated to within a desired $\epsilon$ by decomposition into a sequence of gates from a discrete universal basis with length $O(\log^c(1/\epsilon))$, where $c = 1$ is the asymptotic lower bound~\cite{Knill1995}, and the best-known practical implementation achieves $c\sim3.97$~\cite{Dawson2005}. The algorithm works by finding progressively better approximations of a unitary $U$, through application of the \emph{group commutator} $G_1G_2G_1^{\dagger}G_2^{\dagger}$ for pairs of gates $G_1, G_2$ in the gate set.  The key insight of the theorem is that use of the group commutator converges to $U$ exponentially fast.

Approximations with optimal scaling $O(\log(1/\epsilon))$ are possible.
Fowler proposed an exponential-time algorithm that yields an optimal decomposition with a $T$-gate count of roughly $2.95\log_2(1/\epsilon) + 3.75$~\cite{Fowl04c}, on average.  He used an optimized, but exhaustive search over gate sequences of progressively longer length, stopping at the first sequence within the required distance.  The weakness of this approach is that it is practical for approximations only up to about $\epsilon \geq 10^{-3}$.  Bocharov and Svore have proposed a more efficient method which can be used to extend this range somewhat~\cite{Bocharov2012}.

An ancilla-based method known as ``phase kickback'' provides a computationally efficient and cost-competitive alternative for approximating $Z$-axis rotations~\cite{Kita02}. Phase kickback  involves preparing a special ancilla state based on the quantum Fourier transform and then using addition circuits controlled by the single-qubit input to effect the desired rotation. Optimization of the ancilla state preparation yields a cost scaling which is somewhat higher than Fowler's results~\cite{Jones2012,Jones2013},  but can be made more competitive in certain cases~\cite{Jones2013b}.  Phase kickback offers the possibility of very low circuit depth, as low as $O(\log\log 1/\epsilon)$, but requires a relatively large number of ancilla qubits $O(\log 1/\epsilon)$.

Recently, in a series of breakthroughs, efficient algorithms for asymptotically optimal single-qubit decomposition were discovered~\cite{Kliuchnikov2012a,Selinger2012a,Kliuchnikov2012b}.  These algorithms are based on an earlier algorithm for optimally and exactly decomposing a certain class of unitaries into Clifford and $T$ gates~\cite{Kliuchnikov2012}.  The approximation algorithms work by first rounding the unitary $U$ to the closest $\tilde{U}$ that can be exactly decomposed over $\{\Clifford,T\}$ and then using the exact decomposition algorithm on $\tilde{U}$.  Unlike phase kickback, these algorithms do not require ancilla qubits.

Selinger showed that ancilla-free approximation of a single-qubit $R_Z(\theta)$ rotation to within a distance of $\epsilon$ requires  $4\log_2(1/\epsilon) + 11$ $T$ gates in the worst case~\cite{Selinger2012a}.
For many values of $\theta$, however, the number of $T$ gates can be significantly smaller.
Kliuchnikov, Maslov and Mosca (KMM) gave an efficient algorithm which is shown to scale as $3.21\log_2(1/\epsilon) - 6.93$ for the rotation $R_Z(1/10)$~\cite{Kliuchnikov2012b}.

\section{Non-deterministic decomposition methods
\label{sec:repeat.non-deterministic}
}
A few non-deterministic decomposition techniques have also been developed.
So-called ``programmable ancilla rotations'' (PAR) use a cascading set of specially prepared ancilla states along with gate teleportation~\cite{Jones2012}. The action of each gate teleportation depends on a corresponding measurement outcome.  If the outcome is zero, then the protocol stops. Otherwise gate teleportation is repeated with a new, more complicated ancilla state. Like phase kickback, the number of $T$ gates required by PAR is larger than for ancilla-free methods, but the expected number of resources are comparable in some architectures~\cite{Jones2013b}.  Similar use of non-deterministic circuits to produce a ``ladder'' of non-stabilizer states, and in turn approximate an arbitrary unitary, has also been proposed~\cite{Duclos-Cianci2012}.

RUS circuits have already been proposed for decomposition into an alternate logical gate set. Bocharov, Gurevich and Svore (BGS) showed that arbitrary single-qubit unitaries can be approximated using the gate set $\{H,S,V_3\}$ with a typical scaling of $3\log_5(1/\epsilon)$
in the number of $V_3$ gates~\cite{Bocharov2013}.  They suggest a fault-tolerant implementation of the $V_3$ gate using~\figref{fig:v3-toffoli}, which requires eight $T$ gates, four for each Toffoli (see~\cite{Jones2012d}).  Later, Jones improved this circuit, using only a single Toffoli gate~\cite{Jones2013b}.
Through optimized direct search, we found an alternative RUS circuit for $V_3$ that uses only four $T$ gates and has a lower expectation value than the other two circuits, as shown in~\figref{fig:v3-5.26}.  Further discussion of decomposition with $V_3$ is found in~\secref{sec:repeat.applications.v3}.

Our proposed method of single-qubit unitary decomposition based on RUS circuits is also non-deterministic, of course.  In the next section we describe these circuits in detail and in~\secref{sec:repeat.results} we analyze the results of our optimized direct search. Decomposition algorithms are described in Sections~\ref{sec:repeat.applications.v3} and~\ref{sec:repeat.applications.database}.

For convenience, a summary of single-qubit decomposition methods is given in Tables~\ref{tab:non-axial-decomposition-methods} and~\ref{tab:axial-decomposition-methods}

\begin{table}
\centering
\begin{tabular}{>{\raggedright\arraybackslash}m{2cm}|>{\raggedright\arraybackslash}m{4cm}|>{\raggedright\arraybackslash}m{3.5cm}|>{\raggedright\arraybackslash}m{4cm}}
\hline
\textbf{Method}        & \textbf{Description} & \textbf{$T$ count} & \textbf{Comments}\\
\hline
Solovay-Kitaev& Converging $\epsilon$-net \qquad\qquad based on group commutators. & $O(\log^{3.97}1/\epsilon)$& Computationally efficient, but sub-optimal $T$ count.\\
\hline
Ladder states& Hierarchical distillation based $\ket H$ states.& $O(\log^{1.75}1/\epsilon)$& Some of the cost can be shifted ``offline".\\
\hline
Direct search& Optimized exponential-time search.& $2.95\log_2(1/\epsilon)+3.75$&Optimal ancilla-free $T$ count.\\
\hline
BGS & Direct search decomposition with $V_3$. & $T_V(3 \log_5 1/\epsilon)$& $T_V$ is the $T$ count for choice of fault-tolerant implementation of $V_3$.\\
\hline
\textbf{RUS} ~~ (non-axial)&Database lookup.& $2.4\log_2(1/\epsilon) - 3.3$ &Limited approximation accuracy.\\
\hline
\end{tabular}
\caption[Decomposition methods for arbitrary single-qubit unitaries.]{\label{tab:non-axial-decomposition-methods}
Decomposition methods for arbitrary single-qubit unitaries using the gate set $\{H,S,T\}$.
}
\vspace{1cm}
\begin{tabular}{>{\raggedright\arraybackslash}m{2cm}|>{\raggedright\arraybackslash}m{4cm}|>{\raggedright\arraybackslash}m{3.5cm}|>{\raggedright\arraybackslash}m{4cm}}
\hline
\textbf{Method}        & \textbf{Description} & \textbf{$T$ count} & \textbf{Comments}\\
\hline
Phase \quad kickback&Uses Fourier states and phase estimation.&$O(\log 1/\epsilon)$ ~~\qquad (implementation dependent)&$O(\log 1/\epsilon)$ ancillas. Optimizations make it cost competitive with Selinger and KMM.\\
\hline
PAR & Cascading gate teleportation.& $O(\log 1/\epsilon)$&Constant depth (on average), higher $T$ count than phase kickback.\\
\hline
Selinger&Round-off followed by exact decomposition.& $4\log(1/\epsilon)+11$&$T$ count is optimal for worst case rotations.\\ 
\hline
KMM & Round-off followed by exact decomposition.& $3.21\log_2(1/\epsilon)-6.93$& $T$ count based on scaling for $R_Z(1/10)$.\\
\hline
\textbf{RUS} \qquad (axial)& Database lookup.& $1.26\log_2(1/\epsilon)-3.53$&Approximation to within $\epsilon=10^{-6}$.\\
\hline
\end{tabular}
\caption[Decomposition methods for $Z$-axis rotations.]{\label{tab:axial-decomposition-methods}
Decomposition methods for $Z$-axis rotations using the gate set $\{H,S,T\}$.  Approximation of an arbitrary single-qubit unitary is possible by using the relation $U = R_Z(\theta_1)H R_Z(\theta_2) H R_Z(\theta_3)$.
}
\end{table}

\section{Repeat-until-success circuits
\label{sec:repeat.rus}
}

\begin{figure}
\centering
\includegraphics{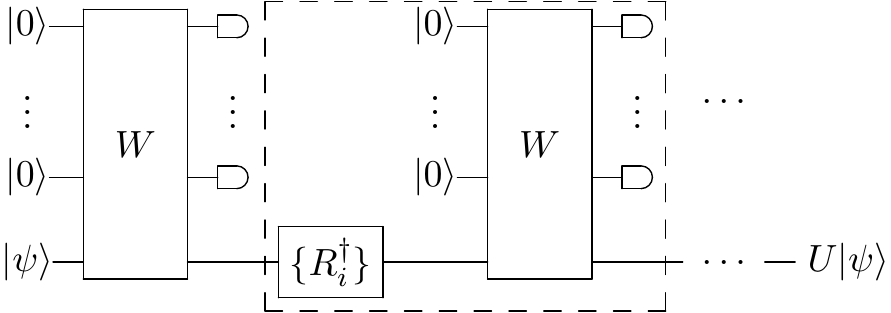}
\caption[General form of an RUS circuit.]{\label{fig:rus-circuit}
A repeat-until-success circuit that implements the unitary $U$.  Ancilla qubits are prepared in $\ket 0$, then the unitary $W$ is performed on both the ancillas and $\ket\psi$.  Upon measuring the ancillas, a unitary operation is effected on $\ket\psi$ which is either $U$ or one of $\{R_i\}$, depending on the measurement outcome.  If the measurement outcome indicates $R_i$, then the recovery operation $R_i^\dagger$ is performed, and the process can be repeated.
}
\end{figure}

The structure of a repeat-until-success (RUS) circuit over a gate set $G$ is as follows.
First, some number $m$ of ancilla qubits are prepared in state $\ket{0^m}$.
Then, given an input state $\ket\psi$ on $n$ qubits, a unitary $W$ is applied to all of the $n + m$ qubits using gates from $G$. Finally, each ancilla qubit is measured in the computational basis.  The output is given by $\Phi_i\ket\psi$, where $\Phi_i$ is a quantum channel---i.e., a unitary plus measurements---on $n$ qubits that depends on the measurement outcome $i \in \{0,1\}^m$.

The measurement outcomes are partitioned into two sets: ``success'' and ``failure''.
Success corresponds to some set of desired operations $\{\Phi_i : i \in \text{success}\}$;
failure corresponds to some set of undesired operations $\{\Phi_i : i \in \text{failure}\}$.
In the case of success, no further action is required.
In the case of failure $i$, a recovery operation $\Phi_i^{-1}$ is applied, and the circuit is repeated.

We restrict to the case in which $\ket\psi$ is a single qubit and the $\{\Phi_i\}$ are unitary.
We also limit to a single ``success'' output $U\ket\psi$, for some unitary $U$, though $U$ may correspond to multiple measurement outcomes.
The operation $W$ is then a $2^{m+1}\times 2^{m+1}$ unitary matrix of the form
\begin{equation}
\label{eq:nondeterministic-unitary-form}
W =
\frac{1}{\sqrt{\sum_i \abs{\alpha_i}^2}}
\begin{pmatrix}
\alpha_0 U & \ldots \\
\alpha_1 R_1 & \ddots \\
\vdots & \\
\alpha_m R_{l} &\\
\end{pmatrix}
\enspace ,
\end{equation}
where $U,R_1,\ldots,R_l$ are $2\times 2$ unitary matrices, and $\alpha_0,\ldots,\alpha_l \in \mathbb{C}$ are scalars.
Since the ancillas are prepared in $\ket{0^m}$, only the first two columns of $W$ are of consequence.
Contents of the remaining columns are essentially unrestricted, except that $W$ must be unitary.
Each of the $l + 1 = 2^{m}$ measurement outcomes corresponds to application of a unitary from $U \cup \{R_i\}$ on the data qubit.  Without loss of generality, we have selected the all zeros outcome to correspond with application of $U$, since outcomes can be freely permuted.
The entire protocol is illustrated in~\figref{fig:rus-circuit}.

For simplicity, we assume that $U \neq R_i ~\forall~ 1\leq i \leq l$.
The case in which $U$ appears multiple times can be easily accommodated.
In order for the circuit to be useful, the remaining matrices $R_1,\ldots,R_l$ should be invertible at a low cost.

In order to be compatible with existing fault-tolerance schemes, we require that $W$ can be synthesized using the gate set $\{\Clifford\} \cup \{T\}$, where $\{\Clifford\}$ denotes the Clifford group generated by $\{H, S, \CNOT\}$.\footnote{Our method is also extensible to other gate sets; however such extensions are not explored here.}
A unitary matrix is exactly implementable by $\{\Clifford,T\}$ if and only if its entries are contained in the ring extension $\KMMring$~\cite{Giles2012}.  Thus, we require that $\alpha_0 U,\alpha_1 R_1, \ldots,\alpha_m R_m$ are matrices over $\KMMring$.
Furthermore, the normalization $1/\sqrt{\sum_i \abs{\alpha_i}^2}$ must also be in the ring.
The unitarity condition on $W$ then requires that
\begin{equation}
\label{eq:HT-unitarity-condition}
\sum_i \abs{\alpha_i}^2 = 2^{k}
\end{equation}
for some integer $k$.

If all of the recovery operations $R_1,\ldots,R_m$ are exactly implementable by $\{\Clifford,T\}$, then we may assume that $\alpha_1,\ldots,\alpha_m\in \KMMring$.
If $\alpha_0$ is an integer, then Lagrange's four-square theorem implies that~\eqnref{eq:HT-unitarity-condition} can be satisfied using at most $n=3$ ancilla qubits.

We pause briefly to note that any element of the ring extension $\KMMring$ can be written as
\begin{equation}
\frac{a + ib + \sqrt{2}(c + id)}{\sqrt{2}^k} \in \KMMring
\enspace,
\end{equation}
for integers $a,b,c,d,k$.  Below we will eliminate the denominator in which case we may write
\begin{equation}
a + ib + \sqrt{2}(c + id) \in \Aring
\enspace .
\end{equation}

\subsection{Characterization of repeat-until-success unitaries
\label{sec:repeat.rus.characterization}
}
Consider a $2\times 2$ unitary matrix $U$ such that
\begin{equation}
\label{eq:rus-unitary}
U = \begin{pmatrix} u_{00} & u_{01} \\ u_{10} & u_{11}\end{pmatrix}
  = \frac{1}{\sqrt{2^k}\alpha}\begin{pmatrix} \beta_{00} & \beta_{01} \\ \beta_{10} & \beta_{11}\end{pmatrix}
  \enspace ,
\end{equation}
for $\alpha \in \mathbb{R}$, $\beta_{00},\ldots,\beta_{11} \in \Aring$ and integer $k \geq 0$.
We are concerned with exactly implementing $U$ only up to a global unit phase $e^{i \phi}$ for some $\phi \in [0,2\pi)$. Accordingly, we may assume without loss of generality that $\alpha$ is real and non-negative since for any $\beta \in \mathbb{C}$, $\frac{\beta \beta^*}{\abs{\beta}} \geq 0$.
The restriction to $\Aring$ rather than $\KMMring$ is also without loss of generality, since $k$ can be chosen to eliminate any denominators.
Then choosing $\alpha_0=\sqrt{2^k}\alpha$ we have
\begin{equation}
\label{eq:rus-scalar-condition}
\alpha_0 = \sqrt{\abs{\beta_{00}}^2 + \abs{\beta_{10}}^2} = \sqrt{x + y\sqrt 2}
\enspace ,
\end{equation}
where $x = a_{00}^2 + c_{00}^2 + a_{10}^2 + c_{10}^2 + 2(b_{00}^2 + d_{00}^2 + b_{10}^2 + d_{10}^2)$, $y = a_{00}b_{00} + c_{00}d_{00} + a_{10}b_{10} + c_{10}d_{10}$ for integers $a_{00}$, $b_{00}$, $c_{00}$, $d_{00}$, $a_{10}$, $b_{10}$, $c_{10}$, $d_{10}$.

Any target unitary $U$ must have this form due to~\eqnref{eq:nondeterministic-unitary-form}.
In other words, the \emph{only} unitaries that can be obtained by $\{\Clifford,T\}$ circuits of the form ~\figref{fig:rus-circuit} are those that can be expressed by entries in $\Aring$ after multiplying by a scalar.
Nonetheless, this restricted class of unitaries can be used to approximate arbitrary unitaries more efficiently than unitaries limited to $\KMMring$, as we show in~\secref{sec:repeat.results} and~\secref{sec:repeat.applications.database}.

In addition to their use in~\cite{Bocharov2013}, repeat-until-success circuits have been considered by Wiebe and Kliuchnikov for small-angle $Z$-axis rotations~\cite{Wiebe2013}.
Whereas Wiebe and Kliuchnikov propose hierarchical RUS circuits over $\{\Clifford,T\}$, we do not a priori restrict to a hierarchical structure or to small $Z$-axis rotations.
RUS circuits have been studied to a limited extent in other contexts, as well.
For example, repeated gate operations have been proposed for use in linear optics to implement a CZ gate~\cite{Lim2004}.
More recently, \cite{Shah2013} adapted deterministic ancilla-driven methods~\cite{Anders2009,Kashefi2009} to allow for non-determinism.

\subsection{Success probability and expected cost
\label{sec:repeat.rus.cost}
}
The success probability, i.e., the probability of obtaining the zero outcome for all ancilla measurements, can be computed from~\eqnref{eq:HT-unitarity-condition} and is given by
\begin{equation}
\Pr[\text{success}] = \frac{\alpha_0^2}{2^k} \leq \frac{\alpha_0^2}{2^{\lceil 2\log_2 \alpha_0 \rceil}}
\enspace ,
\end{equation}
where since $\alpha_0^2 < 2^k$, we may use $k \geq \lceil 2\log_2 \alpha_0 \rceil$.
The circuits in~\figref{fig:v3-circuits}, for example, each yield a value of $\alpha_0=\sqrt{5}$ and therefore a success probability of $5/8$.
On the other hand, if $U$ appears multiple times in~\eqnref{eq:nondeterministic-unitary-form}, then we have
\begin{equation}
\Pr[\text{success}] = \frac{m \alpha_0^2}{2^k} \leq \frac{m \alpha_0^2}{2^{\lceil \log_2 m \alpha_0^2 \rceil}}
\enspace ,
\end{equation}
where $m$ is the number of times that $U$ appears.
This upper bound can be made arbitrarily close to one for large enough $m$.

The expected number of repetitions required in order to achieve success is given by a geometric distribution with expectation value $1/p$, and variance $(1-p)/p^2$, where $p = \Pr[\text{success}]$.
If $C(W)$ is the cost of implementing the unitary $W$, then the expected cost of the RUS circuit is given by $C(W)/p$ with a variance of $C(W)(1-p)/p^2$.  Since the resources required to implement a $\{\Clifford,T\}$ fault-tolerant circuit are often dominated by the cost of implementing the $T$ gate,  we will define $C(W)$ as the number of $T$ gates in the circuit used to implement $W$.

% The $T$-gate count is not the only reasonable cost function.  Other possibilities include circuit size, width, area or volume, or the total number of measurements.  The utility of a particular cost function varies depending on the target quantum computing architecture.  For architectures that use the surface code, for example, total volume can be a more complete metric than $T$ count~\cite{Fowler2013,Jones2013b}.

We choose to use $T$-gate count as the cost function because it is simple, and is consistent with other $\{\Clifford, T\}$ decomposition algorithms~\cite{Kliuchnikov2012,Amy2012,Selinger2012a,Kliuchnikov2012b,Wiebe2013,Gosset2013a}.
However, RUS circuits employ techniques that are not present in the circuits produced by previous decomposition methods.
In particular, rapid classical feedback and control is required.  Moreover, variable time scales for logical single-qubit gates imply the need for active synchronization.  Thus, while $T$ count allows for direct comparison of RUS circuits with other methods, a more complete metric may be required for resource calculations on a particular architecture.

\subsection{Amplifying the success probability
\label{sec:repeat.rus.amplify}
}

We may describe the action of the multi-qubit unitary $W$ by
\begin{equation}
\label{eq:rus-unitary-on-psi}
W\ket{0^m}\ket\psi = \sqrt{p}\ket{0^m}U\ket\psi + \sqrt{1-p}\ket{\Phi^\perp}
\enspace ,
\end{equation}
where $\ket{\Phi^\perp}$ is a state that depends on $\ket\psi$ and satisfies $(\ket{0^m}\bra{0^m}\otimes\Id)\ket{\Phi^\perp} = 0$.  That is, $W$ outputs a state which has amplitude $\sqrt{p}$ on the ``success'' subspace, and amplitude $\sqrt{1-p}$ on the ``failure'' subspace.  We show below that in some cases we may apply the amplitude amplification algorithm to boost the success probability and reduce the expected $T$ count of an RUS circuit.

Traditional amplitude amplification~\cite{Brassard2000} proceeds by applying the operator $(RS)^j$ on the initial state $W\ket{0^m}\ket\psi$ for some integer $j > 0$ and reflections
\begin{equation}\begin{split}
S &= \Id - 2\ket{0^m}\ket\psi \bra{0^m}\bra{\psi}, \\
R &= WSW^\dagger = \Id - 2W\ket{0^m}\ket\psi \bra{0^m}\bra{\psi}W^\dagger
\enspace .
\end{split}\end{equation}
In the two-dimensional subspace spanned by $\{\ket{0^m}U\ket\psi, \ket{\Phi^\perp}\}$, $RS$ acts a rotation by $2\theta$ where $\sin(\theta) = \sqrt{p}$.  Therefore $(RS)^j (W\ket{0^m}\ket\psi) = \sin((2j+1)\theta)\ket{0^m}U\ket\psi + \cos((2j+1)\theta)\ket{\Phi^\perp}$.  The goal then, is to choose $j$ appropriately so as to minimize the expected number of $T$ gates.

The problem in this case is that $\ket\psi$ is unknown, and therefore we cannot directly implement $S$.
We can, however, implement $S' = \text{CZ}(m)\otimes I$, the generalized controlled-$Z$ gate on $m$ qubits defined by $\text{CZ}(m) \ket{x_1,x_2,\ldots,x_m} = (-1)^{x_1 x_2\ldots x_m} \ket{x_1,x_2,\ldots,x_m}$.
We could, therefore, apply $(WS'W^\dagger S')^j$ instead of $(RS)^j$.

\begin{proposition}
\label{prop:amp-amp-cz}
Consider a unitary $W$ that satisfies~\eqnref{eq:rus-unitary-on-psi}.
Amplitude amplification on $\ket{0^m}U\ket\psi$ can be performed using the operator $WS'W^\dagger S'$, where $S' = \text{CZ}(m)\otimes I$.
More precisely,
\begin{equation}
(WS'W^\dagger S')^j (W\ket{0^m}\ket\psi) = \sin((2j+1)\theta)\ket{0^m}U\ket\psi + \cos((2j+1)\theta)\ket{\Phi^\perp}
\enspace ,
\end{equation}
where $\sin(\theta) = \sqrt{p}$.
\end{proposition}

Proof of this claim relies on the $2D$ Subspace Lemma of Childs and Kothari.
\begin{lemma}[\cite{Kothari2013}]
\label{lem:2d-subspace}
Let $W$ be a unitary that satisfies~\eqnref{eq:rus-unitary-on-psi}. Then the state
$$\ket{\Psi^\perp} := W^\dagger (\sqrt{1-p}\ket{0^m}U\ket\psi - \sqrt{p}\ket{\Phi^\perp})$$
 satisfies $(\ket{0^m}\bra{0^m}\otimes\Id)\ket{\Psi^\perp} = 0$.
\end{lemma}

\begin{proof}[Proof of~\propref{prop:amp-amp-cz}]
First, note that both $R$ and $S$ preserve the two-dimensional subspace spanned by $\ket{0^m}U\ket\psi$ and $\ket{\Phi^\perp}$.  That is, the state that results from applying any sequence of $R$ and $S$ on $W\ket{0^m}\ket\psi$ can be written as a linear combination of $\ket{0^m}\ket\psi$ and $\ket{\Phi^\perp}$.  Next, observe that $S'$ also preserves this subspace and is equivalent to $S$ since $S'\ket{0^m}U\ket\psi = -\ket{0^m}U\ket\psi$ and $S'\ket{\Phi^\perp} = \ket{\Phi^\perp}$.

The claim then is that the reflection $WS'W^\dagger$ about the state $W\ket{0^m}\ket\psi$ also preserves the subspace and is equivalent to $R$.  Clearly, $(WS'W^\dagger)W\ket{0^m}\ket\psi = -W\ket{0^m}\ket\psi$. On the other hand, the action of $WS'W^\dagger$ on the state $\ket{\Psi^\perp}$ that is orthogonal to $W\ket{0^m}\ket\psi$ (in the $2D$ subspace) is less obvious and requires \lemref{lem:2d-subspace}, which implies that $(WS'W^\dagger)W\ket{\Psi^\perp} = W\ket{\Psi^\perp}$ as desired.
We therefore conclude that $(WS'W^\dagger S')^j$ is equivalent to ``real'' amplitude amplification on $W\ket{0^m}\ket\psi$ and, in particular, that
\begin{equation*}
(WS'W^\dagger S')^j W\ket{0^m}\ket\psi = \sin((2j+1)\theta)\ket{0^m}U\ket\psi + \cos((2j+1)\theta)\ket{\Phi^\perp}
\enspace .
\end{equation*}
\end{proof}

If $m \leq 2$, then $S'$ can be implemented with only Clifford gates, i.e., $Z$ or $\text{CZ}$.
Then, for a fixed value of $j$, the total number of $T$ gates in the corresponding amplified circuit is given by $(2j+1)T_0$, where $T_0$ is the number of $T$ gates in the unamplified circuit.
In order for amplitude amplification to yield an improvement in the expected number of $T$ gates, we therefore require that
\begin{equation}
(2j+1)\sin^2(\theta) < \sin^2((2j+1)\theta)
\enspace ,
\end{equation}
a condition that holds if and only if $0 \leq p < 1/3$. Thus a sensible course of action is to apply amplitude amplification for all RUS circuits for which $p < 1/3$, and leave higher probability circuits unchanged.

Consider, for example, an RUS circuit that contains $15$ $T$ gates and has a success probability of $0.1$.  In this case, using amplitude amplification with value of $j=1$ yields a new circuit with success probability $0.676$ and $45$ $T$ gates, an improvement in the expected number of $T$ gates by a factor of $2.25$.
The effects of amplitude amplification on our database of RUS circuits are discussed in~\secref{sec:repeat.results}.

Cost analysis of amplitude amplification for circuits with more than two ancilla qubits is more complicated because the reflection operator $S'=\text{CZ}(m)$ is not a Clifford gate. For three ancilla qubits, for example, $S'$ is the controlled-controlled-$Z$ gate, which can be implemented with $4$ $T$ gates~\cite{Jones2012d}.  Larger versions of $\text{CZ}(m)$ could be synthesized directly~\cite{Kliuchnikov2013a,Welch2013}, or by using a recursive procedure~\cite{Nielsen2000}.  The circuits presented in~\secref{sec:repeat.results} use at most two ancilla qubits, however, so more complicated amplification circuits are not an issue in our analysis.

\section{Direct search methods}
\label{sec:repeat.search}
Equations~\eqnref{eq:nondeterministic-unitary-form} and~\eqnref{eq:rus-scalar-condition} restrict the kinds of unitaries that can be obtained from RUS circuits.
However, these conditions say little about how to implement the unitary $W$.
Given $W$ explicitly, it is possible to synthesize a corresponding $\{\Clifford,T\}$ circuit with a minimum number of $T$ gates~\cite{Gosset2013a}, at least for small $W$.
However, given a unitary $U$ of the form~\eqnref{eq:rus-unitary}, there are potentially many choices of $W$. The minimum number of $T$ gates required is therefore unclear and is a direction for future research.

In order to better understand the scope and power of RUS circuits, we design an optimized direct search algorithm that checks for RUS circuits up to a given $T$-gate count. Our direct search algorithm is as follows:
\begin{enumerate}
  \item Select the number of ancilla qubits and the number of gates.
  \item Construct a $\{\Clifford,T\}$ circuit and compute the resulting unitary matrix $W$.
  \item Partition the first two columns of $W$ into $2\times 2$ matrices.
  \item Identify and remove matrices that are proportional to Clifford gates.
  \item If the remaining matrices are all proportional to the same unitary matrix, then keep the corresponding circuit.
\end{enumerate}

We restrict the recovery operations ${R_i}$ of the circuits found by our search to the set of single-qubit Cliffords.
This choice is motivated by our use of the $T$ count as a cost function; Clifford gates, and therefore the recovery operations are assigned a cost of zero.

In order to identify relevant search parameters, we initially performed a random search over a wide range of circuit widths (number of qubits) and sizes (number of gates).  Our search was most successful with small numbers of ancilla qubits, large numbers of $T$ gates, and just one or two entangling gates.  We therefore focused on circuits of the form shown in~\figref{fig:two-cz-canonical}.
These circuits contain just a single ancilla qubit and two CZ gates, interleaved with single-qubit gates.

Naively, the number of circuits of the form~\figref{fig:two-cz-canonical} is $O(3^n)$, where $n$ is the maximum number of (non-CZ) gates in the circuit, and the base of three is the size of the set $\{H,S,T\}$.
In order to reduce the complexity of our search, we constructed each of the single-qubit gate sequences using the canonical form proposed by~\cite{Bocharov2012}. A canonical form sequence is the product of three $2\times2$ unitary matrices $g_2Cg_1$ where $g_1,g_2$ belong to the Clifford group, and $C$ is the product of some number of ``syllables'' $TH$ and $SHTH$.  The canonical form yields a unique representation of all single-qubit circuits over $\{H,T\}$; there are $2^{t-3}+4$ canonical circuits of $T$-count at most $t$.  This yields more than a quadratic improvement compared to the naive search, since the number of $T$ gates is roughly one-half the total number of gates.

In general, the canonical form requires conjugation by the full single-qubit Clifford group, which contains $24$ elements.  Given a product of syllables $C$, each of the $24^2=576$ circuits $g_2Cg_1$ are unique.  However, when multiple canonical form circuits are placed in a larger circuit, as in~\figref{fig:two-cz-canonical}, some combinations of Clifford gates can be eliminated. For example, in $g_2Cg_1\ket{0}$, $g_1$ need only be an element of $\{I,X,SH,SHX,HSH,HSHX\}$ since diagonal gates act trivially on $\ket{0}$.  Similar simplifications for~\figref{fig:two-cz-canonical} are shown in~\figref{fig:canonical-simplifications}.  In total, these Clifford simplifications reduce the search space by a factor of more than $10^5$.

\begin{figure}
\centering
\includegraphics{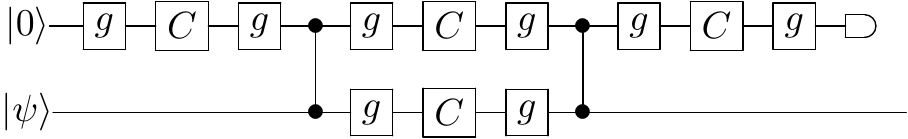}
\caption[General form of circuits in the RUS circuit database.]{\label{fig:two-cz-canonical}
The above circuit illustrates the general form of most of the circuits in our database.  Each of the gates labeled $g$ represents an element of the single-qubit Clifford group.  Each of the gates labeled $C$ represents a single-qubit canonical circuit as defined in~\cite{Bocharov2012}.
}
\vspace{.8cm}
\begin{subfigure}[b]{.49\textwidth}
\centering
\includegraphics[scale=1]{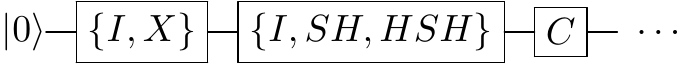}
\caption{}
\end{subfigure}
\begin{subfigure}[b]{.49\textwidth}
\centering
\includegraphics[scale=1]{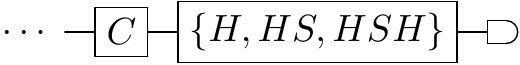}
\caption{}
\end{subfigure}
\begin{subfigure}[b]{\textwidth}
\hfill
\end{subfigure}
\begin{subfigure}[b]{.49\textwidth}
\centering
\includegraphics[scale=1]{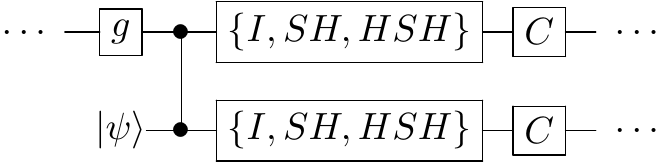}
\caption{}
\end{subfigure}
\begin{subfigure}[b]{.49\textwidth}
\centering
\includegraphics[scale=1]{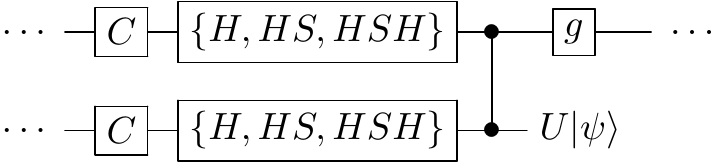}
\caption{}
\end{subfigure}
\caption[Clifford simplifications to~\figref{fig:two-cz-canonical}.]{\label{fig:canonical-simplifications}
Some of the $g$ gates in~\figref{fig:two-cz-canonical} can be restricted to a subset of the single-qubit Clifford group. (a) Circuits that begin with diagonal gates can be eliminated since they add a trivial phase to $\ket 0$. (b) Similarly, diagonal gates have no impact on the $Z$-basis measurement. (c) Pauli gates and $S$ gates can be commuted through the CZ and absorbed into either $\ket\psi$ or the preceding $g$ gate. (d) Analogously, Pauli and $S$ gates occurring before the CZ can be absorbed by the trailing $g$ gate or by the output.
}
\end{figure}

Despite these simplifications, the search time is still exponential in the number of $T$ gates.  To save time, we partitioned the search into thousands of small pieces running in parallel on a large cluster and collected the results in a central database.  We were able to exhaustively search circuits of the form of~\figref{fig:two-cz-canonical} up to a total (raw) $T$ count of $15$. The search took roughly one week running on hundreds of cores.  The results of this search are presented in the next section.

\section{Search results
\label{sec:repeat.results}
}

Our search yielded many circuits that implement the same unitary $U$, but with different $T$-gate counts and success probabilities. To eliminate redundancy we maintained, for a given $U$, a database containing only the circuit with the minimum expected $T$ count. The result is a database containing $2194$ RUS circuits.  Upon success, each circuit exactly implements a unique non-Clifford single-qubit unitary $U$, and otherwise implements a single-qubit Clifford operation.
Database statistics are shown in~\figref{fig:search-results}. For circuits with success probability less than $1/3$, we used amplitude amplification to improve performance (see~\secref{sec:repeat.rus.amplify}). \figref{fig:histogram-expected-tcount} illustrates the impact of amplitude amplification on the expected $T$ count.  Amplification improved the performance of circuits with relatively high expected $T$ count, but did not improve circuits with expected $T$ count of $30$ or less. Note that the database also includes some circuits that were found by preliminary searches not of the form of~\figref{fig:two-cz-canonical}.

\begin{figure}
\centering
\begin{subfigure}[b]{.49\textwidth}
\centering
\includegraphics[height=7.68cm]{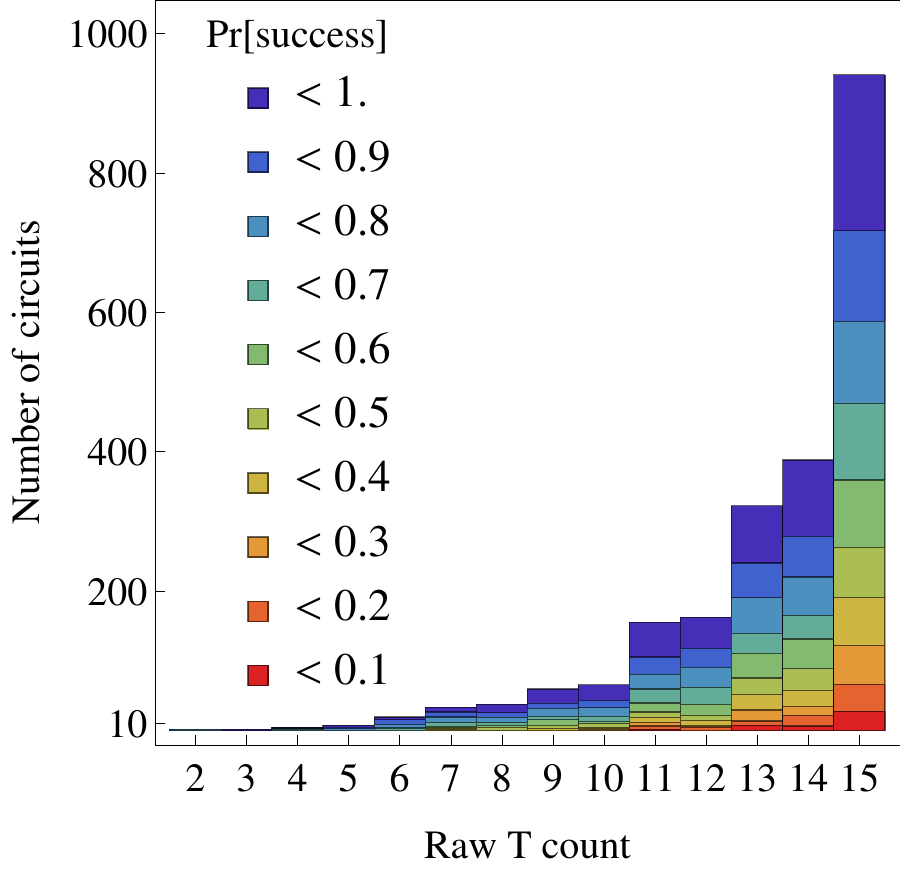}
\caption{\label{fig:histogram-raw-tcount}
}
\end{subfigure}
\hfill
\begin{subfigure}[b]{.49\textwidth}
\centering
\includegraphics[height=7.68cm]{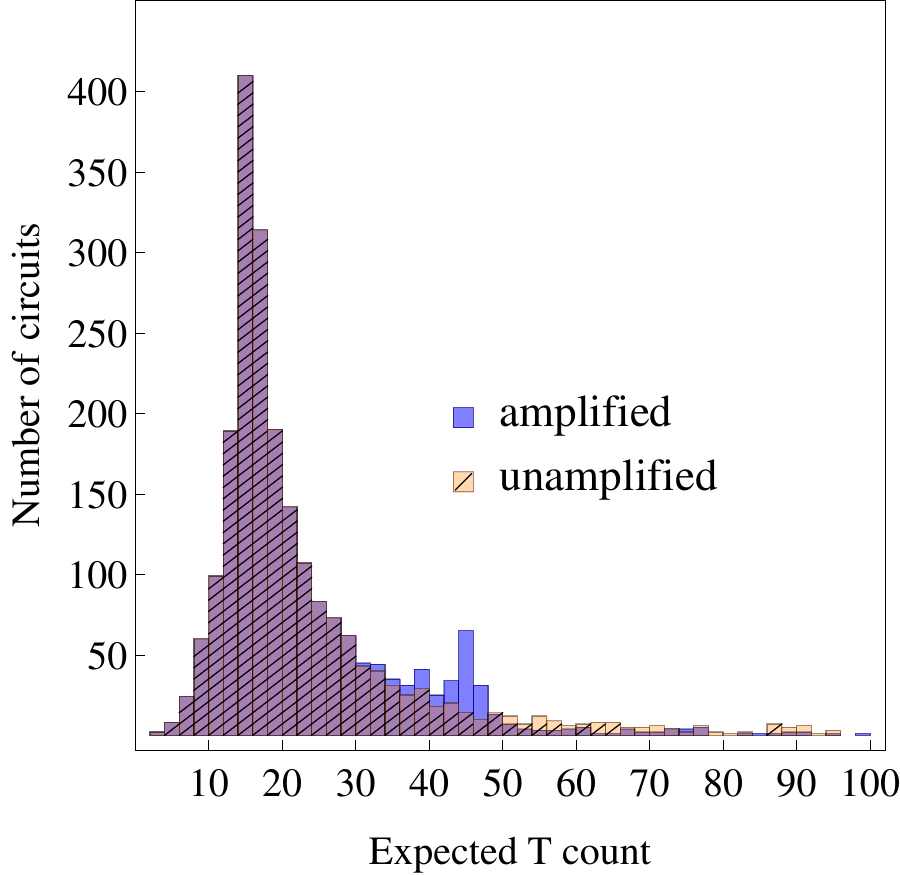}
\caption{\label{fig:histogram-expected-tcount}
}
\end{subfigure}
\caption[Statistics for the database of RUS circuits.]{\label{fig:search-results}
Statistics for the database of repeat-until-success circuits, including all circuits of the form of~\figref{fig:two-cz-canonical} up to a $T$ count of $15$.  (a) The total number of circuits grouped by (raw) $T$ gate count and success probability. (b) The total number of circuits grouped by expected $T$ count, both before amplitude amplification and after amplitude amplification.  The two histograms (before amplification and after amplification) are overlayed, where the darker hatched bars indicate circuits that are unaffected by amplification.  Only circuits with an expected $T$ count of at most $100$ are shown.
}
\end{figure}

The database contains $1659$ axial rotations, i.e., unitaries which, modulo conjugation by Cliffords, are rotations about the $Z$-axis of the Bloch sphere, and $535$ non-axial rotations.  The number of axial rotations is noteworthy since, modulo Clifford conjugation, only one non-trivial single-qubit rotation can be exactly synthesized with $\{\Clifford,T\}$ and without measurement, namely $T$~\cite{Kliuchnikov2012}.  Our results show that \emph{many} axial rotations can be implemented exactly (conditioned on success) when measurement is allowed.

At the same time, the non-axial rotations in our database offer an expected $T$ count that is dramatically better than the $T$ count obtained by approximation algorithms~\cite{Selinger2012a,Kliuchnikov2012b}. For each circuit in the database we computed the number of $T$ gates required to approximate the corresponding unitary to within a distance of $10^{-6}$ using the algorithm of KMM.  \figref{fig:kmm-ratios} shows the ratio of the $T$ count given by KMM vs. the expected $T$ count for the RUS circuit.  Our results show a typical improvement of about a factor of three for axial rotations and a typical improvement of about a factor of about $12$ for non-axial rotations.
The larger improvement for non-axial rotations is expected since the KMM algorithm requires the unitary to be first decomposed into a sequence of three axial rotations.

\begin{figure}
\centering
\includegraphics[scale=1]{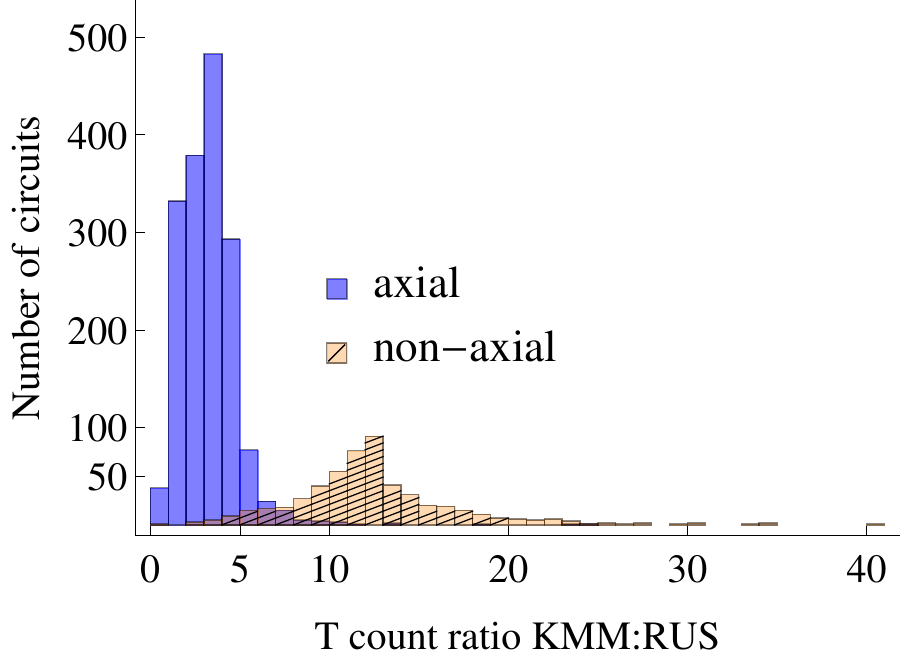}
\caption[Axial and non-axial RUS circuit costs compared to KMM.]{\label{fig:kmm-ratios}
Contents of the RUS database, split into axial and non-axial single-qubit rotations.  For each circuit in the database the number of $T$ gates required to approximate the corresponding ``success'' unitary $U$ to within $10^{-6}$ was calculated using the algorithm of~\cite{Kliuchnikov2012b}.  The $x$-axis represents the ratio of the KMM $T$ count vs. the expected number of $T$ gates for the RUS circuit.
}
\end{figure}

As an example, the RUS circuit shown in~\figref{fig:large-kmm-ratio-circuit} implements the non-axial single-qubit rotation $U = (2X + \sqrt{2}Y + Z)/\sqrt{7}$ with four $T$ gates and a probability of success of $7/8$.  By contrast, approximating $U$ to within $\epsilon = 10^{-6}$ using the KMM algorithm requires a total of $182$ $T$ gates.  Thus~\figref{fig:large-kmm-ratio-circuit} not only implements the intended unitary exactly, but does so at a cost better than $40$ times less than the best approximation methods.

\begin{figure}
\centering
\includegraphics{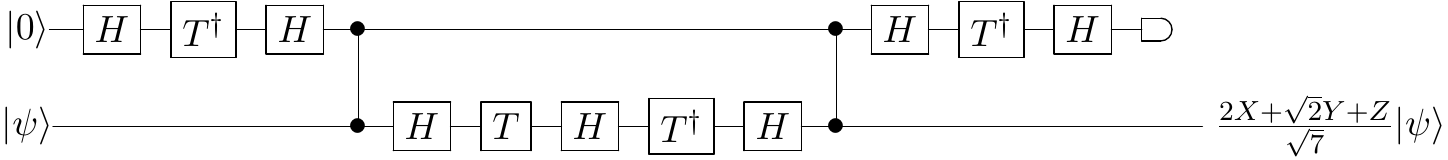}
\caption[Cheap RUS circuit with high KMM $T$ count.]{\label{fig:large-kmm-ratio-circuit}
This RUS circuit implements the unitary $U = (2X + \sqrt{2}Y + Z)/\sqrt{7}$ with probability $7/8$, and otherwise implements $Z$.  Approximation of $U$ without ancillas requires $182$ $T$ gates (roughly $40$ times more) for $\epsilon = 10^{-6}$.
}
\end{figure}

Our database is too large to offer an analysis of each circuit in detail.  Instead, we present some additional noteworthy examples.  The smallest circuit in our database contains two $T$ gates and is shown in~\figref{fig:gosset-unitary}. Upon measuring zero, which occurs with probability $3/4$, the circuit implements $(\Id + \sqrt{2}X)/\sqrt{3}$ and upon measuring one implements $\Id$.
This circuit is notable in that its existence was predicted by Gosset and Nagaj in~\cite{Gosset2013}. They required a $\{\Clifford,T\}$ circuit that exactly implemented $R=(\sqrt{2}\Id-iY)/\sqrt{3}$ with a constant probability of success.  The unitary implemented by~\figref{fig:gosset-unitary} is equivalent to $R$ up to conjugation by Clifford gates.
% The conjugation relation is: R = SX (\sqrt{2}X - \Id) S^\dagger /\sqrt{3}.

\begin{figure}
\centering
\includegraphics{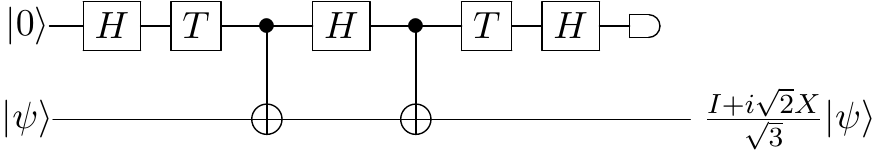}
\caption[The smallest RUS circuit in our database.]{\label{fig:gosset-unitary}
The above circuit is the smallest in our database.  Upon measuring zero, which occurs with probability $3/4$, it implements $(\Id + i\sqrt{2}X)/\sqrt{3}$ on the input state $\ket\psi$.  Upon measuring one, it implements the identity.
}
\end{figure}

As discussed in~\secref{sec:repeat.non-deterministic}, our database contains a circuit that implements $V_3$.  In addition to the circuit shown in~\figref{fig:v3-5.26}, our search also found a circuit that implements $V_3$ with the same number of $T$ gates (four), but just a single ancilla qubit, as shown in~\figref{fig:v3-one-ancilla}.  The expected $T$ count of the single-ancilla circuit is worse than that of~\figref{fig:v3-5.26}, though, since all four of the $T$ gates on the ancilla must be performed ``online''.

\begin{figure}
\centering
\includegraphics{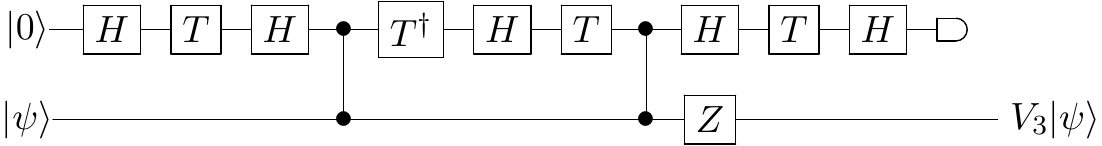}
\caption[A two-qubit RUS circuit for $V_3$.]{Like the circuits in \figref{fig:v3-circuits}, the above circuit implements $V_3$ with probability $5/6$ and identity with probability $1/6$, but with only one ancilla qubit and one measurement.
\label{fig:v3-one-ancilla}
}
\end{figure}

The $V_3$ gate is one of a family of $V$-basis gates for which the normalization factor is $1/\sqrt{5}$.
In addition to single-qubit unitary decomposition based on $V_3$, \cite{Bocharov2013} also offers the possibility of decomposing single-qubit unitaries using $V$-basis gates with normalization factors $1/\sqrt{p}$ where $p$ is a prime.  These ``higher-order'' $V$ gates cover $SU(2)$ more rapidly than $V_3$ and therefore offer potentially more efficient decomposition algorithms.  A number of such $V$-basis gates can be found in our database, including axial versions for $p\in \{13,17,29\}$, as shown in~\figref{fig:high-order-v-circuits}.  The prospect of decomposition algorithms with these circuits is discussed in~\secref{sec:repeat.applications.v3}.

\begin{figure}
\begin{subfigure}[b]{\textwidth}
\includegraphics[width=\textwidth]{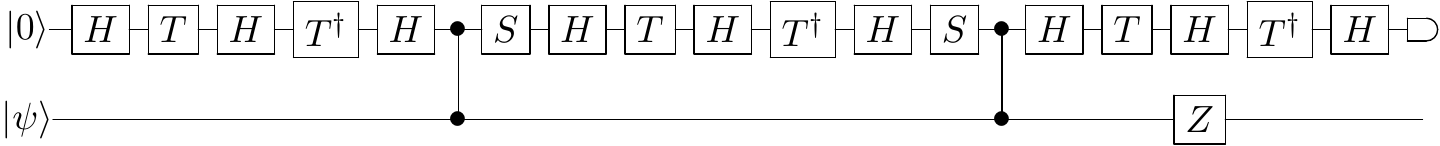}
\caption{$(2Z+3iI)/\sqrt{13}$, $\Pr=13/16$}
\vspace{.2cm}
\end{subfigure}
\begin{subfigure}[b]{\textwidth}
\includegraphics[width=\textwidth]{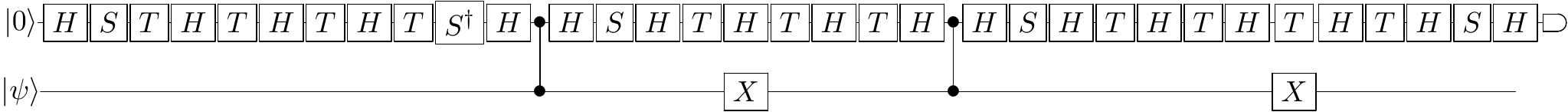}
\caption{$(4I+iZ)/\sqrt{17}$, $\Pr \approx 0.985$}
\vspace{.2cm}
\end{subfigure}
\begin{subfigure}[b]{\textwidth}
\includegraphics[width=\textwidth]{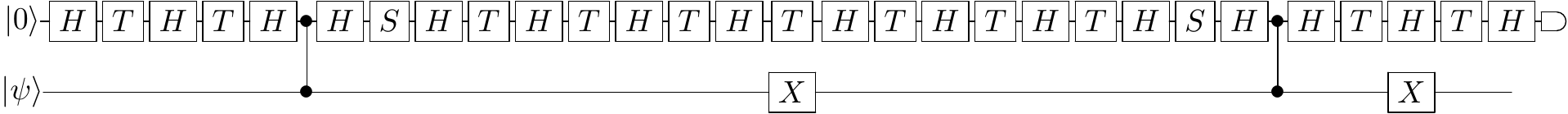}
\caption{$(5I+2iZ)/\sqrt{29}$, $\Pr \approx 0.774$}
\end{subfigure}
\caption[High-order $V$-basis RUS circuits.]{\label{fig:high-order-v-circuits}
RUS circuits for $V$-basis gates with prime normalization factors (a) $p = 13$ (b) $p=17$ and (c) $p = 29$.  The values under each circuit indicate the unitary effected upon success and the success probability, respectively.  Each circuit implements the identity upon failure.
}
\end{figure}

% \section{Applications
% \label{sec:repeat.applications}
% }
% 
% With an extensive database of RUS circuits in hand, we analyze how RUS circuits can be used to improve resource costs of fault-tolerant quantum computation.  The Clifford group with the addition of any one non-Clifford gate is universal for quantum computation. Thus any or all of our RUS circuits can be used to construct a universal gate set.  The question, though, is whether or not RUS circuits can be used to improve resource costs relative to existing unitary approximation methods.
% 
% In this section we show that RUS circuits yield significant improvements for unitary decomposition.  First we discuss the use of our improved $V_3$ circuit for decomposition into $\{\Clifford,V_3\}$.  Then we show how to compose RUS circuits in series in order to expand the size and density of the database.  The expanded database can be used to approximate single-qubit unitaries up to an accuracy that is sufficient for a number of important quantum algorithms.  In particular, in~\secref{sec:repeat.applications.phase-estimation}, we show how to use circuits in our database for applications using the quantum phase estimation algorithm.

\section{Decomposition with \texorpdfstring{$V_3$}{V3}
\label{sec:repeat.applications.v3}
}
Our circuit for $V_3$ in~\figref{fig:v3-5.26} can be used directly in the decomposition algorithm of~\cite{Bocharov2013}.  The BGS direct search algorithm can produce an $\epsilon$-approximation of a randomly chosen single-qubit unitary with a number of $V_3$ gates given by $3\log_5(1/\epsilon)$ in most cases.  Multiplying by an expected $T$-cost of $5.26$ using~\figref{fig:v3-5.26} yields an algorithm with an expected $T$ count of
\begin{equation}
\label{eq:bgs-tcount}
15.78\log_5(1/\epsilon)
\enspace .
\end{equation}
This is an improvement over the estimated $T$ count of $3(3.21\log_2(3/\epsilon)-6.93)$ due to \cite{Kliuchnikov2012b} for all $\epsilon < 0.25$.  This scaling is worse than Fowler's optimal exponential-time search by roughly a factor of two.  However, the exponential nature of Fowler's method means that it can provide approximations in reasonable time only up to roughly $\epsilon=10^{-4}$.  The BGS direct search can provide approximations to within $\epsilon=10^{-10}$.  Thus $V_3$ decomposition appears to be the best option when relatively high precision is required.

% Fault-tolerant implementations provide logical gates that are accurate up to an arbitrarily small error probability $p$.  The total error probability of a fault-tolerant quantum circuit can be calculated using a union bound.  Thus a circuit with $n$ gates will require a fault-tolerant implementation of each gate that has error probability roughly $1/n$.  Using our implementation of $V_3$, therefore, reduces the number of $T$ gates required to approximate a single qubit unitary \emph{and} reduces the necessary resource overhead for each $T$ gate (in expectation).

% Fault-tolerance resource costs of can be estimated by the number of ``physical'' $T$ gates. \figref{fig:phys-tcount-rus-kmm} shows the physical $T$ gate count of $V$-basis and KMM decomposition (assuming a scaling of $9.63\log_2(1/\epsilon)-20.79$).  Both decomposition methods are competitive, but the advantage of the $V$-basis decomposition increases with the rotation accuracy. \todo{Maybe try some alternate values of $\epsilon/p$.}

% \begin{figure}
% \centering
% \includegraphics[width=.5\textwidth]{images/phys-tcount-rus-kmm}
% \caption{\label{fig:phys-tcount-rus-kmm}
% }
% \end{figure}

The database also contains some $V$-basis gates with prime normalization factors larger than $5$. In~\cite{Bocharov2013}, the authors conjecture that the decomposition algorithm for $p=5$ extends to other primes with a $T$-count scaling of
$4\log_p(1/\epsilon)$.
However, whereas $p=5$ requires only the single $V_3$ gate, higher prime values require implementation of multiple $V$ gates. For simplicity, assume that each of the required $V$ gates can be implemented with some number of $T$ gates $T_p$.  Then the decomposition yielded for prime $p$ will be better than that obtained with $V_3$ if
\begin{equation}
\label{eq:high-order-V-ratio}
1 < \frac{5.26}{T_p} \log_5(p)
\enspace .
\end{equation}

Unfortunately, our database contains only a single $V$-basis gate for each of $p = \{13,17,29\}$.  Still we  calculate~\eqnref{eq:high-order-V-ratio} under the optimistic assumption that the remaining $V$ gates can somehow be implemented at the same cost.  Using the circuits in~\figref{fig:high-order-v-circuits} we obtain
\begin{subequations}
\begin{align}
&5.26/7.38 \log_5{13} \approx 1.13,\\
&5.26/11.17 \log_5{17} \approx 0.83,\\
&5.26/14.22 \log_5{17} \approx 0.77
\enspace .
\end{align}
\end{subequations}
Based on these calculations we conclude that, while improved decomposition may be possible using $p=13$, higher values of $p$ are unlikely to yield cost benefits on their own.

On the other hand, given implementations of multiple $V$ gates, there is no reason to limit to a single value of $p$.  One could imagine an algorithm that combined multiple classes of $V$ gates, using largely $V_3$ and using more expensive high-order $V$ gates selectively.  We do not consider such an algorithm directly.  In the next section, however, we study the effect of optimally combining all of the RUS circuits in our database, not just $V$ gates.

\section{Decomposition with the circuit database
\label{sec:repeat.applications.database}
}
% The $V_3$ circuit offers an algorithm for improved decomposition of single-qubit unitaries by alternating $V_3$ with Clifford gates.  A more dramatic improvement can be obtained by exploiting a larger number of circuits from our database.

It is possible to approximate to any desired accuracy, an arbitrary single-qubit unitary using just Clifford gates and the circuits in our database.  But actually finding the optimal sequence among all possible combinations of circuits is a challenging task.  Ideally, we could construct an efficient decomposition algorithm based on algebraic characterization of the set of available circuits, similar to algorithms for more limited gate sets~\cite{Selinger2012a,Kliuchnikov2012b,Bocharov2013}.  But the current theoretical characterization of RUS circuits is limited and is a direction for future work.
Instead, we elect to expand the database by explicitly constructing all possible sequences of circuits.

Construction of the expanded database is similar in nature to the constructions of~\cite{Fowl04c} and~\cite{Bocharov2012}.  Starting with the set of circuits found by our direct search algorithm, we compute all products of pairs of circuits, keeping those that produce a unitary which is not yet in the database.  Triples of circuits can then be constructed from singles and pairs, and so on.  Composite circuits of arbitrary size can be constructed in this way.  Call a circuit a class-$k$ circuit if it is composed of a $k$-tuple of circuits from the original database.  Then the number $N_k$ of class-$k$ circuits is bounded by
\begin{equation}
N_k \leq N_1 \cdot N_{k-1} \leq N_1^k
\enspace ,
\end{equation}
where $N_1$ is the number of circuits in the original database.

To make database expansion more manageable, we keep only those circuits that yield an expected $T$ count of at most some fixed value $T_0$.  This has the simultaneous effect of discarding poorly performing circuits and reducing the value of $N_k$ so that construction of class-$(k+1)$ circuits is less computationally expensive.  Furthermore, circuits can be partitioned into equivalence classes by Clifford conjugation.  The unitaries of the initial set of circuits are of the form $g_0 U g_1$, where $U$ is the unitary obtained from the RUS circuit, and $g_0, g_1$ are Cliffords.  Thus, the product of $k$ such circuits has the form
\begin{equation}
g_0 U_1 g_1 U_2 g_2\ldots U_k g_{k}
\enspace .
\end{equation}
The set of class-(${k+1}$) circuits can then be constructed by using
\begin{equation}
g_0 U_1 g_1 U_2 g_2\ldots U_k g_{k} (g_{k'} U_{k+1} g_{k+1}) = g_0 U_1 g_1 U_2 g_2\ldots U_k g_{k''} U_{k+1} g_{k+1}
\enspace,
\end{equation}
so that the Clifford $g_k$ is unnecessary.  Furthermore, $g_0$ can always be prepended later, and so we instead express each class-$k$ unitary as
\begin{equation}
\label{eq:clifford-equiv-representative}
U_1 g_1 U_2 g_2\ldots U_k
\enspace .
\end{equation}

To find an equivalence class representative of $U$, we first adjust the global phase by multiplying by $u^*/\sqrt{|u|^2}$, where $u$ is the first non-zero entry in the first row of $U$. Next, we conjugate $U$ by all possible pairs of single-qubit Cliffords. The first element of a lexicographical sort then yields the representative $g_1 U g_2$ for some Cliffords $g_1, g_2$.

Once the database has been constructed, the decomposition algorithm is straightforward.  Given a single-qubit unitary $U$ and $\epsilon \in [0,1]$, select all database entries $V$ such that $D(U,V) \leq \epsilon$, where
\begin{equation}
\label{eq:fowler-distance}
D(U,V) = \sqrt{\frac{2-\abs{\Tr(U^\dagger V)}}{2}}
\end{equation}
is the distance metric defined by~\cite{Fowl04c} and also used by~\cite{Selinger2012a, Kliuchnikov2012b,Bocharov2013,Wiebe2013}.
Then, among the selected entries, find and output the circuit with the lowest expected $T$ count.

\subsection{Decomposition with axial rotations
\label{sec:repeat.applications.database.axial}
}

An arbitrary single-qubit unitary can be decomposed into a sequence of three $Z$-axis rotations and two Hadamard gates~\cite{Nielsen2000}.  Therefore, approximate decomposition of $Z$-axis rotations suffices to approximate any single-qubit unitary.
If we limit to $Z$-axis, i.e, diagonal rotations only, then a few additional simplifications are possible.  In particular, each unitary can be represented by a single real number corresponding to the rotation angle in radians.  The result of a sequence of such rotations is then given by the sum of the angles.  Furthermore, up to conjugation by $\{X,S\}$, all $Z$-axis rotations can be represented by an angle in the range $[0, \pi/4]$.  This allows for construction of a database of $Z$-axis rotations which is much larger than a database of arbitrary (non-axial) unitaries.

Using the database expansion procedure described above, we were able to construct a database containing all combinations of RUS circuits with expected $T$ count at most $30$. The maximum distance (according to~\eqnref{eq:fowler-distance}) between any two neighboring rotations is less than $2.8\times 10^{-6}$, and can be improved to $2\times 10^{-6}$ by selectively filling the largest gaps. So the resulting database permits approximation of any $Z$-axis rotation to within $\epsilon = 10^{-6}$.

To approximate a $Z$-axis rotation by an angle $\theta$, we simply select all of the entries that are within the prescribed distance $\epsilon$, and then choose the one with the smallest expected $T$ count.  This procedure is efficient since the database can be sorted according to rotation angle.  Then the subset of entries that are within $\epsilon$ can be identified by binary search.

In order to assess the performance of this method, we approximate, for various values of $\epsilon$, a sample of $10^5$ randomly generated angles in the range $[0,\pi/4]$.
%(Any $Z$-axis rotation can be expressed as a rotation by $\theta = \theta_1 + \theta_2$, where $\theta_1 \in [0,\pi/4]$ and $\theta_2 \in {\pm \pi/2, \pi}$ can be implemented by a Clifford.)
Results are shown in~\figref{fig:rus-axial-decomposition} and~\tabref{tab:rus-axial-decomposition}.  A fit of the mean expected $T$ count for each $\epsilon$ yields a scaling given by~\eqnref{eq:axial-tcount-scaling}, with a slope roughly $2.4$ times smaller than that reported by~\cite{Kliuchnikov2012b} for the rotation $R_Z(1/10)$.

\begin{table}
    \begin{minipage}[b]{.57\textwidth}
        \includegraphics[width=.99\textwidth]{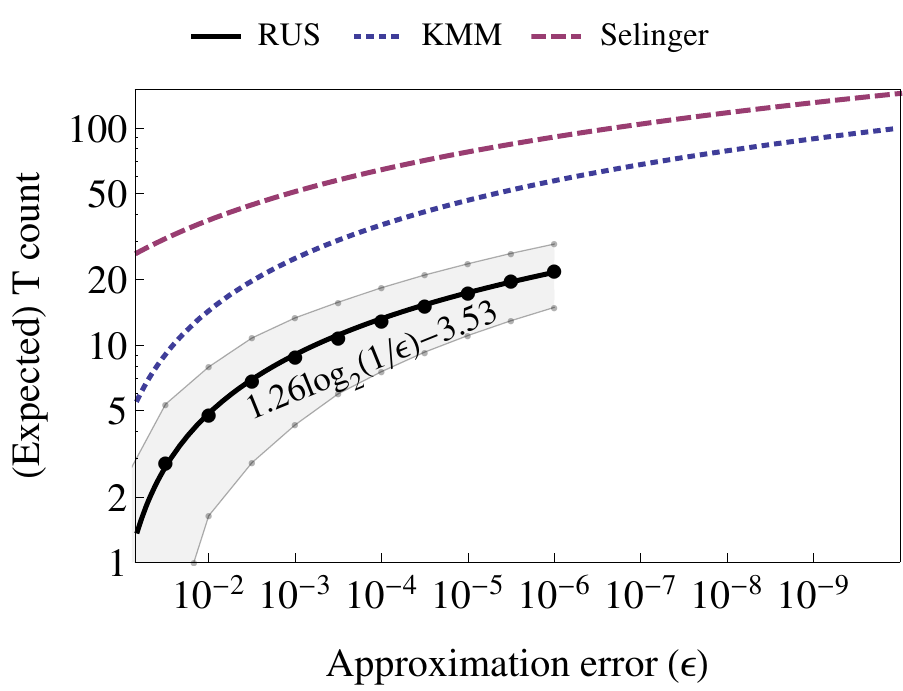}
         \captionof{figure}[$T$ count scaling for approximation of $Z$-axis rotations.]{The above plot shows the expected number of $T$ gates required to approximate a single-qubit $Z$-axis rotation to within a distance $\epsilon$.  The plot was constructed by selecting $10^5$ real numbers in the range $[0,\pi/4]$ uniformly at random.  For each value $\theta$, the RUS circuit with the smallest expected $T$ count within $\epsilon$ of the unitary $R_Z(\theta)$ was selected.  The mean for each value of $\epsilon$ is plotted, yielding a fit-curve of $1.26 \log_2(1/\epsilon) - 3.53$.  The gray region is an estimate of the interval containing the actual number of $T$ gates with probability $95\%$.
Scaling of the Selinger and KMM algorithms are included for reference.
%The red BGS curve is given by~\eqnref{eq:bgs-tcount}, and indicates the expected $T$ count when using $V$-basis decomposition and~\figref{fig:v3-5.26}.
}
         \label{fig:rus-axial-decomposition}
    \end{minipage}
    \hfill
    \begin{minipage}[b]{.4\textwidth}
        \centering
        \begin{tabular}{@{}l@{~~}l@{~~}l@{}}
			$\log_{10}(1/\epsilon)$ & Exp $T$ ($\sigma^2$)& $\pm 95\%$ ($\sigma^2$)\\
			\hline
			$1$  &$1.1$ ($1.1$) &$1.2$ ($3.6$)\\
			$1.5$&$2.9$ ($2.2$) &$2.5$ ($2.9$)\\
			$2$  &$4.8$ ($3.4$) &$3.1$ ($2.9$)\\
			$2.5$&$6.8$ ($3.9$) &$4.0$ ($3.8$)\\
			$3$  &$8.8$ ($4.3$) &$4.5$ ($4.7$)\\
			$3.5$&$10.9$ ($4.6$)&$4.9$ ($5.2$)\\
			$4$  &$12.9$ ($4.8$)&$5.4$ ($5.5$)\\
			$4.5$&$15.1$ ($5.3$)&$5.9$ ($5.7$)\\
			$5$  &$17.4$ ($5.7$)&$6.3$ ($5.8$)\\
			$5.5$&$19.6$ ($6.0$)&$6.7$ ($6.1$)\\
			$6$  &$22.0$ ($6.4$)&$7.1$ ($6.5$)
		\end{tabular}
		\vspace{.5cm}
        \caption[Expected $T$ counts for approximation of random $Z$-axis rotations with RUS circuits.]{Expected $T$ counts for approximation of random $Z$-axis rotations with RUS circuits.  The middle column indicates the expected $T$ count based on a sample of $10^5$ random angles.  The right-hand column indicates the expected $95$ percent confidence interval of the $T$ count for the best RUS circuit, given a random angle $\theta$.  The variance of each expected value is indicated in parenthesis.}
        \label{tab:rus-axial-decomposition}
        \vspace{.5cm}
    \end{minipage}
\end{table}

By way of comparison, Wiebe and Kliuchnikov report a scaling of $1.14\log_2(1/\theta)$ for small angles $\theta$.  However, their RUS circuits are specially designed for small angles.  For arbitrary angles they report an expected $T$ count of about
\begin{equation}
\label{eq:wiebe-kliuchnikov-tcount}
1.14\log_2(10^\gamma) + 8\log_2(10^{-\gamma}/\epsilon)
\enspace,
\end{equation}
where $\theta = a\times 10^{-\gamma}$ for some $a \in (0,1)$ and integer $\gamma > 0$.
Using~\eqnref{eq:wiebe-kliuchnikov-tcount} to calculate costs for the same $10^5$ random angles as above, we obtain a fit function of
\begin{equation}
\label{eq:wiebe-kliuchnikov-coarse-fit}
6\log_2(1/\epsilon) - 2.2
\enspace .
\end{equation}
Formula \eqnref{eq:wiebe-kliuchnikov-coarse-fit} indicates that the efficiency of the circuits in~\cite{Wiebe2013} does not extend to coarse angles.

Equation~\eqnref{eq:axial-tcount-scaling} also implies that RUS $Z$-axis rotations can be used to approximate \emph{arbitrary} single-qubit unitaries with a scaling approaching that of optimal ancilla-free decomposition.  Since an arbitrary unitary can be expressed as a product of three axial rotations, the expected $T$ count for approximating an arbitrary single-qubit unitary is given by
$3.9\log_2(3/\epsilon) - 8.37$.
On the other hand, Fowler calculates an optimal $T$-count of $2.95\log_2(1/\epsilon) + 3.75$ (on average) without using ancillas~\cite{Fowl04c}.

Since our circuits are non-deterministic, we are also concerned with the probability distribution of the number of $T$ gates.  For each composite circuit in the database, we calculate the variance $\sigma^2$ of the $T$ count based on the variance of each individual circuit.  We may then obtain a confidence interval using Chevyshev's inequality
\begin{equation}
\Pr(\abs{\text{Actual}[T] - \text{Exp}[T]} \geq k\sigma) \leq \frac{1}{k^2}
\enspace .
\end{equation}
~\tabref{tab:rus-axial-decomposition} shows the mean of the expected $T$ count for each $\epsilon$.  By also calculating the mean of the variance $\sigma^2$, we obtain an estimate of the corresponding $95\%$ confidence interval, shown by the gray region in~\figref{fig:rus-axial-decomposition}.  That is, for a randomly chosen angle $\theta$, the total number of $T$ gates required to implement $R_Z(\theta)$ is within the given interval around $1.26\log_2(1/\epsilon)-3.53$, with probability $0.95$.

The approximation accuracy permitted by our database is limited by computation time and memory.  To maximize efficiency, we used floating-point rather than symbolic arithmetic. Construction of all RUS circuit combinations up to expected $T$ count of $30$ took roughly $20$ hours and $41$ GB of memory using Mathematica.  \tabref{tab:z-rotation-density} shows the number of circuit combinations and corresponding rotation angle densities for increasing values of the expected $T$ count.  The size and density of the database increases by about an order of magnitude for every five $T$ gates.  We expect that with a more efficient implementation---in C/C++ for example---the worst-case approximation accuracy could be improved.

\begin{table}
\centering
\begin{tabular}{llll}
\hline\hline
Max. exp. &      &               &\\
$T$ count & Size & Mean $D$ & Max $D$ \\
\hline
$5$  & $7$ & $0.04$ & $0.08$\\
$10$ & $134$ & $0.0021$ & $0.0066$\\
$15$ & $2079$ & $0.00013$ & $0.00014$\\
$20$ & $27420$ & $0.00001$ & $0.00017$\\
$25$ & $320736$ & $0.0000009$ & $0.000016$\\
$30$ & $3446708$ & $0.00000008$ & $0.0000028$\\
%$31$ & $4235018$ & $0.00000007$ & $0.0000025$
\hline\hline
\end{tabular}
\caption[Size and density of the $Z$-axis rotation database according to the maximum expected number of $T$ gates.]{\label{tab:z-rotation-density}
Size and density of the $Z$-axis rotation database according to the maximum expected number of $T$ gates.  The mean and the maximum distances between nearest neighbors is given in columns three and four, respectively.
}
\end{table}

\subsection{Decomposition with non-axial rotations
\label{sec:repeat.applications.database.non-axial}
}

Using either the above database, or the methods of KMM or Selinger, decomposition of an arbitrary unitary incurs an additional factor of three in cost because each of the three $Z$-axis rotations are approximated separately.  The increased cost is illustrated in~\figref{fig:kmm-ratios} by the larger ratios for non-axial unitaries.  Indeed~\figref{fig:kmm-ratios} suggests that incorporating both axial and non-axial RUS circuits could yield better approximations than using $Z$-axis rotations alone.

Fowler's method does not incur the additional factor of three for arbitrary unitaries, maintaining a scaling of $2.95\log_2(1/\epsilon) + 3.75$.  But as noted before, RUS circuits offer a larger domain of exactly implementable unitaries than circuits without ancillas.  Just as RUS circuits outperform ancilla-free $Z$-axis decomposition, they could outperform ancilla-free non-axial decomposition.

On the other hand, construction of the database in the non-axial case is significantly more challenging than in the axial case.  Unitaries must be represented by three rotation angles instead of one. Multiplication of circuit combinations is less efficient than for $Z$-axis rotations which only require addition.
Organizing the database for efficient lookup is also more complicated. $Z$-axis rotations can be sorted by rotation angle, but arbitrary unitaries require a more complicated data structure such as a $k$-d tree~\cite{Dawson2005,Amy2013a}.

Despite these limitations, there are some savings to be had.  We may still express each unitary by its Clifford equivalence class representative~\eqnref{eq:clifford-equiv-representative}.  Conjugation by all $576$ pairs of Cliffords is not required however.  First, note that any single-qubit Clifford can be written as a product $g_1g_2$ where $g_1\in G_1$, $g_2 \in G_2$ and
\begin{equation}\begin{split}
G_1 &= \{I,Z,S,S^\dagger\}\\
G_2 &= \{I,H,X,XH,HS,XHS,HSH,XHSH\}
\enspace .
\end{split}\end{equation}
Now, instead of conjugating by the entire Clifford group, we conjugate only by $G_2$.
Then, each resulting unitary can be decomposed into three rotations
\begin{equation}
g_2 U g'_2 = R_Z(\theta_1)R_X(\theta_2)R_Z(\theta_3)
\enspace ,
\end{equation}
where $g_2 \in G_2$ and $g'_2 \in \{g^\dagger ~|~ g\in G_2\}$.
The Cliffords in $G_1$ are diagonal, and only modify $\theta_1$ and $\theta_3$.  Up to conjugation by these remaining Cliffords, we then have
\begin{equation}
R_Z(\theta_1)R_X(\theta_2)R_Z(\theta_3) \equiv R_Z(\theta_1\!\!\!\mod \pi/2)R_X(\theta_2)R_Z(\theta_3\!\!\!\mod \pi/2)
\enspace .
\end{equation}
Choosing $0 \leq \theta_1, \theta_2 < \pi/2$, we can find an equivalence class representative without actually conjugating by $G_1$, saving a factor of $576/64 = 9$.

Even with this optimization, though, our Mathematica implementation is quite slow.  We were able to construct a database of size $45526$ consisting of all RUS circuits with expected $T$ count at most $18$.  We then calculated the best circuit for each of $100$ random single-qubit unitaries for a variety of $\epsilon \geq 8\times 10^{-3}$.  A fit-curve for the data yields a scaling given by~\eqnref{eq:nonaxial-tcount-scaling}.  Based on the slope, the savings is only about $18$ percent over Fowler, but in absolute terms the savings is roughly a factor of two, at least for modest approximation accuracy. See~\figref{fig:rus-nonaxial-decomposition}.

Given the relatively large ratios for non-axial unitaries in~\figref{fig:kmm-ratios}, the scaling given by~\eqnref{eq:nonaxial-tcount-scaling} is perhaps disappointing.  We note, however, that our database contains only a limited subset of possible RUS circuits.  Incorporating a larger set of circuits could improve performance.

\begin{figure}
\centering
\includegraphics{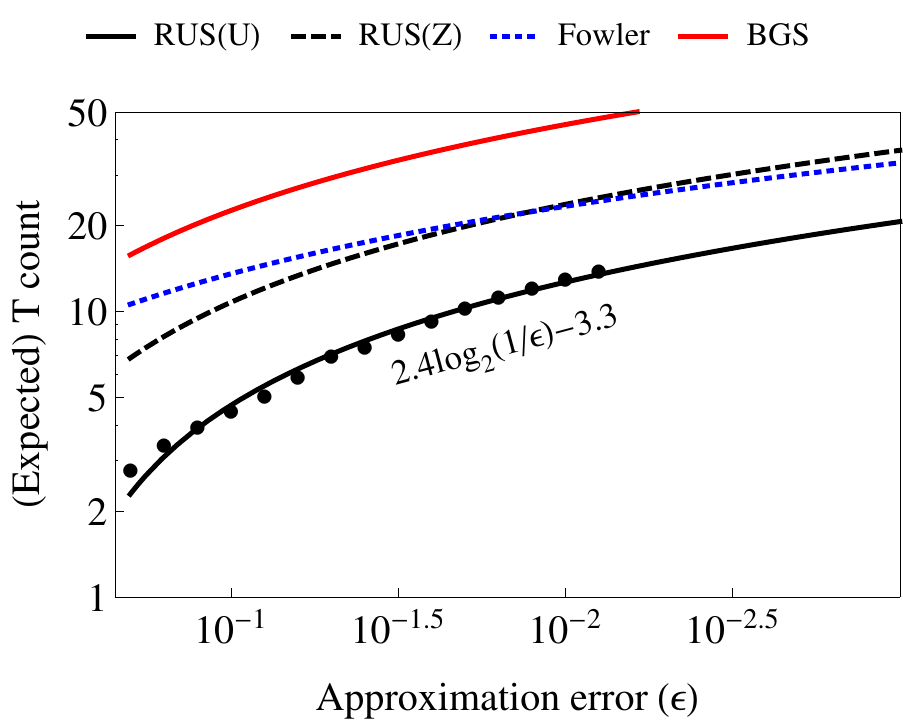}
\caption[$T$ count scaling for approximation of non-axial single-qubit unitaries.]{\label{fig:rus-nonaxial-decomposition}
The above plot shows the expected number of $T$ gates required to approximate an arbitrary single-qubit unitary to within a distance $\epsilon$.  Each point indicates the mean of $100$ random unitaries approximated to the corresponding accuracy with our full database of RUS circuits.  With $95$ percent confidence, the solid black line has slope in the range $[2.29,2.51]$. The dashed black line indicates the estimated cost of first expressing the unitary as a product of axial rotations, and then decomposing each rotation using the $Z$-axis RUS database from~\secref{sec:repeat.applications.database.axial}.  The solid red line indicates the scaling obtained by using the circuit in~\figref{fig:v3-5.26} for $V_3$ decomposition~\cite{Bocharov2013}.  This scaling is worse than the others, but is valid for $\epsilon > 10^{-10}$. The estimated scaling due to Fowler~\cite{Fowl04c} is shown for reference.
}
\end{figure}

\section{Quantum algorithms using coarse angles
\label{sec:repeat.applications.phase-estimation}
}
The accuracy to which the database decomposition methods can reach is limited by the size of the database.  Our $Z$-axis rotation database is capable of approximations to within $10^{-6}$.  If the required accuracy is higher than that, then either the database must be expanded, or an algorithmic decomposition such as Selinger, KMM, or that of~\secref{sec:repeat.applications.v3} must be used.  However, a variety of important quantum algorithms require only relatively coarse accuracy.  Fowler, for example, used numerical analysis to argue that Shor's algorithm requires rotation angles no smaller than $\theta = \pi/64\approx 0.05$ with an with an approximation error of $\epsilon = \pi/256 \approx 0.012$~\cite{Fowl03b}.

Another application of coarse angles is in quantum chemistry.  Consider a Hamiltonian for a molecule expressed in second quantized form, where the objective is to determine the ground state energy of the molecule.\footnote{The second quantized form expresses the quantum system in terms of the number of particles in each possible state.  The specifics are not important for the current discussion, however.}  Wecker et al. \cite{Wecker2013} have developed a technique to scale the coefficients of the non-commuting terms in the Hamiltonian to the maximum coefficient, while maintaining arbitrary accuracy on the estimate of the energy.  This scaling allows one to use large angles within the phase estimation algorithm, where the angles require at most $10^{-6}$ accuracy in practice.
Similarly, Jones et al. show how to optimize quantum chemistry simulations by ignoring terms with small norm~\cite{Jones2012}. They use $Z$-axis rotations with approximation accuracies in the range $\epsilon = 10^{-5}$.

\section{Possible generalizations and limitations}
Traditional methods decompose single-qubit unitaries into deterministic sequences of gates. Wiebe and Kliuchnikov showed that by adding measurements and allowing non-deterministic circuits, decompositions with fewer $T$ gates are possible (in expectation) for very small $Z$-axis rotations~\cite{Wiebe2013}.  Our results extend that conclusion to arbitrary single-qubit unitaries.  By constructing a database of repeat-until-success circuits and then progressively composing those circuits, we can approximate arbitrary single-qubit unitaries to within a distance of $10^{-6}$, which is sufficient for many quantum algorithms.  For a random $Z$-axis rotation, our database yields an approximation which requires as little as one-third as many $T$ gates as~~\cite{Selinger2012a}, \cite{Kliuchnikov2012b} and~\cite{Fowl04c}.  Using all of the circuits in our database (not just the $Z$-axis rotations), the improvement for arbitrary unitaries can be larger, though achieving high approximation accuracy is challenging.

Our results suggest a number of possible areas for improvement and further research.  First, the circuits proposed by~\cite{Wiebe2013} use traditional decomposition algorithms (i.e., Selinger or KMM) to generate the unitaries required for the mantissa $a$ of the angle $a \times 10^{-\gamma}$. Instead, our RUS circuits could be used in order to improve performance.  Indeed, one could consider a hybrid approach that combined all available decomposition methods in order to find the most efficient circuit.
Second, circuits of the form shown in~\figref{fig:two-cz-canonical} make up only a subset of possible RUS circuits.  Expanding the search to include additional types of circuits could improve database density.
Third, the formal theory of RUS circuits is not yet understood.  A better understanding could lead to efficient decomposition algorithms based on RUS circuits and allow for approximation to much smaller values of $\epsilon$.  A tight characterization of RUS circuits would seem to first require a better understanding of $\{\Clifford, T\}$ complexity for multi-qubit unitaries.

One could also consider some relaxations to the RUS circuit framework.  We consider only single-qubit unitaries.  However, multi-qubit unitaries or non-unitary channels may also be of interest.  We also restrict to recovery operations that are Clifford operators.  That restriction could be modified to allow for larger or alternative classes of operations.  On the other hand, fault-tolerance schemes based on stabilizer codes often permit no-cost application of Pauli operators~\cite{Knill2004}.  Thus, it might be sensible to limit recovery operations to only tensor products of Paulis.

Finally, the non-deterministic nature of RUS circuits imposes some additional constraints on the overall architecture of the quantum computer. Many fault-tolerance schemes already use non-deterministic methods such as state distillation to implement certain gates.  But most of the non-determinism occurs ``offline'', without impacting the computational data qubits.  Since RUS circuits are ``online'', the time required to implement a given unitary cannot be determined in advance.  Such asynchronicity could complicate placement and routing techniques (see~\chapref{chap:braidpack}) and classical control logic, thereby increasing resource overhead requirements.  Thorough architecture-specific analysis will be required in order to concretely assess the improvements obtained by using RUS circuits.

\chapter{Global optimization of fault-tolerant quantum circuits
\label{chap:braidpack}
}
This chapter is based on material that appears in~\cite{Paetznick2013}.
\vspace{1cm}

One issue that is generally ignored in fault-tolerant constructions, particularly for concatenated codes and including the one in~\chapref{chap:threshold}, is that realistic proposals for quantum computers impose geometric constraints.  Many proposed architectures involve a two-dimensional lattice of qubits for which interactions are limited to a small set of neighboring locations (see~\secref{sec:fault.architecture}). 
Ultimately, any practical fault-tolerance scheme must account for the particular geometry offered by the quantum computer.

In this chapter we propose two algorithms for efficient placement of fault-tolerant quantum circuits onto a two-dimensional rectangular lattice of qubits.  Our algorithms operate within the context of the surface code and therefore automatically respect nearest-neighbor interaction constraints.  Encoded computation in the surface code is represented by a three-dimensional object in space-time called a \emph{braid}.  Our algorithms are based on the fact that the encoded quantum circuit is invariant under topological transformations of the braid. We may, therefore, smoothly deform the braid according to the dimensions of the quantum computer.

Informally, braid compaction is the problem of topologically deforming a braid so that it fits into a prescribed spacetime volume.
This problem bears a striking resemblance to VLSI placement.  In VLSI placement the goal is to arrange a set of logic elements---represented by rectangles---and wires into the smallest possible area subject to connectivity and distance constraints.  In braid compaction, the task is to pack a set of gates, some of which are represented by boxes, into the smallest possible volume subject to distance and topology constraints. The VLSI placement problem is known to be NP-complete~\cite{Schlag1983}. We conjecture that braid compaction is NP-complete, as well; though attempts at a formal reduction have been unsuccessful.

Correspondingly, our algorithms are constructed from carefully designed heuristics.  The first algorithm is loosely based on physical principles of gravity and tension.  The braid is treated as a physical object that is allowed to slide into a space-time box under its own weight.  Gravity forces direct the braid toward the bottom of the box in order to minimize time, and tension forces keep the braid compact.

Our second algorithm uses the optimization technique of simulated annealing, and is based on a similar algorithm for VLSI placement~\cite{Hsieh}. Each part of the braid is modeled as a cuboid (i.e., a box).  Some cuboids have fixed dimension and some are allowed to expand and contract.  Size, distance and topology constraints are given by sets of linear inequalities on the coordinates of each cuboid. Depending on the shape of the braid, some constraints must be actively enforced, and others need not be enforced. The annealing step consists of swapping constraints in and out of the active set to change the shape and size of the braid.

% \figref{fig:braid-compaction-example} shows an example of a prototype implementation of braid compaction for a circuit consisting of $11$ CNOT gates.
% Ignoring the unnecessary staircases on the top and bottom, the original braid has a bounding box of size ($3 \times 16 \times 34$), whereas the the compacted braid fits in a bounding box of size ($10 \times 13 \times 6$), a factor of four improvement along the time axis.  Improved results are expected once the force-directed and simulated annealing algorithms have been fully implemented.
% 
% \begin{figure}
% \centering
% \begin{subfigure}[b]{.3\linewidth}
%   \centering
%   \includegraphics[width=.8\linewidth,angle=90]{images/y-distillation-circuit}
%   \caption{}
% \end{subfigure}
% \hfill
% \begin{subfigure}[b]{.3\linewidth}
%   \centering
%   \includegraphics[scale=.25]{images/y-distill-canonical}
%   \caption{}
% \end{subfigure}
% \hfill
% \begin{subfigure}[b]{.3\linewidth}
%   \centering
%   \includegraphics[scale=.3]{images/y-distill-packed}
%   \caption{}
% \end{subfigure}
% \caption{\label{fig:braid-compaction-example}
% Optimization of eleven CNOT gates in the surface code with a prototype implementation of the force-directed braid compaction algorithm.
% The circuit (a) is mapped to a canonical surface code braid (b), which is compressed into a smaller but topologically equivalent braid (c).}
% \end{figure}  

%-----------------------------------------------------------------------

\section{Parallelism: optimizing for time
\label{sec:braidpack.parallelism}
}
The main goal of the optimization techniques in previous chapters has been to reduce the space requirements of fault-tolerant quantum computation.  In many cases, these optimizations also lead to smaller time overhead, as well.  To this point, however, time optimization has been a secondary goal. Furthermore, these techniques focus on small but repeated parts of the circuit. 
They do not address, for example, global parallelism concerns.

In our current context, we are instead given a fixed two-dimensional lattice of qubits, and are asked to minimize the time overhead.  If we can minimize the space requirements without increasing time requirements, then we should.  But space that is available but otherwise unoccupied is wasted.

An important goal, therefore, is to parallelize quantum algorithms.  However, many quantum algorithms are serial in nature, leaving large numbers of qubits idle much of the time.  
Low-gate-count arithmetic quantum circuits, for example, form a staircase structure of linear depth~\cite{Cuccaro2004}.
Parallelization of certain procedures, such as the quantum Fourier transform, is possible when extra qubits are available but is typically done on a case by case basis~\cite{Cleve2000a}.

% While it is possible introduce error-correction without inducing time overhead~\cite{Fowler2012g}, any fault-tolerant circuit will require an increase in either time or space as compared to the corresponding ideal quantum circuit. 
% Resource estimates are even larger if geometric constraints of the quantum computer are considered.
% Indeed, many proposals for quantum computers impose a fixed-size two-dimensional lattice, with limited interactions between qubits~\cite{Levy2001,Weinstein2005,Devi08,Helmer2009,Divincenzo2009,Amin10,Jones2012b,Kump11,Levy2011,Kloeffel2013}.

% We propose an automated and global resource optimization solution which is fault-tolerant and accounts for geometry and locality constraints by operating within the surface code.  
% Our strategy is to minimize computation time by smoothly reshaping the computation in order to eliminate wasted space.
% Two algorithms are presented. The first is a force-directed algorithm in which the fault-tolerant quantum circuit is treated as a malleable physical object. The second algorithm is based on simulated annealing.
% Each algorithm takes as input an ideal quantum circuit and a fixed-size two-dimensional qubit lattice and outputs a corresponding compact fault-tolerant quantum circuit that operates entirely within the lattice dimensions.

Some general techniques for pararallelization exist. Typical methods involve local circuit rewriting rules for trading between sequences of gates and additional qubits~\cite{Moore1998,Maslov2008,Saeedi2010}.
Small-depth circuits can be achieved for certain sub-classes of quantum circuits.  Clifford group circuits, for example, can be parallelized to quantum circuits of constant depth followed by log-depth classical post-processing~\cite{Raussendorf2001}.

Others have proposed global circuit optimization procedures that involve a multi-staged transformation to and from the measurement-based quantum computing model~\cite{Broadbent2007,DaSilva2013}. Indeed, there are strong similarities between the measurement-based model and the surface code~\cite{Raussendorf2007}.
However, the template-based and measurement-based optimizations are not fault-tolerant and, except for~\cite{Saeedi2010}, do not explicitly consider geometric constraints imposed by the quantum computer.  It is not clear that the resulting circuits remain compact under such restrictions.

By contrast, since our algorithms operate within the surface code, the output is automatically fault-tolerant and can be easily mapped to a wide variety of two-dimensional nearest-neighbor architectures~\cite{Divincenzo2009,Ghosh2012}.
Furthermore, the rules for topologically transforming surface code braids are conceptually  simple.  There is no need to break up the transformation into multiple stages.
Thus, compared to other proposals, we feel that our approach is easier to understand, implement, and extend.

% That is not to say that braid compaction is \emph{computationally} easy, however.
% In fact, based on known complexities of other similar problems such as VLSI placement~\cite{Schlag1983} and container loading (see, e.g.,~\cite{Scheithauer}), we conjecture that the problem of minimizing the time (height) of a braid on a fixed rectangular lattice of qubits is NP-complete.  Our optimization algorithms are therefore crafted from carefully designed heuristics in order to handle large-scale instances. 

% The first algorithm is loosely based on the physical principles of gravity and tension. Roughly, the braid is treated as a heavy tangle of rubber bands that is allowed to slide into a box under the force of gravity.  The idea is vaguely analogous to techniques used in graph drawing, for example~\cite{Kobourov2012}. The second algorithm uses the technique of simulated annealing.  In simulated annealing the solution landscape is explored by making random incremental changes and favoring those changes which decrease the solution size.  Our algorithm is inspired by a similar procedure for VLSI placement~\cite{Hsieh}.

% We begin with a brief summary of the surface code in~\secref{sec:surface-code} followed by a more formal statement of the braid compaction problem in~\secref{sec:braid-compaction}.  In~\secref{sec:iterative-forcing} we describe the force-directed algorithm in detail, and a C++ implementation that we call Braidpack. The simulated annealing algorithm is presented in~\secref{sec:simulated-annealing}.

\section{The surface code
\label{sec:braidpack.surface-code}
}
The optimization algorithms in~\secref{sec:braidpack.force-directed} and~\secref{sec:braidpack.annealing} are based on fault-tolerant quantum circuits for the surface code. The surface code uses a fundamentally different approach to encoding logical quantum gates than we have previously seen for concatenated codes, and this encoding is key to our optimization approach.
In this section, we give a brief pedagogical introduction to the surface code, with a focus on the mapping from a quantum circuit to a surface code braid.  Other details of the surface code are not essential for understanding our compaction algorithms.
For a comprehensive introduction to the surface code we refer the reader to~\cite{Fowler2012d}.

The surface code has a number of desirable properties.
First, it operates on a two-dimensional rectangular lattice of qubits.  All operations can be performed using only one-qubit gates, and two-qubit gates involving only nearest neighbor qubits.  As a result, the required number of qubits scales much more slowly for the surface code than for concatenated codes on $2$-D nearest-neighbor architectures.  At the same time, the surface code tolerates noisier physical gates than many other quantum error correcting codes.  Reliable computation is possible so long as the noise rate is below roughly $0.6$ percent per gate~\cite{Fowler2012d}.

\subsection{Encoded qubits}
The surface code is a CSS code that can be defined on a $2D$ rectangular lattice graph of degree four.  A qubit is placed on each edge of the graph.  The $X$ stabilizer generators correspond to weight-four operators around each vertex---i.e., each operator has support only on the qubits adjacent to the corresponding vertex. The $Z$ stabilizer generators correspond to weight-four operators around each face of the lattice---i.e., each operator has support only on qubits of the edges that define the face. 
%Each measurement imposes a restriction on the Hilbert space of the lattice. 
Encoded qubits are created by disabling some of the generators, thereby adding new degrees of freedom to the code. 
We choose to define a qubit as a pair of \emph{defects}.  Defects are contiguous regions of the lattice for which the stabilizers are \emph{not} measured. 
There are two types of defects, primal and dual.  Primal defects correspond to operators around vertices of the lattice ($Z$ stabilizer generators), and dual defects correspond to operators around the faces of the lattice ($X$ stabilizer generators).

Error protection is achieved by creating defects of sufficient size, and by keeping defects well separated in space.  For a code distance of $d$, we require that all defects have circumference $d$ and that defects of the same type are separated in $L_\infty$ distance by $d$. For defects of opposite type, the minimum distance depends on the shape of each defect. In all cases a distance of $d/4$ is sufficient (for code distance $d$), though in some cases primal and dual defects may be as close as $d/8$.

\begin{figure}
\centering
\begin{subfigure}[b]{.45\linewidth}
	\centering
	\includegraphics[width=\linewidth]{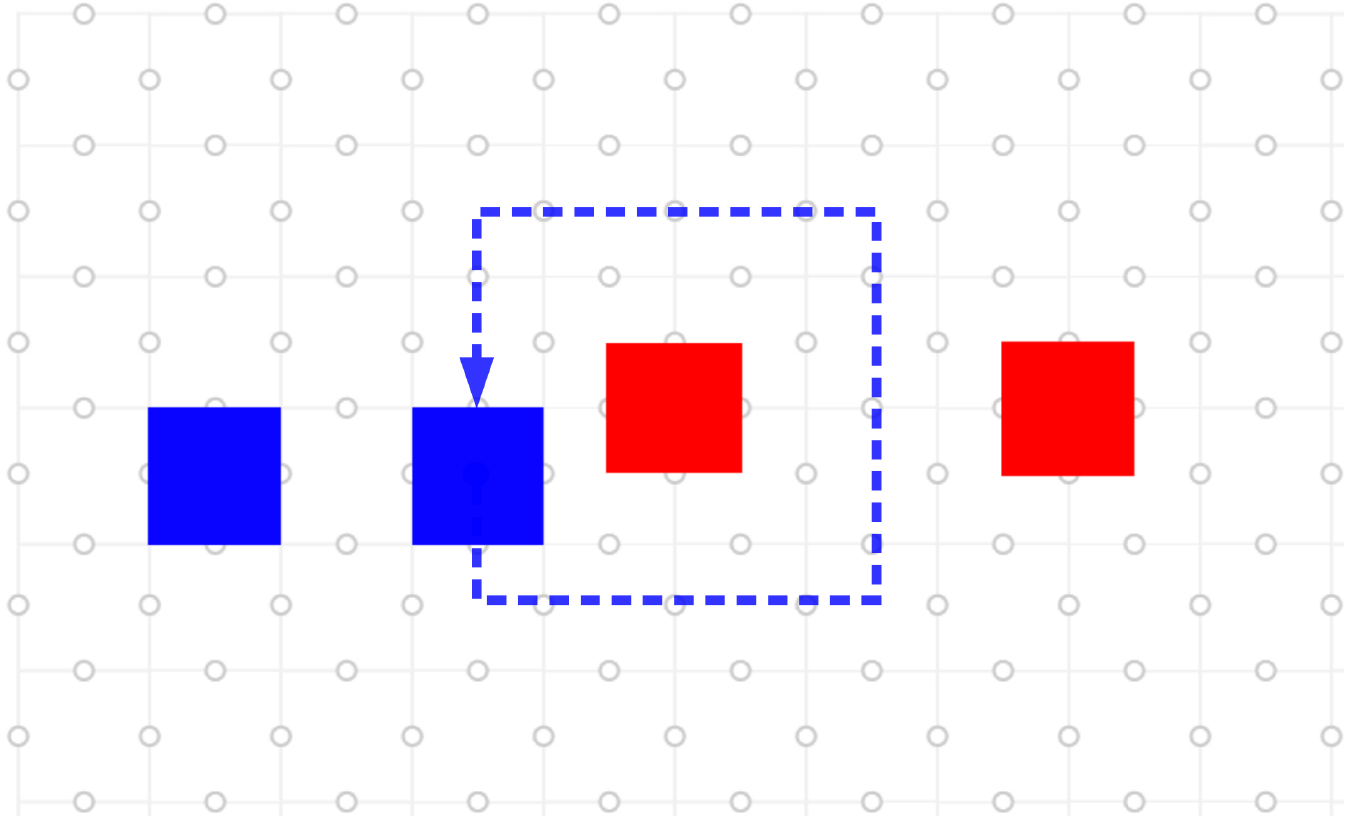}
	\caption{top view}
\end{subfigure}
\begin{subfigure}[b]{.45\linewidth} 
	\centering
	\includegraphics[width=.7\linewidth]{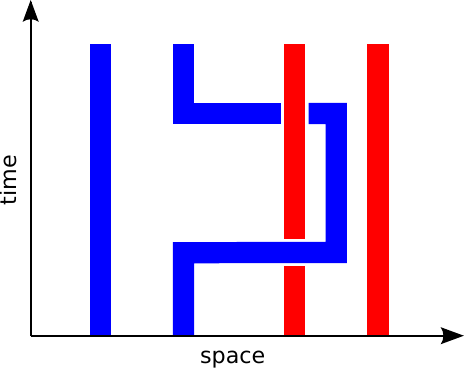}
	\caption{side view}
\end{subfigure}
\caption[Surface code primal-dual CNOT.]{\label{fig:primal-dual-cnot}
(a) Encoded surface code qubits are defined by pairs of defects, either primal (red) or dual (blue). Each defect is composed of multiple physical qubits on the two-dimensional lattice. Operations are performed by moving defects around.  Here, an encoded two-qubit operation is performed by moving one defect from the dual encoded qubit around one of the defects of the primal encoded qubit. (b) The same operation can be written as a space-time diagram in which one of the space axes has been flattened.
}
\end{figure}

\subsection{Encoded operations}

Most encoded operations in the surface code proceed by moving defects around each other.  
Defect movement is achieved by turning off new regions of stabilizer measurements and then turning on other stabilizer measurements.
The movement can be divided into time-slices.
By stacking time-slices on top of each other, the encoded operations are represented by a three-dimensional object in space and time called a \emph{braid}. See~\figref{fig:primal-dual-cnot}.
Transformation of a quantum circuit to a braid can be done systematically by constructing canonical braid elements for each quantum gate.  Preparation of encoded $\ket 0$  is represented by a ``U''-shaped primal defect.  Encoded $Z$-basis measurement is essentially the reverse.  A CNOT operation is performed by a loop of dual defects that wraps around the two associated encoded qubits. See~\tabref{tab:canonical-gate-set}.

\begin{table}
\centering
\def\imagetop#1{\vtop{\null\hbox{#1}}}
\begin{tabular}{c@{\qquad}c@{\qquad}c@{\qquad}c@{\qquad}c}
\includegraphics[width=.05\linewidth]{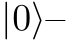}&
\includegraphics[width=.05\linewidth]{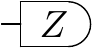}&
\includegraphics[width=.05\linewidth]{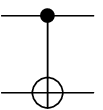}&
\includegraphics[width=.05\linewidth]{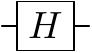}&
\includegraphics[width=.05\linewidth]{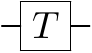}
\\
\imagetop{\includegraphics[width=.04\linewidth]{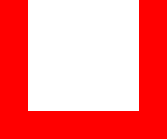}}&
\imagetop{\includegraphics[width=.04\linewidth]{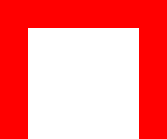}}&
\imagetop{\includegraphics[width=.08\linewidth]{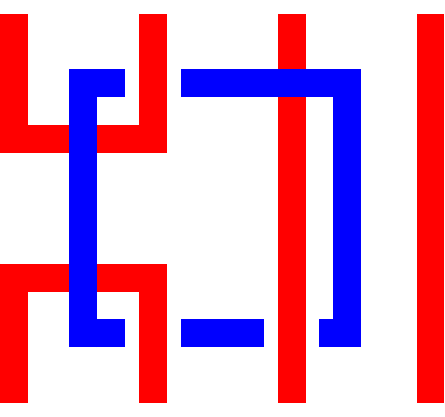}}&
\imagetop{\includegraphics[width=.04\linewidth]{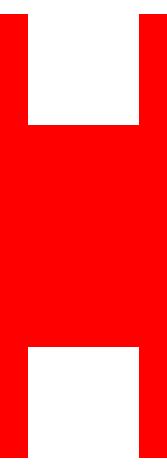}}&
\imagetop{\includegraphics[width=.2\linewidth]{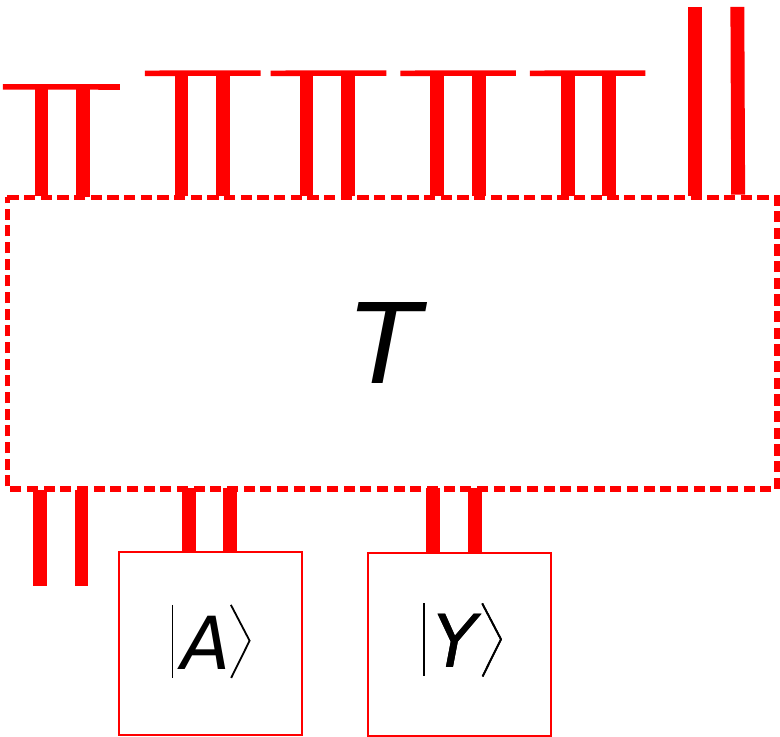}}
\end{tabular}
\caption[Surface code gate set.]{\label{tab:canonical-gate-set}
The surface code gate set (top) and corresponding canonical braids (bottom).  Each braid is a three-dimensional collection of defects.  For visual clarity, the braids have been flattened here into two dimensions.  
}
\end{table}

Braids consisting of these operations are invariant under topological deformation.  That is, a quantum circuit can be represented by a canonical braid, and also by any braid that is topologically equivalent to that canonical braid.  Strings of defects may be smoothly pulled or pushed around in space and time without altering the encoded quantum computation. See~\figref{fig:topological-deformation}.  Note that space and time are symmetric here.  Space can be traded for time and vice versa.

\begin{figure}
\centering
\begin{subfigure}[b]{.15\linewidth}
	\includegraphics[width=.8\linewidth]{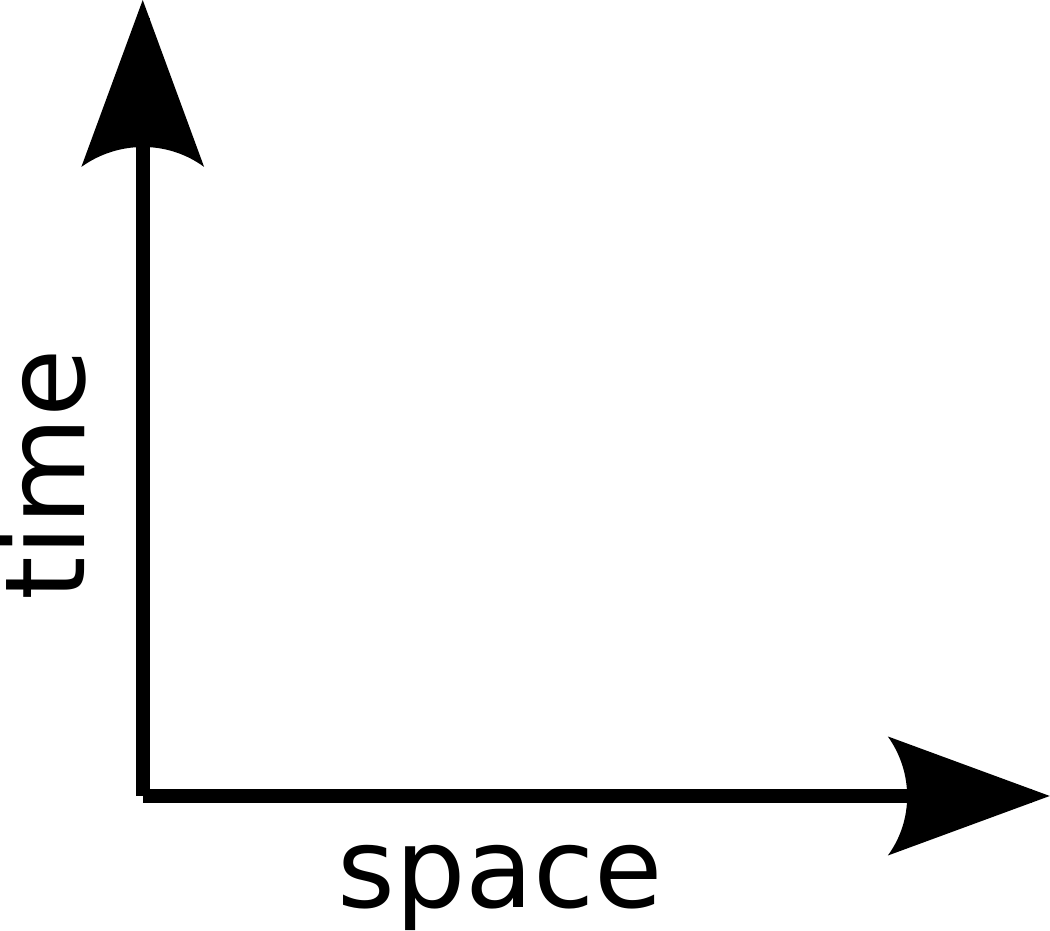}
\end{subfigure}
\begin{subfigure}[b]{.25\linewidth}
	\centering
	\includegraphics[height=2.5cm]{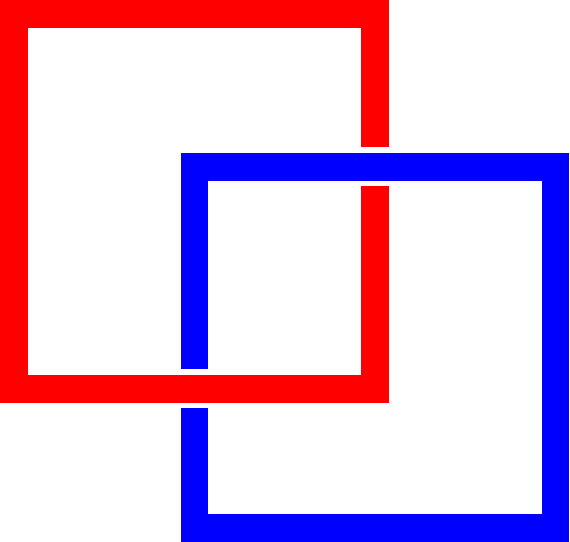}
	\caption{}
\end{subfigure}
\hspace{.1\linewidth}
\begin{subfigure}[b]{.25\linewidth}
	\centering
	\includegraphics[height=2.5cm]{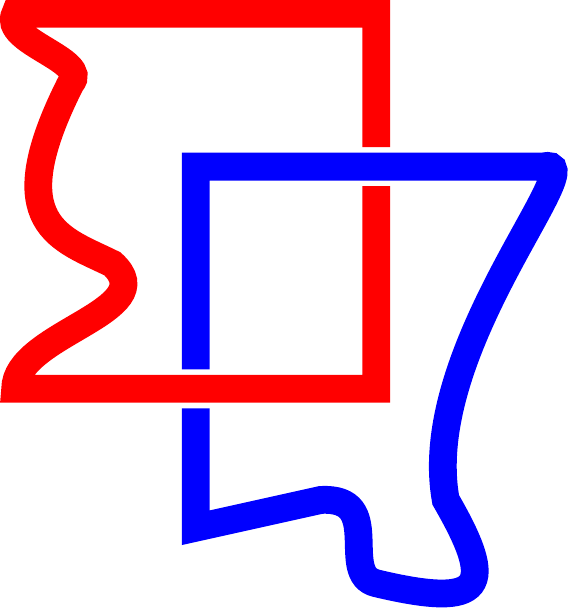}
	\caption{}
\end{subfigure}
\caption[An example of topological deformation.]{\label{fig:topological-deformation} Surface code braids are invariant under topological deformation.  The space-time diagram on the left (a) is topologically equivalent to the diagram on the right (b). Defect strings and loops may be smoothly stretched and contracted without altering the encoded operation.
}
\end{figure}

Not all encoded operations in the surface code can be performed topologically, however.  The encoded Hadamard operation, for example, requires the encoded qubit---i.e., the two corresponding defects---to be placed on a separate lattice, isolated from all other encoded qubits.  This is achieved by first ``cutting out'' part of the lattice around the encoded qubit and then later re-attaching it to the rest of the lattice~\cite{Fowler2012b}. The resulting space-time volume is a cuboid (i.e., a box) of dimension roughly $3d/2 \times 3d/2 \times 5d/2$.  
However, the cuboid contains a variety of boundary types near the surface, thus imposing some restrictions on the configurations of other surrounding defects.
The cuboid can be translated in any direction, or rotated about the time-axis by increments of $\pi/2$, but is otherwise treated as a rigid object.\footnote{In principle, a sideways Hadamard gate is possible and would allow for rotations about the $x$ and $y$ axes.  However, the chosen implementation requires the cuboid to be vertically oriented.}
For concreteness, we adopt the convention that time corresponds to the $z$-axis.

We will also require one other non-topological operation, the encoded $T$-gate. 
This gate cannot be implemented directly in the surface code and is instead constructed by the state distillation protocol described in~\secref{sec:fault.gates.distillation}.  Distillation does not explicitly require the encoded qubit to be cut out of the lattice, as the Hadamard does.  However, both the distillation and gate teleportation involve measurements which are probabilistic. The required circuit changes depending on the measurement outcomes.

Likewise, the corresponding braid cannot be entirely determined ahead of time.
It is possible, however, to shift all of the non-determinism either offline or into logical measurements, which can be performed very efficiently~\cite{Fowler2012g}.  \figref{fig:tgate-time-optimal} shows an alternative circuit that also implements $T$.  In this circuit, an $S$ gate, implemented with the help of a resource state $\ket Y = \frac{1}{\sqrt 2}(\ket 0 + i\ket 1)$, is selectively teleported into the circuit conditioned on the outcome of an $Z$-basis measurement.  Given states $\ket A$ and $\ket Y$, the entire circuit is determined ahead of time except for the measurement bases for selective teleportation.

The circuit in~\figref{fig:tgate-injection-circuit} is smaller than that of~\figref{fig:tgate-time-optimal}.  The latter circuit, however, has the advantage that it can be composed in parallel with any number of additional $T$ gate circuits.  The braid corresponding to the single-qubit unitary $THT$, for example, can be parallelized as shown in~\figref{fig:tgate-time-ordering}. The logical measurements in this braid are implemented differently than previously discussed.  The cap on the defects has been flattened into a wider, but thinner set of defects that looks like a tabletop.  This allows for maximum parallelization of sequences of $T$ gates.

The measurement regions of~\figref{fig:tgate-time-ordering} must obey a relative time ordering.
In particular, the $Z$-basis measurement of the input qubit $\ket\psi$ must be completed before the selective teleportation measurements can be performed.  In addition, the selective teleportation of the \emph{previous} $T$ (if applicable) must be completed before selective teleportation measurements of current $T$ gate can be performed.
In this way, the measurement regions for sequences of $T$ gates form a tree.  Each measurement region must be located strictly later in time than each of its children.

There are a variety of options for preparing the $\ket A$ and $\ket Y$ states required by~\figref{fig:tgate-time-optimal}.  The $\ket A$ state, for example, can be prepared using the $[[15,1,3]]$ state distillation procedure due to Bravyi and Kitaev~\cite{Bravyi2004}, or any of the other proposals presented in~\secref{sec:fault.gates.distillation}.
Efficient surface code braids are known for several of these protocols~\cite{Fowler2012f,Fowler2013}, though we will not discuss the details here.  Rather, for simplicity we abstract the $\ket A$ and $\ket Y$ preparation as rigid cuboids, similar to the Hadamard gate.  This gives us the freedom to define braid compaction algorithms without being coupled to a particular distillation procedure.

\newsavebox{\Smatrix}
\savebox{\Smatrix}{$\left(\begin{smallmatrix}1&0\\0&i\end{smallmatrix}\right)$}

\begin{figure}
\centering
\begin{subfigure}[b]{.4\linewidth}
  \centering
  \includegraphics[width=\linewidth]{images/tgate-injection}
  \caption{\label{fig:tgate-injection-circuit}}
\end{subfigure}
\hfill
\begin{subfigure}[b]{.5\linewidth}
  \centering
  \includegraphics[width=\linewidth]{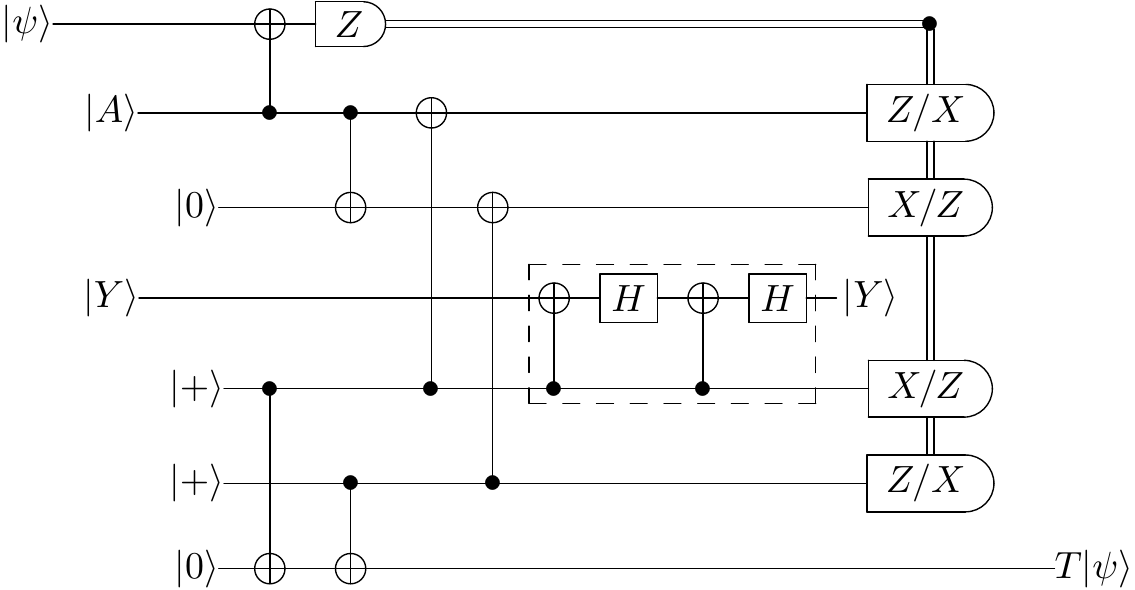}
  \caption{\label{fig:tgate-time-optimal}}
\end{subfigure}
\caption[Surface code $T$ gate.]{
Two circuits that implement the $T$ gate on input state $\ket\psi$. (a) The resource state $\ket A = \ket 0 + e^{i \pi/4}\ket 1$ is constructed by injection and distillation.  Conditioned on the measurement outcome, a corrective 
$S$ rotation may be required, which requires a non-destructive use of an ancilla $\ket Y = \ket 0 + i\ket 1$ state, initially prepared by injection and distillation (not shown).
(b) Instead of performing the conditional $S$ gate directly, selective destination teleportation can be used~\cite{Fowler2012g}.  On one path of the teleportation, the $S$ gate is applied, and on the other path it is not.  The $Z$-basis measurement on $\ket\psi$ determines the bases in which the other four qubits are measured.  The output is $T\ket\psi$, up to Pauli corrections from teleportation.}
\end{figure}

\begin{figure}
\centering
\begin{subfigure}[b]{.25\linewidth}
	\centering
	\includegraphics[width=.9\linewidth]{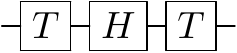}
	\caption{}
\end{subfigure}
\begin{subfigure}[b]{.99\linewidth}
	\centering
	\includegraphics[width=.75\linewidth]{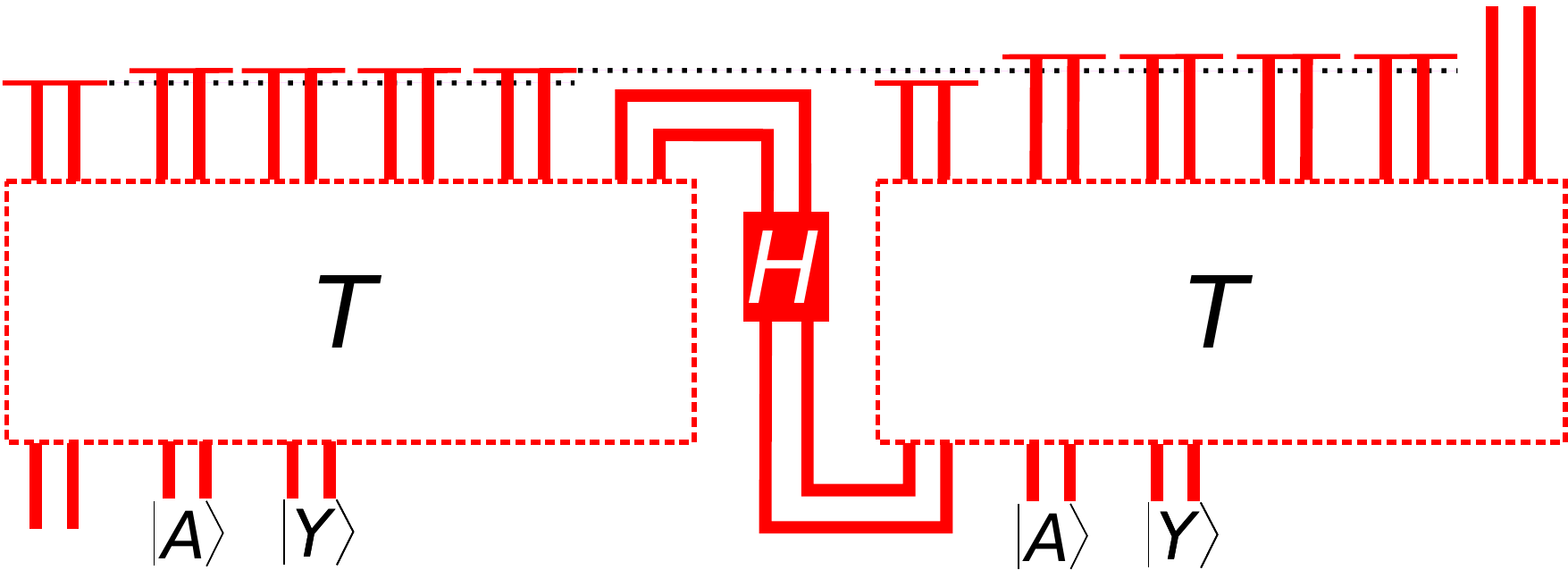}
	\caption{}
\end{subfigure}
\caption[Time ordering of $T$ gates.]{\label{fig:tgate-time-ordering}
(a) A quantum circuit for the single-qubit unitary $THT$ in which time runs left to right. (b) A schematic representation of the corresponding surface code braid in which time runs bottom to top.  For simplicity, the braids corresponding to~\figref{fig:tgate-time-optimal} are shown as boxes, except for the measurements which are shown as thin tabletop structures.  The $T$ and $H$ boxes may be placed in parallel, and the $\ket A$ and $\ket Y$ states may be prepared ahead of time. The first measurement of the $T$ gate must complete before the remaining four selective teleportation measurements can be performed. Selective teleportation measurments between $T$ gates also obey a relative time-ordering as indicated by the black dotted lines.  Any sequence of single-qubit gates from $\{T,H,S\}$ may be parallelized in this way.
}
\end{figure}

The gates listed in~\tabref{tab:canonical-gate-set} are universal for quantum computing.  Thus any quantum circuit can be mapped to a surface code braid by first decomposing it into this gate set, and then sequentially constructing each of the canonical braid elements.

\section{The braid compaction problem
\label{sec:braidpack.compaction}
}

The canonical braid is a fault-tolerant representation of the original circuit, but there is no guarantee that it will fit onto the two-dimensional lattice of qubits that is available.  Indeed, the structure of the canonical braid closely resembles that of the original circuit.  It is essentially a long line of defects that extends out in time.  Even if the braid fits, its two-dimensional shape means that most of the qubits in the quantum computer will be left unused.

Of course, one could try to compile the braid in a different way, so as to use more of the available space.  However, the efficiency of the compilation will depend heavily on the structure of the original circuit.  Qubits that were originally local when arranged linearly might be placed far apart when arranged in two dimensions, thereby increasing the volume required for a CNOT between the two.

We instead choose to optimize the canonical braid by smoothly deforming it.   
So long as the deformations are topological, the optimized braid will be logically equivalent to the original.
Braid compaction, then, is the problem of taking a braid $B$ and converting it into a topologically equivalent braid $B'$ that fits into a smaller bounding volume. 
Alternatively, the problem can be described as follows.
\begin{description}
\item[Braid compaction] Given a braid $B$, code distance $d$, and a rectangular lattice of dimension $A = (x,y)$, find a braid $B'$ that is topologically equivalent to $B$ and such that $B'$ that achieves a minimum code distance of $d$ and is contained in a volume $V = (x,y,z)$ of minimum size.
\end{description}

The $x$ and $y$ dimensions of the bounding volume are are fixed by the size and geometry of the quantum computer. The goal is to efficiently use the provided space in order to minimize computation time.

Abstractly, we can view braid compaction as a process of placing cuboids (Hadamard and $T$ gates) in a large box, subject to certain distance, connectivity and topology constraints.  When viewed in this way, the problem looks strikingly similar to that of VLSI placement~\cite{Schlag1983}.  In the VLSI placement problem, the task is to pack a set of circuit elements---represented by rectangles---on a two-dimensional circuit board of minimum area.  Some of the circuit elements must be connected by wires, and some must be separated from other circuit elements by a minimum distance.

VLSI placement is NP-complete~\cite{Schlag1983}.  Given the close similarities with VLSI placement and with other packing problems, we conjecture that braid compaction is also NP-complete. However, despite their similarities, there are several key differences between VLSI placement and braid compaction.  In particular, the rigid objects in VLSI placement have arbitrary dimension whereas the Hadamard cuboids in the braid are of fixed size.  Thus a naive reduction from VLSI placement to braid compaction is not possible.  Attempts at a more complicated reduction or reduction from related problems such as $3$-Partition and bin packing have so far failed.

\section{A force-directed compaction algorithm
\label{sec:braidpack.force-directed}
}

We now describe our force-directed algorithm, the first of two proposed algorithms for braid compaction. The algorithm employs two complementary ``forces''.  A gravity force acts to pull the braid down toward the bottom of the space-time grid, thereby reducing computation time.  Meanwhile, a tension force prevents the braid from becoming too large and impeding the progress of gravity.

\subsection{Braid representation}
\label{sec:forcing-braid-model}
For our force-directed algorithm, the braid is modeled as a set of plumbing pieces (i.e., pipes) placed on a three-dimensional grid.  For circuits containing preparation, measurement, single-qubit Paulis and CNOT gates, only four types of pipes are required: straight and bent (elbow shaped) pipes, both primal and dual.  See~\figref{fig:pipes}. The braid is then constructed by connecting pipes into interlocking loops.  Junctions can also be supported by merging two or more pipes.

The three-dimensional ($l\times w \times h$) grid is partitioned into $4 \times 4 \times 4$ cells, each of which contain at most one primal pipe and one dual pipe.  Each pipe connects to at least two of the faces of the cell. For each face there is a designated unit cube to which a pipe can connect.  For example, a primal pipe that connects to the $-y$ face must always connect at position $(1,0,2)$ within the cell. Including the empty pipe, there are $2^6=64$ possible primal pipes and $64$ possible dual pipes, for a total of $4096$ possible cell configurations. See~\figref{fig:cells}.

The structure of the cell enforces a minimum distance of a single unit cube between defects of opposite type and a  distance of three unit cubes between distinct defects of the same type.  Thus, if the length of a unit cube is $\delta$, the resulting surface code distance is $d=3\delta$.  A unit cube contains $2\delta$ physical qubits per side (including qubits for stabilizer measurement), so that a single time-slice of a cell contains $64 \delta^2$ qubits. 

\begin{figure}
\centering
\begin{subfigure}[b]{.24\linewidth}
	\centering
	\includegraphics[width=.6\linewidth]{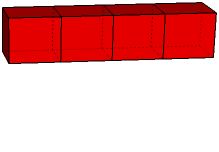}
	\caption{straight primal}
\end{subfigure}
\hfill
\begin{subfigure}[b]{.24\linewidth}
	\centering
	\includegraphics[width=.6\linewidth]{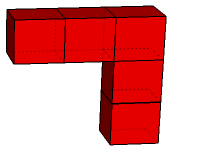}
	\caption{bent primal}
\end{subfigure}
\hfill
\begin{subfigure}[b]{.24\linewidth}
	\centering
	\includegraphics[width=.6\linewidth]{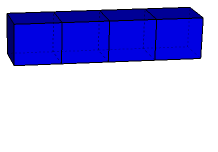}
	\caption{straight dual}
\end{subfigure}
\hfill
\begin{subfigure}[b]{.24\linewidth}
	\centering
	\includegraphics[width=.6\linewidth]{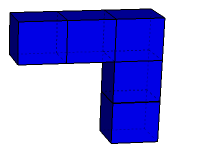}
	\caption{bent dual}
\end{subfigure}
\caption[Primative plumbing pieces.]{\label{fig:pipes}
In the force-directed algorithm, braids are constructed by rotating and connecting the four primative ``plumbing'' pieces shown above.}
\end{figure}

\begin{figure}
\centering
\begin{subfigure}[b]{.1\linewidth}
  \includegraphics[scale=.4]{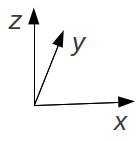}
\end{subfigure}
\begin{subfigure}[b]{.3\linewidth}
  \includegraphics[scale=.2]{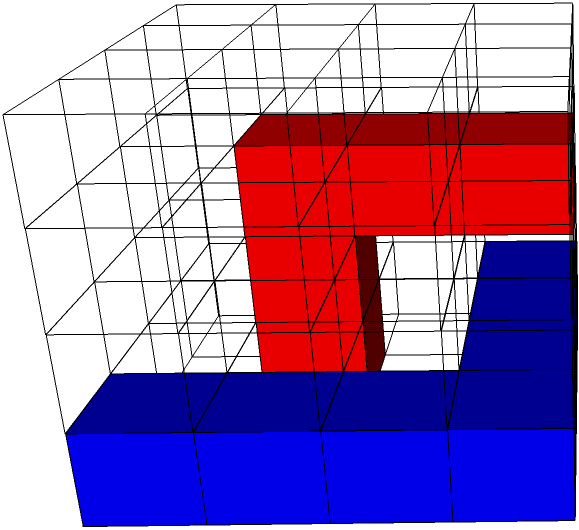}
  \caption{}
\end{subfigure}
\begin{subfigure}[b]{.3\linewidth}
  \includegraphics[scale=.2]{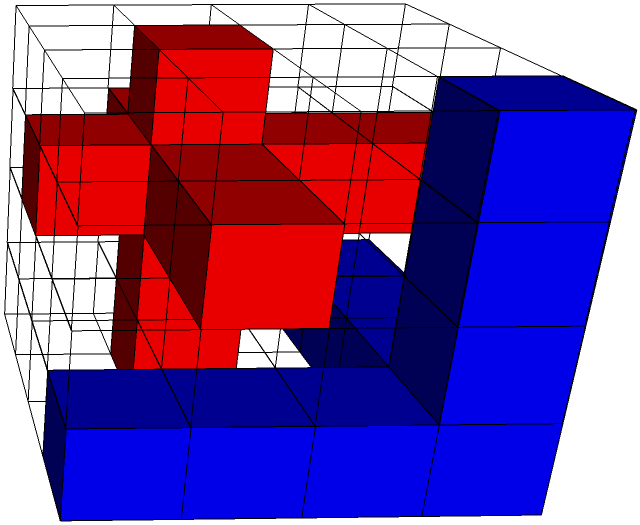}
  \caption{}
\end{subfigure}
\caption[Surface code cells.]{\label{fig:cells}
(a) An example of a $4$ x $4$ x $4$ cell containing both a primal pipe and a dual pipe. The primal pipe connects to the southern face ($-z$) and the eastern face ($+x$).  The dual pipe connects to the the western face ($-x$) and to the far face ($+y$). (b) All possible pipes superimposed on a single cell.  Primal and dual defects are always separated by at least one unit cube.  Neighboring unconnected defects of the same type are always separated by at least three unit cubes.
}
\end{figure}
  
Regions such as Hadamards, and state distillation and tabletop measurement for $T$ gates cannot be represented as a collection of conventional plumbing pieces.
Instead, they are represented by a volume of special purpose pipes which collectively are treated as a contiguous region.  These pipes are much like regular pipes, except that they consume an arbitrary region of the $4\times 4\times 4$ cell.

\subsection{Braid synthesis}
As defined, the braid compaction problem takes an arbitrary braid as input.  Thus our algorithm need not address the synthesis of a quantum circuit into a braid.  Indeed, the force-directed braid model requires only that rigid collections of pipes (i.e., cuboids) be specified along with rotation and time-ordering constraints.

For concreteness, however, we will assume that the initial braid is constructed from a quantum circuit in the canonical way as described in~\secref{sec:braidpack.surface-code}. That is, qubits are represented by pairs of primal defects.  Single-qubit preparation corresponds to two bent pipes connected to form a ``U'' shape and single-qubit measurement is the same, except that the U shape is upside-down.  Hadamards, and $T$ gates are abstracted as cuboids of fixed dimension. CNOT gates are constructed by wrapping a dual loop around corresponding primal loops.
%The resulting initial braid is effectively a two-dimensional object that mirrors the structure of the input circuit.  It is the starting point for our compaction algorithm.

The Hadamard cuboid is three cells wide, four cells deep and four cells high. This cuboid is larger than is strictly necessary to enclose the Hadamard operation. Part of the Hadamard operation involves cutting a boundary around the corresponding logical qubit.  The volume given above provides enough room for the Hadamard operation to take place inside boundary, while enforcing that defects outside of the boundary are a safe distance away. Affixed to opposite faces of the cuboid are pairs of straight pipes representing the input and output logical qubit.

The specifics of the $T$-gate braid depend on the distillation protocol and on the desired gate accuracy, but otherwise follow~\figref{fig:tgate-time-optimal}.  Our compaction algorithm is flexible enough to allow any type of distillation scheme.  For simplicity, we will assume the existence of two cuboid regions for each $T$ gate, one for $\ket A$ and one for $\ket Y$.  Straight pipes representing the output are affixed to the top of each cuboid.

\subsection{Gravity}
The primary ``force'' in the algorithm is a vector field that loosely resembles physical gravity acting on the braid.  With each cell in the grid, we associate two vectors of the form $(a,m)$, specified by an axis $a \in \{x, y, z\}$ and a magnitude $m \in \mathbb Z$.  The first vector represents a force on the primal pipe contained in the cell, and the second vector represents a force on the dual pipe.

There are a number of reasonable ways to initialize and update the gravity field as defects are moved around.  The simplest strategy is to assign a fixed, negative magnitude to each spacetime point and align the vector along the $z$-axis so that the force always points downward.
In order that defects may slide past each other, though, we allow vectors to point sideways along the $x$ and $y$ axes, as well. See~\figref{fig:gravity-example}.  Roughly, gravity vectors are assigned to point to the closest cell from which the defect may then move downward.  For example, a primal pipe occupying cell $(x,y,z)$ may be blocked by a dual pipe in cell $(x,y,z-1)$.  If, however, cells $(x+1,y,z)$ and $(x+1,y,z-1)$ are empty, then the primal gravity vector for cell $(x,y,z)$ is assigned to point along the positive $x$-axis.

% Gravity magnitudes are assigned based on the height of the cell and the direction of the vector.  If the maximum height of the grid is $h$ then a gravity vector pointing along the negative $z$-axis is assigned a magnitude of $z-h+1$.  Similarly, a vector pointing along the $x$- or $y$-axis is assigned a magnitude (in absolute value) of $h-z$.  That is, the absolute value of gravity magnitude increases linearly with height.  Gravity vectors do not point in the positive $z$ direction.

\begin{figure}
\centering
\includegraphics[scale=.2]{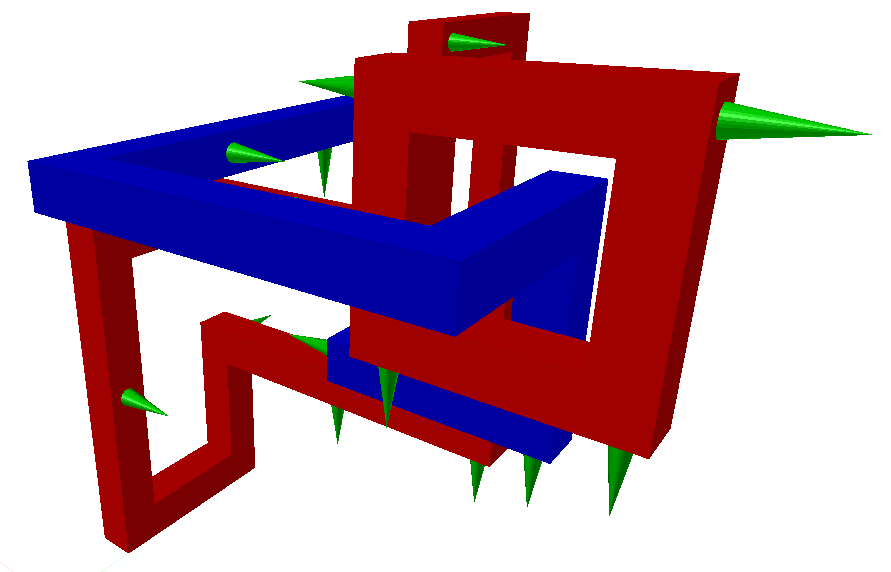}
\caption[An example of gravity forces.]{\label{fig:gravity-example}
Gravity vectors (shown as green cones) generally point downward, but may point in any direction.}
\end{figure}

\subsection{Tension}
The gravity force, while effective at directing pipes toward the bottom of the grid, has the effect of stretching strings and loops, thus increasing the length of the braid.  This happens, for example, when a loop is pulled by gravity in one direction but a small segment of the loop is prevented from moving because other defects are in the way.
When a string or loop becomes very long, it may take up space that could otherwise be occupied by other parts of the braid.  To prevent this behavior we implement a tension force which acts to reduce the length of a string.

Tension is applied to each string of defects independently.  For each pipe in the string, there is a force pulling in the direction of the input face and a force pulling in the direction of the output face.  For example, a pipe connected to the $-x$ and $+z$ faces will experience a force in the $-x$ and $+z$ directions.  The magnitude of the force is proportional to the length of the string, just as for a physical spring.

This choice of tension forces means that the inward and outward forces cancel for straight pipes.  Bent pipes, however, feel an inward force toward the rest of the string.  This inward force tends to decrease the curvature of the string, thereby reducing its length. See~\figref{fig:tension}.

\begin{figure}
\centering
\begin{subfigure}[b]{.45\linewidth}
  \centering
  \includegraphics[scale=.4]{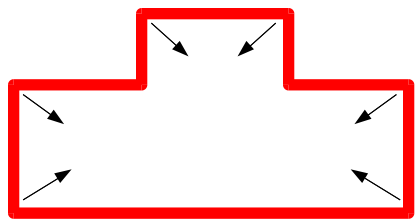}
  \caption{before}
\end{subfigure}
\begin{subfigure}[b]{.2\linewidth}
  \centering
  \includegraphics[scale=.4]{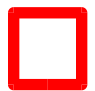}
  \caption{after}
\end{subfigure}
\caption[An example of the tension force.]{\label{fig:tension}
The tension force pulls inward on each of the corners of the loop. The result is a smaller rectangular loop.
}
\end{figure}

Tension forces also act on cuboids. Each of the pipes connected to a cuboid exerts a force that pulls in the direction of the pipe.  Again, the force is proportional to the length of the string to which each connecting pipe belongs.

\subsection{Compaction}
The braid is initially placed above the three-dimensional grid.  Since the braid may be wider than the grid dimensions, a funnel is placed on top of the grid.  This allows the braid to slowly deform according to the geometry of the lattice.
Compaction then proceeds by iterating through each of the cuboids and strings.  Cuboids are translated or rotated as a single rigid object. Other regions of pipes form strings which either connect to cuboids or form loops. Strings are treated as flexible objects in which each pipe can be translated independently.

Associated with each pipe is a velocity vector.  Each pipe in a string is moved by first taking the initial velocity vector and updating it according to the gravity and tension forces at that location.  The pipe is then translated according to the direction and magnitude of the new velocity vector.  During the move, additional pipes may be added or removed in order to maintain connectivity of the string.
  
To translate a cuboid, the total velocity is calculated by summing each of the individual velocity vectors.  Similarly, the gravity and tension forces are calculated by summing the force vectors associated with each pipe.  The cuboid velocity is then updated by dividing the total force by the number of pipes (each pipe is assumed to have the same mass) and then adding to the existing velocity.  Finally, the cuboid is translated according to the direction and magnitude of the velocity vector.

In the case of tabletop measurement translations along the $z$-axis, we must preserve the partial ordering.  When translating a tabletop $m$ along the $-z$-axis we must check the height of the other measurements on which $m$ depends.  Likewise, when translating $m$ along the $+z$-axis, we must check the height of the measurement that depends on $m$.

Cuboid rotations are performed similarly by calculating a rotational velocity according to the moments of each pipe and the torque due to gravity and tension forces.  Rotation about a given axis is performed only if rotation is allowed and the magnitude of the angular velocity is large enough to induce a rotation of $\pi/2$.

Of course, all of the moves performed during compaction must maintain the braid topology.  In particular, we do not allow pipes to intersect nor do we allow a pipe of one type to pass through a pipe of opposite type in order to arrive at its destination. Though we do allow defects of the same type to pass through each other since this does not change the computation. It is possible for the translational or rotational path of a group of pipes to be blocked by other pipes.  When this happens, we say that a collision has occurred.

Collisions are resolved by first calculating the velocity and mass of each of the two objects involved.  In the case that a cuboid collides into multiple pipes, the impeding pipes are treated as a collective object.  The velocities of the two objects are then recalculated according to the equations of motion for a partially inelastic collision.
In this way, distinct parts of the braid are able to communicate with each other.  For example, large objects may shift smaller objects out of the way and linked loops may tug on each other. However, the rules for moving each pipe are still entirely local and relatively simple.

Note that a collision can also occur between time-dependent measurements even when the two cuboids are not located nearby each other.  Such a collision happens if the vertical motion of one of the measurements would cause a violation of relative time-ordering constraints.  The collision is non-local, but can be efficiently identified and resolved by maintaining a dependency tree with the location of measurement.

Since the topology of the braid is preserved at each step, compaction can be terminated at any time. Indeed, there are a number of reasonable termination conditions.  Compaction can be stopped after a fixed number of iterations, or a fixed amount of time.  It can also be stopped when all of the pipes are located below a particular height, or as soon as all of the pipes fit within the dimensions of the lattice.  The termination condition could also be more complicated.  For example, compaction could be halted if the maximum height remains unchanged for a certain fixed number of iterations.

\subsection{Performance and scalability}
\label{sec:performance-scalability}
The complexity of a single compaction iteration scales as the size of the braid. The size of a canonical braid is $O(nm)$ where $n$ is the number of qubits and $m$ is the number of gates in the input circuit.  The number of iterations required to obtain good compaction results depends on the ratio of the lattice area---i.e., the $x$-$y$ plane---to the braid size.  In the case that the lattice area is large compared to the braid size, it seems reasonable to expect the braid to flatten in time proportional to the height of the canonical braid.  If the canonical braid has area large compared to the lattice, then $O(nm)$ iterations may be required in order to funnel then entire braid into the proper bounding box.

For small circuit sizes, a runtime of $O(n^2m^2)$ is reasonable.  But for large circuits consisting of thousands of qubits and possibly millions or billions of gates, we require a better strategy. Indeed, we cannot hope to globally optimize braids for large-scale problem sizes.  Instead, the circuit is partitioned into subcircuits of manageable size and the braid is synthesized and compacted hierarchically.  Just as we treat single-qubit Hadamards as atomic cuboids of fixed size, we may consider sub-braids as fixed size cuboids.

Each sub-braid is represented as a tangle of defects in which some defects are anchored to grid boundaries.  Subject to the anchoring constraints, the sub-braid is compacted as normal.  Once its compacted size is determined, the sub-braid is then treated as a black-box in the larger braid.
If two sub-braids contain measurements that are time-ordered, then the sub-braids must also be time ordered.  But again, this is no different than time ordering restrictions on tabletop measurements in the original model.

We anticipate that the best partitioning strategy will be one that reflects the structure of the input circuit.  Reasonable representations of large input circuits will be hierarchical and it should be possible to mimic this hierarchy for large-scale compaction.
This technique will be particularly useful for highly repetitive circuits.  Repeated sub-circuits can be synthesized and compacted once, and then duplicated in the larger braid. 

\subsection{Implementation and results}
We have implemented the force-directed compaction algorithm in C++ as a tool called Braidpack.  Braidpack takes, as input, a representation of a circuit along with physical space restrictions.  It produces, as output, a compact logically equivalent surface code braid.  

The current implementation is not fully functional, but is capable of synthesizing and compacting arbitrary circuits of CNOT gates, including qubit preparation and measurement. \figref{fig:cnot-compaction} shows the result of compaction on a single CNOT gate.  The tension force first contracts the primal loop on the right-hand-side.  Then gravity flattens the braid.  Compaction in this example was done without implementing collisions between pipes.  As a result, tension is unable to fully contract the dual loop.  With a more complete implementation of the algorithm, we expect the braid to fully flatten and contract.  

\figref{fig:y-distillation-compaction} shows the same prototype implementation of Braidpack for a circuit composed of eleven CNOT gates.  For simplicity of implementation, the qubit preparation and measurements in the canonical braid are arranged in a staircase fashion.  Ignoring the staircases, the canonical braid has a bounding box of size ($3 \times 16 \times 34$), whereas the the compacted braid fits in a bounding box of size ($10 \times 13 \times 6$), a factor of four improvement along the time axis.  Again, we expect improved results with a more complete implementation of Braidpack.

\begin{figure}
\centering
\hfill
\begin{subfigure}[b]{.49\linewidth}
  \centering
  \includegraphics[scale=.25]{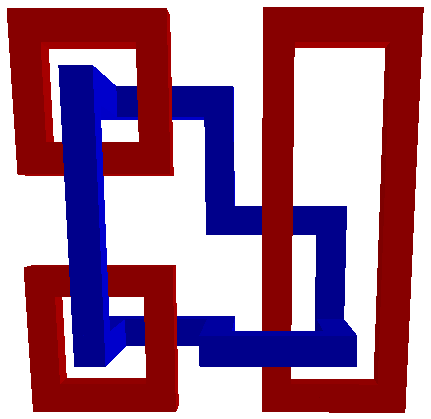}
  \caption{}
\end{subfigure}
\hfill
\begin{subfigure}[b]{.49\linewidth}
  \centering
  \includegraphics[scale=.25]{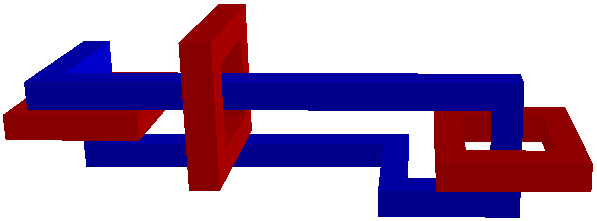}
  \caption{}
\end{subfigure}
\caption[Compaction of a CNOT gate.]{\label{fig:cnot-compaction}
Compaction of a single CNOT using a prototype of the force-directed algorithm. (a) A canonical CNOT braid is initially arranged vertically. (b) After compaction, the braid has been almost completely flattened.}
\end{figure}  

\begin{figure}
\centering
% \begin{subfigure}[b]{.3\linewidth}
%   \includegraphics[width=\linewidth,angle=90]{images/y-distillation-circuit}
%   \caption{}
% \end{subfigure}
\hfill
\begin{subfigure}[b]{.45\linewidth}
  \centering
  \includegraphics[width=.6\linewidth]{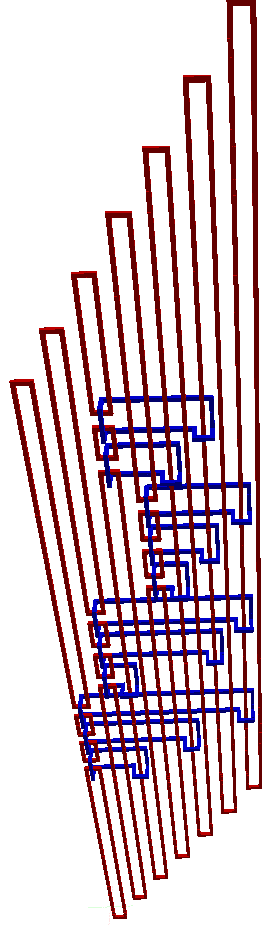}
  \caption{}
\end{subfigure}
\hfill
\begin{subfigure}[b]{.45\linewidth}
  \centering
  \includegraphics[width=.6\linewidth]{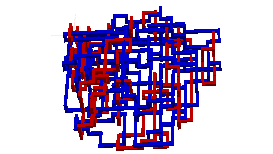}
  \caption{}
\end{subfigure}
\caption[Compaction of eleven CNOT gates.]{\label{fig:y-distillation-compaction}
Compaction of eleven CNOT gates with a prototype implementation of the force-directed algorithm.
The canonical braid (a) is compressed into a smaller but topologically equivalent braid (b).}
\end{figure}  

In order to facilitate debugging, we have developed a braid visualization tool called Braidview.  This tool creates a single file from a braid or sequence of braids.  The file can be viewed in Blender~\cite{Blender}, a third-party open-source $3$-D modeling application.  Braidview is capable of separately rendering primal and dual defects (as in~\figref{fig:y-distillation-compaction}), as well as gravity vectors (see~\figref{fig:gravity-example}). The backbone of Braidview is a set of rendering functions that use the Blender Python-API.  These, and some additional functions, are used by similar visualization tools Nestcheck and Autotune~\cite{Milburn2012,Fowler2012a}.

\section{Compaction by simulated annealing
\label{sec:braidpack.annealing}
}

In this section we describe our second compaction algorithm, which is based on simulated annealing. Simulated annealing is a general optimization technique that has been applied to a wide variety of problems.  The main idea is to explore the solution space by hopping randomly from the current solution to a nearby solution. Hops that result in an improved solution are kept.  In order to avoid local minima, hops that result in a less desirable solution are also kept with some non-zero probability, thus permitting broader exploration of the set of possible solutions.

Our simulated annealing algorithm is based largely on a procedure used for VLSI placement \cite{Hsieh}.  In the VLSI algorithm, circuit elements and wires are represented by rectangles.  Size, distance and connectivity constraints are given by linear inequalities on the coordinates of each rectangle.  Rectangles can be shifted around by swapping linear constraints.  The idea for braids is similar.  Defects are represented by cuboids. Size, distance and topology constraints are given by linear inequalities which can be swapped to perform topological deformation.

\subsection{Definition of the braid}
In the force-directed algorithm, the braid was modeled as a connected configuration of plumbing pieces.  Some collections of pipes formed rigid cuboids.  Other collections of pipes formed flexible strings and loops.
For simulated annealing, we take a different approach.  Each cuboid is represented by a pair of points $(p, p')$ in the three-dimensional lattice.  Point $p$ specifies the point closest to the origin (lower-left corner) and $p'$ specifies the point furthest from the origin (upper-right corner).
Defect strings and loops are also represented by cuboids.  A string of defects is given by a set of overlapping cuboids of arbitrary dimension.  By connecting cuboids it is possible to construct any desired loop or string.

Thus the entire braid is specified by a set of cuboids.
A \emph{layout} of $n$ cuboids is defined by $2n$ three-dimensional integer coordinates.  The $x$, $y$, $z$ dimensions of the layout are defined by the maximum $x$, $y$, and $z$ coordinates respectively.
The layout must satisfy a set of constraints which we group into the following types:
\begin{enumerate}
  \item size constraints,
  \item time-ordering constraints,
  \item minimum distance constraints,
  \item jog constraints,
  \item connectivity constraints and
  \item topological constraints.
\end{enumerate}
Except for the topological constraints, all of the constraints can be directly expressed as sets of linear inequalities.

\subsubsection{Size constraints}
\label{sec:size-contraints}
Minimum dimension constraints of a cuboid are specified by a triple $(\delta_x, \delta_y, \delta_z)$ of non-negative real values and three linear inequalities:
\begin{equation}\begin{split}
\label{eq:cuboid-dimension-constraints}
  x + \delta_x &\leq x' \\
  y + \delta_y &\leq y' \\
  z + \delta_z &\leq z'
  \enspace.
\end{split}
\end{equation}
For string cuboids (those that are not $H$ gates or table-like measurements), $\delta_x = \delta_y = \delta_z = d/4$, where $d$ is the code distance. Hadamard and $T$ gates may be rotated $90$ degrees about the $z$-axis. 
Each gate can take on one of four different rotations $\{0, \pi/2, \pi, -\pi/2\}$.
Rotations $0$ and $\pi$ correspond to the set of constraints given by~\eqnref{eq:cuboid-dimension-constraints}.  
Rotations $\pm \pi/2$ correspond to the same set of constraints in which $\delta_x$ and $\delta_y$ have been exchanged.

We therefore assign one of two sets of constraints to each $H$ and $T$ gate, either the constraints of~\eqnref{eq:cuboid-dimension-constraints} or the permuted version. The corresponding cuboids must satisfy all constraints from at least one of sets.

\subsubsection{Time-ordering constraints}
The non-deterministic implementation of $T$ gates in the surface code induces a partial time-ordering of tabletop measurement regions. As discussed in~\secref{sec:braidpack.surface-code}, this partial ordering requires that, for certain pairs, one tabletop measurement must be located above another tabletop measurement.
The time-ordering constraint for two dependent measurements, $a$, $b$ is given by,
\begin{equation}
  z'_a + 1 \leq z_b
  \enspace .
\label{eq:time-ordering-constraint}
\end{equation}

\subsubsection{Minimum distance constraints}
\label{sec:min-distance-constraints}
Like the size constraints, minimum distances are proportional to $d$, the distance of the code. With a few exceptions (see \secref{sec:jog-nodes} and \secref{sec:connectivity-constraints}), primal defect cuboids must be at least a distance $d$ away from other primal defects. Similarly, dual defect cuboids must be $d$ away from other dual cuboids.  Cuboids of opposite type must be at least $d/4$ apart.

If two cuboids $r_i, r_j$ must be separated by $\delta$, then at least one of the following constraints must be satisfied:
% \begin{equation}
% \begin{tabular}{ccc}
%   $x'_i - x_j \geq \delta$ & $x_i - x'_j + y_i - y'_j \geq \delta$ & $x_j - x'_i + y_j - y'_i + z_j - z'_i \geq \delta$  \\
%   $x'_j - x_i \geq \delta$ & $x_j - x'_i + y'_i - y'_j \geq \delta$ & $x_i - x'_j + y_j - y'_i + z_j - z'_i \geq \delta$ \\
%   $y'_i - y_j \geq \delta$ & $x_j - x'_i + y_j - y'_i \geq \delta$ & $x_i - x'_j + y_i - y'_j + z_j - z'_i \geq \delta$ \\
%   $y'_j - y_i \geq \delta$ & $x_i - x'_j + y_j - y'_i \geq \delta$ & $x_j - x'_i + y_i - y'_j + z_j - z'_i \geq \delta$ \\
%   $z'_i -z_j \geq \delta$ & $x_i - x'_j + z_i - z'_j \geq \delta$ & $x_j - x'_i + y_j - y'_i + z_i - z'_j \geq \delta$ \\
%   $z'_j - z_i \geq \delta$ & $x_j - x'_i + z_i - z'_j \geq \delta$ & $x_i - x'_j + y_j - y'_i + z_i - z'_j \geq \delta$ \\
%   & $x_j - x'_i + z_j - z'_i \geq \delta$ & $x_i - x'_j + y_i - y'_j + z_i - z'_j \geq \delta$\\
%   & $x_i - x'_j + z_j - z'_i \geq \delta$ & $x_j - x'_i + y_i - y'_j + z_i - z'_j \geq \delta$ \\
%   & $y_i - y'_j + z_i - z'_j \geq \delta$ & \\
%   & $y_j - y'_i + z_i - z'_j \geq \delta$ & \\
%   & $y_j - y'_i + z_j - z'_i \geq \delta$ & \\
%   & $y_i - y'_j + z_j - z'_i \geq \delta$ &
% \end{tabular}
% \end{equation}
\begin{equation}
\begin{tabular}{ll}
  $x'_i +\delta \leq x_j$ & $x'_j + \delta \leq x_i$ \\
  $y'_i +\delta \leq y_j$ & $y'_j + \delta \leq y_i$ \\
  $z'_i +\delta \leq z_j$ & $z'_j + \delta \leq z_i$ \\
\end{tabular}
\label{eq:min-distance-constraints}
\end{equation}
Each constraint corresponds to a different relative arrangement of the two cuboids.  The $x'_i +\delta \leq x_j$ constraint, for example, enforces that $r_i$ is placed to the left of $r_j$. Whereas $z'_i +\delta \leq z_j$ requires that $r_i$ be placed below $r_j$.

\subsubsection{Jog nodes}
\label{sec:jog-nodes}
A fixed string of defects may be represented by a set of overlapping cuboids each of which has a fixed orientation along one of the three axes.  However, in order to accommodate topological deformation we require a representation that allows for flexible strings of cuboids. This is analogous to a VLSI instance in which an arbitrary number of jogs are allowed in each wire.  To fulfill this requirement, we introduce an object called a \emph{jog node}.

A jog node is a set of six cuboids, each of which has a particular orientation axis.  The first cuboid is oriented along the $+x$ axis, the second along the $+y$ axis, and the third along the $+z$ axis.  The fourth, fifth and sixth cuboids are oriented along the $-x$, $-y$ and $-z$ axes, respectively.  Each cuboid in the jog node is allowed to expand along its corresponding axis.  Adjacent cuboids are required to overlap so that the entire jog node forms a continuous path.  The constraints for a jog node are given by:
\begin{equation}
\begin{tabular}{cccccc}
  $x_1 \leq x_2$, &$y_1 = y_2$, &$z_1 = z_2$, &$x_1' = x_2'$, &$y_1' \leq y_2'$, &$z_1' = z_2'$,\\
  $x_2 = x_3$, &$y_2 \leq y_3$, &$z_2 = z_3$, &$x_2' = x_3'$, &$y_2' = y_3'$, &$z_2' \leq z_3'$,\\
  $x_3 \geq x_4$, &$y_3 = y_4$, &$z_3 \geq z_4$, &$x_3' = x_4'$, &$y_3' = y_4'$, &$z_3' = z_4'$,\\
  $x_4 = x_5$, &$y_4 \geq y_5$, &$z_4 = z_5$, &$x_4' \geq x_5'$, &$y_4' = y_5'$, &$z_4' = z_5'$,\\
  $x_5 = x_6$, &$y_5 = y_6$, &$z_5 \geq z_6$, &$x_5'=x_6'$, &$y_5' \geq y_6'$, &$z_5' = z_6'$.
\end{tabular}
\label{eq:jog-node-constraints}
\end{equation}

It possible to connect two jog nodes at their endpoints.  Given the sixth cuboid $a6$ of jog node $a$ and the first cuboid $b1$ of jog node $b$ the endpoints are connected by requiring
\begin{equation}
  x_{a6} = x_{b1}, y_{a6} = y_{b1}, z_{a6} = z_{b1}
  \enspace .
\end{equation}
In this way, jog nodes can be connected to form an arbitrary defect path of any length.  It is possible to form both loops and open ended strings.

\def\imagetop#1{\vtop{\null\hbox{#1}}}
\begin{figure}
 \centering
\begin{tabular}{cccc}
  &
  \includegraphics[scale=.25]{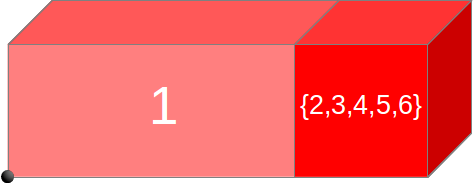} &
  \includegraphics[scale=.25]{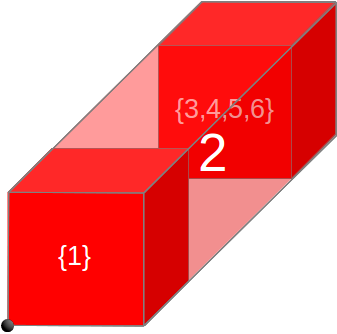} &
  \includegraphics[scale=.25]{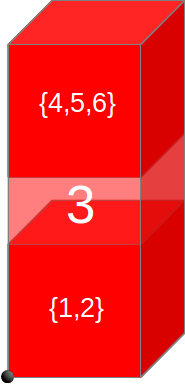}\\
  \raisebox{-2.5cm}{\includegraphics[scale=.5]{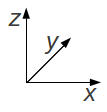}} &
  \imagetop{\includegraphics[scale=.25]{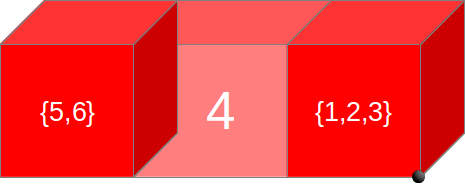}} &
  \imagetop{\includegraphics[scale=.25]{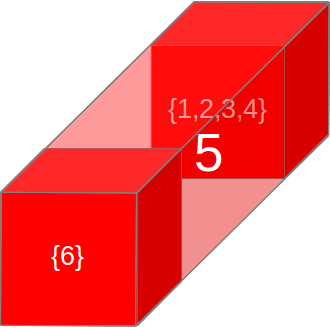}} &
  \imagetop{\includegraphics[scale=.25]{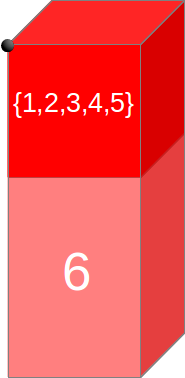}}
\end{tabular}
 \caption[A jog node.]{\label{fig:jog-node-loop}
   A jog node consists of six overlapping cuboids.  Each cuboid is allowed to extend in only one direction, and only one cuboid in the node may be extended.  The six possible jog node configurations are shown above.  The node origin is indicated by a black dot, where visible.
}
\end{figure}

The jog node constraints, as stated, conflict with the minimum distance constraints in~\secref{sec:min-distance-constraints}.  For example, cuboids $a1$ and $a2$ are required by~\eqnref{eq:jog-node-constraints} to be connected, but are required by~\eqnref{eq:min-distance-constraints} to be separated by $\delta$. As a workaround, we first require that each jog node be oriented along at most one axis.
This is accomplished by changing the appropriate inequality constraints to equality constraints.  For example, to force an orientation along the $+x$ axis only, leave the $x_1 \leq x_2$ constraint alone and change all of the other inequalities to equalities.
Then the cuboid corresponding to the $+x$ axis can be of arbitrary size (subject to minimum dimension constraints) and all other cuboids of the node must fit inside of it.
See~\figref{fig:jog-node-loop}.

Next, remove the minimum distance constraints for all jog node cuboids except those that correspond to the orientation axis.  Finally, remove minimum distance constraints between cuboids in adjacent jog nodes. Now, overlapping cuboids within the same jog node or between connected jog nodes are consistent with all other constraints.

A jog node may also be configured to take no orientation.  In this case, all cuboids in the node are constrained to be of minimum size, i.e., $x + \delta_x = x'$, $y + \delta_y = y'$, $z + \delta_z = z'$.  Furthermore, all minimum distance constraints involving the node are removed.  This type of node will either be unconnected to any other node (in which case it can be removed), or it will be contained entirely within another jog node.  In either case, its distance from other objects in the braid is unimportant.

\subsubsection{Connectivity constraints}
\label{sec:connectivity-constraints}
Jog nodes allow for arbitrary defect paths and loops.  We must also define how jog nodes are used to connect to cuboids such as Hadamard gates and state distillation.
Each gate cuboid contains some number of ports to which string defect cuboids are allowed to attach. The locations of the ports are fixed relative to the gate.  However, since gates can be rotated, the constraints that describe the connection must correspond to the permutation of the dimensional constraints from~\secref{sec:size-contraints}.

A port is a rectangle defined by two coordinates on the surface of the gate.  A jog node is connected to a port by requiring that certain coordinates of the jog node cuboid match the coordinates of the port. For example, if the input port $(x, y, z)$, $(x', y', z)$ is located on the top of the gate, then the jog node connection constraint is given by
\begin{equation}
  x_3 = x, y_3 = y, x'_3 = x', y'_3 = y', z_3 = z
  \enspace .
\end{equation}
See~\figref{fig:gate-with-ports}.

\begin{figure}
\centering
\includegraphics[height=4.5cm]{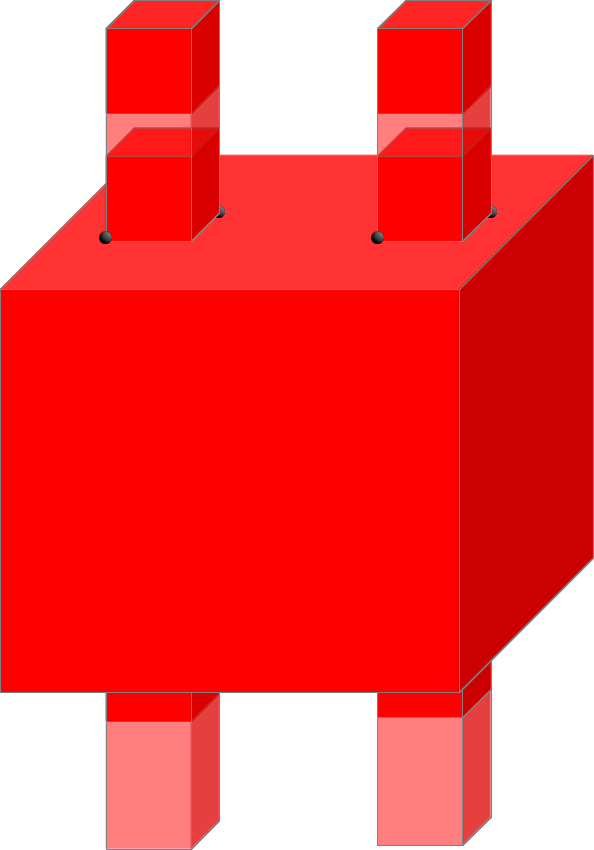}
\caption[A cuboid with ports.]{\label{fig:gate-with-ports}
The above cuboid has four ports defined on its surface, two on top and two on the bottom.  Jog nodes are affixed to the points that define each port.
}
\end{figure}

To maintain consistency, the minimum distance constraints between the gate and the connecting jog node must be eliminated.  Note that it is still possible for two connected gates to achieve a separation of exactly $d$.  In this case, the node connected to the output port of the first gate is also connected to the input port of the second gate, and vice versa.  But since each node is of minimum size, the minimum distance constraints between the node and the gates do not apply (see~\secref{sec:jog-nodes}).

\subsubsection{Topological constraints}
Finally we address the topological constraints.  Informally, these constraints enforce the linking between loops.  Links between loops of the same type are trivial and need not be constrained.  However, certain linking properties between loops of different types must be maintained. In particular, it is sufficient to consider the linking number for each primal-dual loop pair. For each primal-dual pair $(l_p, l_d)$ we have the following constraint
\begin{equation}
 l_{pd} = L_{pd} \mod 2
\end{equation}
where $l_{pd}$ is the linking number of loops $l_p$ and $l_d$ and $L_{pd} \in \{0,1\}$ is an input parameter.

There is a simple linear-time algorithm to compute the linking number between two loops (see, e.g., \cite{Kauffman2001}).  However, in order to efficiently compute the cost function of a layout, we will require that all constraints be linear. See~\secref{sec:annealing-algorithm}.

We impose linear topology constraints separately for loop pairs with odd linking number (i.e., loops that are linked) and loop pairs with even linking number (loops that are not linked).  First consider two loops with odd linking number.  One of the loops consists of primal defects and the other loop consists of dual defects.  To the primal loop, attach a new primal cuboid which we will call a \emph{linking node}.  The linking node has dimension $(5d/4\times d/4\times 5d/4)$. It is attached to the primal loop by connecting one of the jog nodes to the top and connecting an adjacent jog node to the bottom.

The linking node is also attached to jog nodes of the dual loop.  Instead of connecting on the top, the dual jog nodes are connected on either side of the linking node.  The dimensions of the linking node are about twice as large as would otherwise be necessary for maintaining minimum distance constraints between the primal and dual cuboids.  The extra space is used as a placeholder.

As the simulated annealing algorithm proceeds, the linking number between the two loops may change.  The jog nodes that were originally connected to the linking node must remain connected.  But other cuboids from the loops are unrestricted and may cross each other.  At the end of the algorithm the linking node is removed leaving some empty space.  

The primal and dual loops must now be reconnected.  However, we have a choice.  We may either connect the dual loop so that it is inside of the primal loop.  Or we may connect the dual loop so that it is outside of the primal loop.  In effect, the choice of reconnection determines whether the linking number is even or odd.  We may simply choose the configuration that yields an odd linking number.
See~\figref{fig:linking-node}.

\begin{figure}
\centering
\begin{subfigure}[b]{.25\linewidth}
  \includegraphics[width=.7\linewidth]{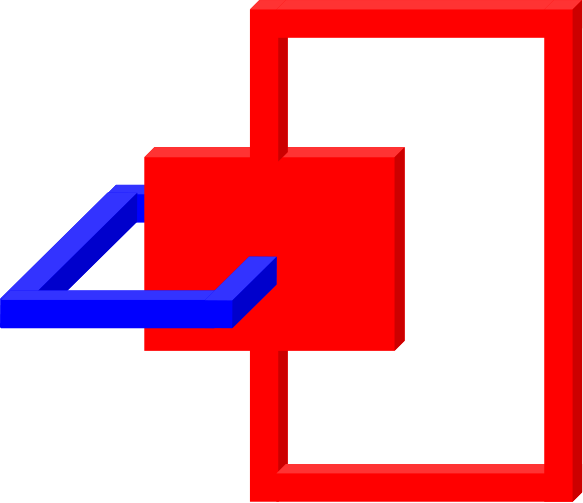}
  \caption{}
\end{subfigure}
\begin{subfigure}[b]{.25\linewidth}
  \includegraphics[width=.7\linewidth]{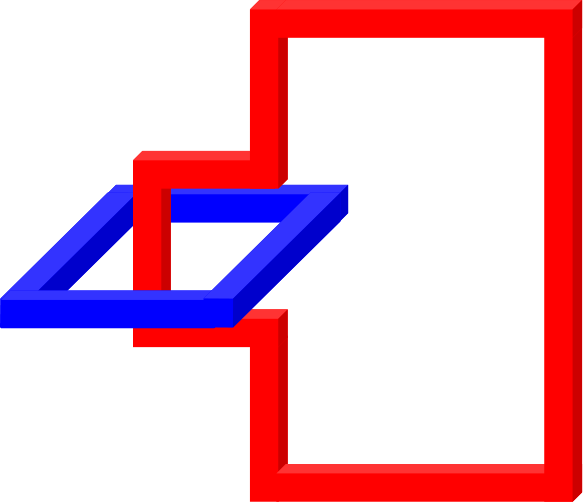}
  \caption{}
\end{subfigure}
\begin{subfigure}[b]{.25\linewidth}
  \includegraphics[width=.7\linewidth]{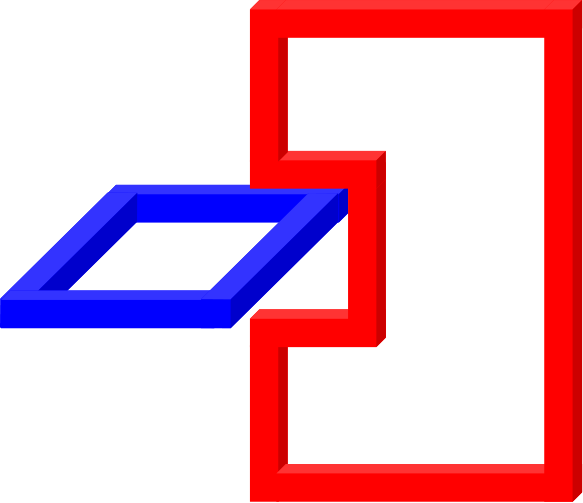}
  \caption{}
\end{subfigure}
\caption[A linking node.]{\label{fig:linking-node}
When a primal and a dual loop are linked in the canonical braid a linking node (a) is inserted and attached to both loops.  Once compaction has completed, the linking node is removed.  The linking number can be left unchanged (b), or toggled (c) if necessary.}
\end{figure}

Now consider a primal loop and a dual loop with even linking number.  Since these loops are unlinked, they may be far apart in spacetime.  Thus the linking node strategy is not practical.  However, we must still ensure that these loops remain unlinked in the output of the algorithm.  We will do this by requiring that the dual loop remain sufficiently far from the primal loop at all times.

Consider the primal loop.  It is composed of a set of connected primal cuboids, some of which are jog nodes and some of which are Hadamard or state distillation cuboids.  Let $x$ be the minimal $x$-coordinate of any cuboid in this set and let $x'$ be the maximal $x$-coordinate of any cuboid in the set.  Similarly define $y$, $y'$, $z$ and $z'$ as the minimal and maximal $y$- and $z$-coordinates.  Then the entire primal loop is contained in a bounding box of dimension $(x'-x, y'-y, z'-z)$.

If all of the cuboids in the dual loop stay outside of the bounding box that encloses the primal loop, then the linking number is guaranteed to be zero. We therefore introduce a new cuboid that encloses the primal loop. For all dual loops which have even linking number with the corresponding primal loop, we add primal-dual minimum distance constraints between the dual cuboids and the enclosing cuboid. See~\figref{fig:enclosing-cuboid}.

\begin{figure}
\centering
\includegraphics[scale=.3]{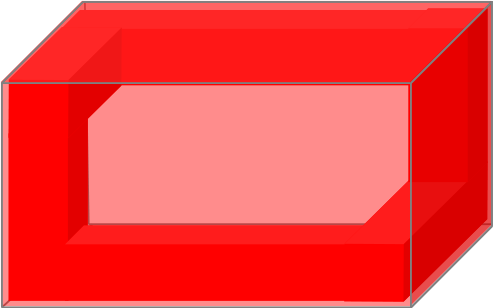}
\caption[An enclosing cuboid around a primal loop.]{\label{fig:enclosing-cuboid}
In order to avoid unwanted links, a cuboid is placed around each primal loop.  Dual loops which do not link with the primal loop are prohibited from entering the enclosing cuboid.
}
\end{figure}

In order to ensure that the new cuboid actually encloses the primal loop, additional variables and constraints are required.  Let $(x,y,z)$ and $(x',y',z')$ be variables describing the enclosing cuboid.  Then for each cuboid $(x_i, y_i, z_i)$, $(x'_i, y'_i, z'_i)$ in the primal loop we require that
\begin{equation}
\begin{tabular}{ll}
  $x \leq x_i$ & $x'_i \leq x'$ \\
  $y \leq y_i$ & $y'_i \leq y'$ \\
  $z \leq z_i$ & $z'_i \leq z'$.
\end{tabular}
\end{equation}

\subsection{The annealing algorithm}
\label{sec:annealing-algorithm}
The algorithm takes a canonical braid as input.  Initialization consists of constructing all of the cuboids and constraint groups. The instance includes a set of coordinates $P$, which can be divided into sets of integers $X$, $Y$ and $Z$ corresponding to the $x$-, $y-$, and $z$-coordinates, respectively.  The constraints can be represented as a set $C$ of integer triples. Some of the constraints, such at time-ordering constraints, must be satisfied for all possible layouts.  Other constraints may be partitioned into subsets for which the layout must satisfy at least one of the constraints in the subset.
Let $C'$ be the set of all constraints that must always be satisfied, let $C''$ be the remaining constraints and let $B$ be the corresponding partition into constraint subsets.  Let $A \subset C$ be the set of ``active'' constraints such that $C' \subset A$ and $A$ contains exactly one constraint from each element of $B$. 

A key element of the algorithm is to calculate the ``cost'' of layout.
There are a number of choices of cost function.  The goal is to construct a braid of small height that fits into an $x$-$y$ area of fixed size.  The first step is to ensure that the braid fits into that area.  We initially set the cost function as the $x$-coordinate of the bounding box.  Once this $x$-coordinate is small enough, we impose a global constraint that the $x$-coordinates of all cuboids must be no greater than that of the bounding box. We then set the cost function as the $y$-coordinate of the bounding box and repeat the procedure.
Finally, once the entire braid fits into the $x$-$y$ area, we minimize over the height.

Start by choosing a set of active constraints such that all constraints in $A$ are satisfied by the canonical braid.
The algorithm then proceeds by repeating the following sequence.
\begin{enumerate}
  \item Randomly select an element $\beta \in B$.
  \item Randomly select a constraint $b \in \beta$ such that $b \not\in A$.
  \item Locate the single constraint in $b' \in A \cap \beta$.  Remove $b'$ from $A$ and replace it with the new constraint $b$.
  \item Compute the new minimum bounding box size and corresponding cost function.
  \item If the new set of active constraints is infeasible, then reject the swap by removing $b$ from $A$ and replacing with $b'$.
  \item If the cost is smaller than before, keep the new constraint.
  \item If the cost is larger than before, then keep the new constraint with probability given by the annealing schedule (see below).
\end{enumerate}

In order for the algorithm to be efficient, we require an efficient way to compute the size of the minimum bounding box.  This can be done using the constraint graph method proposed in~\cite{Liao1983} and used by~\cite{Hsieh}.  First, partition the active constraints into three sets: those that involve only $x$ coordinates, those that involve only $y$ coordinates and those that involve only $z$ coordinates.  Note that there are no constraints that involve coordinates for two different axes.  Consider just the set of $x$-coordinates $X$. We construct a weighted directed graph $G_X = (V_X, E_X)$.  Assign $V_X = X \cup \{x_\emptyset, x_\infty\}$ where $x_\emptyset$ and $x_\infty$ are a boundary coordinates. For each constraint $x_i \leq x_j + d_{ij}$ there is a directed edge from vertex $x_i$ to vertex $x_j$ with weight $d_{ij}$.  The value of each coordinate $x \in X$ is assigned by computing the longest path from $x_\emptyset$ to $x$.  Assuming that the set of constraints can be satisfied, $G_X$ is a acyclic.  Thus the longest path can be computed in linear time by negating the weights and using Dijkstra's algorithm.  Constraint graphs for $y$ and $z$ coordinates are similarly constructed.

The cost of constructing the initial constraint graphs is $O(n^2)$, where $n$ is the number of cuboids.  Once the graphs are constructed, updates can be computed by an online algorithm.  When a constraint swap is performed, only those paths affected by the corresponding vertices need to be recalculated. This algorithm can also detect cycles induced by the new constraint.  If a cycle is detected, then the set of constraints is infeasible and the swap is rejected.

For VLSI placement Hsieh, Leong and Liu use a fixed-ratio temperature schedule in which the temperature is reduced by a constant factor after each time step~\cite{Hsieh}.  This schedule is simple and efficient and can also be used for our algorithm.  Other kinds of schedules could also be used.

\section{Discussion and future work}
The surface code provides a unique opportunity for fault-tolerant quantum circuit optimization by topological deformation.  We have defined the problem of braid compaction subject to geometric constraints, and given two heuristic algorithms.  Our tool Braidpack implements the first of these---the force-directed algorithm--- and small examples indicate that compaction algorithms can lead to significant improvement in spacetime overhead when compared to the canonical braid.

Currently, Braidpack is a proof-of-principle rather than production-ready software tool.  Small-scale results are largely encouraging, but not all of the intended features have been implemented, and larger-scale examples are required to demonstrate the extent of its usefulness.
Implementation of the simulated annealing algorithm is desired in order to compare the performance of the two algorithms.  Indeed, we could also construct a hybrid algorithm which incorporates both techniques.

Our simulated annealing algorithm is inspired from a similar algorithm for VLSI placement.  VLSI also offers a number of other techniques including, genetic algorithms, numerical and partitioning algorithms, and force-directed algorithms that are distinct from our own~(see, e.g., \cite{Shahookar1991}).  
Perhaps some of these additional techniques could be adapted to braid compaction.

Due to similarity with VLSI compaction and other packing problems, we conjecture that braid compaction is NP-complete. A formal reduction has proven elusive, however.  Thus an obvious open problem is to confirm or refute that conjecture. 

Finally, we have focused on topological deformation.  However, other non-topological braid identities exist~\cite{Fowler2012f,Raussendorf2007a}.  Optimization involving these identities has been previously done by hand, but it may be possible to incorporate non-topological techniques into an automated tool such as ours.

\chapter{Concluding thoughts
\label{chap:conclusion}
}
The promise of a reliable large-scale quantum computer is in the exponential speedups that it offers for real-world applications in physics, cryptography and number theory.  Quantum computers do not yet exist in the real world, however.  It is the main objective of the fault-tolerant quantum circuit designer to reduce resource requirements to match the capabilities of current or near-term technology.
In this thesis we have tried to further this objective by optimizing a variety of aspects of fault-tolerance including: encoded gates, error correction, threshold calculations, unitary decomposition and global parallelization.

We can extract a number of themes from these optimizations.  One theme is the circumvention of optimality or no-go theorems by making novel use of the available machinery or by removing unnecessary constraints. \thmref{thm:transversal-hadamard} shows that the Eastin-Knill theorem against transversal universality can be side-stepped at essentially no cost.  Overlap-based stabilizer state preparation break the convention of treating stabilizer generators independently in exchange for reduced circuit size.  A tighter threshold can be obtained by eliminating the need for an adversarial noise model.  Repeat-until-success circuits achieve better-than-optimal scaling by incorporating quantum measurements.

The use of gate teleportation and ancillary qubits has been a theme in quantum fault-tolerance from the earliest protocols due to Shor~\cite{Shor1997}, and we have continued the trend here.  The utility of ancillas is particularly evident in the circuits presented in~\chapref{chap:repeat}. By using ancillas and measurement, suddenly a much wider class of unitary operations can be implemented without expanding the gate set beyond $\{\Clifford,T\}$.  Ancillas and teleportation are used heavily in state distillation and we saw two new distillation protocols, one in~\chapref{chap:transversal} and one in~\chapref{chap:ancilla}.

Another strong theme is the development and use of software tools to aid in circuit design and discovery.  Indeed, except for~\chapref{chap:transversal}, all of the new results presented in this thesis made use of custom computer software in some form or another.  Undoubtedly, software tools will continue to be an important part of fault-tolerance optimizations going forward.  One can imagine a kind of software ``toolchain'' for compiling and optimizing quantum algorithms, taking a high level description of an algorithm and progressively decomposing it into machine-level instructions.

The new results and ideas in this thesis introduce many new questions, and leave room for improvement in several areas.  Given their universal and transversal power, triorthogonal codes appear to have a special place in the theory of fault-tolerant quantum computation.  However, beyond numerical study of the $[[15,1,3]]$ code~\cite{Cross2009}, and codes developed by Bravyi and Haah~\cite{Bravyi2012a}, very little is known about these codes.  A worthy research pursuit is to search for new and better triorthogonal codes.

Similarly, despite the large database compiled in~\chapref{chap:repeat}, little is known about the power of repeat-until-success circuits, and non-deterministic circuits in general.  In particular, what are the cost lower bounds for unitary decomposition when ancilla qubits and non-determinism is allowed?  We considered only a small fraction of possible circuits and it is possible that other kinds of circuits could yield even better performance.

Perhaps the biggest opportunity for improvement and further research is global optimization algorithms such as those presented in~\chapref{chap:braidpack}. The Braidpack tool presented in this thesis represents only a proof-of-concept.  Substantial and quantitative results will require a larger-scale effort in the development of these kinds of tools.  This study of this area has only just begun, and there is much that can be learned from existing classical techniques such as VLSI.

Beyond the ideas considered in this thesis, the field of fault-tolerant quantum computation has much room for exploration.  One particularly appealing option is the use of codes with very high encoding rates.  Very recently, Gottesman has shown that fault-tolerant quantum computation with \emph{constant} overhead may be possible by using certain low-density parity-check (LDPC) quantum codes~\cite{Gottesman2013a}.  However, realization of his claims presume efficient classical decoding algorithms for these codes, algorithms which are are not currently known.

Another exciting, but speculative pursuit is the use of non-abelian anyons for topological quantum computation.  Because of their inherently robust properties, some have likened anyons to ``quantum transistors'' (thereby implying a comparison between quantum circuits and vacuum tubes).  The experimental viability of this method remains to be seen.

At the current time, the surface code seems to be the leader among realistic schemes for fault-tolerant quantum computation.  Its high threshold and $2$D nearest-neighbor properties make it a very appealing option for a variety of proposed quantum computing architectures.  Indeed it has been the subject of intense study in recent years.  We addressed global topological optimization for the surface code in this thesis, but others have also considered optimizations, particularly for state distillation~\cite{Fowler2012f,Fowler2013,Jones2013b}.

The motivation for resource optimization is a strong one, and more improvements are necessary before requirements become low-enough for implementation of quantum algorithms.  To quote Gottesman~\cite{Gottesman2013a}, ``the main thing is not to give up''.  We can be pleased with the optimizations that we discover, but we should not be satisfied until fault-tolerant quantum computing is a reality.

% Add a title page before the appendices and a line in the Table of Contents
\chapter*{Appendices}
\addcontentsline{toc}{chapter}{Appendices}
%======================================================================
\appendix

\chapter{Proof of \texorpdfstring{\claimref{clm:P2-multiplicative}}{Claim~\ref{clm:P2-multiplicative}}
\label{app:counting.scaling-proof}
}

We now prove \claimref{clm:P2-multiplicative}, that the level-two malignant event upper bounds decrease with $\gamma$ according to the distance of the code.  The claim is restated here for convenience.

\begin{reclaim}
  For $0 \leq \epsilon \leq 1$, $\mathcal{P}^{(2)}_E (\epsilon \Gamma^{(1)}(\gamma)) \leq \epsilon^{t+1} \mathcal{P}^{(2)}_E (\Gamma^{(1)}(\gamma))$, where $t=\lfloor (d-1)/2 \rfloor$ and $d$ is the minimum distance of the (unconcatenated) code.
\end{reclaim}

%\begin{proof}
\proof{
From~\eqnref{eq:malig-event-bound} we see that $\mathcal{P}^{(2)}_E$ can be bounded as
  \begin{equation}
    \frac{\Pr[\malig_E,\good]}{\Pr[\accept]} + \Pr[\bad \vert \accept]
    \enspace .
  \end{equation}
  The $\Pr[\malig_E,\good]$ term is expressed as a sum of the form
  \begin{equation}
    \sum_{k=0}^{\kMax} c(k) \Gamma^k
    \label{eq:pr-malig-form}
  \end{equation}
  where all of the coefficients $c(k)$ are non-negative (because there are no non-deterministic components at level-two) and it is understood that $\Gamma$ is a function of $\gamma$.  The $\Pr[\accept]$ term in the denominator is a product of terms of the form
  \begin{equation}
    1 - \sum_{k=0}^{\kMax} c(k) \Gamma^k
    \label{eq:pr-accept-form}
  \end{equation}
  where, again, all $c(k)$ are non-negative.  $\Pr[\bad \vert \accept]$ is a sum of terms similar to \eqnref{eq:pr-malig-form}, some of which contain \eqnref{eq:pr-accept-form} terms in the denominator. 
  
  Strict fault-tolerance of the exRec implies that the coefficients $c(k)$ of \eqnref{eq:pr-malig-form} and the numerator coefficients of $\Pr[\bad \vert \accept]$ are zero for $k \leq t$. Therefore, for $0\leq \epsilon \leq 1$, $\mathcal{P}^{(2)}_E(\epsilon \Gamma)$ is a sum of non-negative terms of the form
  \begin{align}
    \frac{\sum_{k=0}^{\kMax} c(k) (\epsilon \Gamma)^k}{1 - \sum_{k=0}^{\kMax} c(k) (\epsilon \Gamma)^k}
    \leq \frac{\epsilon^t \sum_{k=4}^{\kMax} c(k) \Gamma^k}{1 - \sum_{k=0}^{\kMax} c(k) \Gamma^k}
  \end{align}
  which completes the proof.
}
\bibliographystyle{alpha-eprint}
% This specifies the location of the file containing the bibliographic information.  
% It assumes you're using BibTeX (if not, why not?).
\cleardoublepage % This is needed if the book class is used, to place the anchor in the correct page,
                 % because the bibliography will start on its own page.
                 % Use \clearpage instead if the document class uses the "oneside" argument
\phantomsection  % With hyperref package, enables hyperlinking from the table of contents to bibliography             
% The following statement causes the title "References" to be used for the bibliography section:
\renewcommand*{\bibname}{References}

% Add the References to the Table of Contents
\addcontentsline{toc}{chapter}{\textbf{References}}

\bibliography{library}
% Tip 5: You can create multiple .bib files to organize your references. 
% Just list them all in the \bibliogaphy command, separated by commas (no spaces).

% The following statement causes the specified references to be added to the bibliography% even if they were not 
% cited in the text. The asterisk is a wildcard that causes all entries in the bibliographic database to be included (optional).
%\nocite{*}

\end{document}